\documentclass[acmsmall]{acmart}

\usepackage[most]{tcolorbox}
\usepackage{listings}
\usepackage{caption}           
\usepackage{cleveref}
\usepackage{mathtools}
\usepackage{marvosym}
\usepackage{enumitem}
\usepackage{tabularx}
\usepackage{multirow}
\usepackage{adjustbox}
\usepackage{makecell}
\usepackage{bm}
\usepackage{dsfont}
\usepackage{placeins}
\usepackage{wrapfig}
\usepackage{csquotes}
\usepackage{tikz}
\usepackage[table]{xcolor}

\definecolor{ForestGreen}{cmyk}{0.91,0,0.88,0.12}

\newcommand{\caserule}{%
  \arrayrulecolor{gray!50}\cmidrule{2-3}\arrayrulecolor{black}}

\usepackage[disable]{todonotes}                        
\crefformat{section}{\S#2#1#3}
\Crefformat{section}{\S#2#1#3}

\crefformat{subsection}{\S#2#1#3}
\Crefformat{subsection}{\S#2#1#3}
\crefformat{subsubsection}{\S#2#1#3}
\Crefformat{subsubsection}{\S#2#1#3}
\AtEndPreamble{%
  \ifcsname remark\endcsname\else
    \theoremstyle{remark}
    \newtheorem{remark}[theorem]{Remark}
    \theoremstyle{acmplain}
  \fi
  \AtBeginEnvironment{lemma}{\crefalias{theorem}{lemma}}%
  \AtBeginEnvironment{definition}{\crefalias{theorem}{definition}}%
  \AtBeginEnvironment{remark}{\crefalias{theorem}{remark}}%
  \AtEndEnvironment{lemma}{\crefalias{theorem}{theorem}}%
  \AtEndEnvironment{definition}{\crefalias{theorem}{definition}}%
  \AtEndEnvironment{remark}{\crefalias{theorem}{theorem}}%
}

\usetikzlibrary{patterns}
\usetikzlibrary{fit,positioning,calc, arrows.meta, positioning, shapes.symbols}

\makeatletter
\def\@acmplainindent{0pt}
\def\@acmdefinitionindent{0pt}
\def\@proofindent{\noindent}
\makeatother

\tcbuselibrary{listingsutf8}

\newtcblisting{mybox}[2][]{%
  colback=#2!20,
  colframe=#2!80!black,
  coltitle=white,
  boxrule=1pt,
  arc=6pt,
  fonttitle=\bfseries\footnotesize,
  listing only,
  listing options={
    basicstyle=\ttfamily\footnotesize,
    breaklines=true,
    language=Dice,
    showstringspaces=false
  },
  width=\linewidth,
  valign=top,
  height=3.6cm,
  #1
}

\newtcblisting{programpanel}[2][]{%
  #1,
  listing only,
  title={#2},
  colback=black!2,
  colframe=black!30,
  colbacktitle=black!6,
  coltitle=black,
  fonttitle=\small\itshape,
  boxrule=0.4pt,
  titlerule=0.4pt,
  arc=2pt,
  left=4pt,
  right=4pt,
  top=5pt,
  bottom=5pt,
  valign=center,
  listing options={
    aboveskip=0pt,
    belowskip=0pt,
    escapechar=!,
    basicstyle=\ttfamily\footnotesize}
}
\usepackage[scaled=0.84]{beramono}
\usepackage{mathpartir}
\usepackage{xspace}
\usepackage{todonotes}
\usepackage{subcaption}

\lstdefinelanguage{Dice}{
    morekeywords={let, in, if, then, else, fst, snd, fun, observe, match, with, end, ref, fix},
    morekeywords=[2]{uniform, discrete, gaussian, exponential, beta, flip},
    morekeywords=[3]{bool, float, int},
    basicstyle=\ttfamily\small,
    keywordstyle=\bfseries,
    commentstyle=\itshape,
    stringstyle=\ttfamily,
    lineskip=-1pt,                
    aboveskip=0.1em,                
    belowskip=0.1em,
}

\lstdefinelanguage{OCaml}{
    morekeywords={let, in, if, then, else, fst, snd, fun, match, end, with, rec, observe},
    morekeywords=[2]{uniform, discrete, gaussian, exponential, beta, flip},
    basicstyle=\ttfamily\small,
    keywordstyle=\bfseries,                            
    keywordstyle=[2]\bfseries,                        
    keywordstyle=[3]\color{pink}\bfseries,           
    commentstyle=\color{green!50!black}\itshape,
    stringstyle=\ttfamily,
    morecomment=[s]{(*}{*)}, 
    lineskip=-1pt,                
    aboveskip=0.1em,                
    belowskip=0.1em,
}

\newcommand{\Nats}{\mathbb{N}}
\newcommand{\Reals}{\mathbb{R}}

\newcommand{\letkw}{\textnormal{\ttfamily\bfseries let}}
\newcommand{\inkw}{\textnormal{\ttfamily\bfseries in}}
\newcommand{\ifkw}{\textnormal{\ttfamily\bfseries if}}
\newcommand{\thenkw}{\textnormal{\ttfamily\bfseries then}}
\newcommand{\elsekw}{\textnormal{\ttfamily\bfseries else}}
\newcommand{\uniform}{\textnormal{\ttfamily\bfseries uniform}}
\newcommand{\discrete}{\textnormal{\ttfamily\bfseries discrete}}
\newcommand{\gaussian}{\textnormal{\ttfamily\bfseries gaussian}}
\newcommand{\exponential}{\textnormal{\ttfamily\bfseries exponential}}
\newcommand{\betafn}{\textnormal{\ttfamily\bfseries beta}} 
\newcommand{\fstkw}{\textnormal{\ttfamily\bfseries fst}}
\newcommand{\sndkw}{\textnormal{\ttfamily\bfseries snd}}
\newcommand{\funkw}{\textnormal{\ttfamily\bfseries fun}}
\newcommand{\observekw}{\textnormal{\ttfamily\bfseries observe}}
\newcommand{\matchkw}{\textnormal{\ttfamily\bfseries match}}
\newcommand{\withkw}{\textnormal{\ttfamily\bfseries with}}

\newcommand{\fixkw}{\textnormal{\ttfamily\bfseries fix}}
\newcommand{\matchcase}{\boldsymbol{\mid}} 

\newcommand{\diverge}{\textnormal{\ttfamily\bfseries diverge}}
\newcommand{\logand}{\mathrel{\&\&}}                    
\newcommand{\logor}{\mathrel{||}}                       
\newcommand{\lognot}{\text{not}}                        

\newcommand{\bool}{\textnormal{\ttfamily\bfseries bool}}
\newcommand{\intty}{\textnormal{\ttfamily\bfseries int}}
\newcommand{\float}{\textnormal{\ttfamily\bfseries float}}
\newcommand{\floattype}[2]{\mathnormal{\float}[#1; #2]}
\newcommand{\fin}[1]{\textnormal{\ttfamily\bfseries fin}(#1)}
\newcommand{\listty}[1]{\textnormal{\ttfamily\bfseries list}(#1)}
\newcommand{\unit}{\textnormal{\ttfamily\bfseries unit}}

\newcommand{\finconst}[2]{#1_{\##2}}                    
\newcommand{\finlt}[1]{<_{\#{#1}}}                      
\newcommand{\finleq}[1]{\leq_{\#{#1}}}                  
\newcommand{\fingt}[1]{>_{\#{#1}}}                      
\newcommand{\fingeq}[1]{\geq_{\#{#1}}}                  
\newcommand{\fineq}[1]{==_{\#{#1}}}                     

\definecolor{typecolor}{HTML}{4169E1} 
\newcommand{\codetype}[1]{\textcolor{typecolor}{\ttfamily\bfseries#1}}
\newcommand{\ftnv}[1]{\codetype{float[#1]}}  
\newcommand{\ftb}[2]{\codetype{float[#1; #2]}}  
\newcommand{\ft}[1]{\codetype{float[#1; T]}}        

\newcommand{\expr}{e}

\newcommand{\var}{x}

\newcommand{\true}{\textnormal{true}}
\newcommand{\false}{\textnormal{false}}

\newcommand{\latticeleq}{\sqsubseteq}

\newcommand{\latticeelem}{d}

\newcommand{\cutsets}{\mathcal{B}}
\newcommand{\clt}[1]{<\!\!{#1}}
\newcommand{\cleq}[1]{\leq\!\!{#1}}

\newcommand{\valsets}{\mathcal{V}}
\newcommand{\distinguishes}[2]{\mathsf{recovers}(#1, #2)}
\newcommand{\answerslt}[3]{\mathsf{answers}_<(#1, #2, #3)}
\newcommand{\answersleq}[3]{\mathsf{answers}_{\leq}(#1, #2, #3)}
\newcommand{\subtype}[2]{#1 <: #2}
\newcommand{\has}[2]{\mathsf{has}(#1, #2)}

\newcommand{\intervals}[1]{\text{Intervals}(#1)}

\newcommand{\sem}[1]{\llbracket #1  \rrbracket}
\newcommand{\bigsemstep}[2]{\langle #1  \rangle_{#2}}

\newcommand{\bigsem}[1]{\langle #1  \rangle}

\newcommand{\nonvals}{N}
\newcommand{\vals}{Vals}

\newcommand{\monbind}{\ensuremath{\ensuremath{\gg\!\!=}}}
\newcommand{\dirac}[1]{\delta_{#1}}

\newcommand{\Slice}{\textnormal{\scshape Slice}\xspace}
\newcommand{\Slicebf}{\textbf{\scshape Slice}}
\newcommand{\Storm}{\textnormal{\scshape Storm}\xspace}
\newcommand{\Stormbf}{\textbf{\scshape Storm}}
\newcommand{\Dice}{\textnormal{\scshape Dice}\xspace}
\newcommand{\Roulette}{\textnormal{\scshape Roulette}\xspace}
\newcommand{\SPPL}{\textnormal{\scshape Sppl}\xspace}
\newcommand{\PSI}{\textnormal{\scshape Psi}\xspace}
\newcommand{\Hakaru}{\textnormal{\scshape Hakaru}\xspace}
\newcommand{\Dicebf}{\textbf{\scshape Dice}}
\newcommand{\Roulettebf}{\textbf{\scshape Roulette}}

\newcommand{\R}{\mathbb{R}}

\newcommand{\discretize}[1]{\mathcal{D}\sem{#1}}

\newcommand{\dc}{\ensuremath{\texttt{DC}_{B,B_1,B_2,V_1,V_2}}}

\newcommand{\skel}{Skel}
\newcommand{\holes}{holes}
\newcommand{\decompose}{\mathsf{dec}}

\colorlet{green}{green!70!black}
\colorlet{gray}{gray!70!black}
\colorlet{blue}{blue!50!black}
\colorlet{red}{red!80!black}
\colorlet{purple}{purple!80!black}
\colorlet{orange}{orange}
\colorlet{teal}{teal!80!black}
\colorlet{olive}{olive!70!black}
\colorlet{brown}{brown!90!black}
\colorlet{violet}{violet!70!black}

\newcommand*\circled[1]{\tikz[baseline=(char.base)]{
            \node[shape=circle,draw,inner sep=1.5pt] (char) {#1};}}

\setcopyright{cc}
\setcctype{by}
\acmDOI{10.1145/3839534}
\acmYear{2026}
\acmJournal{PACMPL}
\acmVolume{10}
\acmNumber{OOPSLA2}
\acmArticle{402}
\acmMonth{10}
\acmSubmissionID{oopslab26main-p1463-p}
\received{2026-03-17}
\received[accepted]{2026-06-10}

\begin{document}

\title{Type-Directed Discretization of Probabilistic Programs (Extended Version)}

\author{Katherine Wu}
\orcid{0000-0002-4179-0411}
\affiliation{%
  \institution{Cornell University}
  \city{Ithaca}
  \country{USA}
}
\email{kaw324@cornell.edu}

\author{Jules Jacobs}
\orcid{0000-0003-1976-3182}
\affiliation{%
  \institution{ETH Zurich}
  \city{Zurich}
  \country{Switzerland}
}
\affiliation{%
  \institution{Jane Street}
  \city{New York}
  \country{USA}
}
\email{julesjacobs@gmail.com}

\author{Kevin Batz}
\orcid{0000-0001-8705-2564}
\affiliation{%
  \institution{University of Munster}
  \city{Munster}
  \country{Germany}
}
\email{kevin.batz@uni-muenster.de}

\author{Alexandra Silva}
\orcid{0000-0001-5014-9784}
\affiliation{%
  \institution{Cornell University}
  \city{Ithaca}
  \country{USA}
}
\email{alexandra.silva@cornell.edu}

\addtocontents{toc}{\protect\setcounter{tocdepth}{-10}}

\begin{abstract}
We study \emph{exact discretization} as a semantics-preserving transformation for recursive, higher-order probabilistic programs with continuous distributions. We target programs where continuous values are compared against finitely many constants, so exact inference reduces to a discrete problem. Our central technical contribution is a non-local, type-directed analysis that infers where continuous values can be partitioned into finitely many observationally relevant regions, then rewrites sampling and comparison behavior over those regions. We call this transformation \Slice{}. Because this construction is global and type-directed, correctness requires reasoning beyond the local syntax: we formalize the transformation and prove soundness for boolean queries using a coupling-style logical relations argument over operational semantics. As an application, transformed programs can be executed by discrete engines such as \Dice{}, \Roulette{}, and \Storm{}. Our empirical evaluation shows two complementary strengths of \Slice{} when paired with discrete backends: it enables \emph{exact inference} for challenging continuous programs that lie beyond the reach of previous exact systems, and, on benchmarks where direct comparison is possible, it is competitive with state-of-the-art exact inference systems for continuous programs.
\end{abstract}



\bibliographystyle{ACM-Reference-Format}

\maketitle

\section{Introduction}\label{sec:intro}
Probabilistic programming languages (PPLs) are widely used to model uncertainty~\cite{Moy2025Roulette,Holtzen2020Dice,DeRaedt2007ProbLog,Saad2021SPPL,Carpenter2017Stan,Salvatier2016PyMC3,Bingham2019Pyro,Dillon2017TFP,Tran2016Edward,Tolpin2016Anglican,Goodman2014WebPPL,Pfeffer2009Figaro,Minka2018InferNET,Ge2018Turing,CusumanoTowner2019Gen,Tehrani2020BeanMachine,Goodman2008Church,Narayanan2016Hakaru}. Programs with {\em continuous} distributions capture rich models but require reasoning about uncountable domains, whereas purely {\em discrete} programs often admit fast and precise inference but cannot express real-valued models.

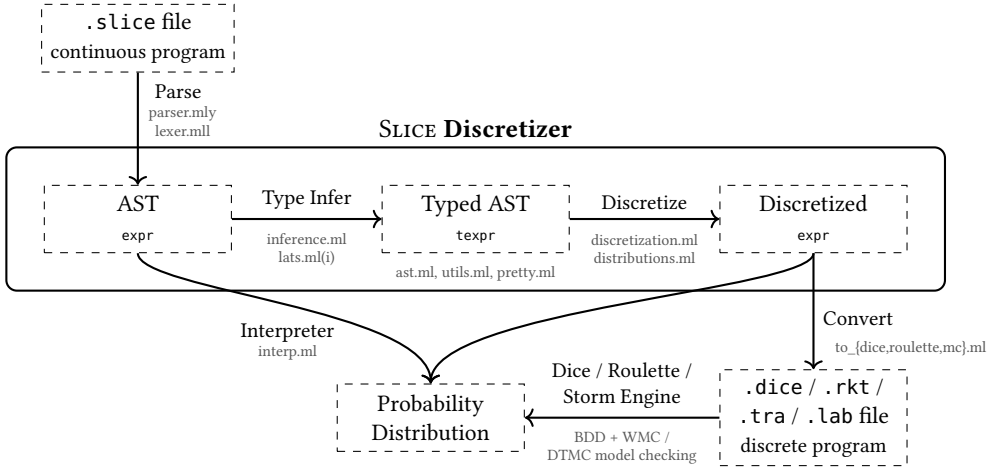
\begin{figure}[t]
    \centering
    \scalebox{.99}{
    \begin{tikzpicture}[
        node distance=1.2cm and 2cm,
        data/.style={rectangle, draw, dashed, minimum width=2.5cm, minimum height=0.9cm, align=center, font=\small},
        arrow/.style={->, thick},
        interp/.style={->, thick},
        operation/.style={font=\footnotesize, midway, above},
        operationside/.style={font=\footnotesize, midway, right},
        operationleft/.style={font=\footnotesize, midway, left},
        bigbox/.style={draw, thick, rounded corners, inner sep=0.3cm},
        impl/.style={font=\tiny, text=gray}
    ]
        \node[data] (ast) {AST\\{\tiny\texttt{expr}}};
        \node[data, right=of ast] (tast) {Typed AST\\{\tiny\texttt{texpr}}};
        \node[data, right=of tast] (dast) {Discretized\\{\tiny \texttt{expr}}};
        
        \node[data, above=1.5cm of ast] (input) {\texttt{.slice} file\\{\footnotesize continuous program}};
        
        \node[bigbox, fit=(ast)(tast)(dast), inner sep=0.5cm, label=above:{\textbf{\Slice{} Discretizer}}] (slicebox) {};
        
        \node[impl] at ($(slicebox.south) + (0.0,0.25)$) {ast.ml, utils.ml, pretty.ml};
        
        \node[data, below=1.55cm of dast] (dice) {\texttt{.dice} / \texttt{.rkt} / \\\texttt{.tra} / \texttt{.lab} file\\{\footnotesize discrete program}};
        
        \node[data, left=2.6cm of dice] (output) {Probability\\Distribution};
        
        \draw[arrow] (input) -- node[operationside, yshift=5mm, xshift=1mm] {Parse} (ast);
        \node[impl, align=center] at ($(input)!0.5!(ast) + (0.6,0.1)$) {parser.mly\\lexer.mll};
        
        \draw[arrow] (ast) -- node[operation] {Type Infer} (tast);
        \node[impl, align=center] at ($(ast)!0.5!(tast) + (0,-0.4)$) {inference.ml\\lats.ml(i)};
    
        \draw[arrow] (tast) -- node[operation] {Discretize} (dast);
        \node[impl, align=center] at ($(tast)!0.5!(dast) + (0,-0.4)$) {discretization.ml\\distributions.ml};
        
        \draw[arrow] (dast) -- node[operationside, yshift=-1mm] {Convert} (dice);
        \node[impl, align=center] at ($(dast)!0.5!(dice) + (1.3,-0.4)$) {to\_\{dice,roulette,mc\}.ml};
        
        \draw[arrow] (dice.west) -- node[operation, align=center] {Dice / Roulette / \\ Storm Engine} (output.east);
        \node[impl, align=center] at ($(dice)!0.5!(output) + (0,-0.4)$) {BDD + WMC / \\ DTMC model checking};
        
        \draw[interp] (ast.south) .. controls +(0,-0.5) and +(0,1) .. 
            node[operationleft, font=\footnotesize, align=center, yshift=-5mm, xshift=8mm] {Interpreter\\[-1mm]{\tiny \textcolor{gray}{interp.ml}}} (output.north);
        \draw[interp] (dast.south) .. controls +(0,-0.5) and +(0,1) .. (output.north);
    \end{tikzpicture}}
    \caption{Pipeline for exact discretization: typed \Slice{} programs are transformed into semantically equivalent discrete programs, then evaluated by discrete engines. The transformation stage (type inference + discretization) is backend-agnostic; \Dice{}, \Roulette{}, and \Storm{} are examples of target backends.}
    \label{fig:architecture}
\end{figure}

\Slice{} studies \emph{exact discretization}: a semantics-preserving source-to-source transformation that converts, if possible, continuous programs into discrete ones. We target higher-order recursive programs with conditioning and continuous sampling, without arithmetic. The key mechanism is \emph{type-directed cut inference}: our type system tracks each constant threshold (or {\em cut}) used in comparisons and uses this information to partition continuous domains into finitely many intervals. For example, $\uniform(0,1) < 0.5$ discretizes to $\discrete(0.5,0.5) < 1$, yielding a discrete program that existing exact inference tools can process.

Cuts may arise far from their source expressions: a sampled value may flow through higher-order functions, conditionals, recursive calls, data structures, and observations before it is eventually compared. To discretize soundly, \Slice{} propagates cut information along these non-local flows, collecting all observationally relevant cut points for each continuous value and using them to partition its domain into finitely many intervals (see~\Cref{sec:examples} for illustrative examples). A type system provides a natural framework for this analysis: it gives a compositional account of how cut information propagates through each language construct (via one typing rule per langauge feature); it separates the analysis result, i.e. a type derivation, from the type inference algorithm used to infer it; and it supports our correctness argument, since semantic preservation can be stated as a logical relation indexed by our types carrying cut and value sets.

Our discretization method is subject to a principled restriction: a continuous value must ultimately be compared against constants; direct comparisons between two unknown reals such as $\uniform(0,1)<\uniform(0,1)$ cannot be discretized exactly and remain continuous. Nonetheless, benchmarks from the literature and our novel case studies show that many programs \emph{can} be discretized by our system automatically.

Discretization itself is not new. Garg et al.~\cite{Garg2024BitBlast} encode real-valued variables via user-chosen bit-precisions, trading accuracy for tractability. Other approximate discretization schemes exist in the literature~\cite{Huang2021AQUA,Beutner2022GuaranteedBounds,Claret2013BayesianDataFlow}. In \Slice{}, \emph{exact discretization} means that the transformation is semantics-preserving: the discrete program preserves the relevant probabilities of the original continuous program. This differs from approximate discretization, which approximates the original semantics and may provide error bounds or convergence guarantees, but does not prove equality of source and transformed program probabilities in the sense of our main theorem.

\Cref{fig:architecture} shows the concrete pipeline from typed source programs to discrete backends. Because discretization is source-to-source, any discrete PPL that supports finite distributions can serve as a  backend; we instantiate this with \Dice{} \cite{Holtzen2020Dice}, \Roulette{} \cite{Moy2025Roulette}, and \Storm{} \cite{storm}.

\smallskip
\noindent\emph{Outline.} Starting from examples (\Cref{sec:examples}) and concluding with related and future work in \Cref{sec:related}, we:
\begin{description}[leftmargin=0.3cm,topsep=3pt,itemsep=1pt]
    \item[Type System (\Cref{sec:language})] Characterize when programs are discretizable.
       \item[Type Inference (\Cref{sec:type-inference})] Infer types automatically to guide discretization.
    \item[Exact Discretization (\Cref{sec:discretization})] Define a type-directed transformation from continuous programs to discrete ones using inferred type information.
    \item[Soundness (\Cref{sec:soundness})] Prove semantic preservation of the discretization transformation with a coupling-style relational argument
    \cite{Bizjak2015Step,Wand2018Contextual,Culpepper2017ContextualScoring}.
    \item[Implementation, Case Studies, and Evaluation (\Cref{sec:implem_eval})] Implement \Slice{} type inference and the discretization transformation in OCaml, and provide \Dice{}, \Roulette{}, and \Storm{} backends for running the discretized programs. We provide (i) case studies combining continuous sampling, higher-order functions, and unbounded recursion; to the best of our knowledge, there is no prior exact inference system that can cover programs \emph{combining} all three of these features and (ii) case studies where \Slice{}+\Dice{} and \Slice{}+\Roulette{} show favorable scaling and empirical performance on new and existing benchmarks~\cite{gehr2016psi,Saad2021SPPL}.
\end{description}

\section{Slice by Example}\label{sec:examples}
To build intuition for \Slice{}, we start with simple examples demonstrating the key ideas, and then provide a more involved example containing higher-order recursion and lists in \Cref{sec:examples:7}. Our pipeline performs \emph{fully automatic exact inference} on all provided programs.

\subsection{Branching on a Single Continuous Variable}\label{sec:examples:1}
Consider a program that samples a real value from a standard Gaussian distribution and then performs two nested comparisons: first checking whether the sample is less than 0.8, then comparing it to either 0.1 or 2.0:

\begin{lstlisting}[aboveskip=1em,belowskip=1em,gobble=4]
    let x = gaussian(0,1) in if x < 0.8 then x < 0.1 else x < 2.0
\end{lstlisting}

\noindent We discretize this program by replacing the real-valued sample with a discrete interval index, as shown below and in \Cref{fig:cuts-gaussian}:

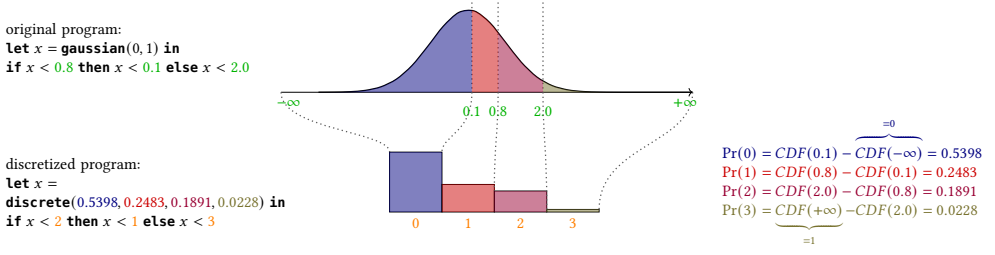
\begin{figure}[t]
  \centering
  \scalebox{0.6}{
  \begin{tikzpicture}[scale=1.1]

  \begin{scope}[xshift=-2.5cm, scale=0.75]

  \draw[->] (-5,0) -- (6,0);
  \node[text=green] at (-4.8, -0.3) {$-\infty$};
  \node[text=green] at (5.8, -0.3) {$+\infty$};

  \node[text=green] at (0.1, -0.5) {0.1};
  \node[text=green] at (0.8, -0.5) {0.8};
  \node[text=green] at (2.0, -0.5) {2.0};

  \draw[thick, domain=-4:6, smooth, variable=\x]
      plot ({\x},{2.2*exp(-(\x*\x)/2)});

  \draw[dotted, thick, gray] (0.1,2.4) -- (0.1,0);
  \draw[dotted, thick, gray] (0.8,2.4) -- (0.8,0);
  \draw[dotted, thick, gray] (2.0,2.4) -- (2.0,0);

  \fill[blue,   opacity=0.5] (-4,0) -- plot[domain=-4:0.1] (\x,{2.2*exp(-(\x*\x)/2)}) -- (0.1,0) -- cycle;
  \fill[red,    opacity=0.5] (0.1,0) -- plot[domain=0.1:0.8] (\x,{2.2*exp(-(\x*\x)/2)}) -- (0.8,0) -- cycle;
  \fill[purple, opacity=0.5] (0.8,0) -- plot[domain=0.8:2.0] (\x,{2.2*exp(-(\x*\x)/2)}) -- (2.0,0) -- cycle;
  \fill[olive, opacity=0.5] (2.0,0) -- plot[domain=2.0:6] (\x,{2.2*exp(-(\x*\x)/2)}) -- (6,0) -- cycle;

  \def\shift{0.7}
  \def\barW{1.4}
  \def\barY{-3.2}

  \def\hA{-1.6}        
  \def\hB{-2.464}      
  \def\hC{-2.639}      
  \def\hD{-3.132}      

  \def\xA{-2*\barW + \shift}
  \def\xB{-1*\barW + \shift}
  \def\xC{0 + \shift}
  \def\xD{\barW + \shift}

  \draw (\xA,\hA) rectangle ++(\barW,\barY-\hA);
  \fill[blue, opacity=0.5] (\xA,\hA) rectangle ++(\barW,\barY-\hA);

  \draw (\xB,\hB) rectangle ++(\barW,\barY-\hB);
  \fill[red, opacity=0.5] (\xB,\hB) rectangle ++(\barW,\barY-\hB);

  \draw (\xC,\hC) rectangle ++(\barW,\barY-\hC);
  \fill[purple, opacity=0.5] (\xC,\hC) rectangle ++(\barW,\barY-\hC);

  \draw (\xD,\hD) rectangle ++(\barW,\barY-\hD);
  \fill[olive, opacity=0.5] (\xD,\hD) rectangle ++(\barW,\barY-\hD);

  \node[text=orange] at (\xA+0.5*\barW,\barY-0.3) {0};
  \node[text=orange] at (\xB+0.5*\barW,\barY-0.3) {1};
  \node[text=orange] at (\xC+0.5*\barW,\barY-0.3) {2};
  \node[text=orange] at (\xD+0.5*\barW,\barY-0.3) {3};

  \draw[dotted, thick, gray]
      (-5,0) .. controls (-5,-1.0) and (\xA,-1.05) .. (\xA,\hA);

  \draw[dotted, thick, gray]
      (0.1,0) .. controls (0.1,-1.1) and (\xB,-1.15) .. (\xB,\hA);

  \draw[dotted, thick, gray]
      (0.8,0) .. controls (0.8,-1.2) and (\xC,-1.25) .. (\xC,\hC);

  \draw[dotted, thick, gray]
      (2.0,0) .. controls (2.0,-1.3) and (\xD,-1.35) .. (\xD,\hD);

  \draw[dotted, thick, gray]
      (6,0) .. controls (6,-1.4) and (\xD+\barW,-1.4) .. (\xD+\barW,\hD);

  \end{scope}

  \node[left, align=left, text width=6cm] at (-6.2,0.9) {original program:\\ $\letkw \; x = \gaussian(0,1) \; \inkw$ \\ $\ifkw \; x < \textcolor{green}{0.8} \; \thenkw \; x < \textcolor{green}{0.1} \; \elsekw \; x < \textcolor{green}{2.0}$};

  \node[left, align=left, text width=6cm] at (-6.2,-2.0) {discretized program:\\ $\letkw \; x =$ \\$\discrete(\textcolor{blue}{0.5398},\textcolor{red}{0.2483},\textcolor{purple}{0.1891},\textcolor{olive}{0.0228}) \; \inkw$ \\ $\ifkw \; x < \textcolor{orange}{2} \; \thenkw \; x < \textcolor{orange}{1} \; \elsekw \; x < \textcolor{orange}{3}$};

  \node[align=left] at (5.2,-1.8) {
    \textcolor{blue}{$\Pr(0)=CDF(0.1)-\overbrace{CDF(-\infty)}^{= 0}=0.5398$}\\
    \textcolor{red}{$\Pr(1)=CDF(0.8)-CDF(0.1)=0.2483$}\\
    \textcolor{purple}{$\Pr(2)=CDF(2.0)-CDF(0.8)=0.1891$}\\
    \textcolor{olive}{$\Pr(3)=\underbrace{CDF(+\infty)}_{= 1}-CDF(2.0)=0.0228$}
  };

  \end{tikzpicture}
  }
  \caption{Discretization of the original program into its semantically equivalent discretized program for~\Cref{sec:examples:1}. The shaded regions show how the cuts at 0.1, 0.8, and 2.0 partition the continuous distribution into four intervals $(-\infty, 0.1)$, $[0.1, 0.8)$, $[0.8, 2.0)$, and $[2.0, +\infty)$, whose masses become the probabilities of outcomes 0, 1, 2, and 3, i.e., interval indices, in the discrete distribution.}
  \label{fig:cuts-gaussian}
\end{figure}

\begin{lstlisting}[aboveskip=1em,belowskip=1em,gobble=4]
    let x = discrete(0.5398, 0.2483, 0.1891, 0.0228) in if x < 2 then x < 1 else x < 3
\end{lstlisting}

\noindent Here, \discrete(0.5398, 0.2483, 0.1891, 0.0228) already represents all four observationally distinguishable outcomes of the Gaussian sample. The precise real value of $x$ never matters to the rest of the program; only the interval in which it falls does. This is the regime we target: continuous quantities observed only through finitely many comparison outcomes. \emph{Slice’s goal is to detect such situations automatically.} We partition the continuous distribution into intervals that are indistinguishable to the rest of the program, then map every real to the interval it inhabits.

Our type system guides the transformation, in particular through a \emph{cut set} --- the set of comparison thresholds extracted from the program. A continuous distribution is annotated with its cut set, which records the values it is compared against. For our example, the type annotation is:

\begin{lstlisting}[aboveskip=1em,belowskip=1em,escapechar=!,gobble=4]
    let x = (gaussian(0,1) : !{\ftnv{\{<0.1,<0.8,<2\}}}!) in if x < 0.8 then x < 0.1 else x < 2.0
\end{lstlisting}

\noindent Intuitively, the cuts $\{\clt{0.1}, \clt{0.8}, \clt{2.0}\}$ partition the real line $\R$ into four intervals $(-\infty, 0.1)$, $[0.1, 0.8)$, $[0.8, 2.0)$, $[2.0, +\infty)$, and only the interval containing the Gaussian sample $x$ affects the rest of the computation. \Slice{} uses these cuts to replace the continuous distribution with a discrete one whose outcomes correspond exactly to these intervals, and each floating-point constant is mapped to the index of the interval that contains it.

\subsection{Branching on Multiple Continuous Variables}\label{sec:examples:2}

\noindent Different continuous quantities may induce different cut sets, and \Slice{} computes the appropriate set for each one independently. Consider an example with multiple continuous variables:

\begin{lstlisting}[aboveskip=1em,belowskip=1em,gobble=4]
    let x = uniform(0, 1) in
    let y = uniform(0, 2) in
    if x < 0.5 then x < 0.2 else y < 0.1
\end{lstlisting}

\noindent In this program, the variable $x$ is drawn from \uniform(0, 1) and compared against two constants: $0.5$ and $0.2$. These cuts $\clt{0.5}$ and $\clt{0.2}$ partition the real line $\R$ into three intervals: $(-\infty, 0.2)$, $[0.2, 0.5)$, and $[0.5, +\infty)$. Similarly, the variable $y$, drawn from \uniform(0, 2), is compared only with $0.1$. This single cut $\clt{0.1}$ partitions the real line $\R$ into $(-\infty, 0.1)$ and $[0.1, +\infty)$.

Type inference performs this analysis by collecting all cuts for each continuous variable. For $x$, it discovers the cuts $\clt{0.2}$ and $\clt{0.5}$, giving it the type \ftnv{\{<0.2,<0.5\}}. For $y$, it only finds $\clt{0.1}$, giving it the type \ftnv{\{<0.1\}}. This information is captured in the program below (left):

\begin{tcbraster}[
    raster columns=2,
    raster equal height=rows,
    raster before skip=1.5em,
    raster after skip=1.5em,
    raster column skip=0.8em]
    \begin{programpanel}{(a) Type-Annotated Program}
let x = uniform(0, 1) : !{\ftnv{\{<0.2,<0.5\}}}! in
let y = uniform(0, 2) : !{\ftnv{\{<0.1\}}}! in
if x < 0.5 then x < 0.2 else y < 0.1
    \end{programpanel}
    \begin{programpanel}{(b) Discretized Program}
let x = discrete(0.2, 0.3, 0.5) in
let y = discrete(0.05, 0.95) in
if x < 2 then x < 1 else y < 1
    \end{programpanel}
\end{tcbraster}
\label{fig:examples:2}

\Slice{} therefore discretizes $x$ into three outcomes \discrete(0.2, 0.3, 0.5) and $y$ into two outcomes \discrete(0.05, 0.95), where each probability in the discrete distribution is the measure of the corresponding original interval. Comparisons then operate over the integer indices of those intervals. The constants $0.5$ and $0.2$ map to indices 2 and 1, respectively, in the partition $(-\infty,0.2)\cup[0.2,0.5)\cup[0.5,+\infty)$ for $x$; the constant $0.1$ maps to index 1 in the partition $(-\infty,0.1)\cup[0.1,+\infty)$ for $y$.

\subsection{Control Flow}\label{sec:examples:4}

We now consider comparisons against a continuous quantity whose value is determined by control flow. Additionally, continuous values can flow through conditionals:

\begin{lstlisting}[aboveskip=1em,belowskip=1em,escapechar=!,gobble=4]
    let x = if flip() then uniform(0,2) else gaussian(0,1) in
    let y = if 1.5 > x then 1.8 else 0.3 in x <= y
\end{lstlisting}

\noindent Here, $y$ evaluates to either $1.8$ or $0.3$ depending on $x$. Because $x \leq y$ may use either value, \Slice{} collects both constants and treats the comparison as ranging over this finite set. This motivates enriching our float types with a \emph{value set}: a finite collection of concrete values that an expression may evaluate to. \Cref{fig:float-type-def} presents the \emph{full} float type implemented by \Slice{}; each floating-point expression is annotated with such a type. The type-annotated program is below (left):

\begin{tcbraster}[
    raster columns=5,
    raster equal height=rows,
    raster before skip=1.5em,
    raster after skip=1.5em,
    raster column skip=0.8em]
    \begin{programpanel}[raster multicolumn=3]{(a) Type-Annotated Program}
let x =
 (if flip() then
   (uniform(0,2) : !{\ft{\{<=0.3,<1.5,<=1.8\}}}!)
 else
   (gaussian(0,1) : !{\ft{\{<=0.3,<1.5,<=1.8\}}}!)
 : !{\ft{\{<=0.3,<1.5,<=1.8\}}}!) in
let y =
 (if (1.5 : !{\ftb{\{<=0.3,<1.5,<=1.8\}}{\{1.5\}}}!) > x
   then
     (1.8 : !{\ftb{\{<=0.3,<1.5,<=1.8\}}{\{1.8\}}}!)
   else
     (0.3 : !{\ftb{\{<=0.3,<1.5,<=1.8\}}{\{0.3\}}}!)
 : !{\ftb{\{<=0.3,<1.5,<=1.8\}}{\{0.3,1.8\}}}!) in
x <= y
    \end{programpanel}
    \begin{programpanel}[raster multicolumn=2]{(b) Discretized Program}
let x =
 if flip() then
   discrete(0.15, 0.6, 0.15, 0.1)
 else
   discrete(0.617911, 0.315281,
            0.030877, 0.035930)
in
let y =
 if 2 > x then 2 else 0
in
x <= y
    \end{programpanel}
\end{tcbraster}
\label{fig:examples:4}

\noindent Value sets refine an expression's cut set by adding cuts induced by comparisons. For example, in the test $x \leq y$, the value set of $y$ contributes boundary points that must also appear in the cut set of $x$. This dependency is non-local: constants introduced in one branch can constrain how a sampled value from another branch must be discretized.

Whenever two expressions are compared, the type system equates their cut sets so they are discretized over the same intervals. These cut sets then propagate backward through the control flow to all contributing subexpressions. \Cref{fig:control-flow-gaussian} illustrates this non-local propagation, which the soundness proof must relate through a semantic coupling relation.

The discretized program is shown on the right. The cuts $\cleq{0.3}$, $\clt{1.5}$, $\cleq{1.8}$ induce the intervals $(-\infty, 0.3]$, $(0.3, 1.5)$, $[1.5, 1.8]$, $(1.8, +\infty)$. Each \discrete{} statement \mbox{lists the mass of these intervals} rounded to six decimal places.

\begin{figure}
	\scalebox{0.9}{
    $
    \textbf{float}
    [
    \underbrace{\text{cut-set}}_{\footnotesize\begin{array}{c}\text{set of cuts}\text{ (e.g., }\{<\!0.1, <\!0.5, \leq\!0.8\}\text{)}\text{ or }\top\text{ (uncut)}\end{array}}
    ;
    \underbrace{\text{value-set}}_{\footnotesize\begin{array}{c}\text{set of concrete values }\text{(e.g., }\{0.1, 0.4, 0.5\}\text{)}\text{ or }\top\text{ (any value)}\end{array}}
    ]
    $
}
    \caption{The type \lstinline{float[cut-set; value-set]} captures two aspects of a floating-point expression. The cut set is determined either by (i) the constants the expression is compared against or (ii) the constant itself if the expression is a float constant. The value set records the concrete values an expression can take. The discretization is driven by the cut set: the $n$ cuts define $n+1$ intervals over the real line.}
    \label{fig:float-type-def}
\end{figure}

\subsection{Discrete Latent Variable}\label{sec:examples:5}

Continuous distributions can have parameters that depend on discrete random choices. In this example, a Gaussian distribution's standard deviation is determined by such a random choice:

\begin{lstlisting}[escapeinside={<@}{@>},aboveskip=1em,belowskip=1em,gobble=4]
    let x = if flip() then 0.5 else 1.5 in gaussian(0, x) < 0.5
\end{lstlisting}

\noindent 
The type annotations show how the variable $x$ determines the Gaussian's standard deviation, and \Cref{fig:control-flow-distinguishes} illustrates how type information is propagated.

\begin{tcbraster}[
    raster columns=5,
    raster equal height=rows,
    raster before skip=1.5em,
    raster after skip=1.5em,
    raster column skip=0.8em]
    \begin{programpanel}[raster multicolumn=3]{(a) Type-Annotated Program}
let x =
  (if flip() then
      (0.5 : !{\ftb{\{<=0.5,<=1.5\}}{\{0.5\}}}!)
    else
      (1.5 : !{\ftb{\{<=0.5,<=1.5\}}{\{1.5\}}}!)
  : !{\ftb{\{<=0.5,<=1.5\}}{\{0.5,1.5\}}}!)
in
((gaussian(0, x) : !{\ft{\{<0.5\}}}!)
  < (0.5 : !{\ftb{\{<0.5\}}{\{0.5\}}}!))
    \end{programpanel}
    \begin{programpanel}[raster multicolumn=2]{(b) Discretized Program}
let x =
  if flip() then 0 else 1
in
(if x == 0 then
    discrete(0.84, 0.16)
  else
    discrete(0.63, 0.37)) < 1
    \end{programpanel}
\end{tcbraster}
\label{fig:examples:5}

\noindent Because $x$ is used as a parameter to the $\gaussian$ distribution, \Slice{} generates a separate case for each possible value that $x$ can take (namely 0.5 and 1.5). To support this, the type system uses non-strict inequalities in cut sets, allowing these values to map to unique integer indices starting at 0. Specifically, 0.5 maps to index 0 via \codetype{\ftb{\{<=0.5,<=1.5\}}{\{0.5\}}}, and 1.5 maps to index 1 via \codetype{\ftb{\{<=0.5,<=1.5\}}{\{1.5\}}}. The resulting discretized program then branches on these values and instantiates the corresponding discrete distribution, depending on which standard deviation was selected by the initial flip, as shown in the program on the right.

\subsection{Conditioning}\label{sec:examples:6}

\Slice{} supports conditioning via $\observekw$. Consider the following: 

%

\begin{tcbraster}[
    raster columns=5,
    raster equal height=rows,
    raster before skip=1.5em,
    raster after skip=1.5em,
    raster column skip=0.8em]
    \begin{programpanel}[raster multicolumn=3]{(a) Type-Annotated Program}
let x =
 (uniform(0.0, 1.0) : !{\ft{\{<0.2,<0.5\}}}!) in
let _ =
 observe(x < (0.5 : !{\ftb{\{<0.2,<0.5\}}{\{0.5\}}}!)) in
x < (0.2 : !{\ftb{\{<0.2,<0.5\}}{\{0.2\}}}!)
    \end{programpanel}
    \begin{programpanel}[raster multicolumn=2]{(b) Discretized Program}
let x = discrete(0.2, 0.3, 0.5) in
let _ = observe (x < 2) in x < 1
    \end{programpanel}
\end{tcbraster}
\label{fig:examples:6}

\noindent
In the discretized program, the observation $x<2$ eliminates the discrete state corresponding to the original interval $[0.5,+\infty)$, thereby preserving the effect of conditioning on $x<0.5$.

Exact equalities are encoded with $\observekw\; (x \leq c \mathrel{\&\&} x \geq c)$. The two cuts $\clt{c}$ and $\cleq{c}$ induce the partition $(-\infty,c), \{c\}, (c,\infty)$, so \Slice{} can track probability mass assigned exactly to $c$ as its own discrete outcome. This supports point observations when the observed value has strictly positive probability mass; \Cref{ap:gpa} gives such an example.

Measure-zero observations remain delicate~\cite{Jacobs2021ParadoxesProbProg}. What matters for $\observekw(e)$ is whether $e$ has positive probability, not whether it is written as an equality. When an observation has probability zero, \Slice{} records observation-failure mass instead of renormalizing (cf.\ \Cref{sec:soundness}); if all observation-satisfying executions have probability zero, the result is the empty subdistribution. Scoring semantics for true measure-zero observations are out of scope, so users should approximate them with narrow intervals when that is the intended model. \Cref{ap:measure-zero-observations} gives details and examples.

\subsection{Unbounded Recursion, Higher-Order Functions, and Lists}\label{sec:examples:7}
\newcommand{\linelabel}[1]{%
  \edef\@currentlabel{\thelstnumber}%
  \label{#1}%
}

\begin{figure}[t]
\centering
\begin{minipage}[t]{0.44\textwidth}
\vspace{0pt}
\begin{lstlisting}[aboveskip=0.4em,belowskip=0.4em,
  language=OCaml,
  escapechar=!,
  numbers=right,
  stepnumber=1,
  numberstyle=\color{gray}\scriptsize,
  numbersep=3pt,
  basicstyle=\ttfamily\footnotesize,
  commentstyle=\color{ForestGreen}\itshape]
(* Returns true if y is in [lo,hi) *)
let within = fun lo -> fun hi -> fun y -> !\linelabel{ln:within}!
  (y >= lo) && (y < hi) in

(* Function for sampling interval bounds *)
let b = fun _ ->  !\linelabel{ln:b}!
  discrete(0.18: 0.05, 0.11: 0.24,
          0.21: 0.37, 0.09: 0.58,
          0.41: 0.89) in

(* Initialize intervals by sampling their
respective (possibly overlapping) bounds *)
let init =!\linelabel{ln:init}!
  (false, within 0.00 (b ())) ::
  (false, within (b ()) (b ())) ::
  (false, within (b ()) (b ())) ::
  (false, within (b ()) 1.00) ::
  nil in

(* Update flags depending on y falling
   within the respective interval *)
let update = fix update y := fun ys -> !\linelabel{ln:update}!
  match ys with
    | nil -> nil
    | h :: t ->
        ((fst h || ((snd h) y)), snd h)
        :: (update y t)
  end in
\end{lstlisting}
\end{minipage}
\hfill
\hspace{-4pt}%
\begin{minipage}[t]{0.48\textwidth}
\vspace{0pt}
\begin{lstlisting}[aboveskip=0.4em,belowskip=0.4em,
  language=OCaml,
  escapechar=!,
  numbers=right,
  stepnumber=1,
  firstnumber=last,
  numberstyle=\color{gray}\scriptsize,
  numbersep=0pt,
  basicstyle=\ttfamily\footnotesize,
  commentstyle=\color{ForestGreen}\itshape]
(* Sample from a gaussian, then use
   the resulting value to update flags *)
let step = fun ys -> !\linelabel{ln:step}!
  update (gaussian(0,1)) ys in

(* Number of samples determined by
   unbounded recursion *)
let cover = fix cover st :=!\linelabel{ln:cover}!
  if uniform(0,1) < 0.1 then st
  else cover (step st) in

(* Have all intervals been covered at
   least once? *)
let is_covered = fix is_covered ys := !\linelabel{ln:iscovered}!
  match ys with
    | nil -> true
    | h :: t -> (fst h) && (is_covered t)
  end in

(* Run entire process on sampled intervals *)
let final = cover init in
  is_covered final !\linelabel{ln:final}!
\end{lstlisting}
\end{minipage}%
\hspace{4pt}

\caption{An infinite stochastic process for random interval covering over $[0,1]$.}
\label{fig:epsball}
\end{figure}
\noindent
We now turn to the more involved example in \Cref{fig:epsball}, which \Slice{} discretizes automatically. This example emphasizes how our discretization procedure operates across higher-order functions, unbounded recursion, and lists.

The program models the following process: We start by randomly sampling $4$ intervals (l.\ \ref{ln:b} and \ref{ln:init}), stored as a list of pairs, where the first component is a flag and the second component is a function that determines whether a given point falls within that interval. We then employ unbounded recursion (l.\ \ref{ln:cover}) to repeat the following a potentially unbounded number of times according to a geometric distribution: sample a point $y$ from a gaussian (l.\ \ref{ln:step}), and update the flag of every interval to $\true$ whenever $y$ falls within that interval (l.\ \ref{ln:update}). The program returns $\true$ iff \emph{all} intervals have been hit at least once (l. \ref{ln:final} and \ref{ln:iscovered}). Exact inference on this program therefore determines the probability that this infinite stochastic process hits every interval at least once.

\Slice{} \emph{automatically} produces a \emph{provably equivalent} (cf.\ \Cref{sec:soundness}) discrete program (shown in \Cref{ap:epsball}). Even though this program still contains unbounded recursion, our back-end automatically determines that its stochastic execution behavior can be modeled as a \emph{finite-state} Markov chain. Our key insight is that exact inference can then be performed via \emph{probabilistic model checking}: we use \Storm{} \cite{storm}, which automatically computes that the sought-after probability is $0.034474...$. See \Cref{sec:case-studies} for details. To the best of our knowledge, our pipeline is the first capable of performing exact inference on suitable programs containing continuous sampling, higher-order functions, and unbounded recursion. In \Cref{sec:casestudy:epsball}, we extend this program by conditioning \emph{inside the recursion}.

%
%
%
%


\section{The Slice Language and Type System}\label{sec:language}
We now describe the language (\Cref{sec:sub_language}), its types (\Cref{sec:types}), and the typing rules for discretization (\Cref{sec:type-system}). Henceforth, we use the terms \emph{expression}, \emph{subexpression}, and \emph{program} somewhat interchangeably when the intended meaning is clear from context.

\subsection{The Slice Language}
\label{sec:sub_language}
\begin{figure}[!t]
	\scalebox{0.8}{
	\begin{minipage}{1\textwidth}
	\begin{align*}
		e ::= &\; x \mid c \mid \text{true} \mid \text{false} \mid \finconst{k}{n} \mid () & \text{constants} \\
		| &\; \letkw \; x = e_1\; \inkw \; e_2  & \text{let binding} \\
		| &\; \uniform(e_1, e_2) \mid \gaussian(e_1, e_2) \mid \exponential(e) \mid \ldots & \text{continuous distributions} \\
		| &\; \discrete(p_0, \ldots, p_{n-1})      & \text{discrete distribution} \\
		| &\; e_1 < e_2 \mid e_1 \leq e_2 \mid e_1 \finlt{n} e_2 \mid e_1 \finleq{n} e_2 & \text{comparisons} \\
		| &\; e_1 \logand e_2 \mid e_1 \logor e_2 \mid \text{not}\; e & \text{Boolean operations} \\
		| &\; \ifkw \; e_1\; \thenkw \; e_2\; \elsekw \; e_3 & \text{conditional} \\
		| &\; (e_1, e_2) \mid \fstkw \; e \mid \sndkw \; e & \text{pairs} \\
		| &\; \funkw \; x \; \rightarrow \; e \mid e_1 \; e_2 \mid \fixkw\; f\; x := e & \text{functions} \\
		| &\; \observekw\; e \mid e_1; e_2 & \text{effects} \\
		| &\; \text{nil} \mid e_1 :: e_2 \mid ( \matchkw \; e \; \withkw \matchcase \text{nil} \rightarrow e_{\text{nil}} \matchcase e_1 :: e_2 \rightarrow e_{\text{cons}}\; {\textnormal{\ttfamily\bfseries end}}) & \text{lists} \\
		| &\; \diverge\; & \text{diverge} \\
	\end{align*}
	\end{minipage}
}
\caption{Syntax of the \Slice{} language. \Slice{} supports a variety of continuous distributions with one and two parameters; for a full list of supported distributions, see~\Cref{appen:supported-distributions}.}
\label{fig:grammar}
\end{figure}

The syntax in \Cref{fig:grammar} extends a standard higher-order functional core with probabilistic choices. Beyond constants, variables, comparisons, Boolean operators, conditionals, pairs, and lists, \Slice{} offers first-class functions, recursion, and observations. The construct $\observekw(e)$ conditions on $e$ holding. Finite-domain constants $\finconst{k}{n}$ range over $\{\finconst{0}{n},\ldots,\finconst{(n-1)}{n}\}$ as discussed below. The expression $\discrete(p_0,\ldots,p_{n-1})$ requires a nonempty list of probabilities: $n \geq 1$, every $p_i \geq 0$, and $\sum_i p_i = 1$. It returns $\finconst{k}{n}$ with probability $p_k$. For random sampling, we include the usual discrete distributions and any continuous distribution with a tractable cumulative distribution function (one- or two-argument variants such as exponential, laplace, uniform, gaussian, beta, gamma). Adding further distributions is straightforward as long as their CDFs are known; \Cref{appen:supported-distributions} lists the ones implemented. For brevity, we include only $<$ and $\leq$ (and resp. $\finlt{n}$ and $\finleq{n}$) in the language, but we desugar $>$ and $\geq$ (and their finite variants $\fingt{n}$ and $\fingeq{n}$) by swapping their operands, as well as $\fineq{n}$ by binding each operand once and testing both $\finleq{n}$ directions.


\subsection{Types}\label{sec:types}
We formalize the types that \Slice{} supports. 
A \emph{value} set $V$ is a finite set of reals or $\top$.
Elements of $\{ \clt{x} ~|~ x \in \R \}	\cup \{ \cleq{x} ~|~ x \in \R \}$ are called \emph{cuts}. Cut sets induce intervals partitioning the reals as follows:
 Given a non-empty cut set $B =\{ \sim_1\!\!x_1, \ldots, \sim_n\!\! x_n \}$ with $x_1 \leq \ldots \leq x_n$, ordering $\clt{c}$ before $\cleq{c}$ when both occur, we define the set $\intervals{B}$ of intervals it induces as expected, e.g., 
 $\intervals{\{\clt{0.1},\cleq{0.5}\}} = \{(-\infty,0.1), [0.1,0.5], (0.5,\infty)\}$.
The strict and non-strict cuts differ only in which side contains the boundary point: $\intervals{\{\clt{c}\}} = \{(-\infty,c), [c,\infty)\}$, whereas $\intervals{\{\cleq{c}\}} = \{(-\infty,c], (c,\infty)\}$. Thus a boundary generated by $e<c$ places $c$ on the right, while a boundary generated by $e\leq c$ places $c$ on the left.
%
%
Moreover, we let $\intervals{\emptyset} = \{(-\infty,\infty)\}$ and $\intervals{\top} = 2^\Reals$. We denote by $\intervals{B}_i$ the $i$-th interval of $\intervals{B}$, assuming $\intervals{B}$ is $0$-indexed, i.e., $i \in \{0,...,n\}$. We denote the set of all cut- and value sets, respectively, by $\cutsets$ and $\valsets$. 
\begin{definition}[Types]
	For $n \in \Nats$, $B \in \cutsets$, and $V \in \valsets$, types $\tau$ adhere to the grammar
	\begin{align*}
		\tau {}::={} &\phantom{|} \textbf{int} 
		\mid \textbf{bool}
		\mid \unit
		 \mid\fin{n} 
		\mid \float[B; V]
		 \mid \tau_1 * \tau_2 
		 \mid \tau_1 \rightarrow \tau_2 
		 \mid\listty{\tau}~.
	\end{align*}
\end{definition}
Standard types cover integers, booleans, ordered pairs, functions, and lists. 
Finite types $\fin{n}$ contain $\{\finconst{0}{n},\ldots,\finconst{(n-1)}{n}\}$; the ${\#n}$-subscript reminds inference backends such as \Dice{} that these values live in a finite domain. The float type $\float[B; V]$ annotates a real-valued expression with its cut set $B$ (an over-approximation of the comparisons it participates in) and value set $V$ (an over-approximation of the values it can take). $B=\top$ means the expression may be compared to arbitrary reals; $B=\emptyset$ means it will never be compared; $V=\top$ indicates an unrestricted result.

\subsection{Typing Rules}\label{sec:type-system}
The typing rules are split into three parts: Rules specific to float types and continuous distributions (\Cref{fig:typing-float}), standard rules for all other language constructs (\Cref{fig:typing-general}), and subtyping (\Cref{fig:subtyping}).

To understand the float typing rules, it is helpful to understand how discretization works.
A type $\float$[B; V] is discretized as follows:
    If $B \neq\top$ is a finite set of cuts, then the type is discretized to a type $\fin{n+1}$ where $n+1 = |\intervals{B}|$ is the number of intervals induced by $B$. For instance, if $B = \{\clt{0.5}, \cleq{0.8}\}$, then the type is discretized to $\fin{3}$, where $\finconst{0}{3}$ represents the interval $(-\infty, 0.5)$, $\finconst{1}{3}$ represents the interval $[0.5, 0.8]$, and $\finconst{2}{3}$ represents the interval $(0.8, \infty)$.
    If $B = \top$, then no discretization is performed, and the type remains $\float$.

To understand the typing rules, keep in mind that the cut sets $B$ need to be sufficiently fine-grained --- i.e., contain enough cuts --- so that the discretized program has enough information to determine how the original continuous program behaved.

\begin{figure}
	\centering
	\scalebox{0.75}{
\begin{mathpar}
    \inferrule[\textsc{Float-Const}]
    {\ \has{V}{c}}
    {\Gamma \vdash c : \float[B; V]}

	    \inferrule[\textsc{Continuous-1}]
	    {\ \Gamma \vdash e : \float[B_1; V_1] \and \distinguishes{B_1}{V_1}
	    \and B_1 \equiv_{\top} B}
	    {\Gamma \vdash \exponential(e) : \float[B; \top]}

	    \makebox[\linewidth][c]{$\displaystyle
	      \inferrule[\textsc{Continuous-2}]
	      {\ \Gamma \vdash e_1 : \float[B_1; V_1] \and \distinguishes{B_1}{V_1}
	         \and B_1 \equiv_{\top} B \and
	         \Gamma \vdash e_2 : \float[B_2; V_2] \and \distinguishes{B_2}{V_2}
	         \and B_2 \equiv_{\top} B}
	      {\Gamma \vdash \uniform(e_1, e_2) : \float[B; \top]}
	    $}

    \inferrule[\textsc{Less}]
    {\Gamma \vdash e_1 : \float[B; V_1] \and \Gamma \vdash e_2 : \float[B; V_2] \and \answerslt{B}{V_1}{V_2}}
    {\Gamma \vdash e_1 < e_2 : \bool}

    \inferrule[\textsc{Less-Equal}]
    {\Gamma \vdash e_1 : \float[B; V_1] \and \Gamma \vdash e_2 : \float[B; V_2] \and \answersleq{B}{V_1}{V_2}}
    {\Gamma \vdash e_1 \leq e_2 : \bool}

\end{mathpar}
}
\caption{Typing rules for float types and continuous distributions ($\uniform$ is a representative for all other two-argument distributions). The notions \emph{recovers} and \emph{answers} are defined in~\Cref{def:distinguishes,def:answers-less,def:answers-less-equal}. Here $\Gamma$ is a typing environment of the form $\{\var_1 : \tau_1,\ldots,\var_n : \tau_n\}$,  where $ \Gamma, x: \tau_1 = (\Gamma \setminus \{ x : \tau ~|~ \tau~\text{type}\} ) \cup \{ x: \tau_1 \}$.}
\label{fig:typing-float}
\end{figure}

\paragraph{Constants.}
Rule \textsc{Float-Const} assigns a constant $c$ the type $\float[B; V]$ whenever $\has{V}{c}$ holds (i.e., $c\in V$ or $V=\top$). Constants do not constrain the cut set, but recording them in $V$ ensures later comparisons know which values must remain distinguishable.

\paragraph{Continuous Distributions.} Rule \textsc{Continuous-1} types a single-argument distribution when its argument $e$ has type $\float[B; V]$ and $\distinguishes{B}{V}$:
We write $B_1 \equiv_{\top} B_2$ iff $B_1=\top \Longleftrightarrow B_2=\top$.
\begin{definition}
	\label{def:distinguishes}
		Let $B\in \cutsets$ and $V\in\valsets$. We say that \emph{$\distinguishes{B}{V}$} if $B = \top$ or ($V \neq \top$ and every interval $I$ induced by $B$ contains at most one value in $V$, i.e., $
	\forall I \in \intervals{B}\colon | V \cap I|  \leq 1
	$).
\end{definition}
Intuitively, $\distinguishes{B}{V}$ means that the discretization induced by $B$ is lossless on $V$: from the resulting interval index, the original value in $V$ can be uniquely recovered. For example, $\exponential(e)$ can be discretized only when $e$ takes finitely many recoverable values. Thus $V$ must be finite and $B$ must separate those values so the discretized program can determine which one occurred. \Cref{sec:discretization} revisits this invariant. The rule for two-argument distributions is analogous. 

\paragraph{Comparisons.} 
For $e < c$ with $e : \float[B_1; \top]$ and $c : \float[B_2; \{c\}]$, the discretized $e$ must still decide the predicate, so $B_1$ must include $\clt{c}$. If $c$ can take any value in $V_2$, every $\clt{v}$ for $v\in V_2$ must appear. Symmetrically, for $c < e$ with $e : \float[B_2; V_2]$, $B_2$ must contain $\cleq{c}$ so the discretized intervals know on which side of $c$ they fall. For example, $2.3<e$ requires cuts separating $(-\infty,2.3]$ from $(2.3,\infty)$; finer partitions are allowed but not required.

In general, the typing rule for $e_1 < e_2$ requires that the expressions share a cut set $B$, and that their respective value sets $V_1, V_2$ are related in a way that allows the comparison to be determined after discretization, via the $\answerslt{B}{V_1}{V_2}$ relation:

\begin{definition}
	\label{def:answers-less}
	Let $B\in \cutsets$ and $V_1, V_2 \in\valsets$.
	We say that \emph{$\answerslt{B}{V_1}{V_2}$} if  (1) $B = \top$, or (2) $B \neq \top$ and $V_2 \neq \top$ and $\forall x \in V_2,\ \clt{x} \in B$, or
		(3) $B \neq \top$ and $V_1 \neq \top$ and $\forall x \in V_1,\ \cleq{x} \in B$.
\end{definition}


To understand this definition, we first want to introduce two different perspectives on the $B$ set. The first perspective is that $B$ specifies a set of intervals, such that values of the expression will be mapped to their respective interval index. The second perspective is that this discretized index \emph{still contains enough information to evaluate all comparisons in $B$ when viewed as predicates}. That is, if $B$ contains $\clt{x}$ then the discretized value of $e : \float[B; V]$ still contains enough information to answer $e < x$ precisely: if after discretization, $e$ has an interval index below the interval containing $x$, then the comparison is true, and otherwise it is false.

Now, if $e_1 : \float[B; V_1]$ and $e_2 : \float[B; V_2]$, then $\answerslt{B}{V_1}{V_2}$ ensures that the comparison result can still be determined after discretization because:

\begin{enumerate}
	\item If $B = \top$, then no discretization is performed, and the comparison can just be performed on the real-valued expression.
	\item If $B \neq \top$ and $V_2 \neq \top$, then the expression is discretized into intervals as determined by $B$, but in this case the right-hand side can only take on values in $V_2$. In this case, all $\clt{x}$ for $x \in V_2$ must be contained in $B$, which means that the value of $e_2$ can be exactly recovered from its interval index.
	Then, given the previous observation, we can determine the comparison result by comparing the interval index of $e_1$ to the interval index of $e_2$.
	\item If $B \neq \top$ and $V_1 \neq \top$, then the situation is similar to the previous case.
\end{enumerate}

If $V_1=V_2=\top$, then $\answerslt{B}{V_1}{V_2}$
(and analogously $\answersleq{B}{V_1}{V_2}$) requires $B=\top$, so neither comparand is discretized.
Additionally, we insist that $B_1 = B_2$, which ensures that both comparands will be discretized to the same type, or not discretized at all.


Finally, for non-strict comparisons, the rule is analogous to \Cref{def:answers-less}:
\begin{definition}
	\label{def:answers-less-equal}
	Let $B\in \cutsets$ and $V_1, V_2 \in\valsets$.
	We say that \emph{$\answersleq{B}{V_1}{V_2}$} if (1) $B = \top$, or (2)  $B \neq \top$ and $V_2 \neq \top$ and $\forall x \in V_2,\ \cleq{x} \in B$, or (3) 
	 $B \neq \top$ and $V_1 \neq \top$ and $\forall x \in V_1,\ \clt{x} \in B$.
\end{definition}

\begin{figure}
	\scalebox{0.75}{
    \begin{mathpar}
        \inferrule[\textsc{Var}]
        {\ }
        {\Gamma, x: \tau \vdash x : \tau}
    
        \inferrule[\textsc{Let}]
        {\Gamma \vdash e_1 : \tau_1 \\
         \Gamma, x: \tau_1 \vdash e_2 : \tau_2}
        {\Gamma \vdash \letkw \; x = e_1 \; \inkw \; e_2 : \tau_2}
    
        \inferrule[\textsc{If}]
        {\Gamma \vdash e_1 : \bool \\
         \Gamma \vdash e_2 : \tau \\
         \Gamma \vdash e_3 : \tau}
        {\Gamma \vdash \ifkw \; e_1 \; \thenkw \; e_2 \; \elsekw \; e_3 : \tau}

        \inferrule[\textsc{Observe}]
        {\Gamma \vdash e : \bool}
        {\Gamma \vdash \observekw\; e : \unit}
    
        \inferrule[\textsc{Discrete}]
        {n \geq 1 \\
            \forall i \in \{0,\ldots,n-1\}.\; p_i \geq 0 \\
            \sum_{i=0}^{n-1} p_i = 1}
        {\Gamma \vdash \discrete(p_0, \ldots, p_{n-1}) : \fin{n}}

        \inferrule[\textsc{Diverge}]
        {\ }
        {\Gamma \vdash \diverge : \tau}

        \inferrule[\textsc{FinConst}]
        {0 \leq k < n}
        {\Gamma \vdash \finconst{k}{n} : \fin{n}}
    
        \inferrule[\textsc{FinLess}]
        {\Gamma \vdash e_1 : \fin{n} \\
         \Gamma \vdash e_2 : \fin{n}}
        {\Gamma \vdash e_1 \finlt{n} e_2 : \bool}
    
        \inferrule[\textsc{FinLessEq}]
        {\Gamma \vdash e_1 : \fin{n} \\
         \Gamma \vdash e_2 : \fin{n}}
        {\Gamma \vdash e_1 \finleq{n} e_2 : \bool}
    \end{mathpar}
}
    \caption{Typing rules for general language constructs}
        \label{fig:typing-general}
    \end{figure}

\paragraph{Subtyping.}
The type system includes a subtyping relation that captures when one type can be replaced by another. The subtyping rules are shown in \Cref{fig:subtyping}. For float types $\float[B;V]$, subtyping follows the lattice ordering on $V$: an expression known to range over some set of values can also be considered to range over a superset of those values. For $B$, we require equality to make sure that the discretized representation of the expression is not changed by subtyping.

All other typing rules are standard and depicted in \Cref{fig:typing-general} and \Cref{appen:typing-rules}, respectively.

\begin{figure}
	\scalebox{0.75}{
\begin{mathpar}
    \inferrule[\textsc{Sub-Refl}]
    {\tau\in\{\intty,\bool,\fin{n},\unit\}}
    {\tau <: \tau}


    \inferrule[\textsc{Sub-Float}]
    {B_1 = B_2 \and V_1 \sqsubseteq V_2}
    {\float[B_1; V_1] <: \float[B_2; V_2]}

    \inferrule[\textsc{Sub-Pair}]
    {\tau_1 <: \tau_1' \and \tau_2 <: \tau_2'}
    {\tau_1 * \tau_2 <: \tau_1' * \tau_2'}

    \inferrule[\textsc{Sub-List}]
    {\tau <: \tau'}
    {\listty{\tau} <: \listty{\tau'}}

    \inferrule[\textsc{Sub-Arrow}]
    {\tau_1' <: \tau_1 \and \tau_2 <: \tau_2'}
    {\tau_1 \rightarrow \tau_2 <: \tau_1' \rightarrow \tau_2'}

    \inferrule[\textsc{Sub-Type}]
    {\tau_1 <: \tau_2 \and \Gamma \vdash e : \tau_1}
    {\Gamma \vdash e : \tau_2}
\end{mathpar}
}
\caption{Subtyping relation. For both cut- and value sets, we let  $			\latticeelem \latticeleq \latticeelem'
				~\text{iff}~
				\begin{cases}
						\latticeelem \subseteq \latticeelem' &\text{if $\latticeelem\neq \top$ and $\latticeelem'\neq \top$} \\
						\latticeelem' = \top & \text{otherwise}~.
					\end{cases}$}
\label{fig:subtyping}
\end{figure}

\section{Type Inference}\label{sec:type-inference}

We now describe our type inference algorithm, whose goal is to annotate each subexpression of a program with a sound type. Specifically, we focus on expressions of type $\float[B;V]$ and show how \Slice{} populates the cut set $B$ and value set $V$, which later determine how and whether the expression can be discretized. Type inference for all other types is standard and follows the classic Hindley-Milner algorithm. For float-typed expressions, we use a two-phase approach: the first phase traverses the program once to generate constraints over cut and value sets; the second phase solves those constraints to produce a fully annotated program.

The distinction between typing and type inference is important. The correctness argument cleanly separates into two pieces: (1) type inference produces a typing derivation matching our typing rules, and (2) given such a typing derivation, discretization is semantics preserving. Thus the proof of discretization does not have to reason directly about the type inference procedure.

\subsection{Constraint Generation}

We use the constraint language in \Cref{fig:constraints-grammar} to express requirements over the type structure and float refinements. We adopt the framework of~\cite{Odersky1999ConstrainedTypes}, introducing a new judgment of the form $\Gamma \vdash e : \tau \mid C$, which reads as \emph{in the environment $\Gamma$, $e$ has type $\tau$ and produces the set of constraints $C$}. In this judgment, $\Gamma$ and $e$ are inputs to the type inference algorithm, while $\tau$ and $C$ are outputs: the algorithm infers the type $\tau$ for expression $e$ and generates the constraints $C$ that must hold for this typing to be valid. The inference rules formalize which constraints must be generated and later solved for a given expression $e$ to have inferred type $\tau$. \Cref{fig:type-inference} gives the constraint generation rules for \Slice{}. We focus on a small but representative subset of constructs in \Slice{}, capturing the key ideas behind constraint generation for the \emph{full} language.

The type inference rules closely mirror the typing rules in \Cref{fig:typing-float}, with two key differences. First, the inference rules generate a constraint set instead of having those constraints appear in the premises. Second, subtyping is handled in a syntax-directed way: rather than using a general subtyping rule, we bake the subtyping relation into each rule.

\begin{figure}[!t]
    \scalebox{0.9}{$\begin{aligned}
    C &::= \; \tau_1 <: \tau_2  
          \mid \distinguishes{B}{V}  
          \mid \answerslt{B}{V_1}{V_2}  
          \mid \answersleq{B}{V_1}{V_2} 
          \mid \has{V}{c} \\
          &\quad \mid B_1 \sqsubseteq B_2
          \mid B_1 = B_2  
          \mid V_1 \sqsubseteq V_2 
          \mid V_1 = V_2  
          \mid B_1 \equiv_{\top} B_2   
          \mid C_1, C_2   
    \end{aligned}$}
\caption{Syntax for our constraint language. Constraints correspond to the semantic requirements in our type system: subtyping ensures type compatibility, $\distinguishes{B}{V}$ ensures that distribution parameters can be recovered after discretization, $\answerslt{B}{V_1}{V_2}$ and $\answersleq{B}{V_1}{V_2}$ ensure that comparison results can be determined after discretization, and $\has{V}{c}$ ensures that constants are properly tracked. $B_1 \equiv_{\top} B_2$ holds exactly when $B_1 = \top \iff B_2 = \top$.}
\label{fig:constraints-grammar}
\end{figure}

\begin{figure}
    \footnotesize

    \begin{mathpar}
        \inferrule[\textsc{TFloat-Const}]
        {\ }
        {\Gamma \vdash c : \float[B;V] \mid \has{V}{c}}

        \inferrule[\textsc{TIf}]
        {\Gamma \vdash e_1 : \bool \mid C_1 \\
         \Gamma \vdash e_2 : \tau_2 \mid C_2 \\
         \Gamma \vdash e_3 : \tau_3 \mid C_3}
        {\Gamma \vdash \ifkw \; e_1 \; \thenkw \; e_2 \; \elsekw \; e_3 : \tau_4 \mid C_1, C_2, C_3, \subtype{\tau_2}{\tau_4}, \subtype{\tau_3}{\tau_4}}

        \inferrule[\textsc{TVar}]
        {\ }
        {\Gamma, x: \tau \vdash x : \tau \mid \emptyset}
        
         \inferrule[\textsc{TLess}]
        {\Gamma \vdash e_1 : \float[B_1; V_1] \mid C_1 \and \Gamma \vdash e_2 : \float[B_2; V_2] \mid C_2}
        {\Gamma \vdash e_1 < e_2 : \bool \mid C_1, C_2, B_1 = B_2, \answerslt{B_1}{V_1}{V_2}}

        \inferrule[\textsc{TContinuous-1}]
        {\Gamma \vdash e : \float[B_1; V_1] \mid C}
        {\Gamma \vdash \exponential(e) : \float[B; V] \mid C, \distinguishes{B_1}{V_1}, B_1 \equiv_{\top} B, V = \top}

        \inferrule[\textsc{TContinuous-2}]
        {\Gamma \vdash e_1 : \float[B_1; V_1] \mid C_1 \\
         \Gamma \vdash e_2 : \float[B_2; V_2] \mid C_2}
        {\Gamma \vdash \uniform(e_1, e_2) : \float[B; V] \mid C_1, C_2, \distinguishes{B_1}{V_1}, B_1 \equiv_{\top} B, \distinguishes{B_2}{V_2}, B_2 \equiv_{\top} B, V = \top}

               \inferrule[\textsc{TLessEq}]
        {\Gamma \vdash e_1 : \float[B_1; V_1] \mid C_1 \and \Gamma \vdash e_2 : \float[B_2; V_2] \mid C_2}
        {\Gamma \vdash e_1 \leq e_2 : \bool \mid C_1, C_2, B_1 = B_2, \answersleq{B_1}{V_1}{V_2}}

        \inferrule[\textsc{TLet}]
        {\Gamma \vdash e_1 : \tau_1 \mid C_1 \\
         \Gamma, x: \tau_1 \vdash e_2 : \tau_2 \mid C_2}
        {\Gamma \vdash \letkw \; x = e_1 \; \inkw \; e_2 : \tau_2 \mid C_1, C_2}

    \end{mathpar}
\caption{Constraint generation rules for type inference for float-typed expressions.}
    \label{fig:type-inference}
\end{figure}

\begin{example}
  To make the constraint generation process concrete, \Cref{fig:constraint-generation} walks through~\Cref{sec:examples:4}. Each subexpression is annotated with type variables for the cut sets ($B_i$) and value sets ($V_i$), and the constraints produced at each step are shown explicitly.
\end{example}

\begin{figure}[t]
    \centering
    \small
    
    \lstdefinestyle{mystyle}{
        basicstyle=\ttfamily\small,
        columns=fullflexible,
        keepspaces=true,
        escapechar=!,
        mathescape=true
    }
    
    \scalebox{0.9}{
    \begin{tabular}{p{0.35\linewidth} p{0.7\linewidth}}

    \textbf{Program} & \textbf{Constraints Emitted} \\[0.4em]

    \lstinline[style=mystyle]|let x = if flip()| 
      & \\
    
    \lstinline[style=mystyle]| then uniform(0 : $\textcolor{blue}{\floattype{B_1}{V_1}}$,| 
      & $+\; \textcolor{blue}{\has{V_1}{0}}$ \\
    
    \lstinline[style=mystyle]|$\hspace{8em}$ 2 : $\textcolor{blue}{\floattype{B_2}{V_2}}$)| 
     & $+\; \textcolor{blue}{\has{V_2}{2}}$ \\

    \lstinline[style=mystyle]|$\hspace{3em}$ : $\textcolor{violet}{\floattype{B_3}{V_3}}$| 
     & $+\; \textcolor{violet}{\distinguishes{B_1}{V_1}, B_1 \equiv_{\top} B_3, \distinguishes{B_2}{V_2}, B_2 \equiv_{\top} B_3, V_3 = \top}$ \\
    
    \lstinline[style=mystyle]| else gaussian(0 : $\textcolor{blue}{\floattype{B_4}{V_4}}$, | 
      & $+\; \textcolor{blue}{\has{V_4}{0}}$ \\

    \lstinline[style=mystyle]|$\hspace{8.4em}$ 1 : $\textcolor{blue}{\floattype{B_5}{V_5}}$)| 
      & $+\; \textcolor{blue}{\has{V_5}{1}}$ \\

    \lstinline[style=mystyle]|$\hspace{3em}$ : $\textcolor{violet}{\floattype{B_6}{V_6}}$ | 
      & $+\; \textcolor{violet}{\distinguishes{B_4}{V_4}, B_4 \equiv_{\top} B_6, \distinguishes{B_5}{V_5}, B_5 \equiv_{\top} B_6, V_6 = \top}$  \\

    \lstinline[style=mystyle]| : $\textcolor{teal}{\floattype{B_7}{V_7}}$ | 
      & $+\; \textcolor{teal}{\subtype{\floattype{B_3}{V_3}}{\floattype{B_7}{V_7}}, \subtype{\floattype{B_6}{V_6}}{\floattype{B_7}{V_7}}}$ \\
    
    \lstinline[style=mystyle]|in| 
      & \\
    
    \lstinline[style=mystyle]|let y = if x : $\textcolor{olive}{\floattype{B_7}{V_7}}$ | 
      & $+\; \textcolor{olive}{\emptyset}$ \\

    \lstinline[style=mystyle]|$\hspace{6em}$ <| 
      & \\

    \lstinline[style=mystyle]|$\hspace{5.5em}$ 1.5 : $\textcolor{blue}{\floattype{B_8}{V_8}}$| 
      & $+\; \textcolor{blue}{\has{V_8}{1.5}}$ \\

    \lstinline[style=mystyle]|$\hspace{5.5em}$ : $\textcolor{brown}{\bool}$| 
      & $+\; \textcolor{brown}{\answerslt{B_7}{V_7}{V_8}, B_7 = B_8}$ \\
    
    \lstinline[style=mystyle]| then 1.8 : $\textcolor{blue}{\floattype{B_9}{V_9}}$| 
      & $+\; \textcolor{blue}{\has{V_9}{1.8}}$ \\
    
    \lstinline[style=mystyle]| else 0.3 : $\textcolor{blue}{\floattype{B_{10}}{V_{10}}}$| 
      & $+\; \textcolor{blue}{\has{V_{10}}{0.3}}$ \\

    \lstinline[style=mystyle]| : $\textcolor{teal}{\floattype{B_{11}}{V_{11}}}$| 
      & $+\; \textcolor{teal}{\subtype{\floattype{B_9}{V_9}}{\floattype{B_{11}}{V_{11}}}, \subtype{\floattype{B_{10}}{V_{10}}}{\floattype{B_{11}}{V_{11}}}}$ \\
    
    \lstinline[style=mystyle]|in| 
      & \\
    
    \lstinline[style=mystyle]| x : $\textcolor{olive}{\floattype{B_7}{V_7}}$| 
      & $+\; \textcolor{olive}{\emptyset}$  \\

    \lstinline[style=mystyle]| <=| 
      & \\

    \lstinline[style=mystyle]| y : $\textcolor{olive}{\floattype{B_{11}}{V_{11}}}$| 
      & $+\; \textcolor{olive}{\emptyset}$  \\

    \lstinline[style=mystyle]| : $\textcolor{brown}{\bool}$| 
      & $\textcolor{brown}{+\; \answersleq{B_7}{V_7}{V_{11}}, B_7 = B_{11}}$ \\
    
    \end{tabular}
}
    
    \caption{Program from~\Cref{sec:examples:4} (left) and its emitted constraints from each line (right) from \textcolor{blue}{\textsc{TFloat-Const}}, \textcolor{violet}{\textsc{TContinuous-2}}, \textcolor{teal}{\textsc{TIf}}, \textcolor{olive}{\textsc{TVar}}, and \textcolor{brown}{\textsc{TLess} and \textsc{TLessEq}} following the constraint generation rules in~\Cref{fig:type-inference}. Lines annotated with $\emptyset$ emit no constraints.}
    \label{fig:constraint-generation}
\end{figure}

\subsection{Constraint Solving}
 \noindent
Once constraint generation completes, we obtain a system of constraints over type variables, cut sets $B$, and value sets $V$. The constraint-solving phase proceeds in two steps:

\begin{figure}[t]
  \centering
  \small

    \begin{subfigure}[t]{\linewidth}
      \centering

      \lstdefinestyle{mystyle}{
          basicstyle=\ttfamily\small,
          columns=fullflexible,
          keepspaces=true,
          escapechar=!,
          mathescape=true
      }
	\scalebox{0.9}{
      \begin{tabular}{p{0.42\linewidth} p{0.6\linewidth}}

      \textbf{Subtyping Constraint}
        & \textbf{After Subtyping Reduction} \\[0.4em]

      \lstinline[style=mystyle]|$\textcolor{teal}{\floattype{B_3}{V_3} <: \floattype{B_7}{V_7}}$|
        & $\textcolor{teal}{B_3 = B_7},\; \textcolor{teal}{V_3 \sqsubseteq V_7}$ \\

      \lstinline[style=mystyle]|$\textcolor{teal}{\floattype{B_6}{V_6} <: \floattype{B_7}{V_7}}$|
        & $\textcolor{teal}{B_6 = B_7},\; \textcolor{teal}{V_6 \sqsubseteq V_7}$ \\

      \lstinline[style=mystyle]|$\textcolor{teal}{\floattype{B_9}{V_9} <: \floattype{B_{11}}{V_{11}}}$|
        & $\textcolor{teal}{B_9 = B_{11}},\; \textcolor{teal}{V_9 \sqsubseteq V_{11}}$ \\

      \lstinline[style=mystyle]|$\textcolor{teal}{\floattype{B_{10}}{V_{10}} <: \floattype{B_{11}}{V_{11}}}$|
        & $\textcolor{teal}{B_{10} = B_{11}},\; \textcolor{teal}{V_{10} \sqsubseteq V_{11}}$ \\

      \end{tabular}
    }

      \caption{\emph{Subtyping reduction}. Each subtyping constraint over float types is decomposed into an equality constraint on cut sets and a subset constraint on value sets.}
  \end{subfigure}

  \vspace{0.8em}

  \begin{subfigure}[t]{0.7\linewidth}
      \centering
      \scalebox{0.9}{
      \begin{tabular}{c p{0.28\linewidth} p{0.38\linewidth}}

      &
      \shortstack{\textbf{Constraints}\\\textbf{Processed}}
      &
      \shortstack{\textbf{Solver Updates}} \\[0.4em]

      & 
        & $\forall i \colon B_i = \emptyset, V_i = \emptyset$ \\

      \multirow{7}{*}{
        \textcolor{black}{
          $\left\{
          \begin{array}{c}
          \\ \\ \\ \\ \\ \\ \\
          \end{array}
          \right.$
        }
        \;\text{\circled{1}}
      }
      & $\textcolor{blue}{\has{V_1}{0}}$
        & $\textcolor{blue}{V_1 := V_1 \sqcup \{0\} = \{0\}}$ \\

      & $\textcolor{blue}{\has{V_2}{2}}$
        & $\textcolor{blue}{V_2 := V_2 \sqcup \{2\} = \{2\}}$ \\

      & $\textcolor{blue}{\has{V_4}{0}}$
        & $\textcolor{blue}{V_4 := V_4 \sqcup \{0\} = \{0\}}$ \\

      & $\textcolor{blue}{\has{V_5}{1}}$
        & $\textcolor{blue}{V_5 := V_5 \sqcup \{1\} = \{1\}}$ \\

      & $\textcolor{blue}{\has{V_8}{1.5}}$
        & $\textcolor{blue}{V_8 := V_8 \sqcup \{1.5\} = \{1.5\}}$ \\

      & $\textcolor{blue}{\has{V_9}{1.8}}$
        & $\textcolor{blue}{V_9 := V_9 \sqcup \{1.8\} = \{1.8\}}$ \\

      & $\textcolor{blue}{\has{V_{10}}{0.3}}$
        & $\textcolor{blue}{V_{10} := V_{10} \sqcup \{0.3\} = \{0.3\}}$ \\[0.6em]

      \multirow{2}{*}{
        \textcolor{black}{
          $\left\{
          \begin{array}{c}
          \\ \\
          \end{array}
          \right.$
        }
        \;\textsc{\circled{2}}
      }
      & $\textcolor{violet}{V_3 = \top}$
        & $\textcolor{violet}{V_3 := \top}$ \\

      & $\textcolor{violet}{V_6 = \top}$
        & $\textcolor{violet}{V_6 := \top}$ \\[0.6em]

      \multirow{4}{*}{
        \textcolor{black}{
          $\left\{
          \begin{array}{c}
          \\ \\ \\ \\ 
          \end{array}
          \right.$
        }
        \;\textsc{\circled{3}}
      }
      & $\textcolor{teal}{V_3 \sqsubseteq V_7}$
        & $\textcolor{teal}{V_7 := V_7 \sqcup \top = \top}$ \\

      & $\textcolor{teal}{V_6 \sqsubseteq V_7}$
        & $\textcolor{teal}{\emptyset}$ \\

      & $\textcolor{teal}{V_9 \sqsubseteq V_{11}}$
        & $\textcolor{teal}{V_{11} := V_{11} \sqcup \{1.8\} = \{1.8\}}$ \\

      & $\textcolor{teal}{V_{10} \sqsubseteq V_{11}}$
        & $\textcolor{teal}{V_{11} := V_{11} \sqcup \{0.3\}= \{0.3, 1.8\}}$ \\[0.6em]

      \multirow{4}{*}{
        \textcolor{black}{
          $\left\{
          \begin{array}{c}
          \\ \\ \\ \\
          \end{array}
          \right.$
        }
        \;\textsc{\circled{4}}
      }
      & $\textcolor{violet}{\distinguishes{B_1}{V_1}}$
        & $\textcolor{violet}{\emptyset}$ \\

      & $\textcolor{violet}{\distinguishes{B_2}{V_2}}$
        & $\textcolor{violet}{\emptyset}$ \\

      & $\textcolor{violet}{\distinguishes{B_4}{V_4}}$
        & $\textcolor{violet}{\emptyset}$ \\

      & $\textcolor{violet}{\distinguishes{B_5}{V_5}}$
        & $\textcolor{violet}{\emptyset}$ \\[0.6em]

      \multirow{2}{*}{
        \textcolor{black}{
          $\left\{
          \begin{array}{c}
          \\ \\ 
          \end{array}
          \right.$
        }
        \;\textsc{\circled{5}}
      }
      & $\textcolor{brown}{\answerslt{B_7}{V_7}{V_8}}$
        & $\textcolor{brown}{B_7 := B_7 \sqcup \{\clt{1.5}\} = \{\clt{1.5}\}}$ \\

      & $\textcolor{brown}{\answersleq{B_7}{V_7}{V_{11}}}$
        & $\textcolor{brown}{\begin{array}[t]{@{}r@{}l@{}}
            B_7 :=\; &B_7 \sqcup \{\cleq{0.3}, \cleq{1.8}\} \\
                =\; &\{\clt{1.5}, \cleq{0.3}, \cleq{1.8}\}
          \end{array}}$ \\[0.6em]

      \multirow{2}{*}{
        \textcolor{black}{
          $\left\{
          \begin{array}{c}
          \\ \\ \\
          \end{array}
          \right.$
        }
        \;\textsc{\circled{6}}
      }
      & $\textcolor{brown}{B_7 = B_8, B_7 = B_{11}}$
        & $\textcolor{brown}{B_7 = B_8, B_7 = B_{11}}$ \\

      & $\textcolor{teal}{B_3 = B_7, B_6 = B_7}$
        & $\textcolor{teal}{B_3 = B_7, B_6 = B_7}$ \\

	      & $\textcolor{teal}{B_9 = B_{10}, B_{10} = B_{11}}$
	        & $\textcolor{teal}{B_9 = B_{11}, B_{10} = B_{11}}$ \\[0.6em]

	      \multirow{2}{*}{
	        \textcolor{black}{
	          $\left\{
	          \begin{array}{c}
	          \\ \\
	          \end{array}
	          \right.$
	        }
	        \;\textsc{\circled{7}}
	      }
	      & $\textcolor{violet}{B_1 \equiv_{\top} B_3, B_2 \equiv_{\top} B_3}$
	        & $\textcolor{violet}{\emptyset}$ \\

	      & $\textcolor{violet}{B_4 \equiv_{\top} B_6, B_5 \equiv_{\top} B_6}$
	        & $\textcolor{violet}{\emptyset}$ \\

	      \end{tabular}
    }
      \caption{\emph{Lattice solving}. Each pass performs monotone updates to $B_i$ or $V_i$. Lines annotated with $\emptyset$ denote constraint is trivially satisfied.}
  \end{subfigure}
  \hfill
  \begin{subfigure}[t]{0.24\linewidth}
      \centering
      \begin{tabular}{p{0.95\linewidth}}

      \textbf{Cut Sets} \\[0.3em]
      $B_1 = \emptyset$ \\
      $B_2 = \emptyset$ \\
      $B_4 = \emptyset$ \\
      $B_5 = \emptyset$ \\
      $B_3 = B_6$ \\
      $ = B_7 = B_8$ \\
      $ = B_9 = B_{10}$ \\
      $ = B_{11} $ \\
      $ = \{\clt{1.5}, \cleq{0.3}, \cleq{1.8}\}$ \\[0.6em]

      \textbf{Value Sets} \\[0.3em]
      $V_1 = \{0\}$ \\
      $V_2 = \{2\}$ \\
      $V_3 = \top$ \\
      $V_4 = \{0\}$ \\
      $V_5 = \{1\}$ \\
      $V_6 = \top$ \\
      $V_7 = \top$ \\
      $V_8 = \{1.5\}$ \\
      $V_9 = \{1.8\}$ \\
      $V_{10} = \{0.3\}$ \\
      $V_{11} = \{0.3, 1.8\}$ \\

      \end{tabular}
      \caption{Final cut sets and value sets obtained.}
  \end{subfigure}

  \caption{Following~\Cref{fig:constraint-generation}, the constraint solving process for~\Cref{sec:examples:4}.}
  \label{fig:constraint-solving}
\end{figure}

\begin{enumerate}
    \item \textbf{Subtyping reduction}: We first process all subtyping constraints of the form $\subtype{\tau_1}{\tau_2}$. Since subtyping for float types follows the rule $\subtype{\float[B_1;V_1]}{\float[B_2;V_2]}$ if and only if $B_1 = B_2$ and $V_1 \sqsubseteq V_2$, we reduce each subtyping constraint to corresponding equality and subset constraints on the cut and value sets. Subtyping elimination continues until only constraints among $B$'s and $V$'s remain.

    \item \textbf{Lattice constraint solving}: After subtyping reduction, we are left with constraints purely over cut sets and value sets: equalities, subset relations, and the specialized constraints $\distinguishes{B}{V}$, $\answerslt{B}{V_1}{V_2}$, $\answersleq{B}{V_1}{V_2}$, and $\has{V}{c}$. These constraints are solved using a monotone fixed-point algorithm over the cut and value lattices.
\end{enumerate}

\paragraph{Subtyping Reduction.}
Subtyping reduction eliminates all remaining type-level subtyping constraints by recursively decomposing them according to type structure.
For base float types, a constraint
$\subtype{\float[B_1;V_1]}{\float[B_2;V_2]}$
is reduced to the conjunction $B_1 = B_2$ and $V_1 \sqsubseteq V_2$.
For compound types, subtyping is handled structurally before reaching floats.
For example, a constraint
$\subtype{\float[B_1;V_1] * \float[B_2;V_2]}
        {\float[B_5;V_5] * \float[B_6;V_6]}$
is decomposed into
$\subtype{\float[B_1;V_1]}{\float[B_5;V_5]}$
and
$\subtype{\float[B_2;V_2]}{\float[B_6;V_6]}$,
each of which is then reduced to equality and subset constraints on the corresponding cut
and value sets.
This process is syntactic and continues until no subtyping constraints remain, leaving only constraints over $B$ and $V$.

\paragraph{Lattice Solving.}
The lattice solver keeps track of the current contents of all $B$ and $V$ variables.
Initially, all $B$ and $V$ variables are set to $\emptyset$.
During the solving process, the variables may go up in the lattice to a finite set, or to $\top$.
Each of the constraints is handled as follows:

\begin{itemize}[leftmargin=*]
    \item \textbf{Value containment} ($\has{V}{c}$): We update $V$ to include the constant $c$, i.e., $V := V \sqcup \{c\}$. If $V$ is already $\top$, the constraint is trivially satisfied. \circled{1}

    \item \textbf{Top assignment} ($V = \top$): We immediately set $V$ to $\top$, representing that the value set contains all possible float values. \circled{2}

    \item \textbf{Subset constraints} ($B_1 \sqsubseteq B_2$ or $V_1 \sqsubseteq V_2$): We ensure that the contents of the first variable are always a subset of the second. When $B_1$ or $V_1$ is updated, we propagate the changes to $B_2$ or $V_2$ by taking the least upper bound. \circled{3}

    \item \textbf{Recoverability} ($\distinguishes{B}{V}$): This constraint ensures that the cut set $B$ can distinguish all values in $V$. If $B = \top$, the constraint is already satisfied. If $V = \top$ and $B \neq \top$, we set $B := \top$, since no finite cut set recovers an unrestriced real value. Otherwise $V$ is finite; we use a heuristic of inserting $\cleq{x}$ for each value $x \in V$. These cuts place distinct elements of $V$ in distinct intervals and therefore satisfy recoverability. This choice can add redundant cuts, but remains sound. \circled{4}

    \item \textbf{Less-than answerability} ($\answerslt{B}{V_1}{V_2}$): Intuitively, this constraint ensures that $B$ contains enough cuts to answer all possible comparisons of the form ``$v_1 < v_2$'' where $v_1 \in V_1$ and $v_2 \in V_2$.
    If $B = \top$, the constraint is trivially satisfied.
    Otherwise, we again have a heuristic choice: we add a cut $\clt{v}$ for each value $v \in V_2$ (if $V_2 \neq \top$) and $\cleq{v}$ for each value $v \in V_1$ (if $V_1 \neq \top$). If $V_1$ and $V_2$ are both $\top$, we set $B$ to $\top$. \circled{5}

    \item \textbf{Less-equal answerability} ($\answersleq{B}{V_1}{V_2}$): Analogous to less-than. \circled{5}

	    \item \textbf{Equality constraints} ($B_1 = B_2$ or $V_1 = V_2$): We unify the variables, treating them as the same variable throughout the solving process. Any update to one variable is immediately reflected in all variables unified with it. \circled{6}

	    \item \textbf{Top consistency} ($B_1 \equiv_{\top} B_2$): If either variable becomes $\top$, set both to $\top$. \circled{7}
\end{itemize}

The solver repeatedly processes constraints until reaching a fixed point, where no constraint causes any further updates to the variables. The monotonicity of the lattice operations ensures that this process terminates. This procedure may add more cuts to $B$s than strictly necessary, but it never sets $B=\top$ unless that is the only way to satisfy all constraints. In other words, the discretization itself may be suboptimal in terms of the number of cuts.

Lattice constraint solving takes time polynomial in the size of the program: at worst, \emph{every} constant in the program gets added to the cut set or value set of every subexpression \mbox{of type $\float[B;V]$.}

\paragraph{Additional Constraint for Contextual Equivalence.}
When the return type of the program involves a float type, either directly or as part of a larger compound type (e.g., $\float[B;V] \to \float[B';V']$), we add an additional constraint to ensure that these float types are not discretized. 
We recursively traverse the return type and set all nested $B$ sets to $\top$. This ensures that the observable behavior of the program is preserved after discretization, regardless of the context in which the program is used. When the return type is basic, such as $\bool$, nothing needs to be done.

\begin{example}
  To make the constraint solving process concrete, \Cref{fig:constraint-solving} walks through~\Cref{sec:examples:4}. Each constraint obtained from the generation phase is processed by the solver, then collected to obtain the final cut and value sets.
\end{example}

\section{Discretization}\label{sec:discretization}
\begin{figure}[t]
    \centering
    \scriptsize
    \setlength{\tabcolsep}{6pt}
    \renewcommand{\arraystretch}{1.2}
    \begin{adjustbox}{max width=\textwidth}
    \begin{tabular}{@{}>{\raggedright\arraybackslash}m{0.27\textwidth}
        @{\hspace{6em}}
        >{\raggedright\arraybackslash}p{0.18\textwidth}
        @{\hspace{1em}}
        >{\raggedright\arraybackslash}p{0.47\textwidth}@{}}
        \toprule
        \toprule
        $\boldsymbol{e : \tau}$ & \textbf{Condition} & $\boldsymbol{\discretize{e : \tau}}$ \\
        \midrule
        ${x : \tau}$ & & $x$ \\
        \midrule
        \multirow[c]{2}{0.27\textwidth}{\raisebox{-0.8ex}{${c : \float[B;V]}$}}
        & $B = \top$ 
		& $c$ \\
        \caserule
        & $B \neq \top$
        & $\finconst{k}{(|B|+1)} \text{ such that } c \in \intervals{B}_k$ \\
        \midrule
        ${\letkw\; x = e_1 : \tau_1\; \inkw\; e_2 : \tau_2 : \tau_2}$
        &
        &
        $\letkw\; x = \discretize{e_1 : \tau_1}\; \inkw\; \discretize{e_2 : \tau_2}$ \\
        \midrule
        ${(\ifkw\; e_1 : \bool\; \thenkw\; e_2 : \tau\; \elsekw\; e_3 : \tau) : \tau}$
        &
        &
        $\ifkw\; \discretize{e_1 : \bool}\; \thenkw\; \discretize{e_2 : \tau}\; \elsekw\; \discretize{e_3 : \tau}$ \\
        \midrule
        \multirow[c]{2}{0.27\textwidth}{\raisebox{-0.8ex}{${ (e_1 : \float[B;V_1] < e_2 : \float[B;V_2]) : \bool }$}}
        & $B \neq \top$
        & $\discretize{e_1 : \float[B;V_1]} \finlt{(|B|+1)} \discretize{e_2 : \float[B;V_2]}$ \\
        \caserule
        & $B = \top$
        & $\discretize{e_1 : \float[B;V_1]} < \discretize{e_2 : \float[B;V_2]}$ \\
        \midrule
        \multirow[c]{2}{0.27\textwidth}{\raisebox{-0.8ex}{${ (e_1 : \float[B;V_1] \leq e_2 : \float[B;V_2]) : \bool }$}}
        & $B \neq \top$
        & $\discretize{e_1 : \float[B;V_1]} \finleq{(|B|+1)} \discretize{e_2 : \float[B;V_2]}$ \\
        \caserule
        & $B = \top$
        & $\discretize{e_1 : \float[B;V_1]} \leq \discretize{e_2 : \float[B;V_2]}$ \\
        \midrule
        \multirow[c]{15}{0.27\textwidth}{{$
		\begin{array}[t]{@{}l@{}}
		\uniform(e_1 : \float[B_1;V_1], \\
		\quad e_2 : \float[B_2;V_2]) : \float[B;V]
		\end{array}$}}
        & $e_1, e_2 \in \mathbb{R},\; e_1 \geq e_2$
        & $\diverge$ \\
        \caserule
        & $B = \top$
        & $\uniform(\discretize{e_1 : \float[B_1;V_1]}, \discretize{e_2 : \float[B_2;V_2]})$ \\
        \caserule
        & $\begin{array}[t]{@{}l@{}}
            e_1, e_2 \in \mathbb{R}, \\
            B = \{\sim_1\!\!b_1, \ldots, \sim_n\!\!b_n\}
          \end{array}$
        & $\begin{array}[t]{@{}l@{}}
            \discrete(p_0, \ldots, p_n) \\
            \quad \text{where } p_k = \mathrm{CDF}(b_{k+1}) - \mathrm{CDF}(b_k)
          \end{array}$ \\
        \caserule
        & $\begin{array}{@{}l@{}}
            B \neq \top, \\
            B_1 \neq \top, \\
            B_2 \neq \top, \\
            V_1 = \{v_1, \ldots, v_n\}, \\
            V_2 = \{w_1, \ldots, w_m\}
          \end{array}$
        & $\begin{array}{@{}l@{}}
            \letkw \; v_1' = \discretize{e_1 : \float[B_1; V_1]} \; \inkw \\
            \letkw \; v_2' = \discretize{e_2 : \float[B_2; V_2]} \; \inkw \\
            \ifkw\; v_1' \fineq{(|B_1|+1)} \discretize{v_1 : \floattype{B_1}{V_1}} \\
            \quad \logand\; v_2' \fineq{(|B_2|+1)} \discretize{w_1 : \floattype{B_2}{V_2}}\; \thenkw \\
            \quad\quad \discretize{\uniform(v_1, w_1) : \float[B;V]} \\
            \ldots ~\text{(enumerate over $V_1 \times V_2$)} ~\ldots \\
            \elsekw\; \ifkw\; v_1' \fineq{(|B_1|+1)} \discretize{v_n : \floattype{B_1}{V_1}} \\
            \quad \logand\; v_2' \fineq{(|B_2|+1)} \discretize{w_m : \floattype{B_2}{V_2}}\; \thenkw \\
            \quad\quad \discretize{\uniform(v_n, w_m) : \float[B;V]} \\
            \elsekw\; \diverge
          \end{array}$ \\
        \bottomrule
        \bottomrule
    \end{tabular}
    \end{adjustbox}
    \caption{Discretization function, where $\mathrm{CDF}(b_{0}) = 0$ and $\mathrm{CDF}(b_{n+1}) = 1$. For brevity, the table writes $\mathrm{CDF}(b_i)$ for each cut $b_i$, by which we mean either $\mathrm{CDF}(b_i^-)$ or $\mathrm{CDF}(b_i)$ depending on if the cut is strict or non-strict. Here, $\mathrm{CDF}(b_i^-)=\Pr[X<b_i]$ denotes the left limit at $b_i$, while $\mathrm{CDF}(b_i)=\Pr[X\leq b_i]$. These coincide for continuous distributions, but may differ for discrete or mixed distributions when $\Pr[X=b_i]>0$.
    For the $\uniform$ distribution, $V = \top$ is implicit, and all remaining cut-set and value-set cases are ruled out by type inference. For the final $\uniform$ case, if $V_1 = \emptyset$ and/or $V_2 = \emptyset$, the case enumeration is empty and thus yields $\diverge$. For all other language constructs, $\discretize{\cdot}$ recurses structurally over them.}
    \label{fig:discretization-code}
\end{figure}
This section presents local transformation rules for our type-directed discretization algorithm, with examples. Our goal is to compute, from a typed expression $e : \bool$ and its type annotations, a \emph{discretized} expression $\discretize{e : \bool} : \bool$ such that $e$ and $\discretize{e : \bool}$ are \emph{semantically equivalent}, i.e., they induce the same distribution. We refer to this semantics-preserving transformation as \emph{exact discretization}, to distinguish it from approaches~\cite{Garg2024BitBlast,Huang2021AQUA} that approximate continuous programs via externally chosen precision parameters. Exactness here is the equality proved by \Cref{thm:maintext:sound}: the transformed program preserves the relevant probabilities of the source program, not merely an approximation with an error bound. Intuitively, when a program uses continuous quantities only through finitely many comparison outcomes, exact inference should reduce to a discrete problem, and exact discretization makes that reduction explicit.

As shown in \Cref{fig:discretization-code}, $\discretize{e : \tau}$ is defined recursively on $e : \tau$ and its type annotations; we focus on continuous uniform distributions, with others handled analogously. The only sub-expressions replaced by $\discretize{\cdot}$ are float constants and sampling instructions. The rule schemata are local, but their selection and parameters come from globally inferred cut/value information; \Cref{sec:soundness} proves that this non-local, type-mediated construction preserves semantics. Before detailing the rules, we outline how type annotations guide discretization. The type system is designed so that every typed sub-expression of the form 
$e = \uniform(e_1, e_2 )$ of type $\float[B;V]$ with $B=\{\sim_1\!\!b_1, \ldots, \sim_n\!\!b_n\} \neq \top$
can be \emph{discretized} to an expression $e' : \fin{n+1}$ such that 
the probability that $e$ evaluates to a value in $\intervals{B}_i$ \emph{coincides} with the probability that $e'$ evaluates to $\finconst{i}{n+1}$. 
In particular, if \emph{all} $\float[B;V]$-typed subexpressions are finite-cut, then the given expression can be \emph{fully discretized}. 

\begin{example}
	\label{ex:discretize_simple}
Consider a uniform distribution whose outcomes are compared to constants:
\begin{lstlisting}[aboveskip=1em,belowskip=1em,escapechar=!,gobble=4]
    let x = uniform(0, 1) : !{\ft{B}}! in
      (x : !{\ft{B}}! < (0.3 : !{\ftb{B}{\{0.3\}}}!)) || 
      (x : !{\ft{B}}! >= (0.5 : !{\ftb{B}{\{0.5\}}}!))
\end{lstlisting}    
where $B= \{ \clt{0.3},\clt{0.5} \}$. The cut set $B$ informs us that the probability of the comparisons evaluating to $\true$ is determined by the probability of $x$ evaluating to one of the intervals 
\[
	\intervals{B}_0 = (-\infty, 0.3) 
	\qquad 
	\intervals{B}_1 =[0.3, 0.5) 
	\qquad 
	\intervals{B}_2 =[0.5, \infty)
\]
induced by $B$. Hence, all expressions of type $\text{\ftb{B}{$\cdot$}}$ can be discretized to expressions of the finite type $\fin{3}$, where $\finconst{k}{3}$ represents the $k$-th interval $\intervals{B}_k$ induced by $B$ for $k \in \{0,1,2\}$. 

Now consider the discretized expression obtained from computing $\discretize{e: \bool}$, given by 
\begin{lstlisting}[aboveskip=1em,belowskip=1em,escapechar=!,mathescape,gobble=4]
    let x = discrete(0.3, 0.2, 0.5) in (x <#3 1#3) || (x >=#3 2#3).
\end{lstlisting}
The probabilities appearing in $\discrete(\ldots)$ are determined by the cumulative distribution function (CDF) of the uniform distribution over $[0,1]$, ensuring that the probability of $x$ evaluating to $\finconst{k}{3}$ in $\discretize{e: \bool}$ \emph{coincides} with the probability of $x$ evaluating to a value in $\intervals{B}_k$ in $e$. Hence, the constants $0.3$ and $0.5$ in $e$ can be replaced by the constants $\finconst{1}{3}$ and $\finconst{2}{3}$  in $\discretize{e: \bool}$.
\end{example}
Crucially, the \emph{answers} constraint in the premise of the typing rules for comparisons (cf.\ \Cref{fig:typing-float}) ensures that $B$ is a sufficiently fine-grained partition of $\Reals$, i.e., there is an injective function
$
	\text{ToInts} \colon \{0.3, 0.5\} \to \intervals{B}
$
that maps each constant $c \in\{0.3,0.5\}$ to the unique interval in $\intervals{B}$ containing $c$. This interval is determined by the cut convention: a cut $\clt{c}$ puts $c$ in the interval to the right of the boundary, while a cut $\cleq{c}$ puts $c$ in the interval to the left. 
\begin{example}
	If the arguments to a uniform sampling are non-constant, the situation becomes more challenging. Consider the following typed expression $e : \bool$ (again only showing relevant types):
\begin{lstlisting}[aboveskip=1em,belowskip=1em,escapechar=!,gobble=4]
    let x = (
      if uniform(0, 1) : !{\ft{\{<0.5\}}}! < (0.5 : !{\ftb{\{<0.5\}}{\{0.5\}}}!) then
        10 : !{\ftb{\{<=10,<=20\}}{\{10\}}}!
      else
        20 : !{\ftb{\{<=10,<=20\}}{\{20\}}}!) : !{\ftb{\{<=10,<=20\}}{\{10,20\}}}! in
    uniform(0 : !{\ftb{\{<=0\}}{\{0\}}}!, x : !{\ftb{\{<=10,<=20\}}{\{10,20\}}}!) : !{\ft{\{<5\}}}!
      < (5 : !{\ftb{\{<5\}}{\{5\}}}!)
\end{lstlisting}	
		Why can this expression be discretized even though the argument to the second uniform sampling instruction is non-constant? The typing information $(x : \text{\ftb{\{<=10,<=20\}}{\{10,20\}}})$ correctly tells us that \emph{$x$ can only take the values $10$ and $20$}. Consequently, $x$ (and the $\ifkw$-expression generating it) can be discretized to a finite type $\fin{3}$ such that the probability of the resulting discrete expression evaluating to $\finconst{0}{3}$ (resp.\ $\finconst{1}{3}$) \emph{coincides} with the probability of the continuous expression evaluating to $10$ (resp.\ $20$). Let us (partially) apply $\discretize{e: \bool}$, resulting in:
\begin{lstlisting}[aboveskip=1em,belowskip=1em,escapechar=!,mathescape,gobble=4]
    let x = if discrete(0.5, 0.5) <#2 1#2 then 0#3 else 1#3 in
    ( let v1' = !$\discretize{\texttt{0 : \ftb{\{<=0\}}{\{0\}}}}$! in
      let v2'  = !$\discretize{\texttt{x : \ftb{\{<=10,<=20\}}{\{10,20\}}}}$! in 
      if v1' ==#2 !$\discretize{\texttt{0 : \ftb{\{<=0\}}{\{0\}}}}$! && 
         v2' ==#3 !$\discretize{\texttt{10 : \ftb{\{<=10,<=20\}}{\{10,20\}}}}$! 
      then !$\discretize{\texttt{\uniform(0,10) : \ft{\{<5\}}}}$!
      else !$\discretize{\texttt{\uniform(0,20) : \ft{\{<5\}}}}$! ) <#2 !$\discretize{\texttt{5 : \ftb{\{<5\}}{\{5\}}}}$!
\end{lstlisting}	
		The continuous $\uniform(0,x)$-expression is replaced by a case distinction on the discretizations of $\uniform(0,10)$ and $\uniform(0,20)$, with each case executed with the correct probability $0.5$. Since this reduces the remaining discretization process to discretizing expressions with \emph{constant parameters only}, the remaining steps are completely analogous to~\Cref{ex:discretize_simple}.
\end{example}
Crucially, the \emph{recoverability} constraints in the \textsc{Continuous}-rules from \Cref{fig:typing-float} ensure that the parameter expressions of the respective sampling expressions can take only finitely many values, and that the cut sets are fine-grained enough to recover the values the parameter expressions evaluate to, as demonstrated in the preceding example.

A similar intuition applies to discretizable comparisons $e_1 < e_2$, where \emph{both} expressions are non-constant but finite-cut. The \emph{answers} constraints in the rule \textsc{Less} (resp.\ \textsc{Less-Equal}) ensure that at least one of $e_1$ and $e_2$ evaluates to only finitely many values, which enables discretization:
\begin{example}
	Consider the following typed expression $(e_1 \leq e_2) : \bool$, given by
\begin{lstlisting}[aboveskip=1em,belowskip=1em,escapechar=!,mathescape,gobble=4]
    (if discrete(0.5, 0.5) <#2 1#2 then 
      uniform(0,10) : !{\ft{B}}! else 2.5 : !{\ftb{B}{\{2.5\}}}!)
    <=
    (if discrete(0.5, 0.5) <#2 1#2 
      then 2.5 : !{\ftb{B}{\{2.5\}}}! else 3.5 : !{\ftb{B}{\{3.5\}}}!)
\end{lstlisting}
	where $B = \{\cleq{2.5}, \cleq{3.5}\}$, and where $e_1 : \text{\ft{B}}$ and $e_2 : \text{\ftb{B}{\{2.5,3.5\}}}$. Intuitively, $e_1$ samples from a mixture of a uniform and a dirac distribution, and $e_2$ samples from a finite distribution. Why can $e_1 \leq e_2$ be discretized even though neither $e_1$ nor $e_2$ is a constant? As suggested by its type, $e_2$ can take only finitely many values, namely $\{2.5,3.5\}$. We apply $\discretize{\cdot}$:
\begin{lstlisting}[aboveskip=1em,belowskip=1em,escapechar=!,mathescape,gobble=4]
    (if discrete(0.5, 0.5) <#2 1#2 then discrete(0.25, 0.1, 0.65) else 0#3)
    <=
    (if discrete(0.5, 0.5) <#2 1#2 then 0#3 else 1#3)
\end{lstlisting}
		Notice that the probability of $e_1$ evaluating to some value in $(-\infty, 2.5]$ indeed \emph{coincides} with the probability of the discretization of $e_1$ evaluating to $\finconst{0}{3}$. For $e_2$, on the other hand, the probability of it evaluating to $2.5$ (resp.\ $3.5$) \emph{coincides} with the probability of its discretization evaluating to $\finconst{0}{3}$ (resp.\ $\finconst{1}{3}$). For this reason, we obtain a semantically equivalent discretization of $e_1 \leq e_2$.
\end{example}

\subsubsection{Discretization of Types}
We conclude by extending $\discretize{\cdot}$ to types. \Cref{fig:discretization-types} shows how type $\tau$ is discretized to type $\discretize{\tau}$. We then obtain the following natural result, proved in~\Cref{proof:type-transformations}:
\begin{figure}[t]
	\centering
	\scalebox{0.9}{
	$
		\discretize{\tau} = 
		\begin{cases}
			\fin{|B|+1} & \text{if } \tau = \float[B; V] \text{ and } B \neq \top \\[0.5em]
			\floattype{B}{V} & \text{if } \tau = \float[B; V] \text{ and } B = \top \\[0.5em]
			\text{(straightforward recursion on type structure)} & \text{otherwise}
		\end{cases}
$
}
	\caption{Transformation of types under $\mathcal{D}$, where $\mathcal{D}$ acts on types.}
	\label{fig:discretization-types}
\end{figure}
\begin{theorem}[Discretization Is Well-Typed]\label{lem:type-transformations}
	Let $\Gamma \vdash e : \tau$. We have
	$
	\discretize{\Gamma} \vdash \discretize{e : \tau} : \discretize{\tau},
	$
	where $\discretize{\Gamma}$ transforms all types in $\Gamma$, i.e., $\discretize{\Gamma}(x) = \discretize{\Gamma(x)}$ for all $x \in \textnormal{dom}(\Gamma)$.
\end{theorem}
%

\section{Operational Semantics and Soundness}
\label{sec:soundness}%
We now establish the soundness of our approach in the following sense: For \emph{every} expression $e:\bool$, possibly including \emph{continuous sampling, higher-order functions, recursion, and conditioning},
\begin{center}
	\emph{$e: \bool$ and $\discretize{e : \bool}:\bool$ (obtained from \Cref{fig:discretization-code}) are semantically equivalent.}
\end{center}
%
The first step is to make formal sense of semantic equivalence. For that, we will define a monadic operational semantics for our language in \Cref{sec:opsem}. We will then proceed by proving the above claim. This is non-trivial: discretization decisions are driven by non-local, type-propagated information, so we cannot proceed by a local rule-by-rule argument. To prove the claim rigorously, we relate the executions of $e$ and $\discretize{e}$ through a coupling-style logical relation in \Cref{sec:opsem}. 
\subsection{Operational Semantics}
\label{sec:opsem}
The type-directedness of our discretization function $\discretize{\cdot}$ requires our soundness argument to depend on the types of the given (sub-)expressions as well. We thus start by defining a \emph{type-dependent semantic functional}, i.e., we define for every type $\tau$ a function 
$
	\bigsem{\cdot}\colon Expr_\tau \to \mathcal{D}(Vals_\tau + \bot)~,
$
mapping each \emph{fully type-annotated} expression $e:\tau$ to a  \emph{sub}-probability measure $\bigsem{e : \tau}$ over values of type $\tau$ or $\bot$. In the presence of possibly unbounded recursion \emph{and} conditioning, our semantics distinguishes between non-termination and observation failure, which we realize as follows: The \enquote{missing} probability mass in $\bigsem{e : \tau}$ is the probability that $e$ diverges. The probability of $\bot$ is the probability of encountering an observation failure. With this model, semantic equivalence indeed captures our sought-after correctness claim: If $\bigsem{e:\bool} = \bigsem{\discretize{e:\bool}:\bool}$, then both expressions have the same probability of producing $\true$, $\false$, or an observation failure.

\begin{remark}
Distinguishing between non-termination and observation failure yields a more fine-grained semantics which does, e.g., \emph{not} semantically equate the two expressions 
\begin{align*}
 \letkw \; x = \uniform(0,1) \; \inkw \; \observekw(x>1)
 \qquad\text{and}\qquad
 \diverge~.
\end{align*}
In particular, our soundness theorem does not deem the latter a valid discretization of the former.
\end{remark}

\subsubsection{Small Step Semantics}
The type-dependent semantic functional $\bigsem{\cdot}$ is defined via a \emph{small-step} semantic functional $\sem{\cdot}\colon Expr_\tau \to \mathcal{D}(Expr_\tau + \bot)$. Defining the latter requires us to carefully design a suitable measurable space over the set of type-annotated expressions of type $\tau$ and $\bot$. All details on these measure-theoretic constructions are presented in \Cref{app:measure_space_expressions}. Let us consider the key idea here.

Let $\tau=\bool$ for the sake of concreteness. Care must be taken since, e.g., sub-expressions of type $\floattype{B}{\{3.5,4.5\}}$ cannot evaluate to \emph{arbitrary} reals but are restricted to the values $\{3.5,4.5\}$. We tackle this by extracting from each type-annotated expression a suitable \emph{skeleton}, where float constants are replaced by \enquote{holes} (hence the term \enquote{skeleton}). For instance, the expression
\[
	e \quad = \quad (x : \floattype{\{\clt{3.5}, \clt{4.5}\}}{\top}
	\;<\;
	y : \floattype{\{\clt{3.5}, \clt{4.5}\}}{\{3.5,4.5 \}})
	: \bool
\]
gives rise to the skeleton
\[
	s \quad = \quad (\square : \floattype{\{\clt{3.5}, \clt{4.5}\}}{\top}
	\;<\;
	\square : \floattype{\{\clt{3.5}, \clt{4.5}\}}{\{3.5,4.5 \}})
	: \bool
\]
and the type annotations tell us that the possible values these holes can be filled with are in $\Reals \times \{3.5,4.5 \}$. We can thus inherit a measurable space over all well-typed expressions arising from this particular skeleton $s$ from the standard Borel measurable space over $\Reals \times \{3.5,4.5 \}$. The measurable space over \emph{all} well-typed expressions of type $\bool$ (i.e., over $Expr_\bool$) is then obtained from the measurable spaces of \emph{all} skeletons that give rise to expressions of type $\bool$.

Equipped with this family of measurable spaces, we define the small-step functional $\sem{\cdot}$ in the expected, type-preserving manner (cf.\ \Cref{tab:smallstep-1,tab:smallstep-2}). For instance, if $e_1$ is not a value, then\footnote{Put formally, $(\mu \gg\!= f)(A) = \int f(x)(A) \cdot \mu(dx)$, see \Cref{def:dist-monad} for details. $\dirac{e}$ denotes the Dirac distribution of $e$.}
\[
	\sem{(\letkw \; x = e_1 : \tau_1 \; \inkw \; e_2 : \tau_2) : \tau_2}
	\quad~{}={}~
	\underbrace{\sem{ e_1 : \tau_1 } \gg\!=  \lambda g. 
		\begin{cases}
			\dirac{\bot} & \text{if $g=\bot$} \\
		\dirac{(\letkw \; x = g \; \inkw \; e_2 : \tau_2) : \tau_2} & \text{otherwise}~.
		\end{cases}
	}_{\text{standard probabilistic monadic bind over suitable sub-Markov kernels}}~,
\]
i.e., $\sem{ e_1 : \tau_1}$ yields the sub-distribution over expressions obtained from executing $e_1$ for \emph{one step}, and the monadic bind $\gg\!=$ \enquote{folds} this distribution into the surrounding $\letkw$-statement. The intuition for the remaining statements is analogous. Crucially, $\gg\!=$ propagates the probability of encountering $\bot$ --- an observation failure --- through executions. For instance, assume that $e_1$ from above is given by $\observekw\; (\false : \bool)  : \unit$. Since $e_1$ encounters an observation failure with probability $1$, i.e., $\sem{\observekw\; (\false : \bool)  : \unit}(\{\bot\}) = 1$, the same holds for the surrounding $\letkw$-statement, i.e., also $\sem{(\letkw \; x = e_1 : \tau_1 \; \inkw \; e_2 : \tau_2) : \tau_2}(\{\bot\})=1$.
All details --- including proofs that all monadic binds are well-behaved since they operate on suitable sub-Markov kernels --- are provided in \Cref{app:small-step}.

\subsubsection{From Small-Step to Big-Step Semantics}
\label{sec:smallsem_to_bigsem}
We now define $\bigsem{\cdot}$ via $\sem{\cdot}$. For that, we first lift the one-step semantics $\sem{\cdot}$ to an $n$-step semantics $\bigsemstep{\cdot}{n}$ in the expected way:
\begin{align*}
	\bigsemstep{e : \tau}{0} =\delta_{e : \tau}  
	\qquad\text{and}\qquad
	\bigsemstep{e : \tau}{n+1} ~{}={}~ 
		\bigsemstep{e : \tau}{n} \monbind 
			\lambda e'. 
		\begin{cases}
			\delta_\bot & \text{if $e'=\bot$}\\
		\sem{e' : \tau}  & \text{otherwise}
		\end{cases} 
\end{align*}
This construction is well-defined since the functional $\sem{\cdot}$ itself gives rise to a suitable sub-Markov kernel (\Cref{lem:small-step-kernel}). We then define for every type $\tau$ the functional $\bigsem{\cdot}\colon Expr_\tau \to \mathcal{D}(Vals_\tau + \bot)$ as
$
	\bigsem{e:\tau} = \lambda A. \lim_{n\to\infty} \bigsemstep{e:\tau}{n}(A)~,
$
which is well-defined by the monotone convergence theorem: the sequence $(\bigsemstep{e : \tau}{n})_{n\in\Nats}$ is monotonically increasing for every measurable set $A$ over $Vals_\tau + \bot$ (cf.\ \Cref{lem:bigstepsem}) and upper-bounded by $1$: executing a program for more and more steps can only increase the probability of terminating in a given set of values or an observation failure.

%
%
%

\subsection{Soundness via Probabilistic Couplings}
\label{C}
We first state the central theorem, then outline our coupling-style correctness argument. Throughout, $e : \tau$ denotes a closed, fully annotated expression satisfying $\emptyset \vdash e : \tau$ under the declarative typing rules.
\begin{theorem}
	\label{thm:maintext:sound}
	Let $e:\bool$. Then $\bigsem{e:\bool} = \bigsem{\discretize{e:\bool} : \bool}$.
\end{theorem}
Intuitively, we prove that the two programs can execute in a synchronized manner such that their probabilities of terminating in value $v$ coincide for each $v \in \{\true,\false,\bot\}$. 

Towards this end, define for every type $\tau$ the set $P_\tau = \{ (e:\tau, \discretize{e : \tau} : \discretize{\tau}) ~|~ e:\tau \}$. We define a \emph{coupled small-step semantics} $\sem{\cdot,\cdot} \colon P_\tau \to \mathcal{D}(P_\tau + \bot)$ (cf.\ \Cref{tab:coupling_smallstep_1,tab:coupling_smallstep_2}), which models the simultaneous execution of $e$ and $\discretize{e:\tau}$. The key property of this coupled semantics is that the semantics of both $e$ and $\discretize{e:\tau}$ can be recovered in a sense we formalize next.

Given $\mu \in \mathcal{D}(P_\tau + \bot)$, we define the \emph{projection of $\mu$ onto the $i$-th component} for $i\in\{ 1,2 \}$ as
\[
	\pi_i(\mu) = \lambda A. \int_{x \in P_\tau + \bot} 
	\begin{cases}
		 \delta_{e_i}(A) & x = (e_1,e_2)  \\
		 \delta_\bot(A)  & x = \bot
	\end{cases}
	\mu(dx)~,
\]
%
With projections, we recover the semantics of $e$ and $\discretize{e:\tau}$ as follows:
\begin{lemma}
	\label{lem:maintext:small_step_coupling_props}
	Let $e:\tau$. We have:
	\begin{enumerate}
		\item\label{lem:maintext:small_step_coupling_props1} 
		$\pi_1(\sem{e:\tau, \discretize{e : \tau} : \discretize{\tau}}) = \sem{e:\tau}$.
		\item\label{lem:maintext:small_step_coupling_props2} 
		$\pi_2(\sem{e:\tau, \discretize{e : \tau} : \discretize{\tau}}) = \sem{\discretize{e : \tau}:\discretize{\tau}}$.
	\end{enumerate}
\end{lemma}
%

Now, using a limit construction analogous to \Cref{sec:smallsem_to_bigsem} for obtaining a coupled \emph{big-step} semantics $\bigsem{\cdot,\cdot} \colon P_\tau \to \mathcal{D}(V_\tau + \bot)$, where $V_\tau = \{ (v:\tau, \discretize{v} : \discretize{\tau}) ~|~ v:\tau~\text{is a value} \}$, we prove a suitable continuity property of projection (cf.\ \Cref{lem:proj_cont}) to conclude that \Cref{lem:maintext:small_step_coupling_props} transfers to $\bigsem{\cdot,\cdot}$. Finally, observing that $V_\bool + \bot = \{(\true,\true), (\false,\false), \bot\}$ yields the desired claim:

\emph{Proof of \Cref{thm:maintext:sound}}. Let $v \in \{\true,\false,\bot\}$. We have
\begin{align*}
	& \bigsem{e : \bool}(v) \\
	{}={}~& \pi_1(\bigsem{e:\bool, \discretize{e : \bool} : {\bool}})(v)
		\tag{\Cref{lem:maintext:small_step_coupling_props}.\ref{lem:maintext:small_step_coupling_props1} for $\bigsem{\cdot,\cdot}$} \\
    {}={}~& \pi_1(\bigsem{e:\bool, \discretize{e : \bool} : {\bool}})(v)
    \tag{$P_\bool$ is discrete} \\
    {}={} & \bigsem{e : \bool, \discretize{e : \bool}:{\bool}}((\true,\true))\cdot \delta_{\true}(v) \\
    &{}+ 
    \bigsem{e : \bool, \discretize{e : \bool}:{\bool}}((\false,\false))\cdot \delta_{\false}(v) \\
    &{}+ 
    \bigsem{e : \bool, \discretize{e : \bool}:{\bool}}(\bot)\cdot \delta_{\bot}(v) 
    \tag{definition of $\pi_1$, argument is discrete} \\
    {}={}~& \pi_2(\bigsem{e:\bool, \discretize{e : \bool} : {\bool}})(v)
    \tag{definition of $\pi_2$, argument is discrete} \\
    {}={}~& \bigsem{\discretize{e : \bool}:{\bool}}(v)~.
    \tag{\Cref{lem:maintext:small_step_coupling_props}.\ref{lem:maintext:small_step_coupling_props2} for $\bigsem{\cdot,\cdot}$}
\end{align*}

\section{Implementation, Case Studies, and Evaluation}\label{sec:implem_eval}

We have implemented the \Slice{} type inference and discretization procedures in OCaml. To compute the interval masses of continuous distributions, we use the OCaml GSL package~\cite{gsl} for cumulative distribution function calculations. All experiments were run on an Apple M3 with 16\,GB of RAM.

This section is driven by two research questions:
\begin{enumerate}
    \item Does \Slice{} enable exact inference on programs combining continuous sampling, higher-order functions, and unbounded recursion?
    \item On programs that are within the scope of existing tools for exact inference, is the combination of \Slice{} and discrete inference backends competitive? 
\end{enumerate}

For the first question, we provide a collection of case studies that involve continuous distributions together with recursion, higher-order functions, lists, or conditioning in \Cref{sec:case-studies}. For these case studies, we use \Slice{} to transform the original programs to finite discrete models and analyze the resulting program with \Storm{} \cite{storm} --- a probabilistic model checker.

For the second question, we evaluate \Slice{}+\Dice{} and \Slice{}+\Roulette{} on two classes of benchmarks in \Cref{sec:synthetic-benchmarks,sec:fairness-benchmarks}: (i) synthetic scaling benchmarks that let us vary program size and dependency structure in a controlled way, and (ii) benchmark suites from prior work \cite{Saad2021SPPL,gehr2016psi}. We compare against \SPPL \cite{Saad2021SPPL}, an exact inference engine supporting continuous probabilistic programs, on benchmarks that \emph{both} \Slice and \SPPL support, and we additionally report \PSI{}~\cite{gehr2016psi} and \Hakaru{}~\cite{Narayanan2016Hakaru} timings for the representative prior-work benchmarks. We remark, however, that \SPPL supports arithmetic, as well as (soft) conditioning on measure zero events, which \Slice does not support. Our evaluation merely serves the purpose of investigating (i) the overhead of discretization and (ii) the benefits of being able to use efficient \emph{BDD-based} inference back-ends (i.e., \Dice and \Roulette) for exact inference on \emph{continuous programs}.

\subsection{\Slicebf+\Stormbf: Case Studies}\label{sec:case-studies}
There is no exact inference system that, to our knowledge, can cover programs combining all three features: continuous sampling, higher-order functions, and unbounded recursion. As a consequence, there is no exact inference engine for the programs in \Cref{sec:examples:7,sec:casestudy:epsball,sec:casestudy:bloom}. Instead, we use \emph{probabilistic model checking}: The small-step semantics of our language (cf.\ \Cref{tab:smallstep-1,tab:smallstep-2}) for discrete programs gives rise to a \emph{Markov chain} modeling the  stochastic process of program execution, inspired by~\cite{Jansen2016BoundedModelCheckingProbPrograms} on imperative programs. If that Markov chain is finite, we can use the probabilistic model checker \Storm \cite{storm} for exact inference. Let us start with an illustrative example.

%

\begin{figure}[!t]
    \centering
    
    \begin{minipage}[t]{0.34\textwidth}
    \vspace{0pt}
    
    \textbf{Continuous Program}
    \begin{lstlisting}[aboveskip=0.4em,belowskip=0.4em,gobble=4,
    escapechar=!,
    numbers=left,
    stepnumber=1,
    numberstyle=\color{gray}\scriptsize,
    numbersep=7pt,
    basicstyle=\ttfamily\footnotesize]
    (fix resample x :=
      let _ = observe(x >= 0.2) in !\linelabel{ln:cr-obs}!
      if x <= 0.5 then true !\linelabel{ln:cr-true}!
      else if x <= 0.8 then resample x !\linelabel{ln:cr-loop}!
      else
        let y = uniform(0,1) in !\linelabel{ln:cr-sample}!
          resample y) 1 !\linelabel{ln:cr-init}!
    \end{lstlisting}
    
    \textbf{Discretized Program}
    \begin{lstlisting}[aboveskip=0.4em,belowskip=0.4em,gobble=4,
    escapechar=!,
    numbers=left,
    stepnumber=1,
    firstnumber=last,
    numberstyle=\color{gray}\scriptsize,
    numbersep=7pt,
    basicstyle=\ttfamily\footnotesize]
    (fix resample x :=
      let _ = observe(x >= 1) in !\linelabel{ln:dcr-obs}!
      if x <= 1 then true !\linelabel{ln:dcr-true}!
      else if x <= 2 then resample x !\linelabel{ln:dcr-loop}!
      else
        let y = discrete(0.2,0.3,0.3,0.2) in !\linelabel{ln:dcr-sample}!
          resample y) 3 !\linelabel{ln:dcr-init}!
    \end{lstlisting}
    \end{minipage}
    \hspace{0.05\textwidth}
    \begin{minipage}[t]{0.52\textwidth}
        \vspace{0pt}
        \raggedright
        \textbf{Markov Chain Construction \& Model Checking}
        
        \vspace{0.5em}
        
        \begin{tikzpicture}[
            scale=0.82,
            transform shape,
            >=Stealth,
            every node/.style={font=\footnotesize},
            state/.style={
                draw,
                rounded corners,
                align=center,
                minimum width=2.15cm,
                minimum height=0.62cm,
                inner sep=2pt
            },
            term/.style={draw, circle, minimum size=7.5mm, inner sep=0pt},
            lab/.style={font=\normalsize, fill=white, inner sep=1pt},
            box/.style={
                draw,
                rounded corners,
                align=center,
                inner sep=6pt
            }
        ]
        
        \node[state] (s3) at (0.4,0) {\texttt{(fix resample x := \dots)\ 3}};
        \node[state] (s0) at (-2.0,-1.35) {\texttt{(fix \dots)\ 0}};
        \node[state] (s1) at (0.4,-1.35) {\texttt{(fix \dots)\ 1}};
        \node[state] (s2) at (2.8,-1.35) {\texttt{(fix \dots)\ 2}};
        \node[term] (fail) at (-2.0,-2.45) {$\bot$};
        \node[term] (true) at (0.4,-2.45) {\texttt{true}};
        
        \draw[->] (0.4,0.75) -- (s3);
        
        \draw[->] ([xshift=-4mm]s3.south) -- node[lab, pos=0.58, left=1pt] {$0.2$} ([xshift=-1mm]s0.north);
        \draw[->] (s3.south) -- node[lab, pos=0.58, right=1pt] {$0.3$} (s1.north);
        \draw[->] ([xshift=4mm]s3.south) -- node[lab, pos=0.58, right=1pt] {$0.3$} ([xshift=1mm]s2.north);
        
        \draw[->] (s3.north west) to[out=145,in=215,looseness=4]
            node[lab, left, xshift=-2pt] {$0.2$} (s3.south west);
        
        \draw[->] (s0) -- node[lab, right, xshift=2pt] {$1$} (fail);
        \draw[->] (s1) -- node[lab, right, xshift=2pt] {$1$} (true);
        
        \draw[->] (s2.north east) to[out=35,in=-35,looseness=4]
            node[lab, right, xshift=2pt] {$1$} (s2.south east);

        \draw[
            -{Implies[length=7mm,width=7mm]},
            double,
            double distance=1.2pt,
            line width=0.9pt
        ] (0.4,-3.15) -- (0.4,-4.0);

        \node[
            draw,
            cloud,
            cloud puffs=14,
            cloud puff arc=120,
            aspect=2.0,
            inner xsep=0.1pt,
            inner ysep=0.1pt,
            text width=1.5cm,
            align=center
        ] (storm) at (0.4,-5.0)
        {
            \textbf{STORM}\\
            {\scriptsize probabilistic model checker}
        };

        \node[font=\small, anchor=west, align=center] at (storm.east)
        {
        \emph{Model Checking Result}: \\
        $\Pr(\lozenge \texttt{true} \mid 
        \square \lnot \bot) = 0.5$
        };
        \end{tikzpicture}
    \end{minipage}
    \caption{Conditional resampling loop and its discretization. The continuous program is transformed into a finite-state discrete program (\Slice{} alone takes less than 0.001s), which induces a finite Markov chain (generated in 0.001s) and can then be analyzed by \Storm (with a model checking time of 0.024s). The $\bot$-state in that Markov chain is entered when encountering an observation failure.}
    \label{fig:conditional-resampling}
\end{figure}
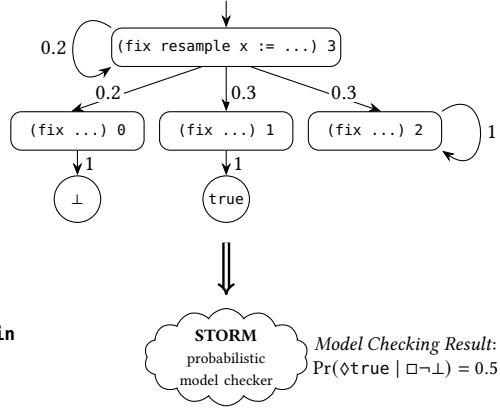

\subsubsection{Conditional Resampling Loop}
Consider the continuous program depicted in \Cref{fig:conditional-resampling}, containing continuous sampling, unbounded recursion, and conditioning inside of the recursion. \Slice automatically discretizes this program, also shown in \Cref{fig:conditional-resampling}. We generate the corresponding Markov chain by unfolding the small-step semantics in a breadth-first manner. In the presence of unbounded recursion, this Markov chain \emph{might} be infinite and this process might thus diverge. If, however, it terminates, the resulting Markov chain is finite and precisely captures the execution behavior of the discretized program. For this example, this is the case: Even though the discretized program contains unbounded recursion, its execution behavior can be modelled as a finite-state Markov chain\footnote{We apply a straightforward optimization: Intermediate execution steps that do not branch probabilistically can be collapsed into a single transition. This is the reason why all states in \Cref{fig:conditional-resampling} correspond to the header of the recursive function.}, where $\bot$ is a special state entered when encountered an observation failure.

Now, using the probabilistic model checker \Storm, we can perform exact inference: It is immediate that the Markov chain's probability $\Pr(\lozenge \texttt{true} \mid 
\square \lnot \bot)$  of eventually reaching $\true$ (in LTL: $\lozenge\true$)  \emph{given} that we never enter the observation-failure state $\bot$ (in LTL: $\square \lnot \bot$) equals the measure of $\true$ in the program's normalized posterior distribution. 

%

\subsubsection{Random Interval Covering of \ensuremath{[0,1]} with Conditioning}
\label{sec:casestudy:epsball}

Our second case study extends the random interval covering example from \Cref{sec:examples:7}  by conditioning inside of the recursion: We condition on \emph{every} sample falling within \emph{at most} one interval. The difference to the program in \Cref{sec:examples:7} is:
%
\begin{center}
\begin{minipage}[t]{0.48\textwidth}
    \vspace{0pt}
    \begin{lstlisting}[aboveskip=0.4em,belowskip=0.4em,gobble=4,
      escapechar=!,
      numbers=left,
      stepnumber=1,
      numberstyle=\color{gray}\scriptsize,
      numbersep=7pt,
      xleftmargin=12.5pt,
      xrightmargin=-12.5pt,
      basicstyle=\ttfamily\footnotesize]
    let non_overlap = fix non_overlap y := !\linelabel{ln:cov-nonoverlap}!
      fun seen -> fun ys -> 
        match ys with
        | nil -> true
        | h :: t ->
            let hit = (snd h) y in
            if hit && seen then false
            else non_overlap y (seen || hit) t
        end in
    \end{lstlisting}
    \end{minipage}
    \hfill
    \begin{minipage}[t]{0.48\textwidth}
    \vspace{0pt}
    \begin{lstlisting}[aboveskip=0.4em,belowskip=0.4em,gobble=4,
      escapechar=!,
      numbers=left,
      stepnumber=1,
      firstnumber=last,
      numberstyle=\color{gray}\scriptsize,
      numbersep=7pt,
      basicstyle=\ttfamily\footnotesize]
    let cover = fix cover st := !\linelabel{ln:cov-cover}!
      if uniform(0,1) < 0.1 then st
      else
        let y = gaussian(0,1) in
        let _ = observe(non_overlap y false st) in !\linelabel{ln:cov-obs}!
        cover (step y st)
    in
    let final = cover init in
    is_covered final !\linelabel{ln:cov-final}!
    \end{lstlisting}
\end{minipage}
\end{center}
\noindent
The key new component is the recursive function \lstinline{non_overlap} (l.\ \ref{ln:cov-nonoverlap}), which checks whether a sampled point $y$ hits more than one interval of the initially sampled intervals. The loop \lstinline{cover} (l.\ \ref{ln:cov-cover}) repeatedly samples a Gaussian point and immediately conditions on this predicate via \lstinline{observe} (l.\ \ref{ln:cov-obs}), so only executions in which the point lies in at most one interval are retained. The measure of $\true$ in the conditional posterior distribution is thus the probability that this unbounded process hits every interval at least once \emph{given} that each sampled point falls in at most one interval.

The program remains discretizable: each sampled point is still used only through comparisons against finitely many sampled interval boundaries, and \lstinline{non_overlap} depends only on the resulting hit pattern, not on the exact real value of $y$. \Slice{} replaces each Gaussian draw by its interval index, yielding a discrete program (shown in \Cref{fig:app:epsball-observe:disc}) with a finite exact state space. \Slice{} takes 0.001s, and we generate a Markov chain consisting of $268878$ states in $80$ seconds. \Storm takes $4.7$ seconds for model checking, and reports that the sought-after probability is 0.000227....

Currently, our Markov chain construction process is rather naive: extending the number of sampling points or intervals yields the Markov chain to be infeasibly large. Developing efficient symbolic discrete inference engines for higher-order programs with unbounded recursion is not in the scope of this paper. However, our approach is an excellent motivation for this future direction. 

\subsubsection{Bloom Filter}
\label{sec:casestudy:bloom}
We model a probabilistically populated variant of a Bloom filter \cite{Bloom1970SpaceTime}, which is a data structure using hash functions for storing data space efficiently: one can insert data into a Bloom filter and subsequently query whether a given data point is stored in it. Due to aggressive compression via hashing, this process can be faulty: while it can never produce false-negatives, it can produce false-positives. Now conisder the following program (full program shown in~\Cref{fig:app:bloom-filter}):
\begin{center}
\begin{minipage}[t]{0.48\textwidth}
\begin{lstlisting}[aboveskip=0.4em,belowskip=0.4em,
  language=OCaml,
  escapechar=!,
  numbers=left,
  stepnumber=1,
  numberstyle=\color{gray}\scriptsize,
  numbersep=7pt,
  xleftmargin=12.5pt,
  xrightmargin=-12.5pt,
  basicstyle=\ttfamily\footnotesize]
let insert = fun q -> fun h1 -> !\linelabel{ln:bloom-insert}!
    fun h2 -> fun bf -> 
  let index1 = h1 q in
  let index2 = h2 q in
  let bf1 = set_true bf index1 in
  set_true bf1 index2 in

let search = fun q -> fun h1 -> !\linelabel{ln:bloom-search}!
    fun h2 -> fun bf -> 
  let index1 = h1 q in
  let index2 = h2 q in
  (get bf index1) && (get bf index2) in

let h1 = fun q -> !\linelabel{ln:bloom-h1}!
  if q <= -2 then 0
  else if q <= 0 then 1
  else if q <= 2 then 2
  else 3 in
\end{lstlisting}
\end{minipage}
\hfill
\begin{minipage}[t]{0.48\textwidth}
\begin{lstlisting}[aboveskip=0.4em,belowskip=0.4em,
  language=OCaml,
  escapechar=!,
  numbers=left,
  stepnumber=1,
  firstnumber=last,
  numberstyle=\color{gray}\scriptsize,
  numbersep=7pt,
  basicstyle=\ttfamily\footnotesize]
let h2 = fun q -> !\linelabel{ln:bloom-h2}!
  if q <= -2 then 3
  else if q <= 0 then 2
  else if q <= 2 then 1
  else 0 in

let init = (false, (false, (false, false))) in !\linelabel{ln:bloom-init}!
let key = gaussian(0,1) in !\linelabel{ln:bloom-key}!

let search_before_full = fix loop bf :=
  if (full bf) then bf
  else if (search key h1 h2 bf) then bf
  else
    let q = gaussian(0,1) in !\linelabel{ln:bloom-samples}!
    loop (insert q h1 h2 bf) in

let bf_final = (search_before_full init) in 
  not (full bf_final) !\linelabel{ln:bloom-final}!
\end{lstlisting}
\end{minipage}
\end{center}
The program samples a Gaussian key (l. \ref{ln:bloom-key}) and starts from an empty four-bit Bloom filter (l. \ref{ln:bloom-init}). It then builds the filter recursively by drawing fresh Gaussian samples (l.\ \ref{ln:bloom-samples}) and inserting each one into the filter by hashing the sampled real twice using the comparison-based hash functions \lstinline{h1} and \lstinline{h2} (l.\ \ref{ln:bloom-h1} and \ref{ln:bloom-h2}), following a locality-sensitive hashing technique that hashes similar input items into the same buckets. Each insertion sets the corresponding bits in the filter using \lstinline{insert} (l.\ \ref{ln:bloom-insert}). The unbounded recursion continues until either the filter becomes full or the key is reported present by \lstinline{search} (l.\ \ref{ln:bloom-search}). The final Boolean result, \lstinline{not (full bf_final)} (l.\ \ref{ln:bloom-final}), therefore asks whether the key was found strictly before the filter became full.

The key simplification is that the hash functions never inspect the exact value of a sample. They only ask which region the sample falls into, using the cut points $-2$, $0$, and $2$. These cuts partition the real line into four regions, which already coincide with the Bloom filter's index space. \Slice{} therefore replaces each Gaussian draw by its bucket index while preserving exactly the probability of every collision pattern relevant to insertion and search. The fully discretized program is shown in~\Cref{fig:app:bloom-filter:disc}. \Slice{} takes less than 0.001s, and we generate a Markov chain consisting of $14$ states in $0.023$ seconds. \Storm takes $0.001$ seconds for model checking, and reports the probability $0.91314...$.

%
%

\subsection{\Slicebf+\Dicebf~and \Slicebf+\Roulettebf: New Benchmarks for Scaling}\label{sec:synthetic-benchmarks} 
In the absence of unbounded recursion, we can use \Dice or \Roulette as inference backends.
We designed scalable benchmarks of conditionally independent continuous programs with varying variable dependencies. Formally, all programs in our benchmark can be described by the structure: $X_i \perp\!\!\!\perp \{\,X_j : j < i, j \neq \pi(i)\,\}
\;\Big|\; X_{\pi(i)}\}$, where $\pi(i) \in \{1, \ldots, i-1\}$ is the parent index of $X_i$. Each variable $X_i$ depends directly on only one earlier variable, namely its parent $X_{\pi(i)}$; conditioned on $X_{\pi(i)}$, it is independent of all other preceding variables. Detailed listings and graphs for variants of these benchmarks appear in \Cref{app:synthetic-sketches}. Results are depicted in \Cref{fig:alt-benchmarks-main}.

We report runtime (seconds) for (i) \SPPL{}, (ii) \Slice{} (discretization only), (iii) \Slice{}+\Dice{} (discretization plus inference via \Dice{}), and (iv) \Slice{}+\Roulette{} (discretization plus inference via \Roulette{}). The timeout is 120 seconds. Program sizes are evaluated in increments of 5, from 0 to 50. For each size, we generate $10$ random instances and report average runtime. We observe that discretization remains small across sizes, and takes less than $0.01$ seconds for each program. End-to-end runtimes for \Slice{} with discrete backends grow more mildly on these families; \SPPL reaches a timeout frontier at larger instances. We conclude that using efficient BDD-based back-ends for exact inference on continuous programs via discretization can pay off for programs whose control flow induces strong conditional independence structure, allowing the resulting discrete models to remain compact even as program size grows.
 

\begin{figure}[!t]
    \centering

    \begin{subfigure}{0.24\textwidth}
        \centering
        \includegraphics[width=\linewidth]{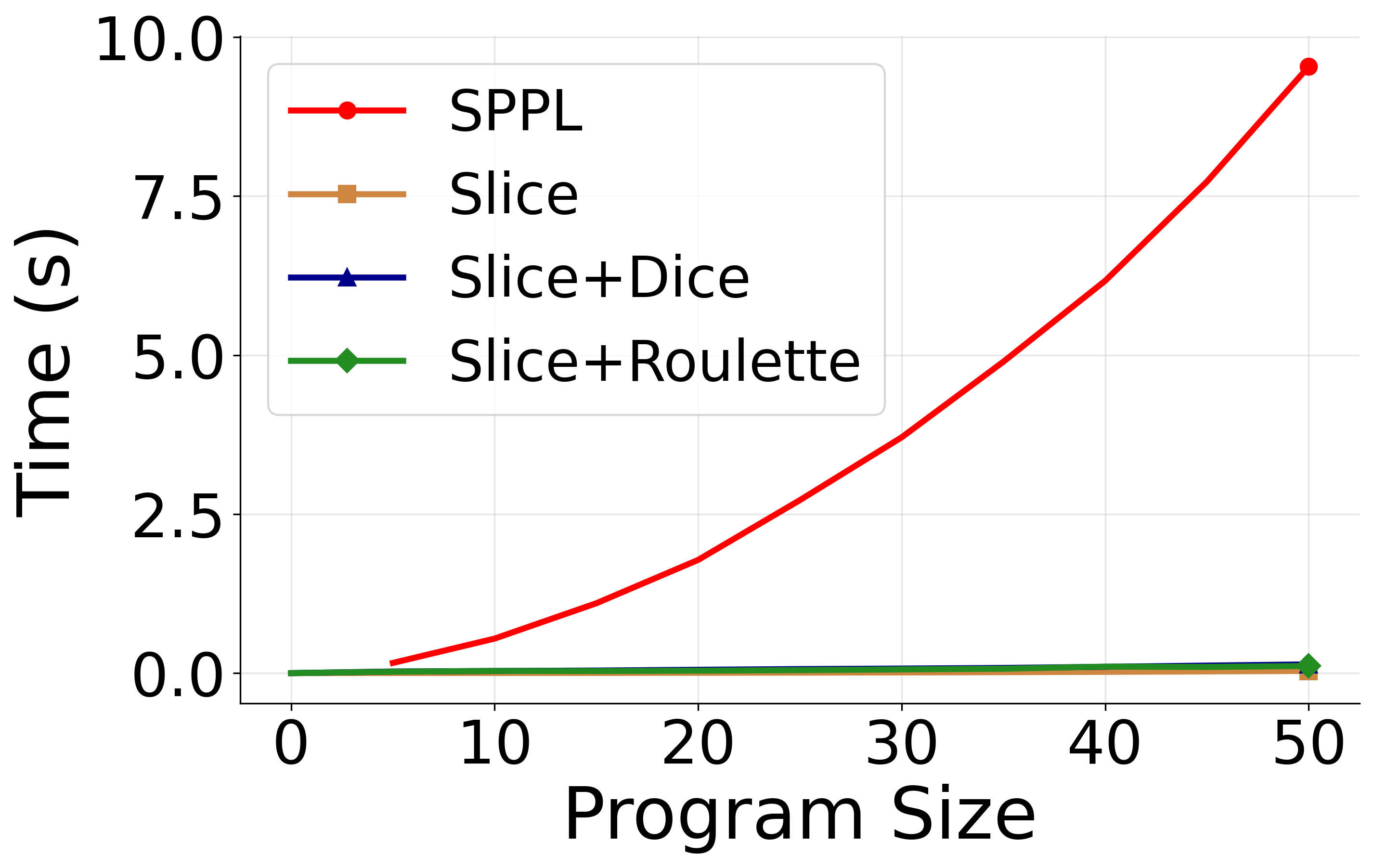}
        \caption{Cond. Ind. -- Fork}
        \label{fig:cond-benchmarks-b-main}
    \end{subfigure}\hfill
    \begin{subfigure}{0.24\textwidth}
        \centering
        \includegraphics[width=\linewidth]{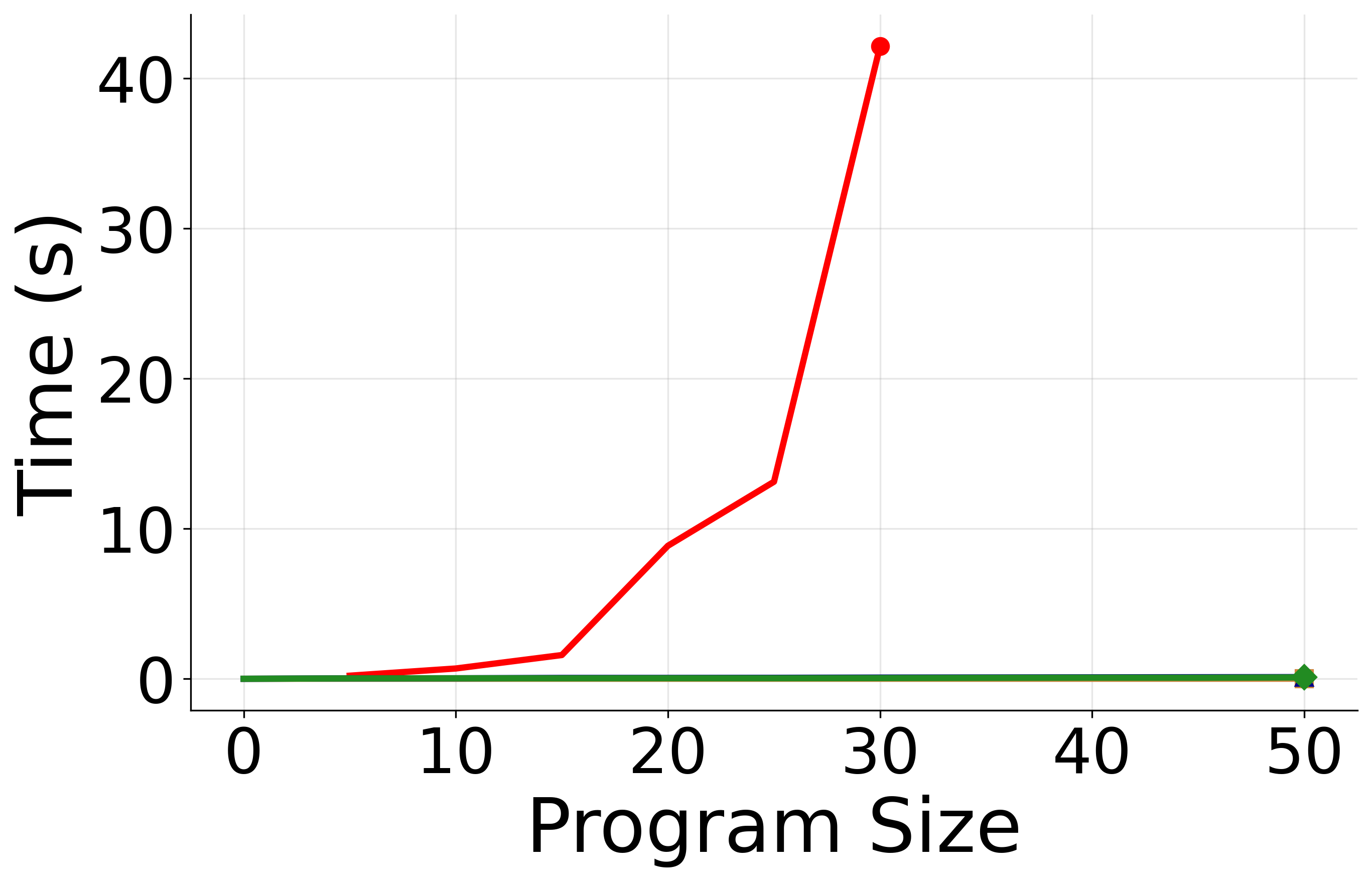}
        \caption{Cond. Ind. -- Rand.}
        \label{fig:cond-benchmarks-c-main}
    \end{subfigure}\hfill
    \begin{subfigure}{0.24\textwidth}
        \centering
        \includegraphics[width=\linewidth]{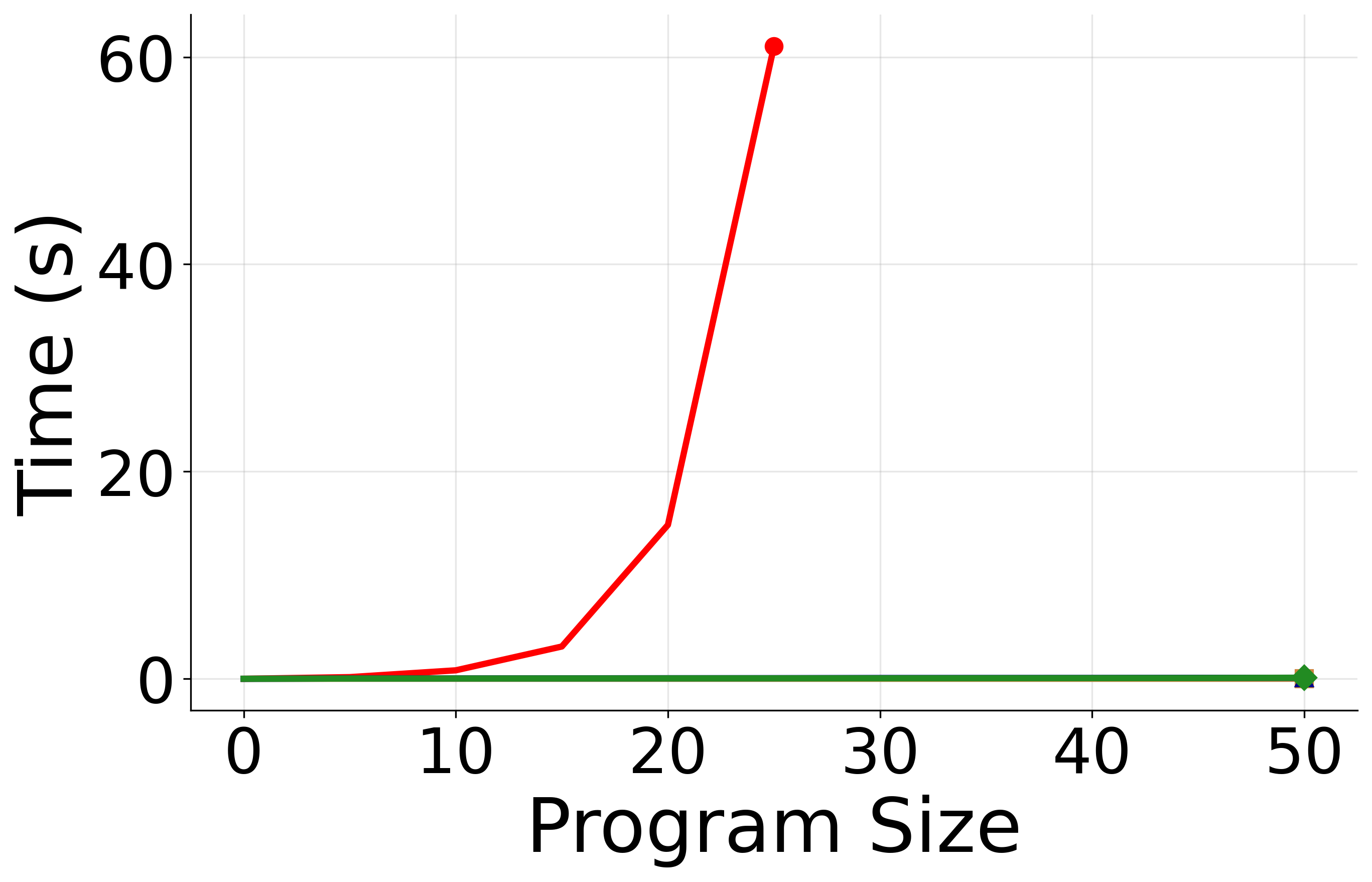}
        \caption{Alt. Guard}
        \label{fig:alt-benchmarks-a-main}
    \end{subfigure}\hfill
    \begin{subfigure}{0.24\textwidth}
        \centering
        \includegraphics[width=\linewidth]{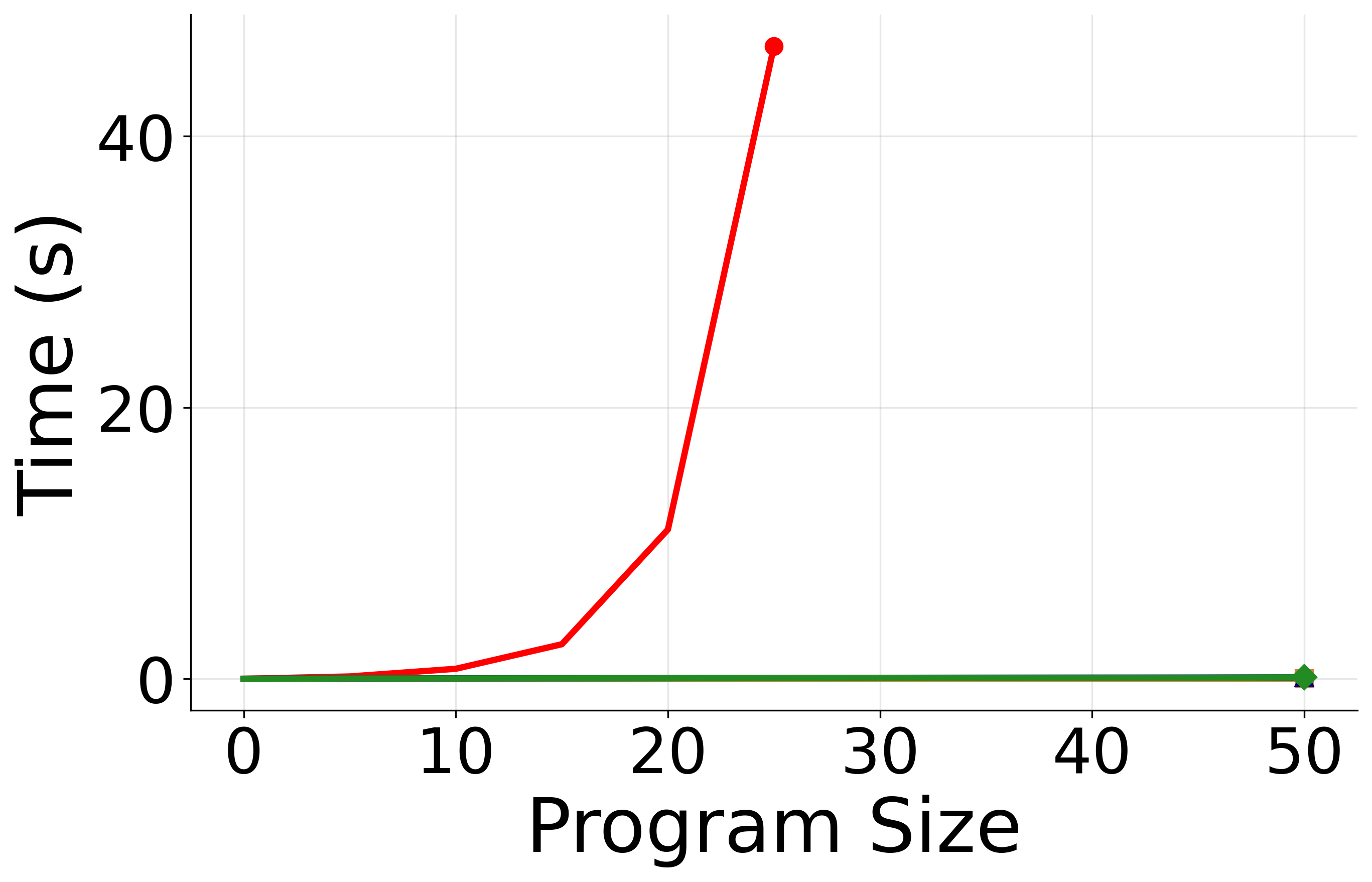}
        \caption{Alt. Guard -- Rand.}
        \label{fig:alt-benchmarks-d-main}
    \end{subfigure}

\caption{Scaling behavior on representative benchmark families (additional families in~\Cref{app:synthetic-sketches}). The plots emphasize asymptotic growth-rate differences and the resulting timeout frontier as program size increases. Time-out=120s; a disappearing line indicates timeout.}
    \label{fig:alt-benchmarks-main}
\end{figure}

\subsection{\Slicebf+\Dicebf~and \Slicebf+\Roulettebf: Existing Benchmarks}\label{sec:fairness-benchmarks}

We next use benchmark suites from prior work, detailed in \Cref{tab:benchmarks_dt,tab:benchmarks_dt_extra,tab:benchmarks_psi} in \Cref{app:benchmarks}.

The decision tree benchmarks from~\cite{albarghouthi2017fairsquare} assess fairness properties under different population models. In~\Cref{tab:benchmarks_dt,tab:benchmarks_dt_extra}, DT$_n$ denotes a tree with $n$ conditionals, the second column specifies the population model, and the remaining columns report runtimes. We split \SPPL{} execution into translation and fairness-judgment phases to make phase costs explicit. We additionally report runtimes for \PSI{} and \Hakaru{}. For \Hakaru{}, we separate Disintegrate, Simplify, and Total runtimes. Hakaru's disintegration phase transforms the original probabilistic program into a new Hakaru program representing a conditional distribution, which is then simplified by Maple to perform exact inference through computer algebra. Thus, the Total column reports the timing for the full exact inference pipeline. We observe that discretization itself (\Slice) becomes expensive for larger $n$, eventually yielding a timeout. Regarding end-to-end runtimes, \Slice{}+\Dice{} mostly outperforms \SPPL{}, whereas \Slice{}+\Roulette{} is mostly inferior to \SPPL{} for larger $n$.

\Cref{tab:benchmarks_psi} reports on benchmarks drawn from the \PSI{} and \SPPL{} benchmark sets. For benchmarks adapted by \SPPL{} from \PSI{}, we use the modified versions provided in the \SPPL{} Github repo, which replace certain continuous choices with finite ones. Translating these versions to \Slice{} enables a direct comparison between \Slice{} and \SPPL{}. For completeness, we also report \PSI{}/\Hakaru{} results on the original, unmodified programs, which remain outside the class supported by both \Slice{} and \SPPL{}. Moreover, programs requiring numeric transformations or arithmetic \emph{remain outside the language features of \Slice}. Like in the decision tree benchmarks, we report runtimes for (i) \Slice{}+\Dice{}, (ii) \Slice{}+\Roulette{}, (iii) \SPPL{}, (iv) \PSI{}, (v) \Hakaru{} (Disintegrate), and (v) \Hakaru{} (Simplify) in seconds. We observe that \Slice{}+\Dice{} and \Slice{}+\Roulette{} are competitive with \SPPL{}, \PSI{}, and \Hakaru{} (Disintegrate + Simplify).

\section{Related and Future Work}
\label{sec:related}
We start with related work and conclude with future work.

\smallskip
\noindent\textbf{Discretization approaches.}
\citet{Holtzen2018Abstraction} are the closest prior work on exact discretization. They develop a probabilistic predicate abstraction framework that, given predicates of interest, constructs an abstract probabilistic program and proves soundness for events expressible in the predicate domain. For a first-order, loop-free fragment of our language, \Slice{} can be viewed similarly: inferred cuts such as $x<c$ and $x\le c$ induce predicates over which the program is discretized. The key difference is that \Slice{} \emph{automatically infers and propagates} the relevant cut/value sets through higher-order functions, recursion, and observations, then compiles to a discrete program for existing exact backends. By contrast, \citet{Holtzen2018Abstraction} assume predicates are supplied a priori, explicitly identify useful predicate discovery as a hard problem, and restrict their semantics to measurable functions, excluding arbitrary unbounded loops and higher-order functions. Beyond this, other discretization methods use approximate discretization~\cite{Huang2021AQUA,Garg2024BitBlast,Claret2013BayesianDataFlow}. Recent work also uses approximate discretization to compute posterior bounds for recursive higher-order programs~\cite{Beutner2022GuaranteedBounds}. Vice versa, Leios~\cite{Laurel2020ContinualizationCorrection} approximates a fully continuous program from a discrete one. 

\smallskip
\noindent\textbf{Discrete-only probabilistic languages.}
Roulette~\cite{Moy2025Roulette} extends Racket with first-class support for finitely-supported distributions and leverages symbolic evaluation to compile queries into weighted model-counting problems, enabling scalable exact inference on discrete programs. Dice~\cite{Holtzen2020Dice} follows a similar philosophy in an OCaml DSL, compiling discrete programs to weighted model counting. 
Pluck~\cite{Bowers2025StochasticLazyKnowledgeCompilation} enables exact and approximate inference in discrete probabilistic programs based on lazy evaluation. Earlier work in probabilistic logic programming, most prominently ProbLog~\cite{DeRaedt2007ProbLog}, annotates Prolog facts with probabilities and reduces inference to weighted Boolean formulas. PERPL is a system that performs exact inference on recursive calls over recursive data~\cite{Chiang2023PERPL}. Caesar~\cite{Schroer2023DeductiveVerificationInfrastructure} is a quantitative program verification infrastructure for discrete probabilistic programs that together with the approach from \cite{Batz2025FoundationsDeductiveVerification} enables the verification of continuous programs. 

\smallskip
\noindent\textbf{Exact inference on continuous programs.} 
Many existing systems perform exact inference directly in continuous/hybrid languages. Psi~\cite{gehr2016psi} and Hakaru~\cite{Narayanan2016Hakaru} support symbolic exact inference capabilities with different algebraic strengths, though symbolic generality can make some tractable instances challenging. Symbo~\cite{Martires2019ExactApproxWMI} leverages weighted model integration~\cite{Belle2015ProbabilisticInferenceHybridWMI} for exact inference in hybrid domains, using sentinel decision diagrams as the representation and the PSI solver to symbolically integrate over continuous variables. SPPL~\cite{Saad2021SPPL} emphasizes efficient exact symbolic inference for its fragment and supports features we do not currently automate, notably numeric transformations and conditioning on measure-zero events. $\lambda$PSI~\cite{Gehr2020LambdaPSI} supports exact inference for higher-order continuous programs with nested inference and bounded recursion. Generating-function systems such as Genfer~\cite{Zaiser2023Genfer}, Geni~\cite{Li2025Geni}, and Prodigy~\cite{Klinkenberg2024Prodigy} collectively support exact inference for infinite-support models and restricted unbounded-control models, but not general higher-order programs. PERPL~\cite{Chiang2023PERPL} supports higher-order functions and unbounded recursion, but not continuous random variables. GuBPI~\cite{Beutner2022GuaranteedBounds} supports recursive higher-order programs with continuous distributions, but computes guaranteed converging bounds rather than exact posteriors. \Slice{} targets the combination needed for our case studies: exact inference for programs that combine continuous sampling, higher-order functions, and unbounded recursion.

\smallskip
\noindent\textbf{Relational proof techniques.}
Our soundness argument uses a coupling-style semantic relation between source and discretized executions. This connects directly to logical-relations and relational-semantics techniques for probabilistic programs~\cite{Bizjak2015Step,Wand2018Contextual,Culpepper2017ContextualScoring}, adapted here to a global, non-local, type-guided transformation. Notably,~\cite{Bizjak2015Step} considers only discrete probabilistic programs without conditioning, while~\cite{Wand2018Contextual,Culpepper2017ContextualScoring} support continuous distributions and scoring. More recently, ~\cite{Zhang2022Reasoning} considered probabilistic programs with nested queries and recursion using a step-indexed model.

\paragraph{Future Work}
We would like to extend our work to support tractable symbolic integration for richer classes of continuous distributions (i.e. conjugate priors or other distributions admitting closed-form expressions), thus broadening the class of programs that \Slice{} can discretize exactly. We also plan to investigate symbolic techniques à la \citet{Wang2024StaticPosterior} to speed-up inference in the presence of recursion.

\section*{Data-Availability Statement}
A reproducible artifact for the evaluation in~\Cref{sec:implem_eval} is available on Zenodo~\cite{Wu2026SliceArtifact}. 

\begin{acks}
We are grateful to our OOPSLA reviewers who helped us improve the paper. We also thank Cheng Zhang for providing feedback on an initial draft of our paper, as well as the Silva Lab retreat guests for helpful discussions. This work was supported by the NSF Graduate Research Fellowship Program under Grant No. DGE-2139899. Additionally, this research was developed with funding from the Defense Advanced Research Projects Agency (DARPA) and the
Office of Naval Research (ONR). The views, opinions, and/or findings contained in this material are those of the authors
and should not be interpreted as representing the official views or policies of the Department of the Navy, DARPA, or the U.S. Government. 
\end{acks}

\bibliography{refs}

\addtocontents{toc}{\protect\setcounter{tocdepth}{2}}
\clearpage
\appendix
\tableofcontents

\section{\Slicebf{} Examples (Continued)}\label{ap:mot}
This section is a continuation of~\Cref{sec:examples}, describing more in detail the following examples: \Cref{sec:examples:4,sec:examples:5,sec:examples:6,sec:examples:7}. Additionally, we include an extra example based on the Indian GPA problem, a canonical example in the probabilistic programming literature~\cite{Wu2018DiscreteContinuousMixtures}.

\subsection{Control Flow}\label{app:examples:4}

We revisit the control-flow example from \Cref{sec:examples:4}:
\begin{lstlisting}[aboveskip=1em,belowskip=1em,escapechar=!,gobble=4]
    let x = if flip() then uniform(0,2) else gaussian(0,1) in
    let y = if 1.5 > x then 1.8 else 0.3 in x <= y
\end{lstlisting}
The main feature of this program is that the comparison $x \leq y$ is governed by control flow: the right-hand side is not a single constant, but may evaluate to either $0.3$ or $1.8$ depending on the outcome of the earlier test $1.5 > x$. Consequently, the discretization of $x$ must account for both these values.

\begin{figure}[ht]
  \centering
  \begin{minipage}[c]{\textwidth}
    \begin{lstlisting}[escapeinside={<@}{@>},aboveskip=0.1em,belowskip=0.1em]
    let x = if flip() <@\textcolor{purple}{then}@> uniform(0,2) <@\textcolor{purple}{else}@> gaussian(0,1) in
    let y = if 1.5 <@\textcolor{red}{>}@> x <@\textcolor{purple}{then}@> 1.8 <@\textcolor{purple}{else}@> 0.3 in x <@\textcolor{orange}{<=}@> y
    \end{lstlisting}      
  \end{minipage}
        
  \vspace{0.8em}

  \begin{minipage}[c]{\textwidth}
  \centering
  \scalebox{0.58}{
  \begin{tikzpicture}[
      floatnode/.style={draw, rounded corners, minimum width=3.3cm, minimum height=1.1cm, align=center},
      flow/.style={->, very thick, draw=purple},
      constraintgt/.style={->, very thick, dashed, draw=red},
      constraintleq/.style={->, very thick, dashed, draw=orange}
  ]


  \node[floatnode] (u) at (-6, 2) {
      $\uniform(0,2)$ \\
      \phantom{$\floattype{\{\cleq{0.3}, \clt{1.5}, \cleq{1.8}\}}{\top}$}
  };

  \node[floatnode] (g) at (-6, -2) {
      $\gaussian(0,1)$ \\
      \phantom{$\floattype{\{\cleq{0.3}, \clt{1.5}, \cleq{1.8}\}}{\top}$}
  };

  \node[floatnode] (x) at (-1, 0) {
      $x$ \\
      \phantom{$\floattype{\{\cleq{0.3}, \clt{1.5}, \cleq{1.8}\}}{\top}$}
  };

  \node[floatnode] (c15) at (-1, -2) {
      1.5 \\
      \phantom{$\floattype{\{\cleq{0.3}, \clt{1.5}, \cleq{1.8}\}}{\{1.5\}}$}
  };

  \node[floatnode] (y) at (8, 0) {
      $y$ \\
      \phantom{$\floattype{\{\cleq{0.3}, \clt{1.5}, \cleq{1.8}\}}{\{0.3,1.8\}}$}
  };

  \node[floatnode] (c18) at (4.0, 2) {
      1.8 \\
      \phantom{$\floattype{\{\cleq{0.3}, \clt{1.5}, \cleq{1.8}\}}{\{1.8\}}$}
  };

  \node[floatnode] (c03) at (4.0, -2) {
      0.3 \\
      \phantom{$\floattype{\{\cleq{0.3}, \clt{1.5}, \cleq{1.8}\}}{\{0.3\}}$}
  };


  \draw[flow] (u) -- node[below left] {\small then-branch flow} (x);
  \draw[flow] (g) -- node[above left] {\small else-branch flow} (x);

  \draw[flow] (c18) -- node[above right] {\small then-branch flow} (y);
  \draw[flow] (c03) -- node[below right] {\small else-branch flow} (y);


  \draw[constraintgt] (c15) -- node[right] {$>$} (x);
  \draw[constraintleq] (x) -- node[above] {$\leq$} (y);

  \end{tikzpicture}
  }
  \end{minipage}
\caption{Constraint graph for float types. Solid arrows represent data flow, and dashed arrows constraints.}
\label{fig:constraint-graph-gaussian}
\end{figure}
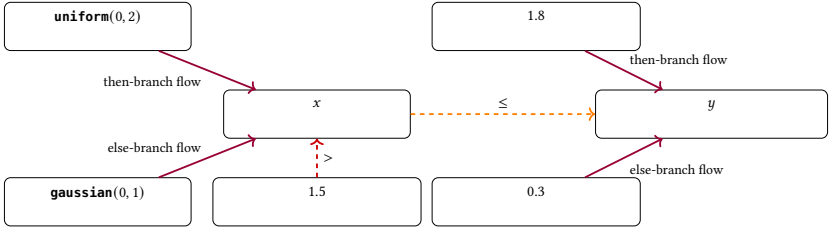

We show, at a high-level, how the float type information for this program is propagated. First, \Cref{fig:constraint-graph-gaussian} makes these dependencies explicit using a constraint graph over float-typed subexpressions. Solid edges denote data flow through conditionals, while dashed edges denote comparison constraints. In this example, the comparisons $1.5 > x$ and $x \leq y$ jointly determine the cut structure relevant to $x$. The data flow edges record how continuous samples from $\uniform(0,2)$ and $\gaussian(0,1)$ flow into $x$, while the constants $1.8$ and $0.3$ flow into $y$ through its two branches.

\input{figures/cf-propagation.tex}

The propagation process of float type information proceeds in four stages (\Cref{fig:control-flow-gaussian}). First, literal constants contribute singleton value sets, while continuous samples are assigned value set $\top$ (\Cref{fig:control-flow-gaussian-a}). Second, value sets propagate forward along data-flow edges, yielding in particular that $y$ has value set $\{0.3,1.8\}$ (\Cref{fig:control-flow-gaussian-b}). Third, comparison sites generate cut information: $1.5 > x$ contributes the threshold $1.5$, and $x \leq y$ contributes the possible values of $y$, namely $0.3$ and $1.8$ (\Cref{fig:control-flow-gaussian-c}). Since both operands of a comparison must be discretized against compatible interval boundaries, these cut sets are unified across the compared expressions. Finally, the resulting cut set propagates backward along data flow edges to all contributing subexpressions (\Cref{fig:control-flow-gaussian-d}).

The final cut set is
\[
\{\cleq{0.3}, \clt{1.5}, \cleq{1.8}\},
\]
inducing the intervals $(-\infty,0.3]$, $(0.3,1.5)$, $[1.5,1.8]$, and $(1.8,\infty)$.

\subsection{Discrete Latent Variable}\label{app:examples:5}

We revisit the discrete latent variable example from \Cref{sec:examples:5}:
\begin{lstlisting}[escapeinside={<@}{@>},aboveskip=1em,belowskip=1em]
    let x = if flip() then 0.5 else 1.5 in gaussian(0, x) < 0.5
\end{lstlisting}
The essential feature here is that the parameter of the continuous distribution is itself determined by a discrete choice. Thus, the discretization of the program must also account for the fact that the Gaussian is instantiated at two distinct parameter values.

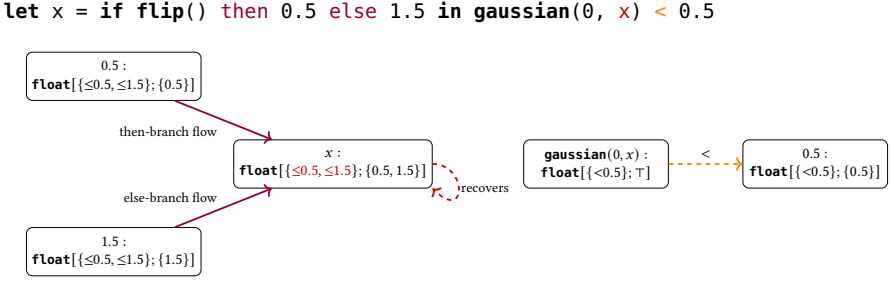
\begin{figure}[ht]
    \centering
  
    \begin{minipage}[c]{\textwidth}
      \begin{lstlisting}[escapeinside={<@}{@>},aboveskip=0.1em,belowskip=0.1em]
      let x = if flip() <@\textcolor{purple}{then}@> 0.5 <@\textcolor{purple}{else}@> 1.5 in gaussian(0, <@\textcolor{red}{x}@>) <@\textcolor{orange}{<}@> 0.5
      \end{lstlisting}
    \end{minipage}
  
    \vspace{0.8em}
  
    \begin{minipage}[c]{\textwidth}
        \centering
        \scalebox{0.58}{
        \begin{tikzpicture}[
            floatnode/.style={draw, rounded corners, minimum width=3.3cm, minimum height=1.1cm, align=center},
            flow/.style={->, very thick, draw=purple},
            constraintlt/.style={->, very thick, dashed, draw=orange}
        ]
        
        
        \node[floatnode] (v05) at (-6, 2) {
            0.5 : \\
            $\floattype{\{\cleq{0.5}, \cleq{1.5}\}}{\{0.5\}}$
        };
        
        \node[floatnode] (v15) at (-6, -2) {
            1.5 : \\
            $\floattype{\{\cleq{0.5}, \cleq{1.5}\}}{\{1.5\}}$
        };
        
        
        \node[floatnode] (x) at (-1, 0) {
            $x$ : \\
            $\floattype{\{\textcolor{red}{\cleq{0.5}}, \textcolor{red}{\cleq{1.5}}\}}{\{0.5,1.5\}}$
        };
        
        
        \node[floatnode] (g) at (5, 0) {
            $\gaussian(0, x)$ : \\
            $\floattype{\{\clt{0.5}\}}{\top}$
        };
        
        
        \node[floatnode] (c05) at (10, 0) {
            0.5 : \\
            $\floattype{\{\clt{0.5}\}}{\{0.5\}}$
        };
        
        
        \draw[flow] (v05) -- node[below left] {\small then-branch flow} (x);
        \draw[flow] (v15) -- node[above left] {\small else-branch flow} (x);
        
        
        \draw[flow, color=red, dashed, very thick]
            (x.east) .. controls +(0.8,0) and +(0.8,-0.8) ..
            node[right, text=black, inner sep=1pt] {\small recovers}
            (x.south east);
        
        
        \draw[constraintlt] (g) -- node[above] {$<$} (c05);
        
        \end{tikzpicture}
        }
    \end{minipage}
  
    \caption{Constraint graph populated with float type information. As before, solid arrows represent data flow, and dashed arrows represent comparison constraints. Using $x$ as a Gaussian parameter induces a \emph{recovers} constraint, the red dashed self-loop, which ensures a unique discrete index for each value of $x$.}
    \label{fig:control-flow-distinguishes}
 \end{figure}

\Cref{fig:control-flow-distinguishes} presents the corresponding constraint graph. The variable $x$ is formed by control flow from the two constants $0.5$ and $1.5$, and therefore has value set $\{0.5,1.5\}$. Because $x$ appears as a parameter to $\gaussian(0,x)$, \Slice{} must recover which concrete value of $x$ was selected in order to instantiate the appropriate Gaussian distribution. The figure records this dependency by the \emph{recovers} constraint on $x$.

Operationally, the recovers constraint requires that the values in the value set of $x$ be represented by distinct discrete indices. This is why the type system uses non-strict cuts for such values: the cut set $\{\cleq{0.5},\cleq{1.5}\}$ induces a partition in which $0.5$ and $1.5$ occupy distinguishable regions, allowing them to be mapped to separate cases in the discretized program. The final comparison $\gaussian(0,x) < 0.5$ contributes the cut $\clt{0.5}$ for the sampled Gaussian result, while the value-set information on $x$ determines which discretized Gaussian is selected.

\subsection{Indian GPA Problem}\label{ap:gpa}

This example involves both discrete and continuous random variables and models a student's GPA based on nationality (India or USA) and whether they have a perfect record. The GPA is deterministic for perfect students (10.0 for India, 4.0 for USA) but drawn from a uniform distribution otherwise. \Slice{} correctly handles conditioning on a positive-measure atom at 10:

\begin{lstlisting}[aboveskip=1em,belowskip=1em,escapechar=!]
  let nationality = discrete(0.5, 0.5) in
  let perfect = discrete(0.01, 0.99) in
  let gpa = if nationality <= 0 then
          if perfect <= 0 then
            (10.0 : !{\ftb{\{<10,<=10\}}{\{10\}}}!)
          else
            (uniform(0, 10) : !{\ft{\{<10,<=10\}}}!)
        else
          if perfect <= 0 then
            (4.0 : !{\ftb{\{<10,<=10\}}{\{4\}}}!)
          else
            (uniform(0, 4) : !{\ft{\{<10,<=10\}}}!)
        : !{\ft{\{<10,<=10\}}}! in
  let _ = observe(gpa : !{\ft{\{<10,<=10\}}}! <= 10 : !{\ftb{\{<10,<=10\}}{\{10\}}}!
            && gpa : !{\ft{\{<10,<=10\}}}! >= 10 : !{\ftb{\{<10,<=10\}}{\{10\}}}!) in 
  nationality > 0
\end{lstlisting}

\noindent
The type system propagates the cuts $\clt{10}$ and $\cleq{10}$ for all GPA expressions. These cuts induce the three regions:
\[
(-\infty,10), \qquad \{10\}, \qquad (10,\infty).
\]
Under this partition, both continuous branches, \texttt{uniform(0,10)} and \texttt{uniform(0,4)}, place all of their probability mass in the first region. The constant 10.0 is mapped to the singleton region $\{10\}$, while the constant 4.0 is mapped to the region $(-\infty,10)$. The resulting discretized program is therefore:

\begin{lstlisting}[aboveskip=1em,belowskip=1em]
  let nationality = discrete(0.5, 0.5) in
  let perfect = discrete(0.01, 0.99) in
  let gpa = if nationality <= 0 then
          if perfect <= 0 then 1
          else
            discrete(1.0, 0.0, 0.0)
        else
          if perfect <= 0 then 0
          else
            discrete(1.0, 0.0, 0.0) in
  let _ = observe(gpa <= 1 && gpa >= 1)
  nationality > 0
\end{lstlisting}

\noindent
After discretization, the equality test $\texttt{gpa} = 10$ becomes the discrete test that the GPA index is exactly the one corresponding to the singleton region $\{10\}$. 

\subsection{Point and Measure-Zero Observations}\label{ap:measure-zero-observations}

\Slice{} can condition on an exact value when that event has strictly positive probability mass, as in the Indian GPA example in \Cref{ap:gpa}. In that case, conditioning is well-defined and \Slice{} returns the correct answer. The exact set of programs supported by \Slice{} is therefore somewhat larger than one would initially expect, while not covering full point observations on continuous samples.

The mechanism is the same as in \Cref{ap:gpa}: the cuts $\clt{c}$ and $\cleq{c}$ induce the partition
\[
(-\infty,c), \qquad \{c\}, \qquad (c,\infty),
\]
so any probability mass assigned exactly to $c$ is tracked as its own discrete outcome.

More precisely:

\begin{enumerate}
  \item \Slice{} cannot handle programs where every control-flow path conditions on an event of measure zero, and therefore programs with exact observations on continuous samples are in general out of scope, i.e. the programs one typically writes in e.g. Stan are out of scope unless one widens exact conditioning to a small but finite interval.

  \item \Slice{} can handle programs where at least one control-flow path conditions only on strictly-positive-probability events, even if that is a point observation, and even if other paths condition only on measure-zero events.

  \item It does not matter whether an event is measure-zero because it is point-conditioning on a continuous sample, or whether it is measure-zero for another reason, e.g. $\observekw(\uniform(0,1) > 3)$ or simply $\observekw(\false)$. The critical distinction for $\observekw(E)$ is whether $\Pr(E)=0$, not whether $E$ is of the form ``$x=3.2$''.

  \item \Slice{} will never give the wrong answer, even when all paths condition on measure-zero events. It will give the empty subdistribution, indicating that under a rejection-sampling interpretation, all samples are rejected with probability~$1$. Thus the user is notified, and \Slice{} never silently gives a wrong answer.
\end{enumerate}

For instance, this program:
\begin{lstlisting}[aboveskip=1em,belowskip=1em]
  let x = uniform(0,1) in
  if flip() then
    let _ = observe(x >= 0.5 && x <= 0.5) in true
  else
    let _ = observe(x <= 0.3 && x >= 0.3) in false
\end{lstlisting}
results in:
\[
\Pr(\true) = 0
\qquad
\Pr(\false) = 0.
\]

Whereas this program:
\begin{lstlisting}[aboveskip=1em,belowskip=1em]
  let x = uniform(0,1) in
  if flip() then
    let _ = observe(x >= 0.5 && x <= 0.5) in true
  else
    let _ = observe(x <= 0.3001 && x >= 0.3) in false
\end{lstlisting}
results in:
\[
\Pr(\true) = 0
\qquad
\Pr(\false) = 1.
\]
This is the correct behavior: consider repeatedly running the program under a rejection-sampling interpretation until an observation succeeds. With probability~$1$, the first successful test will be $\observekw(x \leq 0.3001 \mathrel{\&\&} x \geq 0.3)$ and the program will therefore return $\false$ with probability~$1$.

\subsection{Random Interval Covering Over $\boldmath{[0,1]}$}\label{ap:epsball}
The discretized program for the case study from \Cref{sec:examples:7} is shown in \Cref{fig:app:epsball}. \Slice{} takes 0.000s, and we generate a Markov chain consisting of $192365$ states in $38.69$ seconds. \Storm takes $0.69$ seconds for model checking, and reports that the sought-after probability is $0.034474...$.

\begin{figure}[t]
    \centering
    \begin{minipage}[t]{0.46\textwidth}
    \vspace{0pt}
    \begin{lstlisting}[aboveskip=0.4em,belowskip=0.4em,
      language=OCaml,
      escapechar=!,
      numbers=right,
      stepnumber=1,
      numberstyle=\color{gray}\scriptsize,
      numbersep=7pt,
      basicstyle=\ttfamily\footnotesize,
      commentstyle=\color{ForestGreen}\itshape,
      gobble=4]
    (* Returns true if y is in [lo,hi) *)
    let within = fun lo -> fun hi -> fun y -> !\linelabel{ap-ln:within}!
      (y >= lo) && (y < hi) in 
    
    (* Function for sampling interval bounds *)
    let b = fun _ ->  !\linelabel{ap-ln:b}!
      discrete(0.18: 0.05, 0.11: 0.24, 
              0.21: 0.37, 0.09: 0.58, 
              0.41: 0.89) in
    
    (* Initialize intervals by sampling their 
    respective (possibly overlapping) bounds *)
    let init =!\linelabel{ap-ln:init}!
      (false, within 1 (b ())) ::
      (false, within (b ()) (b ())) ::
      (false, within (b ()) (b ())) ::
      (false, within (b ()) 8) ::
      nil in 
    
    (* Update flags depending on y falling 
       within the respective interval *)
    let update = fix update y := fun ys -> !\linelabel{ap-ln:update}!
      match ys with 
        | nil -> nil
        | h :: t ->
            ((fst h || ((snd h) y)), snd h)
            :: (update y t)
      end in
    \end{lstlisting}
    \end{minipage}
    \hfill
    \hspace{-4pt}%
    \begin{minipage}[t]{0.46\textwidth}
    \vspace{0pt}
    \begin{lstlisting}[aboveskip=0.4em,belowskip=0.4em,
      language=OCaml,
      escapechar=!,
      numbers=right,
      stepnumber=1,
      firstnumber=last,
      numberstyle=\color{gray}\scriptsize,
      numbersep=7pt,
      basicstyle=\ttfamily\footnotesize,
      commentstyle=\color{ForestGreen}\itshape,
      gobble=4]
    (* Sample from a discretized gaussian, 
        then use the resulting value to update
        flags *)
    let step = fun ys -> 
       update (discrete(0.5, 0.0199388, 
                        0.0748961, 0.0494739, 
                        0.0747339, 0.0632619, 
                        0.0309625, 0.0280777, 
                        0.158655)) ys in
     
     (* Number of samples determined by 
        unbounded recursion *)
     let cover = fix cover st :=
       if discrete(0.1, 0.9) < 1 then st
       else cover (step st) in
    
    (* Have all intervals been covered at 
       least once? *)
    let is_covered = fix is_covered ys := !\linelabel{ap-ln:iscovered}!
      match ys with 
        | nil -> true 
        | h :: t -> (fst h) && (is_covered t)
      end in 
    
    (* Run entire process on sampled intervals *)
    let final = cover init in
      is_covered final !\linelabel{ap-ln:final}!
    \end{lstlisting}
    \end{minipage}%
    \hspace{4pt}
    
    \caption{A \emph{discretized} infinite stochastic process for random interval covering over $[0,1]$, involving continuous sampling, higher-order recursion, and lists.}
    \label{fig:app:epsball}
\end{figure}

\section{\Slicebf{} Case Studies (continued)}

\subsection{Random Interval Covering Over $\boldmath{[0,1]}$ with Conditioning}\label{sec:epsball-observe}
The full program for the case study from \Cref{sec:casestudy:epsball} is shown in \Cref{fig:app:epsball-observe}. The discretized program is shown in \Cref{fig:app:epsball-observe:disc}.

\begin{figure}[t]
    \centering
    \begin{minipage}[t]{0.46\textwidth}
    \vspace{0pt}
    \begin{lstlisting}[aboveskip=0.4em,belowskip=0.4em,
      language=OCaml,
      escapechar=!,
      numbers=right,
      stepnumber=1,
      numberstyle=\color{gray}\scriptsize,
      numbersep=7pt,
      basicstyle=\ttfamily\footnotesize,
      commentstyle=\color{ForestGreen}\itshape]
(* Returns true if y is in [lo,hi) *)
let within = fun lo -> fun hi -> fun y ->
  (y >= lo) && (y < hi) in 

(* Function for sampling interval bounds *)
let b = fun _ ->
  discrete(0.18: 0.05, 0.11: 0.24,
          0.21: 0.37, 0.09: 0.58,
          0.41: 0.89) in

(* Initialize intervals by sampling their
   respective (possibly overlapping) bounds *)
let init =
  (false, within 0.00 (b ())) ::
  (false, within (b ()) (b ())) ::
  (false, within (b ()) (b ())) ::
  (false, within (b ()) 1.00) ::
  nil in 

(* Update flags depending on y falling 
   within the respective interval *)
let update = fix update y := fun ys ->
  match ys with
    | nil -> nil
    | h :: t ->
        ((fst h || ((snd h) y)), snd h)
        :: (update y t)
  end in

(* Ensure sampled y hits at most one 
   interval *)
let non_overlap = 
  fix non_overlap y := fun seen -> fun ys ->
    match ys with
      | nil -> true
      | h :: t ->
          let hit = (snd h) y in
          if hit && seen then false
          else non_overlap y (seen || hit) t
    end in
    \end{lstlisting}
    \end{minipage}
    \hfill
    \hspace{-4pt}%
    \begin{minipage}[t]{0.46\textwidth}
    \vspace{0pt}
    \begin{lstlisting}[aboveskip=0.4em,belowskip=0.4em,
      language=OCaml,
      escapechar=! ,
      numbers=right,
      stepnumber=1,
      firstnumber=last,
      numberstyle=\color{gray}\scriptsize,
      numbersep=7pt,
      basicstyle=\ttfamily\footnotesize,
      commentstyle=\color{ForestGreen}\itshape]
(* Use the sampled value to update flags *)
let step = fun y -> fun ys ->
  update y ys in

(* Number of samples determined by
   unbounded recursion *)
let cover = fix cover st :=
  if uniform(0,1) < 0.1 then st
  else
    let y = gaussian(0,1) in
    let _ = observe(non_overlap y false st) in
    cover (step y st) in

(* Have all intervals been covered
   at least once? *)
let is_covered = fix is_covered ys :=
  match ys with
    | nil -> true
    | h :: t -> (fst h) && (is_covered t)
  end in

(* Run the process and query whether
   the whole interval was covered *)
let final = cover init in
  is_covered final
    \end{lstlisting}
    \end{minipage}%
    \hspace{4pt}
    
\caption{An infinite stochastic process for random interval covering over $[0,1]$, involving conditioning, continuous sampling, higher-order recursion, and lists.}
    \label{fig:app:epsball-observe}
\end{figure}

\begin{figure}[t]
    \centering
    \begin{minipage}[t]{0.46\textwidth}
    \vspace{0pt}
    \begin{lstlisting}[aboveskip=0.4em,belowskip=0.4em,
      language=OCaml,
      escapechar=!,
      numbers=right,
      stepnumber=1,
      numberstyle=\color{gray}\scriptsize,
      numbersep=7pt,
      basicstyle=\ttfamily\footnotesize,
      commentstyle=\color{ForestGreen}\itshape]
(* Returns true if y is in [lo,hi) *)
let within = fun lo -> fun hi -> fun y ->
  (y >= lo) && (y < hi) in 

(* Function for sampling interval bounds *)
let b = fun _ ->
  discrete(0.18: 0.05, 0.11: 0.24,
          0.21: 0.37, 0.09: 0.58,
          0.41: 0.89) in

(* Initialize intervals by sampling their
   respective (possibly overlapping) bounds *)
let init =
  (false, within 1 (b ())) ::
  (false, within (b ()) (b ())) ::
  (false, within (b ()) (b ())) ::
  (false, within (b ()) 8) ::
  nil in 

(* Update flags depending on y falling 
   within the respective interval *)
let update = fix update y := fun ys ->
  match ys with
    | nil -> nil
    | h :: t ->
        ((fst h || ((snd h) y)), snd h)
        :: (update y t)
  end in

(* Ensure sampled y hits at most one 
   interval *)
let non_overlap = 
  fix non_overlap y := fun seen -> fun ys ->
    match ys with
      | nil -> true
      | h :: t ->
          let hit = (snd h) y in
          if hit && seen then false
          else non_overlap y (seen || hit) t
    end in
    \end{lstlisting}
    \end{minipage}
    \hfill
    \hspace{-4pt}%
    \begin{minipage}[t]{0.46\textwidth}
    \vspace{0pt}
    \begin{lstlisting}[aboveskip=0.4em,belowskip=0.4em,
      language=OCaml,
      escapechar=! ,
      numbers=right,
      stepnumber=1,
      firstnumber=last,
      numberstyle=\color{gray}\scriptsize,
      numbersep=7pt,
      basicstyle=\ttfamily\footnotesize,
      commentstyle=\color{ForestGreen}\itshape]
(* Use the sampled value to update flags *)
let step = fun y -> fun ys ->
  update y ys in

(* Number of samples determined by
   unbounded recursion *)
let cover = fix cover st :=
  if discrete(0.1, 0.9) < 1 then st
  else
    let y = discrete(0.5, 0.0199388, 
                    0.0748961, 0.0494739, 
                    0.0747339, 0.0632619, 
                    0.0309625, 0.0280777, 
                    0.158655) in
    let _ = observe(non_overlap y false st) in
    cover (step y st) in

(* Have all intervals been covered
   at least once? *)
let is_covered = fix is_covered ys :=
  match ys with
    | nil -> true
    | h :: t -> (fst h) && (is_covered t)
  end in

(* Run the process and query whether
   the whole interval was covered *)
let final = cover init in
  is_covered final
    \end{lstlisting}
    \end{minipage}%
    \hspace{4pt}
    
    \caption{A \emph{discretized} infinite stochastic process for random interval covering over $[0,1]$, involving conditioning, continuous sampling, higher-order recursion, and lists.}
    \label{fig:app:epsball-observe:disc}
\end{figure}

\subsection{Bloom Filter}\label{ap:casestudy:bloom}
The full program for the case study from \Cref{sec:casestudy:bloom} is shown in \Cref{fig:app:bloom-filter}. The discretized program is shown in \Cref{fig:app:bloom-filter:disc}.

\begin{figure}[t]
    \centering
    \begin{minipage}[t]{0.46\textwidth}
    \vspace{0pt}
    \begin{lstlisting}[aboveskip=0.4em,belowskip=0.4em,
      language=OCaml,
      escapechar=!,
      numbers=right,
      stepnumber=1,
      numberstyle=\color{gray}\scriptsize,
      numbersep=7pt,
      basicstyle=\ttfamily\footnotesize,
      commentstyle=\color{ForestGreen}\itshape]
(* Set bf[idx] := true, where idx is 0..3 and
   bf = (b0, (b1, (b2, b3))) *)
let set_true = fun bf -> fun idx ->
  let b0 = fst bf in
  let tail1 = snd bf in
  let b1 = fst tail1 in
  let tail2 = snd tail1 in
  let b2 = fst tail2 in
  let b3 = snd tail2 in
  if idx <= 0 then
    (true, (b1, (b2, b3)))
  else if idx <= 1 then
    (b0, (true, (b2, b3)))
  else if idx <= 2 then
    (b0, (b1, (true, b3)))
  else
    (b0, (b1, (b2, true))) in

(* Get bf[idx], where idx is 0..3 *)
let get = fun bf -> fun idx ->
  let b0 = fst bf in
  let tail1 = snd bf in
  let b1 = fst tail1 in
  let tail2 = snd tail1 in
  let b2 = fst tail2 in
  let b3 = snd tail2 in
  if idx <= 0 then b0
  else if idx <= 1 then b1
  else if idx <= 2 then b2
  else b3 in

(* Insert element q into bf using hash
   functions h1 and h2 *)
let insert = fun q -> fun h1 
            -> fun h2 -> fun bf ->
  let index1 = h1 q in
  let index2 = h2 q in
  let bf1 = set_true bf index1 in
  let bf2 = set_true bf1 index2 in
  bf2 in

(* Search for q in bf using hash
   functions h1 and h2 *)
let search = fun q -> fun h1 
            -> fun h2 -> fun bf ->
  let index1 = h1 q in
  let index2 = h2 q in
  let b01 = get bf index1 in
  let b02 = get bf index2 in
  b01 && b02 in
    \end{lstlisting}
    \end{minipage}
    \hfill
    \hspace{-4pt}%
    \begin{minipage}[t]{0.46\textwidth}
    \vspace{0pt}
    \begin{lstlisting}[aboveskip=0.4em,belowskip=0.4em,
      language=OCaml,
      escapechar=!,
      numbers=right,
      stepnumber=1,
      firstnumber=last,
      numberstyle=\color{gray}\scriptsize,
      numbersep=7pt,
      basicstyle=\ttfamily\footnotesize,
      commentstyle=\color{ForestGreen}\itshape]
(* Hash functions *)
let h1 = fun q ->
  if q <= -2 then 0
  else if -2 < q && q <= 0 then 1
  else if 0 < q && q <= 2 then 2
  else 3 in

let h2 = fun q ->
  if q <= -2 then 3
  else if -2 < q && q <= 0 then 2
  else if 0 < q && q <= 2 then 1
  else 0 in

(* Test whether the Bloom filter is full *)
let full = fun bf ->
  let b0 = fst bf in
  let t1 = snd bf in
  let b1 = fst t1 in
  let t2 = snd t1 in
  let b2 = fst t2 in
  let b3 = snd t2 in
  b0 && b1 && b2 && b3 in

(* Initial, empty Bloom filter *)
let init = (false, (false, (false, false))) in

(* Query key *)
let key = gaussian(0,1) in

(* Insert gaussian samples until bf is full
   or the key is reported present *)
let search_before_full = fix loop bf :=
  if full bf then bf
  else if search key h1 h2 bf then bf
  else
    let q = gaussian(0,1) in
    loop (insert q h1 h2 bf) in

(* Query whether search terminated before
   the filter became full *)
let bf_final = search_before_full init in
  not (full bf_final)
    \end{lstlisting}
    \end{minipage}%
    \hspace{4pt}
    
    \caption{A 4-bit Bloom filter program with continuous samples, locality-sensitive hash functions, and search and insert functions.}
    \label{fig:app:bloom-filter}
\end{figure}

\begin{figure}[t]
    \centering
    \begin{minipage}[t]{0.46\textwidth}
    \vspace{0pt}
    \begin{lstlisting}[aboveskip=0.4em,belowskip=0.4em,
      language=OCaml,
      escapechar=!,
      numbers=right,
      stepnumber=1,
      numberstyle=\color{gray}\scriptsize,
      numbersep=7pt,
      basicstyle=\ttfamily\footnotesize,
      commentstyle=\color{ForestGreen}\itshape]
(* Set bf[idx] := true, where idx is 0..3 and
   bf = (b0, (b1, (b2, b3))) *)
let set_true = fun bf -> fun idx ->
  let b0 = fst bf in
  let tail1 = snd bf in
  let b1 = fst tail1 in
  let tail2 = snd tail1 in
  let b2 = fst tail2 in
  let b3 = snd tail2 in
  if idx <= 0 then
    (true, (b1, (b2, b3)))
  else if idx <= 1 then
    (b0, (true, (b2, b3)))
  else if idx <= 2 then
    (b0, (b1, (true, b3)))
  else
    (b0, (b1, (b2, true))) in

(* Get bf[idx], where idx is 0..3 *)
let get = fun bf -> fun idx ->
  let b0 = fst bf in
  let tail1 = snd bf in
  let b1 = fst tail1 in
  let tail2 = snd tail1 in
  let b2 = fst tail2 in
  let b3 = snd tail2 in
  if idx <= 0 then b0
  else if idx <= 1 then b1
  else if idx <= 2 then b2
  else b3 in

(* Insert element q into bf using hash
   functions h1 and h2 *)
let insert = fun q -> fun h1 
            -> fun h2 -> fun bf ->
  let index1 = h1 q in
  let index2 = h2 q in
  let bf1 = set_true bf index1 in
  let bf2 = set_true bf1 index2 in
  bf2 in

(* Search for q in bf using hash
   functions h1 and h2 *)
let search = fun q -> fun h1 
            -> fun h2 -> fun bf ->
  let index1 = h1 q in
  let index2 = h2 q in
  let b01 = get bf index1 in
  let b02 = get bf index2 in
  b01 && b02 in
    \end{lstlisting}
    \end{minipage}
    \hfill
    \hspace{-4pt}%
    \begin{minipage}[t]{0.46\textwidth}
    \vspace{0pt}
    \begin{lstlisting}[aboveskip=0.4em,belowskip=0.4em,
      language=OCaml,
      escapechar=!,
      numbers=right,
      stepnumber=1,
      firstnumber=last,
      numberstyle=\color{gray}\scriptsize,
      numbersep=7pt,
      basicstyle=\ttfamily\footnotesize,
      commentstyle=\color{ForestGreen}\itshape]
(* Hash functions *)
let h1 = fun q ->
  if q <= 0 then 0
  else if 0 < q && q <= 1 then 1
  else if 1 < q && q <= 2 then 2
  else 3 in

let h2 = fun q ->
  if q <= 0 then 3
  else if 0 < q && q <= 1 then 2
  else if 1 < q && q <= 2 then 1
  else 0 in

(* Test whether the Bloom filter is full *)
let full = fun bf ->
  let b0 = fst bf in
  let t1 = snd bf in
  let b1 = fst t1 in
  let t2 = snd t1 in
  let b2 = fst t2 in
  let b3 = snd t2 in
  b0 && b1 && b2 && b3 in

(* Initial, empty Bloom filter *)
let init = (false, (false, (false, false))) in

(* Query key *)
let key = discrete(0.0227501, 
                    0.47725, 
                    0.47725, 
                    0.0227501) in

(* Insert discretized gaussian samples until 
    bf is full or the key is reported
    present *)
let search_before_full = fix loop bf :=
  if full bf then bf
  else if search key h1 h2 bf then bf
  else
    let q = discrete(0.0227501, 0.47725, 
                    0.47725, 0.0227501) in
    loop (insert q h1 h2 bf) in

(* Query whether search terminated before
   the filter became full *)
let bf_final = search_before_full init in
  not (full bf_final)
    \end{lstlisting}
    \end{minipage}%
    \hspace{4pt}
    
    \caption{A \emph{discretized} 4-bit Bloom filter program with continuous samples, locality-sensitive hash functions, and search and insert functions.}
    \label{fig:app:bloom-filter:disc}
\end{figure}
\clearpage
\section{List of Supported \Slicebf\xspace Distributions}\label{appen:supported-distributions}
The list of distributions supported by \Slice{}, along with their probability density / mass functions, is given in~\Cref{fig:implemented-distribution-laws}. \Slice{} supports discrete distributions by applying the same cut-set and value-set analysis used for continuous distributions, \emph{provided that the discrete distribution admits a well-defined CDF}. The extension to discrete distributions does not introduce any fundamental changes to the discretization algorithm.

\begin{figure}[t]
\small
\setlength{\tabcolsep}{4pt}
\renewcommand{\arraystretch}{1.05}
\begin{tabularx}{\textwidth}{@{}>{\raggedright\arraybackslash}p{.18\textwidth} >{\raggedright\arraybackslash}p{.35\textwidth} >{\raggedright\arraybackslash}p{.14\textwidth} >{\raggedright\arraybackslash}p{.27\textwidth}@{}}
    \toprule
    \textbf{Source constructor} & \textbf{Density or PMF} & \textbf{Support} & \textbf{Parameter domain and convention} \\
    \midrule
    \multicolumn{4}{@{}l}{\textit{Continuous distributions}} \\
    \addlinespace[2pt]
    \texttt{exponential($\lambda$)} & $\lambda e^{-\lambda x}$ & $[0,\infty)$ & Rate $\lambda>0$ \\
    \texttt{uniform($a,b$)} & $\frac{1}{b-a}\mathbf{1}_{[a,b)}(x)$ & $[a,b]$ & Endpoints $a<b$ \\
    \texttt{laplace($b$)} & $\frac{1}{2b}e^{-|x|/b}$ & $\mathbb R$ & Centered; scale $b>0$ \\
    \texttt{gaussian($\mu,\sigma$)} & $\frac{1}{\sqrt{2\pi}\sigma}e^{-(x-\mu)^2/(2\sigma^2)}$ & $\mathbb R$ & Location $\mu\in\mathbb R$; standard deviation $\sigma>0$ \\
    \texttt{cauchy($\gamma$)} & $\frac{\gamma}{\pi(x^2+\gamma^2)}$ & $\mathbb R$ & Centered; scale $\gamma>0$ \\
    \texttt{beta($\alpha,\beta$)} & $\frac{x^{\alpha-1}(1-x)^{\beta-1}}{B(\alpha,\beta)}$ & $[0,1]$ & Shapes $\alpha,\beta>0$ \\
    \texttt{tdist($\nu$)} & $\frac{\Gamma((\nu+1)/2)}{\sqrt{\nu\pi}\,\Gamma(\nu/2)}(1+x^2/\nu)^{-(\nu+1)/2}$ & $\mathbb R$ & Centered; degrees of freedom $\nu>0$ \\
    \texttt{lognormal($\mu,\sigma$)} & $\frac{1}{x\sigma\sqrt{2\pi}}e^{-(\log x-\mu)^2/(2\sigma^2)}$ & $[0,\infty)$ & Log-location $\mu\in\mathbb R$; log-scale $\sigma>0$ \\
    \texttt{chi2($\nu$)} & $\frac{x^{\nu/2-1}e^{-x/2}}{2^{\nu/2}\Gamma(\nu/2)}$ & $[0,\infty)$ & Real degrees of freedom $\nu>0$ \\
    \texttt{gamma($\alpha,\theta$)} & $\frac{x^{\alpha-1}e^{-x/\theta}}{\Gamma(\alpha)\theta^\alpha}$ & $[0,\infty)$ & Shape $\alpha>0$; scale $\theta>0$ \\
    \texttt{logistic($s$)} & $\frac{e^{-x/s}}{s(1+e^{-x/s})^2}$ & $\mathbb R$ & Centered; scale $s>0$ \\
    \texttt{pareto($a,b$)} & $a b^a x^{-a-1}$ & $[b,\infty)$ & GSL shape $a>0$; threshold/scale $b>0$ \\
    \texttt{rayleigh($\sigma$)} & $\frac{x}{\sigma^2}e^{-x^2/(2\sigma^2)}$ & $[0,\infty)$ & Scale $\sigma>0$ \\
    \texttt{weibull($a,b$)} & $\frac{b}{a}(x/a)^{b-1}e^{-(x/a)^b}$ & $[0,\infty)$ & GSL scale $a>0$; exponent $b>0$ \\
    \texttt{gumbel1($a,b$)} & $a b\,e^{-(b e^{-a x}+a x)}$ & $\mathbb R$ & GSL parameters $a,b>0$ \\
    \texttt{gumbel2($a,b$)} & $a b\,x^{-a-1}e^{-b x^{-a}}$ & $[0,\infty)$ & GSL shape $a>0$; coefficient $b>0$ \\
    \texttt{exppow($a,b$)} & $\frac{1}{2a\Gamma(1+1/b)}e^{-(|x|/a)^b}$ & $\mathbb R$ & Scale $a>0$; exponent $b>0$ \\
    \midrule
    \multicolumn{4}{@{}l}{\textit{Discrete distributions}} \\
    \addlinespace[2pt]
    \texttt{poisson($\mu$)} & $e^{-\mu}\mu^k/k!$ & $\mathbb N_0$ if $\mu>0$; $\{0\}$ if $\mu=0$ & Mean $\mu\geq0$ \\
    \texttt{binomial($p,n$)} & $\binom nk p^k(1-p)^{n-k}$ & $\{0,\ldots,n\}$, with endpoint degeneracies & $0\leq p\leq1$; integer $0\leq n\leq\min(2^{32}-1,\texttt{max\_int})$ \\
    \bottomrule
\end{tabularx}
\caption{List of supported distributions.
The binomial support is $\{0\}$ when $n=0$ or $p=0$, and $\{n\}$ when
$p=1$; otherwise it is $\{0,\ldots,n\}$.}
\label{fig:implemented-distribution-laws}
\end{figure}

\section{Typing Rules for Full \Slicebf{} Language}\label{appen:typing-rules}

The typing rules for general language constructs is given in \Cref{fig:typing-general-all}.

\begin{figure}
\begin{mathpar}
    \inferrule[\textsc{Var}]
    {\ }
    {\Gamma, x: \tau \vdash x : \tau}

    \inferrule[\textsc{Let}]
    {\Gamma \vdash e_1 : \tau_1 \\
     \Gamma, x: \tau_1 \vdash e_2 : \tau_2}
    {\Gamma \vdash \letkw \; x = e_1 \; \inkw \; e_2 : \tau_2}

    \inferrule[\textsc{If}]
    {\Gamma \vdash e_1 : \bool \\
     \Gamma \vdash e_2 : \tau \\
     \Gamma \vdash e_3 : \tau}
    {\Gamma \vdash \ifkw \; e_1 \; \thenkw \; e_2 \; \elsekw \; e_3 : \tau}

    \inferrule[\textsc{Discrete}]
    {n \geq 1 \\
     \forall i \in \{0,\ldots,n-1\}.\; p_i \geq 0 \\
     \sum_{i=0}^{n-1} p_i = 1}
    {\Gamma \vdash \discrete(p_0, \ldots, p_{n-1}) : \fin{n}}

    \inferrule[\textsc{LessEq}]
    {\Gamma \vdash e : \intty}
    {\Gamma \vdash e \leq i : \bool}

    \inferrule[\textsc{Pair}]
    {\Gamma \vdash e_1 : \tau_1 \\
     \Gamma \vdash e_2 : \tau_2}
    {\Gamma \vdash (e_1, e_2) : \tau_1 * \tau_2}

    \inferrule[\textsc{Fst}]
    {\Gamma \vdash e : \tau_1 * \tau_2}
    {\Gamma \vdash \fstkw \; e : \tau_1}

    \inferrule[\textsc{Snd}]
    {\Gamma \vdash e : \tau_1 * \tau_2}
    {\Gamma \vdash \sndkw \; e : \tau_2}

    \inferrule[\textsc{Fun}]
    {\Gamma, x: \tau_1 \vdash e : \tau_2}
    {\Gamma \vdash \funkw \; x \; \rightarrow \; e : \tau_1 \rightarrow \tau_2}

    \inferrule[\textsc{App}]
    {\Gamma \vdash e_1 : \tau_1 \rightarrow \tau_2 \\
     \Gamma \vdash e_2 : \tau_1}
    {\Gamma \vdash e_1 \; e_2 : \tau_2}

    \inferrule[\textsc{True}]
    {\ }
    {\Gamma \vdash \text{true} : \bool}

    \inferrule[\textsc{False}]
    {\ }
    {\Gamma \vdash \text{false} : \bool}

    \inferrule[\textsc{And}]
    {\Gamma \vdash e_1 : \bool \\
     \Gamma \vdash e_2 : \bool}
    {\Gamma \vdash e_1 \logand e_2 : \bool}

    \inferrule[\textsc{Or}]
    {\Gamma \vdash e_1 : \bool \\
     \Gamma \vdash e_2 : \bool}
    {\Gamma \vdash e_1 \logor e_2 : \bool}

    \inferrule[\textsc{Not}]
    {\Gamma \vdash e : \bool}
    {\Gamma \vdash \text{not}\; e : \bool}

    \inferrule[\textsc{Unit}]
    {\ }
    {\Gamma \vdash () : \unit}

    \inferrule[\textsc{FinConst}]
    {0 \leq k < n}
    {\Gamma \vdash \finconst{k}{n} : \fin{n}}

    \inferrule[\textsc{FinLess}]
    {\Gamma \vdash e_1 : \fin{n} \\
     \Gamma \vdash e_2 : \fin{n}}
    {\Gamma \vdash e_1 \finlt{n} e_2 : \bool}

    \inferrule[\textsc{FinLessEq}]
    {\Gamma \vdash e_1 : \fin{n} \\
     \Gamma \vdash e_2 : \fin{n}}
    {\Gamma \vdash e_1 \finleq{n} e_2 : \bool}

    \inferrule[\textsc{Observe}]
    {\Gamma \vdash e : \bool}
    {\Gamma \vdash \observekw\; e : \unit}

    \inferrule[\textsc{Seq}]
    {\Gamma \vdash e_1 : \tau_1 \\
     \Gamma \vdash e_2 : \tau_2}
    {\Gamma \vdash e_1; e_2 : \tau_2}

    \inferrule[\textsc{Fix}]
    {\Gamma, f: \tau_1 \rightarrow \tau_2, x: \tau_1 \vdash e : \tau_2}
    {\Gamma \vdash \text{fix}\; f\; x := e : \tau_1 \rightarrow \tau_2}

    \inferrule[\textsc{Nil}]
    {\ }
    {\Gamma \vdash \text{nil} : \listty{\tau}}

    \inferrule[\textsc{Cons}]
    {\Gamma \vdash e_1 : \tau \\
     \Gamma \vdash e_2 : \listty{\tau}}
    {\Gamma \vdash e_1 :: e_2 : \listty{\tau}}

    \inferrule[\textsc{Match}]
    {\Gamma \vdash e : \listty{\tau_1} \\
     \Gamma \vdash e_1 : \tau_2 \\
     \Gamma, h: \tau_1, t: \listty{\tau_1} \vdash e_2 : \tau_2}
    {\Gamma \vdash \text{match}\; e\; \text{with}\; \text{nil} \rightarrow e_1 \mid h :: t \rightarrow e_2\; \text{end} : \tau_2}




    \inferrule[\textsc{Diverge}]
    {\ }
    {\Gamma \vdash \diverge : \tau}
\end{mathpar}
\caption{Typing rules for general language constructs}
\label{fig:typing-general-all}
\end{figure}



\clearpage
\section{Discretization Proof}
\subsection{Well-Typedness of Discretization}

\begin{proof}[Proof of~\Cref{lem:type-transformations}]\label{proof:type-transformations}
  \setlength{\parskip}{0.3em plus 0.1em}
  We proceed by induction on the structure of the typing derivation $\Gamma \vdash e : \tau$, matching each case of the discretization function $\discretize{\cdot}$ defined in \Cref{fig:discretization-code}.

  \par\medskip\noindent
  \textbf{Case} $e = x$.
  \par\nopagebreak
  Since $\Gamma(x) = \tau$, we have $\discretize{\Gamma}(x) = \discretize{\tau}$. Moreover, $\discretize{x : \tau} = x$, so $\discretize{\Gamma} \vdash x : \discretize{\tau}$.

  \par\medskip\noindent
  \textbf{Case} $e = (c : \float[B;V])$.
  \par\nopagebreak
  The \textsc{Float-Const} rule was applied, so $\has{V}{c}$ holds.
  \begin{itemize}[leftmargin=1.5em, itemsep=0.35em, topsep=0.35em]
      \item If $B = \top$: then $\discretize{c : \float[B;V]} = c$ and $\discretize{\float[B;V]} = \float[B;V]$. Since $\has{V}{c}$ still holds, $c$ is well-typed at $\float[B;V]$ by \textsc{Float-Const}.
      \item If $B \neq \top$: then $\discretize{c : \float[B;V]} = \finconst{k}{|B|+1}$ for the unique $k$ such that $c \in \intervals{B}_k$, and $\discretize{\float[B;V]} = \fin{|B|+1}$. Since $0 \leq k < |B|+1$, the finite constant $\finconst{k}{|B|+1}$ is well-typed at $\fin{|B|+1}$ by \textsc{FinConst}.
  \end{itemize}

  \par\medskip\noindent
  \textbf{Case} $e = (\letkw\; x = e_1 : \tau_1 \;\inkw\; e_2 : \tau_2)$.
  \par\nopagebreak
  By the inductive hypothesis,
  \[
      \discretize{\Gamma} \vdash \discretize{e_1 : \tau_1} : \discretize{\tau_1}
      \quad\text{and}\quad
      \discretize{\Gamma, x : \tau_1} \vdash \discretize{e_2 : \tau_2} : \discretize{\tau_2}.
  \]
  Note that $\discretize{\Gamma, x : \tau_1} = \discretize{\Gamma}, x : \discretize{\tau_1}$. Applying \textsc{Let} gives
  \[
      \discretize{\Gamma} \vdash \letkw\; x = \discretize{e_1 : \tau_1} \;\inkw\; \discretize{e_2 : \tau_2} : \discretize{\tau_2}.
  \]

  \par\medskip\noindent
  \textbf{Case} $e = (\ifkw\; e_1 : \bool \;\thenkw\; e_2 : \tau \;\elsekw\; e_3 : \tau)$.
  \par\nopagebreak
  By the inductive hypothesis, $\discretize{\Gamma} \vdash \discretize{e_i : \tau_i} : \discretize{\tau_i}$ for $i = 1, 2, 3$, where $\tau_1 = \bool$ and $\tau_2 = \tau_3 = \tau$. Since $\discretize{\bool} = \bool$, applying \textsc{If} gives
  \[
      \discretize{\Gamma} \vdash \ifkw\; \discretize{e_1 : \bool} \;\thenkw\; \discretize{e_2 : \tau} \;\elsekw\; \discretize{e_3 : \tau} : \discretize{\tau}.
  \]

  \par\medskip\noindent
  \textbf{Case} $e = (e_1 : \float[B; V_1] < e_2 : \float[B; V_2])$.
  \par\nopagebreak
  The \textsc{Less} rule was applied, so both subexpressions share the same cut set $B$ and $\answerslt{B}{V_1}{V_2}$ holds. By the inductive hypothesis, $\discretize{\Gamma} \vdash \discretize{e_i : \float[B;V_i]} : \discretize{\float[B;V_i]}$ for $i = 1, 2$.
  \begin{itemize}[leftmargin=1.5em, itemsep=0.35em, topsep=0.35em]
      \item If $B = \top$: then $\discretize{\float[B;V_i]} = \float[B;V_i]$ for both $i$, and the standard float comparison gives type $\bool = \discretize{\bool}$.
      \item If $B \neq \top$: then both subexpressions are discretized to type $\fin{|B|+1}$, and the finite comparison $<_{\#(|B|+1)}$ is well-typed at $\bool = \discretize{\bool}$ by \textsc{FinLess}.
  \end{itemize}

  \par\medskip\noindent
  \textbf{Case} $e = (e_1 : \float[B; V_1] \leq e_2 : \float[B; V_2])$.
  \par\nopagebreak
  Analogous to the $<$ case, using $\answersleq{B}{V_1}{V_2}$ and \textsc{FinLessEq}.

  \par\medskip\noindent
  \textbf{Case} $e = (\uniform(e_1 : \float[B_1; V_1],\; e_2 : \float[B_2; V_2]) : \float[B; \top])$.
  \par\nopagebreak
  The \textsc{Continuous-2} rule was applied, so $\distinguishes{B_1}{V_1}$ and $\distinguishes{B_2}{V_2}$ hold. By the inductive hypothesis, $\discretize{\Gamma} \vdash \discretize{e_i : \float[B_i;V_i]} : \discretize{\float[B_i;V_i]}$ for $i = 1, 2$. Also, $B_i \equiv_{\top} B$, so $B_i=\top$ iff $B=\top$. We consider each sub-case of the discretization function:
  \begin{itemize}[leftmargin=1.5em, itemsep=0.5em, topsep=0.35em]
      \item $e_1, e_2 \in \mathbb{R}$ and $e_1 \geq e_2$: then $\discretize{\uniform(\ldots) : \float[B;\top]} = \diverge$, which is well-typed at any type by \textsc{Diverge}.
  
      \item $B = \top$: then $\discretize{\float[B;\top]} = \float[B;\top]$, and
      \[
          \discretize{\uniform(e_1, e_2) : \float[B;\top]} = \uniform\!\bigl(\discretize{e_1 : \float[B_1;V_1]},\, \discretize{e_2 : \float[B_2;V_2]}\bigr).
      \]
	      Here $B_1=B_2=\top$, so the inductive hypotheses give float-typed arguments; \textsc{Continuous-2} yields $\float[\top;\top]$.
  
      \item $B \neq \top$ and $e_1, e_2 \in \mathbb{R}$: then $\discretize{\float[B;\top]} = \fin{|B|+1}$, and
      \[
          \discretize{\uniform(e_1, e_2) : \float[B;\top]} = \discrete(p_0, \ldots, p_{|B|})
      \]
      where $p_k = \mathrm{CDF}(b_{k+1}) - \mathrm{CDF}(b_k)$. Each $p_k \in [0,1]$ and $\sum_k p_k = 1$, so this is well-typed at $\fin{|B|+1}$ by \textsc{Discrete}.
  
      \item If $B,B_1,B_2 \neq \top$ and $V_1=\emptyset$ or $V_2=\emptyset$, then the enumeration over $V_1\times V_2$ is empty, so the generated expression evaluates $\discretize{e_1}$ and $\discretize{e_2}$ through its two outer $\letkw$-bindings and then yields $\diverge$; the bindings are well-typed by the inductive hypotheses and weakening, and their body is well-typed at $\fin{|B|+1}$ by \textsc{Diverge}.

      \item $B,B_1,B_2 \neq \top$, $V_1,V_2\neq\emptyset$, and
      $e_1\notin\mathbb{R}$ or $e_2\notin\mathbb{R}$: the algorithm generates conditionals over all pairs $(v_i, w_j)$, each branch dispatching to a constant parameter in $\uniform$. The inductive hypotheses and weakening type the two outer $\letkw$-bindings. In each guard, both operands of each finite equality have the same finite type; since finite equality is desugared using \textsc{FinLessEq}, each equality has type $\bool$, as does their conjunction. By a previous subcase, each branch $\discretize{\uniform(v_i, w_j) : \float[B;\top]}$ is well-typed at $\fin{|B|+1} = \discretize{\float[B;\top]}$. The final $\elsekw\;\diverge$ branch has the same type by \textsc{Diverge}; hence the conditional structure is well-typed by \textsc{If} (as in the $\ifkw$ case above), and the full expression is well-typed at $\discretize{\float[B;\top]}$.
  \end{itemize}

  \par\medskip\noindent
  \textbf{Case} a one-argument continuous distribution. Analogous, using $B_1 \equiv_{\top} B$.

  \par\medskip\noindent
  \textbf{Case} $e$ is any other language construct (functions, pairs, lists, \emph{etc.}).
  \par\nopagebreak
  By definition of $\discretize{\cdot}$, the discretization recurses structurally over $e$ and $\tau$. That is, $\discretize{e : \tau}$ is formed by applying $\discretize{\cdot}$ to each subexpression and subtype of $e : \tau$, and $\discretize{\tau}$ is formed by applying $\discretize{\cdot}$ to each component type of $\tau$. The result then follows immediately by applying the inductive hypothesis to each subterm, analogous to the above cases.

  \par\medskip\noindent
  In every case, the typing derivation is preserved under $\discretize{\cdot}$. Therefore,
  \[
      \discretize{\Gamma} \vdash \discretize{e : \tau} : \discretize{\tau}. \qedhere
  \]
\end{proof}

\section{Soundness Proof}
\label{sec:app:soundness_new}

\subsection{Measure Theory Preliminaries}

We first give a brief overview of basic measure-theoretic terms and lemmas, and refer the reader to standard textbooks for further background~\cite{SteinShakarchi2005RealAnalysis,Tao2011MeasureTheory}.

\paragraph{Topological space} A topological space is a pair $(X, \mathcal{T})$ where $\mathcal{T}$ is a collection of subsets of $X$, called open sets, such that $\emptyset, X \in \mathcal{T}$, and $\mathcal{T}$ is closed under arbitrary unions and finite intersections.

\paragraph{$\sigma$-algebra} A $\sigma$-algebra on a set $X$ is a collection $\Sigma$ of subsets of $X$ that includes $X$ itself and is closed under complement and under countable unions. By implication, a $\sigma$-algebra includes the empty set $\emptyset$ and is closed under countable intersections. 
Given a topological space $(X, \mathcal{T})$, the Borel $\sigma$-algebra $\mathcal{B}(X)$ is the smallest $\sigma$-algebra containing all open sets in $\mathcal{T}$.

\paragraph{Measurable space} A measurable space is a pair $(X,\Sigma)$ consisting of a set $X$ and a $\sigma$-algebra $\Sigma$ on $X$.
When $\Sigma$ is the Borel $\sigma$-algebra on a topological space $X$, we call $(X, \Sigma)$ a Borel measurable space.

\paragraph{Measure} Let $(X,\Sigma)$ be a measurable space. A function $\mu : \Sigma \to [0,\infty]$ is called a measure if it satisfies the following properties: nonnegativity, $\mu(\emptyset) = 0$, and countable additivity.
When $\Sigma$ is a Borel $\sigma$-algebra, $\mu$ is called a Borel measure.

\paragraph{Measure space} The triple $(X, \Sigma, \mu)$ is called a measure space.
 When $\Sigma$ is the Borel $\sigma$-algebra and $\mu$ is a Borel measure, $(X, \Sigma, \mu)$ is called a Borel measure space.

\paragraph{Subprobability measure} Let $(X, \Sigma, \mu)$ be a measure space. When $\mu(X) \leq 1$, $\mu$ is called a (Borel) subprobability measure. When $\mu(X) = 1$, $\mu$ is called a (Borel) probability measure and $(X, \Sigma, \mu)$ a (Borel) probability space.

\paragraph{Product measurable space} The product of two measurable spaces $(X_1, \Sigma_1)$ and $(X_2, \Sigma_2)$ is also a measurable space, namely $(X_1 \times X_2, \Sigma_1 \times \Sigma_2)$, where $\Sigma_1 \times \Sigma_2$ is the $\sigma$-algebra generated by sets of the form $A_1 \times A_2$ with $A_1 \in \Sigma_1$ and $A_2 \in \Sigma_2$.

\paragraph{Generated $\sigma$-algebra}
Let $X$ be a set and let $\mathcal{G} \subseteq \mathcal{P}(X)$ be a collection of subsets of $X$. The $\sigma$-algebra generated by $\mathcal{G}$, denoted $\sigma(\mathcal{G})$, is the smallest $\sigma$-algebra on $X$ that contains $\mathcal{G}$. Equivalently, it is the intersection of all $\sigma$-algebras on $X$ that contain $\mathcal{G}$.

\paragraph{Subprobability distribution} Let $(X,\Sigma)$ be a measurable space. We write $\mathcal{D}(X)$ for the set of all subprobability measures on $(X,\Sigma)$.

\paragraph{Dirac measure} 
Let $(X,\Sigma)$ be a measurable space and let $x \in X$. The
	\emph{Dirac measure} at $x$ is the probability measure
	\[
		\delta_x : \Sigma \to [0,1]
	\]
	defined by
	\[
		\delta_x(A) =
		\begin{cases}
		1 & \text{if } x \in A \\
		0 & \text{if } x \notin A.
		\end{cases}
	\]

\paragraph{Measurable function} A function $f \colon X \to Y$ between measurable spaces $(X, \mathcal{F}_X)$ and $(Y, \mathcal{F}_Y)$ is called measurable if the preimage of any measurable set in $Y$ is measurable in $X$, i.e. if 
\[
	f^{-1}(A) = \{x \in X \mid f(x) \in A\} \in \mathcal{F}_X
\]
for all $A \in \mathcal{F}_Y$.

\begin{definition}[Subspace $\sigma$-algebra]\label{def:subspace-sigma-algebra}
	Let $(X, \Sigma)$ be a Borel measurable space and let $Y \subseteq X$. 
	The subspace (or trace) $\sigma$-algebra on $Y$ is defined as
	\[
	\Sigma_Y = \{ Y \cap A \mid A \in \Sigma \}.
	\]
	Then $(Y, \Sigma_Y)$ is a Borel measurable space.
\end{definition}

\begin{definition}[$\sigma$-algebra generated by a function]\label{def:pullback}
	Let $f \colon X \to Y$ be a function and let $(Y, \Sigma_Y)$ be a measurable space. The $\sigma$-algebra generated by $f$ is defined by:

	\[
		f^{-1}(\Sigma_Y) := \{f^{-1}(A) \mid A \in \Sigma_Y\}.
	\]

	Then $(X, f^{-1}(\Sigma_Y))$ is a measurable space.
\end{definition}

\begin{lemma}[Bijection induces isomorphism]\label{lem:pullback-iso}
	Let $f : X \to Y$ be a bijection and let $(Y,\Sigma_Y)$ be a measurable space.
	Equip $X$ with the $\sigma$-algebra generated by $f$:
	\[
	\Sigma_X \;:=\; f^{-1}(\Sigma_Y) \;=\; \{ f^{-1}(A) \mid A \in \Sigma_Y\}.
	\]

	Then,
	\begin{enumerate}
		\item $(X, \Sigma_X)$ is a measurable space, and
		\item $f$ induces an isomorphism of measurable spaces: 
		\[
			(X,\Sigma_X) \cong (Y,\Sigma_Y).
		\]
	\end{enumerate}
\end{lemma}
	
\begin{proof}
	We first verify that i) $(X, \Sigma_X)$ is a measurable space, then show that ii) $f$ induces an isomorphism of measurable spaces $(X,\Sigma_X) \cong (Y,\Sigma_Y)$.

	\medskip

	\emph{Proof of i)}
	We show that 
	$\Sigma_X = f^{-1}(\Sigma_Y)$ is a $\sigma$-algebra on $X$.

	\begin{itemize}
		\item \textbf{(Contains $X$):} $X = f^{-1}(Y) \in \Sigma_X$ since $Y \in \Sigma_Y$.
		\item \textbf{(Complements):} If $f^{-1}(A) \in \Sigma_X$, then 
			$X \setminus f^{-1}(A) = f^{-1}(Y \setminus A) \in \Sigma_X$ 
			since $Y \setminus A \in \Sigma_Y$.
		\item \textbf{(Countable unions):} If $f^{-1}(A_i) \in \Sigma_X$ for 
			$i \in \mathbb{N}$, then $\bigcup_i f^{-1}(A_i) = 
			f^{-1}\!\left(\bigcup_i A_i\right) \in \Sigma_X$ since 
			$\bigcup_i A_i \in \Sigma_Y$.
	\end{itemize}

	Hence $(X, \Sigma_X)$ is a measurable space.

	\medskip

	\emph{Proof of ii)}
	Two measurable spaces $(X,\Sigma_X)$ and $(Y,\Sigma_Y)$ are isomorphic iff there exists a bijection $f \colon X \to Y$ such that $f$ and $f^{-1}$ are measurable.

	We first show that $f$ is measurable. Measurability of $f$ means that $f^{-1}(A)\in \Sigma_X$ for every $A\in \Sigma_Y$. Let $A \in \Sigma_Y$ be arbitrary. By definition of $\Sigma_X$,
	\[
	f^{-1}(A) \in f^{-1}(\Sigma_Y) = \Sigma_X.
	\]
	It follows that $f$ is measurable.
	
	Next we show that $f^{-1}$ is measurable. Measurability of $f^{-1}$ means that $(f^{-1})^{-1}(B)\in \Sigma_Y$ for every $B\in \Sigma_X$. Let $B \in \Sigma_X$ be arbitrary. By definition of $\Sigma_X$, there exists some $A \in \Sigma_Y$ such that $B = f^{-1}(A)$.

	Observe that:
	\begin{align*}
		(f^{-1})^{-1}(B) &= f(B) \tag{preimage of $f^{-1}$ is $f$} \\
		&= f(f^{-1}(A)) \tag{substituting $B = f^{-1}(A)$} \\
		&= A \tag{$f$ is a bijection}
	\end{align*}
	Since $A \in \Sigma_Y$, it follows that $f^{-1}$ is measurable.
	
	Therefore $f$ is a measurable bijection with measurable inverse, so it is an isomorphism of measurable spaces.
\end{proof}

\begin{lemma}[Minimality of generated $\sigma$-algebras]\label{lem:minimality-sigma-algebra}
	Let $X$ be a set and let $\mathcal{G} \subseteq \mathcal{P}(X)$. 
	Suppose $\mathcal{C}$ is a $\sigma$-algebra on $X$ such that $\mathcal{G} \subseteq \mathcal{C}$. 
	Then
	\[
	\sigma(\mathcal{G}) \subseteq \mathcal{C}.
	\]
\end{lemma}

\begin{proof}
	Recall that $\sigma(\mathcal{G})$ is defined as the smallest $\sigma$-algebra 
	on $X$ containing $\mathcal{G}$, given concretely by the intersection of all 
	$\sigma$-algebras on $X$ containing $\mathcal{G}$:
	\[
	\sigma(\mathcal{G})
	=
	\bigcap_{\substack{\mathcal{F} \text{ is a } \sigma\text{-algebra on } X \\ \mathcal{G} \subseteq \mathcal{F}}}
	\mathcal{F}.
	\]
	Since $\mathcal{C}$ is a $\sigma$-algebra on $X$ containing $\mathcal{G}$, it 
	appears in this intersection. Therefore,
	\[
	\sigma(\mathcal{G}) \subseteq \mathcal{C}.
	\]
\end{proof}

\begin{definition}[Sub-Markov kernels]
Given measurable spaces $(X,\Sigma_X)$ and $(Y,\Sigma_Y)$, a
\emph{sub-Markov kernel} from $X$ to $Y$ is a function
\[
P:X\times\Sigma_Y\to[0,1]
\]
such that:
\begin{enumerate}
    \item for every $x\in X$, $P(x,-)$ is a subprobability measure on $Y$;
    \item for every $A\in\Sigma_Y$, the map $x\mapsto P(x,A)$ is measurable.
\end{enumerate}
A Markov kernel is a sub-Markov kernel satisfying $P(x,Y)=1$ for every
$x\in X$.
\end{definition}


\begin{definition}[Projection with error]
	Let $(X + \bot, \Sigma_{X + \bot})$ and $(Y + \bot, \Sigma_{Y + \bot})$ be measurable spaces, with $(X, \Sigma_X)$ and $(Y, \Sigma_Y)$ measurable spaces as well. Let $\mu \in \mathcal{D}((X \times Y) + \bot)$. We define the first and second projection measures as:
	\begin{enumerate}
	\item $\pi_1(\mu)(A) := \mu(A \times Y)$ for every measurable $A \in \Sigma_X$, and $\pi_1(\mu)({\bot}) := \mu({\bot})$
	\item $\pi_2(\mu)(B) := \mu(X \times B)$ for every measurable $B \in \Sigma_Y$, and $\pi_2(\mu)({\bot}) := \mu({\bot})$
	\end{enumerate}
	Then, $\pi_1(\mu)$ is a measure on $(X + \bot, \Sigma_{X + \bot})$ and $\pi_2(\mu)$ is a measure on $(Y + \bot, \Sigma_{Y + \bot})$.
\end{definition}

\subsection{Measure Space on Expressions}
\label{app:measure_space_expressions}
To make the probabilistic semantics precise, we construct an appropriate measure space structure on the set of expressions by factoring out real numbers using the notion of what we call \emph{expression skeletons}.

\begin{definition}[Expression Skeleton]
	An expression skeleton $s$ is a syntactic expression obtained from a program by replacing each occurrence of a float constant with a hole of the form $\square$. The set of skeletons is generated by the grammar:
	\begin{align*}
		s ::= \; 
		&\square \mid x \mid \true \mid \false \mid \finconst{k}{n} \mid \diverge \mid () \\
		&\mid \letkw \; x = s_1 \; \inkw \; s_2 \mid \ifkw \; s_1 \; \thenkw \; s_2 \; \elsekw \; s_3 \\
		&\mid s_1 < s_2 \mid s_1 \leq s_2 \mid s_1 \finlt{n} s_2 \mid s_1 \finleq{n} s_2 \\
		&\mid \uniform(s_1, s_2) \mid \gaussian(s_1, s_2) \mid \exponential(s) \mid \betafn(s_1, s_2) \mid \discrete(p_0, \ldots, p_{n-1}) \\
		&\mid s_1\; \logand\; s_2 \mid s_1\; \logor\; s_2 \mid \lognot\; s \mid \observekw(s) \mid s_1; s_2 \\
		&\mid (s_1, s_2) \mid \fstkw \; s \mid \sndkw \; s \\
		&\mid \funkw \; x \to s \mid s_1 \; s_2 \mid \fixkw\; f\; x := s \\
		&\mid \textnormal{\ttfamily nil} \mid s_1 :: s_2 \\
		&\mid (\matchkw \; s \; \withkw \matchcase \textnormal{\ttfamily nil} \rightarrow s_{\textnormal{nil}} \matchcase h::t \rightarrow s_{\textnormal{cons}} \; {\textnormal{\ttfamily\bfseries end}})
	\end{align*}
	We denote by $\skel$ the set of all expression skeletons, and by $\holes(s)$ the number of distinct holes in skeleton $s$.
\end{definition}

\begin{example}
	The expression $\uniform(0.0, 1.0) < 0.5$ has the skeleton $\uniform(\square, \square) < \square$ with three holes. 
\end{example}

\begin{definition}[Filling a skeleton]
	Let $s \in \skel$ and let $\vec{r} = (r_1, \ldots, r_{\holes(s)}) \in \mathbb{R}^{\holes(s)}$. Define the expression $s[\vec{r}] \in Expr$ to be the result of replacing each hole $\square_i$ in $s$ with the corresponding real constant $r_i$, for $1 \le i \le \holes(s)$.
\end{definition}

To recover a \emph{typed} notion of skeletons, we extend the typing judgment to
skeletons by adding a distinguished typing rule for holes.

\begin{definition}[Well-typed expression skeleton]
	Define a typing judgment for skeletons, written $\Gamma \vdash s : \tau$, with the understanding that all typing rules from expressions are preserved for skeletons, and that holes are typed as $\floattype{B}{V}$ according to the rule below:

	\[
		\inferrule{ }{\square : \floattype{B}{V}}
	\]

	We denote by $\skel_{\tau}$ the set of all well-typed expression skeletons, where:

	\[
		\skel_{\tau} := \{s \in \skel \mid s : \tau\}.
	\]
\end{definition}

\begin{remark}
	Both $\skel_\tau$ and $Expr_\tau$ are to be understood as the set of all well-typed skeletons and expressions, respectively, \emph{equipped
	with their full subtype information}.

	For example,
	\[
	\square : \floattype{B}{\{0.5\}} < \square : \floattype{B}{\{0.5\}}
	\quad\text{and}\quad
	\square : \floattype{B}{\{0.2\}} < \square : \floattype{B}{\{0.8\}}
	\]
	are distinct elements of $\skel_{\bool}$, while
	\[
	0.5 : \floattype{B}{\{0.5\}} < 0.5 : \floattype{B}{\{0.5\}}
	\quad\text{and}\quad
	0.2 : \floattype{B}{\{0.2\}} < 0.8 : \floattype{B}{\{0.8\}}
	\]
	are distinct elements of $Expr_{\bool}$.

	In particular, both $\skel_\tau$ and $Expr_\tau$ are in general
	uncountable, since float subtype annotations of the form
	$\floattype{B}{V}$ may vary over uncountably many real values.
\end{remark}

\begin{definition}[Decomposing an expression]
	The decomposition function
	\[
		\decompose \colon Expr_{\tau}
		\;\longrightarrow\;
		\bigcup_{s \in Skel_{\tau}} \left( \{s\} \times \mathbb{R}^{\holes(s)} \right)
	\]
	is defined as the inverse of filling a skeleton, i.e., for every expression $e : \tau$,
	\[
		\decompose(e) = (s, \vec{r})
		\quad\text{such that}\quad
		e = s[\vec{r}].
	\]
	The explicit recursive definition of $\decompose$ is shown in \Cref{fig:decompose-definition-1,fig:decompose-definition-2}.
\end{definition}

\begin{figure}
    \small
	\centering
	\begin{align*}
		\decompose(e : \tau) =
		\begin{cases}
			(x : \tau, ())
				& \text{if } e = x : \tau \\[8pt]
			(\true : \bool, ())
				& \text{if } e = \true : \bool \\[8pt]
			(\false : \bool, ())
				& \text{if } e = \false : \bool \\[8pt]
			(\finconst{k}{n} : \fin{n}, ())
				& \text{if } e = \finconst{k}{n} : \fin{n} \\[8pt]
			(\diverge : \tau, ())
				& \text{if } e = \diverge : \tau \\[8pt]
			(() : \unit, ())
				& \text{if } e = () : \unit \\[8pt]
			((\square, \floattype{B}{V}) : \floattype{B}{V}, c)
				& \text{if } e = c : \floattype{B}{V} \\[8pt]
			(\discrete(p_0,\ldots,p_{n-1}) : \fin{n}, ())
				& \text{if } e = \discrete(p_0,\ldots,p_{n-1}) : \fin{n} \\[8pt]
			(\uniform(\decompose(e_1)_1, \decompose(e_2)_1) : \floattype{B}{V}, \\
			(\decompose(e_1)_2, \decompose(e_2)_2))
				& \text{if } e = \uniform(e_1, e_2) : \floattype{B}{V} \\[8pt]
			(\gaussian(\decompose(e_1)_1, \decompose(e_2)_1) : \floattype{B}{V}, \\
			(\decompose(e_1)_2, \decompose(e_2)_2))
				& \text{if } e = \gaussian(e_1, e_2) : \floattype{B}{V} \\[8pt]
			(\exponential(\decompose(e_1)_1) : \floattype{B}{V}, \\
            \decompose(e_1)_2)
				& \text{if } e = \exponential(e_1) : \floattype{B}{V} \\[8pt]
			(\betafn(\decompose(e_1)_1, \decompose(e_2)_1) : \floattype{B}{V}, \\
			(\decompose(e_1)_2, \decompose(e_2)_2))
				& \text{if } e = \betafn(e_1, e_2) : \floattype{B}{V} \\[8pt]
			(\decompose(e_1)_1 \sim \decompose(e_2)_1 : \bool, \\
			(\decompose(e_1)_2, \decompose(e_2)_2))
				& \text{if } e = e_1 \sim e_2 : \bool, \sim = <, \leq, \finlt{n}, \finleq{n} \\[8pt]
            (\ifkw\; \decompose(e_1)_1\; \thenkw\; \decompose(e_2)_1\; \elsekw\; \decompose(e_3)_1 : \tau, \\
            (\decompose(e_1)_2, \decompose(e_2)_2, \decompose(e_3)_2))
                & \text{if } e = \ifkw\; e_1\; \thenkw\; e_2\; \elsekw\; e_3 : \tau \\[8pt]
            (\letkw\; x = \decompose(e_1)_1\; \inkw\; \decompose(e_2)_1 : \tau, \\
            (\decompose(e_1)_2, \decompose(e_2)_2))
                & \text{if } e = \letkw\; x = e_1\; \inkw\; e_2 : \tau \\[8pt]
            (\observekw(\decompose(e_1)_1 : \bool) : \unit, \decompose(e_1)_2)
                & \text{if } e = \observekw(e_1 : \bool) : \unit \\[8pt]
		\end{cases}
	\end{align*}
	\caption{Recursive definition of $\decompose$, a decomposition function, over typed expressions. $\decompose(\cdot)_1$ and $\decompose(\cdot)_2$ denote first and second projections of $\decompose$, respectively. $\decompose(\cdot)_2$ is canonically bijective to a flat vector in $\R^n$ via coordinate concatenation over holes.}
	\label{fig:decompose-definition-1}
\end{figure}

\begin{figure}
    \small
	\centering
	\begin{align*}
		\decompose(e : \tau) =
		\begin{cases}
            (\decompose(e_1)_1 \logand \decompose(e_2)_1 : \bool, \\
			(\decompose(e_1)_2, \decompose(e_2)_2))
				& \text{if } e = e_1 \logand e_2 : \bool \\[8pt]
			(\decompose(e_1)_1 \logor \decompose(e_2)_1 : \bool, \\
			(\decompose(e_1)_2, \decompose(e_2)_2))
				& \text{if } e = e_1 \logor e_2 : \bool \\[8pt]
			(\lognot\; \decompose(e_1)_1 : \bool, \decompose(e_1)_2)
				& \text{if } e = \lognot\; e_1 : \bool \\[8pt]
			(\decompose(e_1)_1; \decompose(e_2)_1 : \tau_2, \\
			(\decompose(e_1)_2, \decompose(e_2)_2))
				& \text{if } e = e_1; e_2 : \tau_2 \\[8pt]
			((\decompose(e_1)_1, \decompose(e_2)_1) : \tau_1 * \tau_2, \\
			(\decompose(e_1)_2, \decompose(e_2)_2))
				& \text{if } e = (e_1, e_2) : \tau_1 * \tau_2 \\[8pt]
			(\fstkw\; \decompose(e_1)_1 : \tau_1, \decompose(e_1)_2)
				& \text{if } e = \fstkw\; e_1 : \tau_1 \\[8pt]
			(\sndkw\; \decompose(e_1)_1 : \tau_2, \decompose(e_1)_2)
				& \text{if } e = \sndkw\; e_1 : \tau_2 \\[8pt]
			(\funkw\; x \to \decompose(e_1)_1 : \tau_1 \to \tau_2, \decompose(e_1)_2)
				& \text{if } e = \funkw\; x \to e_1 : \tau_1 \to \tau_2 \\[8pt]
			(\fixkw\; f\; x := \decompose(e_1)_1 : \tau_1 \to \tau_2, \decompose(e_1)_2)
				& \text{if } e = \fixkw\; f\; x := e_1 : \tau_1 \to \tau_2 \\[8pt]
			(\decompose(e_1)_1\; \decompose(e_2)_1 : \tau_2, \\
			(\decompose(e_1)_2, \decompose(e_2)_2))
				& \text{if } e = (e_1 : \tau_1 \to \tau_2\;\; e_2 : \tau_1) : \tau_2 \\[8pt]
			(\textnormal{\ttfamily nil} : \listty{\tau_1}, ())
				& \text{if } e = \textnormal{\ttfamily nil} : \listty{\tau_1} \\[8pt]
			(\decompose(e_1)_1 :: \decompose(e_2)_1 : \listty{\tau_1}, \\
			(\decompose(e_1)_2, \decompose(e_2)_2))
				& \text{if } e = e_1 :: e_2 : \listty{\tau_1} \\[8pt]
			(\matchkw\; \decompose(e_0)_1\; \withkw\; \matchcase \textnormal{\ttfamily nil} \rightarrow \decompose(e_{\textnormal{nil}})_1 \\
			\matchcase h::t \rightarrow \decompose(e_{\textnormal{cons}})_1\; {\textnormal{\ttfamily\bfseries end}} : \tau_2, 
                & \text{if } e = (\matchkw\; e_0 : \listty{\tau_1}\; \withkw\; \matchcase \textnormal{\ttfamily nil} \rightarrow e_{\textnormal{nil}} : \tau_2 \\
			(\decompose(e_0)_2, \decompose(e_{\textnormal{nil}})_2, \decompose(e_{\textnormal{cons}})_2))
                &\qquad\qquad\;\;\;\matchcase h::t \rightarrow e_{\textnormal{cons}} : \tau_2\; {\textnormal{\ttfamily\bfseries end}}) : \tau_2
		\end{cases}
	\end{align*}
	\caption{Recursive definition of $\decompose$ over typed expressions (continued).}
	\label{fig:decompose-definition-2}
\end{figure}

\begin{definition}[Hole assignments of a skeleton]
	For a given $s : \tau \in \skel_{\tau}$, define the set of hole assignments of $s$ as:
	\[
		A_{s : \tau} := 
		\bigcup_{e : \tau \in Expr_{\tau}} \left\{ \decompose(e : \tau)_2 \mid \decompose(e : \tau)_1 = s : \tau \right\}
		\subseteq \R^{\holes(s)}
	\]

	Intuitively, $A_{s;\tau}$ is the set of all real tuples that can be substituted into the holes of $s$ to yield a well-typed expression of type $\tau$.
\end{definition}

\begin{example}
	Consider the skeleton
	\[
		s : \tau 
		=
		(\square : \floattype{\{\clt{3.5}, \clt{4.5}\}}{\top}
		\;<\;
		\square : \floattype{\{\clt{3.5}, \clt{4.5}\}}{\{3.5,4.5\}})
		: \bool.
	\]
	
	This skeleton has two holes:
	\begin{itemize}
		\item The left hole has type 
		$\floattype{\{\clt{3.5}, \clt{4.5}\}}{\top}$,
		so it may be instantiated by any real number.
		\item The right hole has type 
		$\floattype{\{\clt{3.5}, \clt{4.5}\}}{\{3.5,4.5\}}$,
		so it may only be instantiated by one of the values 
		$3.5$ or $4.5$.
	\end{itemize}
	
	An expression $e : \bool$ has skeleton $s$ iff it is of the form
	\[
		e = c_1 < c_2
	\]
	where $c_1 \in \mathbb{R}$ and 
	$c_2 \in \{3.5,4.5\}$.
	
	Therefore, the admissible parameter tuples are precisely
	\[
		A_{s:\tau}
		=
		\{ (c_1,c_2) \mid c_1 \in \mathbb{R},\; c_2 \in \{3.5,4.5\} \}.
	\]
	
	Equivalently,
	\[
		A_{s:\tau}
		=
		\mathbb{R} \times \{3.5,4.5\}
		\subseteq
		\mathbb{R}^2.
	\]
	
	Thus, expressions with this syntactic shape are parameterized by a two-dimensional space whose first coordinate ranges over $\mathbb{R}$ and whose second coordinate ranges over a finite set.
\end{example}


\begin{definition}[$\sigma$-algebra construction on $A_{s;\tau}$]\label{def:sigma-algebra-construction-holes}
	Let $s : \tau \in \skel_{\tau}$. Endow $A_{s;\tau}$ with the subspace $\sigma$-algebra inherited from the standard Borel $\sigma$-algebra on $\R^{\holes(s)}$:

	\[
		\mathcal{F}_{s;\tau} := \left\{A_{s;\tau} \cap M \mid M \in \mathcal{B}(\R^{\holes(s)})\right\}.
	\]
\end{definition}

\begin{lemma}\label{lem:measurable-a-st}
	$(A_{s;\tau}, \mathcal{F}_{s;\tau})$ is a Borel measurable space.
\end{lemma}

\begin{proof}
	For a fixed $s : \tau \in \skel_{\tau}$, $(\R^{\holes(s)}, \mathcal{B}(\R^{\holes(s)}))$ is a Borel measurable space and $A_{s; \tau} \subseteq \R^{\holes(s)}$. Thus, by \Cref{def:subspace-sigma-algebra}, $(A_{s;\tau}, \{A_{s;\tau} \cap M \mid M \in \mathcal{B}(\R^{\holes(s)})\}) = (A_{s;\tau}, \mathcal{F}_{s;\tau})$ is a Borel measurable space.
\end{proof}

Now, we map this measurable space to the space of actual expressions that share that skeleton. 

\begin{definition}[Expressions of a skeleton]
	For a given $s : \tau \in \skel_{\tau}$, define the set of well-typed expressions having skeleton $s$ as:
	\[
		Expr_{s;\tau}
		:= 
		\left\{\, e : \tau \in Expr_{\tau} \mid \decompose(e : \tau)_1 = s \,\right\}
		\subseteq
		Expr_\tau.
	\]
\end{definition}

Intuitively, $Expr_{s;\tau}$ consists exactly of those expressions whose syntactic shape is fixed by $s$, and whose real constants vary according to the admissible hole assignments in $A_{s;\tau}$.

\begin{example}
	Consider the skeleton
	\[
		s : \tau 
		=
		(\square : \floattype{\{\clt{3.5}, \clt{4.5}\}}{\top}
		\;<\;
		\square : \floattype{\{\clt{3.5}, \clt{4.5}\}}{\{3.5,4.5\}})
		: \bool.
	\]
	
	An expression $e : \bool$ belongs to $Expr_{s;\tau}$ iff it is of the form
	\[
		e = c_1 < c_2
	\]
	where $c_1 \in \mathbb{R}$ and $c_2 \in \{3.5,4.5\}$.
	
	Therefore,
	\[
		Expr_{s;\tau}
		=
		\left\{
			c_1 < c_2
			\;\middle|\;
			c_1 \in \mathbb{R},\; c_2 \in \{3.5,4.5\}
		\right\}.
	\]
\end{example}

\begin{definition}[$\sigma$-algebra construction on $Expr_{s;\tau}$]\label{def:sigma-algebra-construction-skeletons}
	Let $s : \tau \in \skel_{\tau}$. Endow $Expr_{s;\tau}$ with the following $\sigma$-algebra:

	\[
		\mathcal{F}_{Expr_{s;\tau}}
		:=
		\left\{
			\{e : \tau \in Expr_{s;\tau} \mid \decompose(e : \tau)_2 \in B\} \mid B \in 
			\mathcal{F}_{s;\tau}
		\right\}.
	\]
\end{definition}

\begin{lemma}\label{lem:decompose-restriction-bijection}
	For a fixed $s : \tau \in \skel_{\tau}$, the map 
	\[
		\decompose_{s;\tau} \colon Expr_{s;\tau} \to A_{s;\tau},
		\qquad
		e \mapsto \decompose(e)_2.
	\]
	is a bijection.
\end{lemma}

\begin{proof}
	Let $s : \tau \in \skel_{\tau}$. 

	\paragraph{(i) Injectivity:}

	Suppose $e_1, e_2 \in Expr_{s;\tau}$ satisfy $\decompose_{s;\tau}(e_1) = \decompose_{s;\tau}(e_2)$. By definition of $\decompose_{s;\tau}$, $\decompose_{s;\tau}(e_1) = \decompose_{s;\tau}(e_2) = \vec{r}$ for some $\vec{r} \in A_{s;\tau}$. Since both $e_1$ and $e_2$ have skeleton $s$ by assumption, we have $e_1 = s[\vec{r}]$ and $e_2 = s[\vec{r}]$, hence $e_1 = e_2$.

	\paragraph{(ii) Surjectivity:}

	Let $\vec{r} \in A_{s;\tau}$. By definition of $A_{s;\tau}$, there exists some expression $e \in Expr_{\tau}$ such that $\decompose(e) = (s,\vec{r})$. Then $e \in Expr_{s;\tau}$ by construction of $Expr_{s;\tau}$, and $\decompose_{s;\tau}(e) = \vec{r}$.

	\medskip

	$\decompose_{s;\tau}$ is injective and surjective, hence a bijection.
\end{proof}

\begin{lemma}[Isomorphism property]\label{lem:isomorphism}
	For $s \in \skel_{\tau}$,
	\[
		(Expr_{s;\tau}, \mathcal{F}_{Expr_{s;\tau}})
		\cong
		(A_{s;\tau}, \mathcal{F}_{s;\tau}).
	\]
\end{lemma}

\begin{proof}
	We know $(A_{s;\tau}, \mathcal{F}_{s;\tau})$ is a Borel measurable space by~\Cref{lem:measurable-a-st}. First, verify that $\mathcal{F}_{Expr_{s;\tau}}$ coincides with the $\sigma$-algebra generated by $\decompose_{s;\tau}$, i.e. $\mathcal{F}_{Expr_{s;\tau}} = \decompose_{s;\tau}^{-1}(\mathcal{F}_{s;\tau})$:
	\begin{align*}
		\mathcal{F}_{Expr_{s;\tau}} := \quad&\left\{
			\{e : \tau \in Expr_{s;\tau} \mid \decompose(e : \tau)_2 \in B\} \mid B \in 
			\mathcal{F}_{s;\tau}
		\right\} \tag{\Cref{def:sigma-algebra-construction-skeletons}} \\
		:= \quad&\left\{
			\decompose_{s;\tau}^{-1}(B) \mid B \in \mathcal{F}_{s;\tau}
		\right\} \tag{unfolding the preimage of $\decompose_{s;\tau}$} \\
		:= \quad&\decompose_{s;\tau}^{-1}(\mathcal{F}_{s;\tau}) \tag{\Cref{def:pullback}}
	\end{align*}
	By~\Cref{lem:decompose-restriction-bijection}, $\decompose_{s;\tau} \colon Expr_{s;\tau} \to A_{s;\tau}$ is a bijection. By~\Cref{lem:pullback-iso}, $\decompose_{s;\tau}$ induces an isomorphism $(Expr_{s;\tau}, \mathcal{F}_{Expr_{s;\tau}}) \cong (A_{s;\tau}, \mathcal{F}_{s;\tau})$.
\end{proof}

\begin{lemma}\label{lem:measurable-expr-st}
	$(Expr_{s;\tau}, \mathcal{F}_{Expr_{s;\tau}})$ is a Borel measurable space.
\end{lemma}

\begin{proof}
	By~\Cref{lem:isomorphism}.
\end{proof}

Now we combine the measurable spaces for each distinct skeleton into a single measurable space of all well-typed expressions of type $\tau$. The set of all well-typed expressions of type $\tau$ is thus the disjoint union:
\[
	Expr_{\tau} = \bigsqcup_{s : \tau \in \skel_{\tau}} Expr_{s;\tau}.
\]
This union is disjoint because every typed expression $e : \tau$ has a unique skeleton, namely $\decompose(e)_1$, so $e$ belongs to exactly one set $Expr_{s;\tau}$.

\begin{definition}[$\sigma$-algebra construction on $Expr_{\tau}$]\label{def:sigma-algebra-construction-expressions}
	Endow $Expr_\tau$ with the following $\sigma$-algebra: 
	\[
		\mathcal{F}_{Expr_{\tau}}
		:= \{A \subseteq Expr_{\tau} \mid \forall s \in \skel_{\tau} .\; A \cap Expr_{s;\tau} \in \mathcal{F}_{Expr_{s;\tau}}\}
	\]
\end{definition}

\begin{lemma}\label{lem:measurable-expr-t}
	$(Expr_{\tau}, \mathcal{F}_{Expr_{\tau}})$ is a measurable space.
\end{lemma}

\begin{proof}
	We verify that $\mathcal{F}_{Expr_{\tau}}$ is indeed a $\sigma$-algebra on $Expr_{\tau}$.

	\paragraph{(i) Contains $Expr_{\tau}$:}
	
	We show $Expr_{\tau} \in \mathcal{F}_{Expr_\tau}$. By \Cref{def:sigma-algebra-construction-expressions}, it suffices to show that $\forall s \in \skel_\tau$, $Expr_{\tau} \cap Expr_{s;\tau} \in \mathcal{F}_{Expr_{s;\tau}}$.

	Observe that:
	\begin{align*}
		&Expr_{\tau} \cap Expr_{s;\tau} \\
		=\; &Expr_{s;\tau} \tag{since $Expr_{s;\tau} \subseteq Expr_{\tau}$}
	\end{align*}

	Because $\mathcal{F}_{Expr_{s;\tau}}$ is a $\sigma$-algebra, it contains the set $Expr_{s;\tau}$. Hence,
	\[
		Expr_{s;\tau} \in \mathcal{F}_{Expr_{s;\tau}}.
	\]
	Therefore, $Expr_{\tau} \cap Expr_{s;\tau} \in \mathcal{F}_{Expr_{s;\tau}}$ for all $s \in \skel_{\tau}$. Therefore, $Expr_{\tau} \in \mathcal{F}_{Expr_\tau}$.

	\paragraph{(ii) Closure under complements:}

	Let $A \in \mathcal{F}_{Expr_\tau}$. We show that $Expr_\tau \setminus A \in \mathcal{F}_{Expr_\tau}$. By \Cref{def:sigma-algebra-construction-expressions}, it suffices to show that $\forall s \in \skel_\tau$, $(Expr_\tau \setminus A) \cap Expr_{s;\tau} \in \mathcal{F}_{Expr_{s;\tau}}$.
	
	Fix $s \in \skel_\tau$.  Observe that:
	\begin{align*}
		&(Expr_\tau \setminus A) \cap Expr_{s;\tau} \\
		=\; &(Expr_\tau \cap Expr_{s;\tau}) \setminus (A \cap Expr_{s;\tau}) \tag{$\cap$ distributes over $\setminus$} \\
		=\; &Expr_{s;\tau} \setminus (A \cap Expr_{s;\tau}). \tag{since $Expr_{s;\tau} \subseteq Expr_\tau$, so $Expr_\tau \cap Expr_{s;\tau} = Expr_{s;\tau}$}
	\end{align*}

	Since $A \in \mathcal{F}_{Expr_\tau}$ by assumption, by \Cref{def:sigma-algebra-construction-expressions}, we have $A \cap Expr_{s;\tau} \in \mathcal{F}_{Expr_{s;\tau}}$, and because $\mathcal{F}_{Expr_{s;\tau}}$ is a $\sigma$-algebra, it is closed under complements relative to $Expr_{s;\tau}$. Hence
	\[
		Expr_{s;\tau} \setminus (A \cap Expr_{s;\tau}) \in \mathcal{F}_{Expr_{s;\tau}}.
	\]
	Therefore $(Expr_\tau \setminus A) \cap Expr_{s;\tau} \in \mathcal{F}_{Expr_{s;\tau}}$ for all $s \in \skel_{\tau}$. Therefore, $Expr_\tau \setminus A \in \mathcal{F}_{Expr_\tau}$.

	\paragraph{(iii) Closure under countable unions:}

	Let $(A_n)_{n \in \Nats}$ be a sequence with $A_n \in \mathcal{F}_{Expr_\tau}$ for all $n$. We show that $\bigcup_{n \in \Nats} A_n \in \mathcal{F}_{Expr_\tau}$. By \Cref{def:sigma-algebra-construction-expressions}, it suffices to show that $\forall s \in \skel_\tau$,
	$\left(\bigcup_{n \in \Nats} A_n\right) \cap Expr_{s;\tau} \in \mathcal{F}_{Expr_{s;\tau}}$.
	
	Fix $s \in \skel_\tau$. Using distributivity of $\cap$ over $\cup$, we have
	\[
		\left(\bigcup_{n \in \Nats} A_n\right) \cap Expr_{s;\tau}
		=
		\bigcup_{n \in \Nats} \left(A_n \cap Expr_{s;\tau}\right).
	\]
	Since each $A_n \in \mathcal{F}_{Expr_\tau}$ by assumption, by \Cref{def:sigma-algebra-construction-expressions}, we have $A_n \cap Expr_{s;\tau} \in \mathcal{F}_{Expr_{s;\tau}}$ for all $n$. Because $\mathcal{F}_{Expr_{s;\tau}}$ is a $\sigma$-algebra, it is closed under countable unions, hence
	\[
		\bigcup_{n \in \Nats} \left(A_n \cap Expr_{s;\tau}\right) \in \mathcal{F}_{Expr_{s;\tau}}.
	\]
	Therefore $\left(\bigcup_{n \in \Nats} A_n\right) \cap Expr_{s;\tau} \in \mathcal{F}_{Expr_{s;\tau}}$ for all $s$, and thus $\bigcup_{n \in \Nats} A_n \in \mathcal{F}_{Expr_\tau}$.

	\medskip
	Since $\mathcal{F}_{Expr_{\tau}}$ is a $\sigma$-algebra on $Expr_{\tau}$, $(Expr_{\tau}, \mathcal{F}_{Expr_{\tau}})$ is a measurable space.
\end{proof}

\begin{definition}[Extension with error element]\label{def:meas-space-with-error}
	Let $(Expr_\tau,\mathcal F_{Expr_\tau})$ be the measurable space of well-typed
	expressions of type $\tau$. Define
	\[
	Expr_\tau + \bot := Expr_\tau \cup \{\bot\},
	\]
	equipped with the $\sigma$-algebra
	\[
	\mathcal F_{Expr_\tau+\bot}
	:=
	\{A,\; A\cup\{\bot\} \mid A\in \mathcal F_{Expr_\tau}\}.
	\]
\end{definition}

\begin{lemma}\label{lem:measurable-expr-t-error}
    $(Expr_\tau+\bot,\mathcal F_{Expr_\tau+\bot})$ is a measurable space.
\end{lemma}

\begin{proof}
	We verify that $\mathcal{F}_{Expr_\tau+\bot}$ is indeed a $\sigma$-algebra on $Expr_\tau+\bot$.

	By construction, every element of $\mathcal{F}_{Expr_\tau+\bot}$ is of the form $A$ or $A \cup \{\bot\}$ for some $A \in \mathcal{F}_{Expr_{\tau}}$.

	\paragraph{(i) Contains $Expr_{\tau} + \bot$:}

	Take $A = Expr_{\tau} \in \mathcal{F}_{Expr_{\tau}}$. Then, 
	\begin{align*}
		A \cup \{\bot\} = Expr_{\tau} \cup \{\bot\} = Expr_{\tau} + \bot \in \mathcal{F}_{Expr_{\tau} + \bot}.
	\end{align*}

	\paragraph{(ii) Closure under complements:}
	Let $C \in \mathcal{F}_{Expr_{\tau} + \bot}$. There are two cases:

	\begin{itemize}
		\item If $C = A$: Then, $(Expr_{\tau} + \bot) \setminus A = (Expr_{\tau} \setminus A) \cup \{\bot\}$. Since $\mathcal{F}_{Expr_{\tau}}$ is a $\sigma$-algebra, $Expr_{\tau} \setminus A \in \mathcal{F}_{Expr_{\tau}}$, so $(Expr_{\tau} \setminus A) \cup \{\bot\}$ is of the form $A \cup \{\bot\}$ for some $A \in \mathcal{F}_{Expr_{\tau}}$ and hence lies in $\mathcal{F}_{Expr_{\tau} + \bot}$.
		
		\item If $C = A \cup \{\bot\}$: Then, $(Expr_{\tau} + \bot) \setminus (A \cup \{\bot\}) = Expr_{\tau} \setminus A$. Since $\mathcal{F}_{Expr_{\tau}}$ is a $\sigma$-algebra, $Expr_{\tau} \setminus A \in \mathcal{F}_{Expr_{\tau}}$, so $Expr_{\tau} \setminus A$ is of the form $A$ for some $A \in \mathcal{F}_{Expr_{\tau}}$ and hence lies in $\mathcal{F}_{Expr_{\tau} + \bot}$.
	\end{itemize}

	In both cases the complement lies in $\mathcal{F}_{Expr_{\tau} + \bot}$.

	\paragraph{(iii) Closure under countable unions:}

	Let $(C_n)_{n \in \Nats}$ be a sequence with $C_n\in\mathcal{F}_{Expr_\tau+\bot}$ for all $n$. For each $n$, define $A_n := C_n \cap Expr_\tau \in \mathcal{F}_{Expr_\tau}$, where the latter holds since $C_n$ is of the form $A$ or $A \cup \{\bot\}$ for some $A \in \mathcal{F}_{Expr_\tau}$, and in either case $C_n \cap Expr_\tau = A$.

	Let $A := \bigcup_{n\in\mathbb{N}} A_n$. Since $\mathcal{F}_{Expr_\tau}$ is a $\sigma$-algebra, $A\in\mathcal{F}_{Expr_\tau}$. Now observe:
	\[
		\bigcup_{n\in\mathbb{N}} C_n = A \cup \left(\bigcup_{n\in\mathbb{N}} (C_n\setminus Expr_\tau)\right).
	\]
	The second factor is either $\emptyset$ (if every $C_n = A_n$) or $\{\bot\}$ (if some $C_n = A_n \cup \{\bot\}$). In the first sub-case, $\bigcup_n C_n = A$ is of the form $A$ for some $A \in \mathcal{F}_{Expr_{\tau}}$; in the second, $\bigcup_n C_n = A\cup\{\bot\}$ is of the form $A \cup \{\bot\}$ for some $A \in \mathcal{F}_{Expr_{\tau}}$. In both cases $\bigcup_n C_n \in \mathcal{F}_{Expr_\tau+\bot}$.
\end{proof}


Finally, we can equip this measurable space $(Expr_{\tau} + \bot, \mathcal{F}_{Expr_{\tau} + \bot})$ with a subprobability measure such that $\mu(Expr_{\tau} + \bot) \leq 1$. 

We give some basic lemmas:

\begin{lemma}\label{lem:measurable-subsets}
	Let $(Expr_{\tau} + \bot, \mathcal{F}_{Expr_{\tau} + \bot})$ be the Borel measurable space of well-typed expressions with error from \Cref{def:meas-space-with-error}. Let $\vals_{\tau} \subset Expr_{\tau}$ be the set of all well-typed value expressions. Let $N_{\tau} \subset Expr_{\tau}$, defined as $N = Expr_{\tau} \setminus \vals_{\tau}$, be the set of all well-typed non-value expressions. Let $\{\bot\}$ be the singleton set containing the error element $\bot$. The following are true:
	\begin{enumerate}
		\item\label{lem:measurable-vals} The set $\vals_{\tau} \in \mathcal{F}_{Expr_{\tau} + \bot}$.
		\item\label{lem:measurable-error} The set $\{\bot\} \in \mathcal{F}_{Expr_{\tau} + \bot}$.
		\item\label{lem:measurable-nonvals} The set $\nonvals_{\tau} \in \mathcal{F}_{Expr_{\tau} + \bot}$.
	\end{enumerate}
\end{lemma}

\begin{proof}
	We prove the measurability of $\vals_\tau$, $\{\bot\}$, and $N_{\tau}$ in separate parts. For each, we show that it belongs to the $\sigma$-algebra $\mathcal{F}_{Expr_{\tau} + \bot}$ on the extended expression space $Expr_{\tau} + \bot$.

	\begin{enumerate}
		\item \textbf{(Measurability of $\vals_{\tau}$):} Every expression skeleton $s \in \skel_{\tau}$ is either a value skeleton or a non-value skeleton, depending on whether the syntactic form of $s$ is a value. An expression $e \in Expr_{\tau}$ is a value iff its skeleton is a value skeleton.
		
		To show $\vals_{\tau} \in \mathcal{F}_{Expr_{\tau}}$, by~\Cref{def:sigma-algebra-construction-expressions}, it suffices to show that $\forall s \in \skel_{\tau}$, $\vals_{\tau} \cap Expr_{s;\tau} \in \mathcal{F}_{Expr_{s;\tau}}$.

		Fix $s \in \skel_{\tau}$:
		\begin{itemize}
			\item If $s$ is a value skeleton, then every expression in $Expr_{s;\tau}$ is a value, so $\vals_{\tau} \cap Expr_{s;\tau} = Expr_{s;\tau} \in \mathcal{F}_{Expr_{s;\tau}}$.
			\item If $s$ is a non-value skeleton, then no expression in $Expr_{s;\tau}$ is a value, so $\vals_{\tau} \cap Expr_{s;\tau} = \emptyset \in \mathcal{F}_{Expr_{s;\tau}}$.
		\end{itemize} 

		Thus, $\vals_{\tau} \cap Expr_{s;\tau} \in \mathcal{F}_{Expr_{s;\tau}}$ for any $s \in \skel_{\tau}$, so $\vals_{\tau} \in \mathcal{F}_{Expr_{\tau}} \subseteq \mathcal{F}_{Expr_{\tau} + \bot}$.

		\item \textbf{(Measurability of $\{\bot\}$):} Taking $A = \emptyset$ in 
	        \Cref{def:meas-space-with-error} gives $\{\bot\} \in \mathcal{F}_{Expr_\tau + \bot}$.

        \item \textbf{(Measurability of $N_{\tau}$):} $N_{\tau} = Expr_\tau \setminus \vals_{\tau}$ is the 
        complement of $\vals_{\tau}$ in $Expr_\tau$. Since $\vals_{\tau} \in \mathcal{F}_{Expr_\tau}$ 
        and $\mathcal{F}_{Expr_\tau}$ is a $\sigma$-algebra, $N_{\tau} \in 
        \mathcal{F}_{Expr_\tau} \subseteq \mathcal{F}_{Expr_\tau + \bot}$.
	\end{enumerate}
\end{proof}

Note that in general, given $(X, \Sigma_X)$ and $(Y, \Sigma_Y)$ to be measurable spaces, $((X \times Y) + \bot)$ forms a measurable space, since the product $(X \times Y)$ is measurable and for any measurable space $(Z, \Sigma_Z)$, the space $(Z + \bot, \Sigma_{Z + \bot})$ with $\Sigma_{Z + \bot}$ generated by $\Sigma_Z \cup {{\bot}}$ is measurable. 

\subsection{Distribution Monad with Error}

Let $X$ be a set equipped with a suitable measurable space so that $\mathcal{D}(X + \bot)$ denotes the set of subprobability measures over $X + \bot$. We will make heavy use of the following generic monad construction, for which one easily verifies the monad laws:

\begin{definition}[Distribution Monad with Error]
	\label{def:dist-monad}
	We define 
	\begin{align*}
		& \text{return} \colon X \to \mathcal{D}(X + \bot), \quad \text{return}(x) = \delta_x \\
		& \text{bind} \colon \mathcal{D}(X + \bot) \to ((X + \bot) \to \mathcal{D}(Y + \bot)) \to \mathcal{D}(Y + \bot), \\
		&\text{bind}(\mu)(f)(A) = \int_{X + \bot} f(x)(A) \cdot \mu(dx)~,
	\end{align*}
	where $\delta_x$ denotes the Dirac measure for $x$, and $f$ is a sub-Markov kernel.
%
\end{definition}

\begin{definition}[$\bot$-extension]\label{def:bot-extension}
	Define the $\bot$-extension of $g : X \to \mathcal{D}(Y + \bot)$ to be $g_{\bot} \colon (X + \bot) \to \mathcal{D}(Y + \bot)$, obtained from $g$ by:
	\[
	g_\bot(e) =
	\begin{cases}
		\delta_\bot & \text{if } e =\bot \\
		g(e) & \text{otherwise}~.
	\end{cases}
	\]
\end{definition}

\begin{lemma}[$\bot$-extension is a sub-Markov kernel]\label{lem:bot-extension-kernel}
	Given a sub-Markov kernel $f : X \to \mathcal{D}(Y + \bot)$, its $\bot$ extension $f_{\bot} : (X + \bot) \to \mathcal{D}(Y + \bot)$ is also a sub-Markov kernel.
\end{lemma}
	
\begin{proof}
	We verify the two conditions for $f_{\bot}$ to be a sub-Markov kernel.
	
	\paragraph{(1) Subprobability measure:}
	We show for each $e \in X + \bot$, the map 
	\[
		A \mapsto f_{\bot}(e)(A)
	\]
	is a subprobability measure.

	Fix $e \in X + \bot$. 
	\begin{itemize}
	\item If $e = \bot$, then $f_{\bot}(\bot) = \delta_{\bot}$, the Dirac measure at $\bot$, which is a subprobability measure on $Y + \bot$.
	
	\item If $e \in X$, then $f_{\bot}(e) = f(e)$, and since $f$ is a sub-Markov kernel by assumption, $f(e)$ is a subprobability measure on $Y + \bot$.
	\end{itemize}
	
	Thus, for all $e \in X + \bot$, the map $A \mapsto f_{\bot}(e)(A)$ is a subprobability measure.
	
	\paragraph{(2) Measurable function:}
	We show for each $A \in \mathcal{F}_{Y + \bot}$, the map $e \mapsto f_{\bot}(e)(A)$ is a measurable function.

	Fix $A \in \mathcal{F}_{Y + \bot}$. Define
	\[
		g \colon (X + \bot) \to [0,1], \qquad g(e) := f_{\bot}(e)(A).
	\]
	We prove $g$ is measurable by showing $g^{-1}(B) \in \mathcal{F}_{X + \bot}$ for every $B \in \mathcal{B}([0,1])$. It suffices to prove that for every $a \in \R$,
	\[
		g^{-1}((a,\infty) \cap [0,1]) \in \mathcal{F}_{X + \bot}
	\]
	because the family of sets $\{(a,\infty) \cap [0,1] : a \in \R\}$ generates $\mathcal{B}([0,1])$. 

	Because $g(e)\in[0,1]$ for all $e\in X+\bot$ by definition, we have for every $a\in\mathbb{R}$,
	\begin{align*}
		g^{-1}\big((a,\infty)\cap[0,1]\big)
		&= \{e\in X+\bot \mid g(e)\in (a,\infty)\cap[0,1]\} \\
		&= \{e\in X+\bot \mid (g(e)\in (a,\infty)) \wedge (g(e)\in [0,1])\} \\
		&= \{e\in X+\bot \mid (g(e)>a) \wedge (g(e)\in [0,1])\} \\
		&= \{e\in X+\bot \mid g(e)>a\},
	\end{align*}
	where the last step uses that $g(e)\in[0,1]$ always holds.
	Hence it suffices to show that $\{e\in X+\bot \mid g(e)>a\}\in\mathcal{F}_{X+\bot}$
	for every $a\in\mathbb{R}$.

	\begin{itemize}
		\item If $a \geq 1$: $\{e\in X+\bot \mid g(e)>a\} = \emptyset \in \mathcal{F}_{X + \bot}$.
		\item If $a < 0$: $\{e\in X+\bot \mid g(e)>a\} = X + \bot \in \mathcal{F}_{X + \bot}$.
		\item If $0 \leq a < 1$: $\{e \in X + \bot \mid g(e) > a \}$.
	\end{itemize}

	The only nontrivial case is the last. It splits into two sets:
	\begin{align*}
		\{e \in X + \bot \mid g(e) > a \} = &\{e \in X + \bot \mid f_{\bot}(e)(A) > a \} \\
		= &\{e \in X \mid f_{\bot}(e)(A) > a \} \cup \{\bot \mid f_{\bot}(e)(A) > a \} \\
		= &\{e \in X \mid f(e)(A) > a \} \cup \{\bot \mid \dirac{\bot}(A) > a \}
	\end{align*}

	The first set is measurable in $X$ (and hence $X + \bot$) since $f$ is a sub-Markov kernel by assumption. The second set is either $\emptyset \in \mathcal{F}_{X + \bot}$ or $\{\bot\} \in \mathcal{F}_{X + \bot}$. Since $ \mathcal{F}_{X + \bot}$ is closed under countable unions, then $g^{-1}(B) \in \mathcal{F}_{X + \bot}$ for every $B \in \mathcal{B}([0,1])$. Thus, for all $A \in \mathcal{F}_{Y + \bot}$, the map $e \mapsto f_{\bot}(e)(A)$ is a measurable function.
	
	\medskip
	
	Since both conditions hold, $f_{\bot}$ is a sub-Markov kernel.
\end{proof}

\begin{lemma}[Dirac return is a sub-Markov kernel]\label{lem:dirac-return-kernel}
	For any measurable space $(X,\mathcal{F}_X)$, the map
	\[
		\mathsf{return} : X \to \mathcal{D}(X+\bot),
		\qquad x \mapsto \delta_x .
	\]
	is a sub-Markov kernel.
\end{lemma}

\begin{proof}
	We verify the two conditions for $\mathsf{return}$ to be a sub-Markov kernel.

	\paragraph{(1) Subprobability measure:}
	We show for each $x \in X$, the map 
	\[
		A \mapsto \mathsf{return}(x)(A) = \delta_x(A)
	\]
	is a subprobability measure.

	Fix $x \in X$. $\dirac{x}$ is the Dirac measure at $x$, hence a probability measure and therefore a subprobability measure on $X+\bot$. Thus, for every $x\in X$, the map $A \mapsto \delta_x(A)$ is a subprobability measure.

	\paragraph{(2) Measurable function:}
	We show that for each $A \in \mathcal{F}_{X+\bot}$, the map
	\[
		x \mapsto \mathsf{return}(x)(A)=\delta_x(A)
	\]
	is a measurable function.

	Fix $A \in \mathcal{F}_{X+\bot}$. Define
	\[
		g : X \to [0,1], \qquad g(x) := \delta_x(A).
	\]
	We prove $g$ is measurable by showing $g^{-1}(B)\in\mathcal{F}_X$ for every
	$B\in\mathcal{B}([0,1])$. It suffices to prove that for every $a\in\mathbb{R}$,
	\[
		g^{-1}\big((a,\infty)\cap[0,1]\big)\in\mathcal{F}_X,
	\]
	because the family $\{(a,\infty)\cap[0,1] : a\in\mathbb{R}\}$ generates
	$\mathcal{B}([0,1])$.

	Because $g(e)\in[0,1]$ for all $e\in X+\bot$, we have for every $a\in\mathbb{R}$,
	\begin{align*}
		g^{-1}\big((a,\infty)\cap[0,1]\big)
		&= \{e\in X \mid g(e)\in (a,\infty)\cap[0,1]\} \\
		&= \{e\in X \mid (g(e)\in (a,\infty)) \wedge (g(e)\in [0,1])\} \\
		&= \{e\in X \mid (g(e)>a) \wedge (g(e)\in [0,1])\} \\
		&= \{e\in X \mid g(e)>a\},
	\end{align*}
	where the last step uses that $g(e)\in[0,1]$ always holds.
	Hence it suffices to show that $\{e\in X \mid g(e)>a\}\in\mathcal{F}_{X}$
	for every $a\in\mathbb{R}$.

	\begin{itemize}
		\item If $a \geq 1$: $\{x\in X \mid g(x)>a\} = \emptyset \in \mathcal{F}_{X}$.
		\item If $a < 0$: $\{x\in X \mid g(x)>a\} = X \in \mathcal{F}_{X}$.
		\item If $0 \leq a < 1$: $\{x \in X \mid g(x) > a \}$.
	\end{itemize}

	The only nontrivial case is the last. Using $g(x)=\delta_x(A)$ and the fact that
	$\delta_x(A)\in\{0,1\}$ for all $x\in X$, we have for $0\le a<1$:
	\begin{align*}
		\{x \in X \mid g(x) > a \}
		&= \{x \in X \mid \delta_x(A) > a \} \\
		&= \{x \in X \mid \delta_x(A) = 1 \} \\
		&= \{x \in X \mid x \in A\} \\
		&= A \cap X.
	\end{align*}
	Since $A\in \mathcal{F}_{X+\bot}$ and $\mathcal{F}_X = \{X \cap A \mid A \in \mathcal{F}_{X+\bot}\}$ is the subspace $\sigma$-algebra
	on $X\subseteq X+\bot$, it follows that $A\cap X \in \mathcal{F}_X$.

	Therefore $\{x\in X \mid g(x)>a\}\in\mathcal{F}_X$ for all $a\in\mathbb{R}$, so
	$g^{-1}\big((a,\infty)\cap[0,1]\big)\in\mathcal{F}_X$ for all $a$.
	Hence $g^{-1}(B)\in\mathcal{F}_X$ for every $B\in\mathcal{B}([0,1])$, and $g$ is measurable.
\end{proof}

\begin{lemma}[Zero kernel is a sub-Markov kernel]\label{lem:zero-kernel}
    Let $(X, \mathcal{F}_X)$ and $(Y, \mathcal{F}_Y)$ be measurable spaces. The constant-zero map
    \[
        f : X \to \mathcal{D}(Y), \qquad f(x)(A) := 0
    \]
    is a sub-Markov kernel.
\end{lemma}

\begin{proof}
    We verify the two kernel conditions.

    \paragraph{(1) Subprobability measure.}
	We show for each $x \in X$, the map
	\[
		A \mapsto f(x)(A)
	\]
	is a subprobability measure.

    Fix any $x \in X$. Then $f(x)$ is the zero measure: it satisfies 
    $\sigma$-additivity since
    \[
        f(x)\!\left(\bigcup_{i} A_i\right)
        = 0 = \sum_i 0
        = \sum_i f(x)(A_i)
    \]
    for every pairwise disjoint family $(A_i)_i$, and total mass 
    $f(x)(Y) = 0 \le 1$.

    \paragraph{(2) Measurable function.}
    For any $A \in \mathcal{F}_Y$, the map $x \mapsto f(x)(A) = 0$ is the 
    constant-zero function on $X$, hence measurable.

    Thus both conditions hold, so $f$ is a sub-Markov kernel.
\end{proof}

\begin{lemma}[Composition of sub-Markov kernels]\label{lem:kernel-composition}
	Let $(X, \Sigma_X)$, $(Y, \Sigma_Y)$, and $(Z, \Sigma_Z)$ be measurable spaces.
	If 
	\[
		P : X \to \mathcal{D}(Y)
		\quad\text{and}\quad
		Q : Y \to \mathcal{D}(Z)
	\]
	are sub-Markov kernels, then their composition
	\[
		(P \monbind Q)(x,A)
		\;:=\;
		\int_Y Q(y,A)\,P(x,dy)
	\]
	defines a sub-Markov kernel
	\[
		P \monbind Q : X \to \mathcal{D}(Z).
	\]
\end{lemma}

\begin{proof}
	For each $x \in X$, the map $A \mapsto (P \monbind Q)(x,A)$ is a measure by the standard integral construction. Its total mass is bounded by
	\[
		(P \monbind Q)(x,Z)
		=
		\int_Y Q(y,Z)\,P(x,dy)
		\leq
		\int_Y 1\,P(x,dy)
		=
		P(x,Y)
		\leq 1.
	\]
	For each $A \in \Sigma_Z$, measurability of $y \mapsto Q(y,A)$ and the standard measurability theorem for integration against a kernel imply that $x \mapsto \int_Y Q(y,A)\,P(x,dy)$ is measurable; see~\cite[Lemma 1.38]{kallenberg1997foundations}. Hence the composition is a sub-Markov kernel.
\end{proof}

\subsection{Small-Step Semantics}
\label{app:small-step}

\begin{table}[t]
	\allowdisplaybreaks
	\caption{Small-step operational semantics. In the last case, the existence of suitable $v_i$ and $w_j$ is guaranteed since $\distinguishes{B_i}{V_i}$ and $e_i\in V_i$ for $i \in \{1,2\}$. Here, every function of the form $f(g) = \lambda g. \dirac{\ldots}$ is extended to $f_{\bot}$ in the sense of ~\Cref{def:bot-extension}.}
	\label{tab:smallstep-1}%
	\begin{center}
		\begin{adjustbox}{width=\textwidth}
		\begin{tabular}{l l l l l}
			\toprule
			\toprule
			$\mathbf{e \boldsymbol{:} \boldsymbol{\tau}}$ & \textbf{constraints}  & $\boldsymbol{\sem{e : \tau}}$   \\
			\hline
			%
			\multirow{1}{*}{$\diverge : \tau$} 
			
			& & $\lambda A.0$  \\
			\midrule
			%
		    \multirow{1}{*}{$v : \tau$} 
			
			& & $\dirac{v:\tau}$ \\
			\midrule
			%
			%
			%
			\multirow{2}{*}{$\letkw \; x = e_1 : \tau_1 \; \inkw \; e_2 : \tau_2 : \tau_2$} 
			
			& $e_1=v$ & $\dirac{e_2[v/x] : \tau_2}$ \\
			& $e_1$ no val.\ & $\sem{ e_1 : \tau_1 } \gg\!=  \lambda g. \dirac{ \letkw \; x = g \; \inkw \; e_2 : \tau_2 : \tau_2} $ \\
			\midrule 
			%
			\multirow{3}{*}{$\ifkw \; e_1 : \bool \; \thenkw \; e_2 : \tau \; \elsekw \; e_3 : \tau : \tau$} 
			
			& $e_1=\true$ & $\dirac{\expr_2 : \tau} $ \\
			& $e_1=\false$ & $\dirac{\expr_3 : \tau} $ \\
			& o.w.\ & $\sem{ e_1 : \bool} \gg\!= \lambda g. \dirac{ \ifkw \; g \; \thenkw \; e_2 : \tau \; \elsekw \; e_3 : \tau : \tau} $ \\
			\midrule 
			%
			\multirow{3}{*}{$\observekw\; (e : \bool)  : \unit$} 
			
			& $e = \true$ & $\dirac{\unit}$ \\
			& $e=\false$ & $\dirac{\bot}$ \\
			& o.w.\ & $\sem{e: \bool} \gg\!= \lambda g. \dirac{\observekw\; (g : \bool)  : \unit}$ \\
			\midrule
			%
			\multirow{4}{*}{\makecell{$e_1 : \tau_1 < e_2 : \tau_2 : \bool$ \\
			for $\tau_1 = \floattype{B}{V_1}, \tau_2 = \floattype{B}{V_2} $ \\
			or \\
			$\tau_1 = \tau_2 = \fin{n} $ }} 
			 
			& $e_1 = v_1, e_2 = v_2, v_1 < v_2 $ & $\dirac{\true : \bool}$ \\
			& $e_1 = v_1, e_2 = v_2, v_1 \geq v_2$ & $\dirac{\false : \bool}$  \\
			&$e_1 = v_1, e_2 \neq v_2$ & $\sem{ e_2 : \tau_2 } \gg\!= \lambda g. \dirac{ v_1 : \tau_1 < g : \bool}$ \\
			& o.w.\ & $\sem{ e_1 : \tau_1 } \gg\!= \lambda g. \dirac{ g < e_2 : \tau_2 : \bool}$  \\
			\midrule
			%
			\multirow{3}{*}{$\uniform(e_1 : \tau_1, e_2 : \tau_2) : \floattype{B}{V}$} 
			
			& $e_1 = v_1, e_2 = v_2, v_1 < v_2$ & $\text{Uniform}(v_1, v_2) \gg\!= \lambda v. \dirac{(v : \floattype{B}{V})}$ \\
			& $e_1 = v_1, e_2 = v_2, v_1 \geq v_2$ & $\sem{\diverge : \tau}$ \\
			& $e_1 = v_1, e_2 \neq v_2$ & $\sem{ e_2 : \tau_2 } \gg\!= \lambda g. \dirac{ \uniform(v_1 : \tau_1, g) : \tau }$ \\
			& o.w.\ & $\sem{ e_1 : \tau_1 } \gg\!= \lambda g. \dirac{ \uniform(g, e_2 : \tau_2) : \tau } $ \\
			\midrule
			%
			\multirow{1}{*}{$\discrete(v_0, \ldots, v_{n-1}) : \fin{n}$} 
			
			& & $\text{Discrete}{(v_0, \ldots, v_{n-1})}$  \\
			%
			%
			\midrule 
			%
			\multirow{5}{*}{\makecell{$\dc(e_1 : \tau_1, e_2 : \tau_2) : \fin{|B| + 1}$ \\
			$\tau = \fin{|B| + 1}$, \\
			$\tau_1 = \fin{|B_1| + 1}, \tau_2 = \fin{|B_2| + 1}$,\\
			$V_1 = \{v_1,\ldots,v_n\}, V_2 = \{w_1,\ldots,w_m\}$}} 
			
			& \makecell[l]{$e_1=v',\ e_2=w'$,\\ $V_1=\emptyset \lor V_2=\emptyset$} & $\sem{\diverge : \tau}$ \\
			& \makecell[l]{$e_1 = \discretize{v_i : \floattype{B_1}{V_1}}$,\\ $e_2 = \discretize{w_j : \floattype{B_2}{V_2}}$} & $\sem{\discretize{\uniform(v_i, w_j) : \floattype{B}{\top}} : \tau}$ \\
			& $ e_1 = v', e_2 \neq w'$ & $\sem{ e_2 : \tau_2 } \gg\!= \lambda g. \dirac{ \dc(v' : \tau_1, g) }$ \\
			& o.w.\ & $\sem{ e_1 : \tau_1 } \gg\!= \lambda g. \dirac{ \dc(g, e_2 : \tau_2) }$ \\
			& & \\
			%
			\bottomrule
			\bottomrule
		\end{tabular}
	\end{adjustbox}
	\end{center}
\end{table}%

\begin{table}[t]
	\allowdisplaybreaks
	\caption{Small-step operational semantics (continued). Here, every function of the form $f(g) = \lambda g. \dirac{\ldots}$ is extended to $f_{\bot}$ in the sense of \Cref{def:bot-extension}.}
	\label{tab:smallstep-2}%
	\begin{center}
		\begin{adjustbox}{width=\textwidth}
		\begin{tabular}{l l l l l}
			\toprule
			\toprule
			$\mathbf{e \boldsymbol{:} \boldsymbol{\tau}}$ & \textbf{constraints}  & $\boldsymbol{\sem{e : \tau}}$   \\
			\hline 
			\multirow{4}{*}{\makecell{$e_1 : \tau_1 \logand e_2 : \tau_2 : \bool$ \\
			$\tau_1 = \bool, \tau_2 = \bool$}} 
			
			& $e_1 = b_1 = \true, e_2 = b_2 = \true $ & $\dirac{\true : \bool}$ \\
			& $e_1 = b_1, e_2 = b_2 $ & $\dirac{\false : \bool}$  \\
			&$e_1 = b_1, e_2 \neq b_2$ & $\sem{ e_2 : \tau_2 } \gg\!= \lambda g. \dirac{ b_1 : \tau_1 \logand g : \bool}$ \\
			& o.w.\ & $\sem{ e_1 : \tau_1 } \gg\!= \lambda g. \dirac{ g \logand e_2 : \tau_2 : \bool}$  \\
			\midrule
			%
			\multirow{4}{*}{\makecell{$e_1 : \tau_1 \logor e_2 : \tau_2 : \bool$ \\
			$\tau_1 = \bool, \tau_2 = \bool$}} 
			
			& $e_1 = b_1 = \false, e_2 = b_2 = \false $ & $\dirac{\false : \bool}$ \\
			& $e_1 = b_1, e_2 = b_2 $ & $\dirac{\true : \bool}$  \\
			&$e_1 = b_1, e_2 \neq b_2$ & $\sem{ e_2 : \tau_2 } \gg\!= \lambda g. \dirac{ b_1 : \tau_1 \logor g : \bool}$ \\
			& o.w.\ & $\sem{ e_1 : \tau_1 } \gg\!= \lambda g. \dirac{ g \logor e_2 : \tau_2 : \bool}$  \\
			\midrule
			%
			\multirow{3}{*}{\makecell{$\lognot\; e : \tau : \bool$ \\
			$\tau = \bool$}} 
			
			& $e = \true $ & $\dirac{\false : \bool}$ \\
			& $e = \false $ & $\dirac{\true : \bool}$  \\
			& o.w.\ & $\sem{ e : \tau } \gg\!= \lambda g. \dirac{ \lognot\; g : \bool}$  \\
			\midrule
			%
			\multirow{3}{*}{\makecell{$(e_1 : \tau_1, e_2 : \tau_2) : \tau_1 * \tau_2$}} 
			
			& $e_1 = v_1, e_2 = v_2 $ & $\dirac{(v_1,v_2) : \tau_1 * \tau_2}$ \\
			&$e_1 = v_1, e_2 \neq v_2$ & $\sem{ e_2 : \tau_2 } \gg\!= \lambda g. \dirac{ (v_1 : \tau_1, f) : \tau_1 * \tau_2}$ \\
			& o.w.\ & $\sem{ e_1 : \tau_1 } \gg\!= \lambda g. \dirac{ (f, e_2 : \tau_2) : \tau_1 * \tau_2}$  \\
			\midrule
			%
			\multirow{2}{*}{\makecell{$\fstkw\; (e : \tau_1 * \tau_2) : \tau_1$ }} 
			
			& $e = (v_1, v_2)$ & $\dirac{v_1: \tau_1}$ \\ 
			& o.w. & $\sem{e : \tau_1 * \tau_2} \monbind \lambda g. \dirac{\fstkw\; e : \tau_1}$ \\
			\midrule
			%
			\multirow{2}{*}{\makecell{$\sndkw\; (e : \tau_1 * \tau_2) : \tau_2$ }} 
			
			& $e = (v_1, v_2)$ & $\dirac{v_2: \tau_2}$ \\ 
			& o.w. & $\sem{e : \tau_1 * \tau_2} \monbind \lambda g. \dirac{\sndkw\; e : \tau_2}$ \\
			\midrule
			%
			\multirow{3}{*}{\makecell{$(e_1 : \tau_1 \to \tau_2 \; e_2 : \tau_1) : \tau_2$ }} 
			
			& $e_1 = \funkw\; x \to e \text{ or } \fixkw\; g \; x := e, e_2 = v_2$ & $\dirac{e[v_2/x] : \tau_2}$ \\
			& $e_1 = \funkw\; x \to e \text{ or } \fixkw\; g \; x := e, e_2 \neq v_2$ & $\sem{e_2 : \tau_2} \monbind \lambda g. \dirac{e_1 : \tau_1 \to \tau_2 \; g : \tau_2}$ \\
			& o.w. & $\sem{e_1 : \tau_1 \to \tau_2} \monbind \lambda g. \dirac{f\; e_2 : \tau_1 : \tau_2}$ \\
			\midrule
			%
			\multirow{1}{*}{\makecell{$\text{nil} : \listty{\tau}$ }} 
			
			& & $\dirac{\text{nil} : \listty{\tau}}$ \\
			\midrule
			%
			\multirow{2}{*}{\makecell{$(e_1 : \tau_1 :: e_2 : \listty{\tau}) : \listty{\tau}$ }} 
			
			& $e_1 = v_1, e_2 = v_2$ & $\dirac{v_1::v_2 : \listty{\tau}}$ \\
			& $e_1 = v_1, e_2 \neq v_2$ & $\sem{e_2 : \listty{\tau}} \monbind \lambda g.\dirac{v_1 : \tau_1 ::g : \listty{\tau}}$ \\
			& o.w. & $\sem{e_1 : \tau_1} \monbind \lambda g.\dirac{f::e_2 : \listty{\tau} : \listty{\tau}}$ \\
			\midrule
			%
			\multirow{3}{*}{\makecell{$( \matchkw \; e : \listty{\tau_1}\; \withkw$ \\ 
			$\matchcase \text{nil} \rightarrow e_1 : \tau_2$ \\
			$\matchcase h :: t \rightarrow e_2 : \tau_2$ \\
			${\textnormal{\ttfamily\bfseries end}}) : \tau_2$ }} 
			
			& $e = \text{nil}$ & $\dirac{e_1 : \tau_2}$ \\
			& $e = v_h :: v_t $ & $\dirac{e_2[v_h/h, v_t/t] : \tau_2}$  \\
			& o.w. & $\sem{e : \listty{\tau_1}} \monbind \lambda g.\dirac{\matchkw \; g \; \withkw ... : \tau_2}$  \\

			& & \\
			%
			\bottomrule
			\bottomrule
		\end{tabular}
	\end{adjustbox}
	\end{center}
\end{table}%

Throughout this section, we denote by $e :\tau$ a well-typed expression of type $\tau$ and implicitly assume that all sub-expressions are annotated with types, as well. Moreover, for the typed expression $ e: \tau = \uniform(e_1 : \float[B_1; V_1], e_2 :\float[B_2;V_2]) : \float[B,V]$ with $B,B_1,B_2 \neq \top$ and $V_1,V_2\neq\top$, we introduce an auxiliary statement $\dc$ and let 
\[
	\discretize{ e: \tau} = \dc(\discretize{e_1 : \float[B_1; V_1]}, \discretize{e_2 :\float[B_2;V_2]})~.
\]
We define the semantics of $\dc$ in such a way that it behaves like the discretized program in the last case of \Cref{fig:discretization-code} but resolves the let-bindings and the conditional choices in a single step, which will simplify our proof. We equip this statement with the following typing rule :
\begin{align*}
	\inferrule[\textsc{TDiscreteCase}]
	{e_1 : \fin{|B_1| + 1} \\
		e_2 : \fin{|B_2| + 1} \\ \distinguishes{B_1}{V_1} \\ \distinguishes{B_2}{V_2}}
	{\dc(e_1, e_2) : \fin{|B|+1}}
\end{align*}

%
%
%
%
%
Equipped with these notions, we define the small-step semantics of our language:
\begin{definition}
	The small-step semantics $\sem{e :\tau} \in \mathcal{D}(Expr_{\tau} + \bot)$ is defined recursively on the structure of $e:\tau$ by the rules in \Cref{tab:smallstep-1} and \Cref{tab:smallstep-2}. 
\end{definition}
The missing probability mass in $\sem{e :\tau}$ is the probability of encountering an error due to ill-formed parameters fed into distributions (e.g., $\uniform(0.8,0.1)$). $\sem{e :\tau}(\bot)$ is the probability of encountering an observation failure.


\begin{lemma}[Measurability of constant extension map]\label{lem:constant-extension-measurable}
	Let $A \subseteq \R^n$ be equipped with the subspace $\sigma$-algebra
	\[
	\mathcal{F}_A := \{A \cap M \mid M \in \mathcal{B}(\R^n)\},
	\]
	and let $B \subseteq \R^d$ be equipped with the subspace $\sigma$-algebra
	\[
	\mathcal{F}_B := \{B \cap N \mid N \in \mathcal{B}(\R^d)\}.
	\]
	
	Fix constant vectors $c_1 \in \R^{m_1},\dots,c_k \in \R^{m_k}$.
	Define
	\[
	g : \R^n \to \R^d
	\]
	to be any map obtained by concatenating blocks consisting of
	
	\begin{itemize}
	\item copies of the input vector $\vec r \in \R^n$, and
	\item the constant vectors $c_1,\dots,c_k$,
	\end{itemize}
	
	in some fixed order. That is, $g(\vec r)$ is formed by inserting the vectors
	$\vec r$ and $c_j$ into a fixed sequence of blocks.
	
	Assume $g(A) \subseteq B$ and define $\phi := g|_A : A \to B$.
	Then $\phi$ is measurable.
\end{lemma}

\begin{example}[Constant extension maps]\label{ex:coord-wiring}
	The constant extension map allows constants to be inserted and the input vector
	$\vec r \in \R^n$ to appear multiple times. The following are examples of such maps:
	\begin{itemize}
		\item $g : \R^n \to \R^n, \qquad g(\vec r) = \vec r$

		\item $g : \R^n \to \R^{n + m_1}, \qquad g(\vec r) = (\vec r, c_1), \quad c_1 \in \R^{m_1}$

		\item $g : \R^n \to \R^{m_1 + n + m_2}, \qquad g(\vec r) = (c_1, \vec r, c_2), \quad c_1 \in \R^{m_1},\; c_2 \in \R^{m_2}$

		\item $g : \R^n \to \R^{n + m_2}, \qquad g(\vec r) = (\vec r, c_2), \quad c_2 \in \R^{m_2}$

		\item $g : \R^n \to \R^{m_1 + n + m_2 + n + m_3 + n}, \qquad g(\vec r) = (c_1, \vec r, c_2, \vec r, c_3, \vec r), \quad c_i \in \R^{m_i}$
	\end{itemize}
\end{example}

\begin{proof}
	We show that $\phi^{-1}(S) \in \mathcal{F}_A$ for each $S \in \mathcal{F}_B$. 

	Fix $S \in \mathcal{F}_B$. By definition of $\mathcal{F}_B$, there exists
	$N \in \mathcal{B}(\R^d)$ such that $S = B \cap N$. We compute:
    \begin{align*}
        \phi^{-1}(S)
        &= \{\vec{r} \in A \mid \phi(\vec{r}) \in B \cap N\} \\
		&= \{\vec{r} \in A \mid \phi(\vec{r}) \in B \cap \phi(\vec{r}) \in N\} \tag{definition of $\cap$} \\
        &= \{\vec{r} \in A \mid \phi(\vec{r}) \in N\} \tag{$\phi(\vec{r}) \in B$ for all $\vec{r} \in A$}\\
		&= \{\vec{r} \in A \mid g(\vec{r}) \in N\} \tag{definition of $\phi$}\\
		&= \{\vec{r} \in A \mid g(\vec{r}) \in N\} \cap \R^n \tag{$\vec{r}$ ranges over $A \subseteq \R^n$}\\
        &= A \cap \{\vec{r} \in \R^n \mid g(\vec{r}) \in N\} \tag{rewrite as $\cap$ with $A$}.
    \end{align*}
    It remains to show that $\{\vec{r} \in \R^n \mid (\vec{r},\vec{c}) \in N\} \in \mathcal{B}(\R^n)$.
    Define the total constant-extension map
    \[
        g \colon \R^n \to \R^d.
    \]
    We claim $g$ is continuous. A map into a product space is continuous if and only
    if each coordinate function is continuous~\cite[Theorem~19.6]{Munkres2000Topology}. The coordinate functions of $g$ are:
    \begin{itemize}
        \item the identity $\mathrm{id} \colon \R^n \to \R^n$, $\vec{r} \mapsto \vec{r}$, and
        \item the constant map $\kappa_{\vec{c}_i} \colon \R^n \to \R^{m_i}$, $\vec{r} \mapsto \vec{c}_i$.
    \end{itemize}
    Both are continuous: for any open $U \subseteq \R^n$, $\mathrm{id}^{-1}(U) = U$ is open;
    and for any open $V \subseteq \R^{m_i}$,
    \[
        \kappa_{\vec{c}_i}^{-1}(V) =
        \begin{cases} \R^n & \text{if } \vec{c}_i \in V, \\ \emptyset & \text{if } \vec{c}_i \notin V, \end{cases}
    \]
    both of which are open. Hence $g$ is continuous, and therefore Borel measurable.
    Applying this to $N \in \mathcal{B}(\R^d)$:
    \[
        \{\vec{r} \in \R^n \mid g(\vec{r}) \in N\}
        = g^{-1}(N) \in \mathcal{B}(\R^n).
    \]
    Hence $\phi_{\vec{c}}^{-1}(S) = A \cap g^{-1}(N) \in \mathcal{F}_A$
    by definition, so $\phi_{\vec{c}}$ is measurable.
\end{proof}

\begin{lemma}[Substitution is measurable]\label{lem:subst-measurable}
    Fix a typed expression $e : \tau_2$ and a variable $x : \tau_1$. The substitution map
    \[
        \mathrm{Sub}_{e,x} \colon Val_{\tau_1} \to Expr_{\tau_2},
        \qquad
        \mathrm{Sub}_{e,x}(v) := e[v/x],
    \]
    is measurable, where $Val_{\tau_1} \subseteq Expr_{\tau_1}$ carries the subspace $\sigma$-algebra inherited
    from $(Expr_{\tau_1}, \mathcal{F}_{Expr_{\tau_1}})$.
\end{lemma}

\begin{proof}
	We show that $\mathrm{Sub}_{e,x}^{-1}(A) \in \mathcal{F}_{Val_{\tau_1}}$ for all $A \in \mathcal{F}_{Expr_{\tau_2}}$. 
	
	Fix $A \in \mathcal{F}_{Expr_{\tau_2}}$. We know:
	\begin{align*}
		\mathrm{Sub}_{e,x}^{-1}(A) = &\{v \in Val_{\tau_1} \mid e[v/x] \in A\} \\
		= &\{v \in Val_{\tau_1} \mid e[v/x] \in A\} \cap Expr_{\tau_1} \tag{$v$ ranges over $Val_{\tau_1} \subseteq Expr_{\tau_1}$}\\
		= &Val_{\tau_1} \cap \{u \in Expr_{\tau_1} \mid e[u/x] \in A\} \tag{rewrite as $\cap$ with $Val_{\tau_1}$}\\
	\end{align*}
	Now, define:
    \[
        C := \{u \in Expr_{\tau_1} \mid e[u/x] \in A\}.
    \]
	
	We claim proving $C \in \mathcal{F}_{Expr_{\tau_1}}$ is sufficient because by~\Cref{def:subspace-sigma-algebra}, this gives $\mathrm{Sub}_{e,x}^{-1}(A) = Val_{\tau_1} \cap C \in \mathcal{F}_{Val_{\tau_1}}$. To show $C \in \mathcal{F}_{Expr_{\tau_1}}$, by~\Cref{def:sigma-algebra-construction-expressions}, it suffices to show that $\forall s \in \skel_{\tau_1}, C \cap Expr_{s;\tau_1} \in \mathcal{F}_{Expr_{s;\tau_1}}$. 
	
	Fix $s \in \skel_{\tau_1}$. We must prove:
	\begin{align*}
		&C \cap Expr_{s;\tau_1} \in \mathcal{F}_{Expr_{s;\tau_1}} \\
		\iff &\{u \in Expr_{\tau_1} \mid e[u/x] \in A\} \cap Expr_{s;\tau_1}\in \mathcal{F}_{Expr_{s;\tau_1}} \tag{definition of $C$} \\
		\iff &\{u \in Expr_{s; \tau_1} \mid e[u/x] \in A\} \in \mathcal{F}_{Expr_{s;\tau_1}}  \tag{$Expr_{s;\tau_1} \subseteq Expr_{\tau_1}$}\\
		\iff &\{s[\vec{r}] \in Expr_{s; \tau_1} \mid e[s[\vec{r}]/x] \in A\} \in \mathcal{F}_{Expr_{s;\tau_1}} \tag{rewrite $u$}\\
		\iff &\{\vec{r} \in A_{s;\tau_1} \mid e[s[\vec{r}]/x] \in A\} \in \mathcal{F}_{s;\tau_1} \tag{\Cref{lem:isomorphism}}
	\end{align*}
	Define
    \[
        R_s := \{\vec{r} \in A_{s;\tau_1} \mid e[s[\vec{r}]/x] \in A\}.
    \]

    Let $t \in \skel_{\tau_2}$ be the skeleton obtained from $e$ by replacing each
    occurrence of $x$ with $s$. Define a constant extension map:
	\[
		\psi_s \colon A_{s;\tau_1} \to A_{t;\tau_2}, \qquad \psi_s(\vec{r}) = \phi(\vec{r}, \vec{c_1}, ..., \vec{c_k})
	\]
	where $\vec{c_1}, ..., \vec{c_k} \in \R$ are the float constants appearing in $e$ (other than at occurrences of $x$), and $\phi$ permutes/interleaves the blocks into the order dictated by a left-to-right traversal of the holes of $t$. 

	More precisely, since $t$ is obtained from $e$ by replacing each free occurrence of $x$ with $s$, the holes of $t$ consist exactly of:

	\begin{itemize}
		\item the holes of $s$, repeated once per occurrence of $x$ in $e$, and
		\item the holes coming from the float constants of $e$ itself.
	\end{itemize}
	
	By~\Cref{lem:constant-extension-measurable}, $\psi_s$ is measurable.

    By construction, for all $\vec{r} \in A_{s;\tau_1}$:
    \[
        e[s[\vec{r}]/x] = t[\psi_s(\vec{r})].
    \]

    Define the measurable set of hole-assignments corresponding to $A$ at skeleton $t$:
    \begin{align*}
		A_t := &\{\vec{u} \in A_{t;\tau_2} \mid t[\vec{u}] \in A\} \\
		:= &\{\vec{u} \in A_{t;\tau_2} \mid t[\vec{u}] \in A \cap Expr_{t;\tau_2}\} 
        \tag{$t[\vec{u}]$ ranges over $Expr_{t;\tau_2}$}
	\end{align*}

	    By~\Cref{def:sigma-algebra-construction-expressions}, $A \cap Expr_{t;\tau_2} \in \mathcal{F}_{Expr_{t;\tau_2}}$, and by \Cref{lem:pullback-iso}, $A_t \in \mathcal{F}_{t;\tau_2}$. We compute:
    \begin{align*}
        R_s
        &= \{\vec{r} \in A_{s;\tau_1} \mid e[s[\vec{r}]/x] \in A\} \\
        &= \{\vec{r} \in A_{s;\tau_1} \mid t[\psi_s(\vec{r})] \in A\} \\
        &= \{\vec{r} \in A_{s;\tau_1} \mid \psi_s(\vec{r}) \in A_t\} \\
        &= \psi_s^{-1}(A_t).
    \end{align*}
    Since $\psi_s$ is measurable and $A_t \in \mathcal{F}_{t;\tau_2}$, we have
    $R_s = \psi_s^{-1}(A_t) \in \mathcal{F}_{s;\tau_1}$.

    Since $R_s \in \mathcal{F}_{s;\tau_1}$ for every $s \in \skel_{\tau_1}$, we get
    $C \in \mathcal{F}_{Expr_{\tau_1}}$, and therefore
    \[
        \mathrm{Sub}_{e,x}^{-1}(A) = C \cap Val_{\tau_1} \in \mathcal{F}_{Val_{\tau_1}}.
    \]
    Hence $\mathrm{Sub}_{e,x}$ is measurable.
\end{proof}

\begin{lemma}[Monadic bind is a sub-Markov kernel]\label{lem:bind-kernel}
	For every occurrence of a monadic bind $\mu \monbind f$ given in the small-step operational semantics in~\Cref{tab:smallstep-1,tab:smallstep-2}, $f : (Expr_{\tau_1} + \bot) \to \mathcal{D}(Expr_{\tau_2} + \bot)$ is a sub-Markov kernel.
\end{lemma}

\begin{proof}
	We first prove that $f : Expr_{\tau_1} \to \mathcal{D}(Expr_{\tau_2} + \bot)$ is a sub-Markov kernel. Let $e$ be an expression whose evaluation rule contains a monadic bind $\mu \monbind f$. We proceed by case analysis on $e$:

	\begin{itemize}
		\item Consider $e = (\letkw\; y = x : \tau_1 \; \inkw\; e_2 : \tau_2) : \tau_2$. 
		
		This means proving that $\lambda x.\, \dirac{\letkw\; y = x\; \inkw\; e_2}$ is a sub-Markov kernel. Formally, define
		\[
			f \colon Expr_{\tau_1} \to \mathcal{D}(Expr_{\tau_2} + \bot),
			\qquad
			f(x) = \dirac{\letkw\; y = x\; \inkw\; e_2}.
		\]
		We verify the two conditions for $f$ to be a sub-Markov kernel.
	
		\paragraph{(i) Subprobability measure:} We show for each $x \in Expr_{\tau_1}$, the map $A \mapsto f(x)(A)$ is a subprobability measure.
	
		\medskip
	
		Fix $x \in Expr_{\tau_1}$. The map $A \mapsto f(x)(A)$ equals $A \mapsto \dirac{\letkw\; y = x\; \inkw\; e_2}(A)$, which satisfies:
		\[
			f(x)(A) =
			\begin{cases}
				1 & \text{if } \letkw\; y = x\; \inkw\; e_2 \in A, \\
				0 & \text{otherwise.}
			\end{cases}
		\]
		A Dirac measure is a subprobability measure, so this condition holds.
	
		\paragraph{(ii) Measurable function:} 
	
		We show for each $A \in \mathcal{F}_{Expr_{\tau_2} + \bot}$, the map $x \mapsto f(x)(A)$ is a measurable function.
	
		\medskip

		Fix $A \in \mathcal{F}_{Expr_{\tau_2} + \bot}$. Define
		\[
			g \colon Expr_{\tau_1} \to [0,1],
			\qquad
			g(x) := f(x)(A) = \dirac{\letkw\; y = x\; \inkw\; e_2}(A).
		\]
		We prove $g$ is measurable by showing $g^{-1}(B)\in\mathcal{F}_{Expr_{\tau_1}}$ for every
		$B\in\mathcal{B}([0,1])$. It suffices to prove that for every $a\in\mathbb{R}$,
		\[
			g^{-1}\big((a,\infty)\cap[0,1]\big)\in\mathcal{F}_{Expr_{\tau_1}},
		\]
		because the family $\{(a,\infty)\cap[0,1] : a\in\mathbb{R}\}$ generates
		$\mathcal{B}([0,1])$.

		Because $g(x)\in[0,1]$ for all $x\in Expr_{\tau_1}$, we have for every $a\in\mathbb{R}$,
		\begin{align*}
			g^{-1}\big((a,\infty)\cap[0,1]\big)
			&= \{x\in Expr_{\tau_1} \mid g(x)\in (a,\infty)\cap[0,1]\} \\
			&= \{x\in Expr_{\tau_1} \mid (g(x)\in (a,\infty)) \wedge (g(x)\in [0,1])\} \\
			&= \{x\in Expr_{\tau_1} \mid (g(x)>a) \wedge (g(x)\in [0,1])\} \\
			&= \{x\in Expr_{\tau_1} \mid g(x)>a\},
		\end{align*}
		where the last step uses that $g(x)\in[0,1]$ always holds.
		Hence it suffices to show that $\{x\in X \mid g(x)>a\}\in\mathcal{F}_{Expr_{\tau_1}}$
		for every $a\in\mathbb{R}$.

		\begin{itemize}
			\item If $a \geq 1$: $\{x\in Expr_{\tau_1} \mid g(x)>a\} = \emptyset \in \mathcal{F}_{Expr_{\tau_1}}$.
			\item If $a < 0$: $\{x\in Expr_{\tau_1} \mid g(x)>a\} = Expr_{\tau_1} \in \mathcal{F}_{Expr_{\tau_1}}$.
			\item If $0 \leq a < 1$: $\{x \in Expr_{\tau_1} \mid g(x) > a \}$.
		\end{itemize}

		The only nontrivial case is the last. Using $g(x)=\dirac{\letkw\; y = x\; \inkw\; e_2}$ and the fact that
		$g(x) \in\{0,1\}$ for all $x\in Expr_{\tau_1}$, we have for $0\le a<1$:
		\begin{align*}
			\{x \in Expr_{\tau_1} \mid g(x) > a \}
			&= \{x \in Expr_{\tau_1} \mid g(x) = 1 \} \\
		\end{align*}

		Define
		\[
			B := \{x \in Expr_{\tau_1} \mid \letkw\; y = x\; \inkw\; e_2 \in A\}.
		\]
		We must show $B \in \mathcal{F}_{Expr_{\tau_1}}$. By~\Cref{def:sigma-algebra-construction-expressions}, it suffices to show that $\forall s \in \skel_{\tau_1},\; B \cap Expr_{s;\tau_1} \in \mathcal{F}_{Expr_{s;\tau_1}}$.

		Fix $s \in \skel_{\tau_1}$. We must prove:
		\begin{align*}
			&B \cap Expr_{s;\tau_1} \in \mathcal{F}_{Expr_{s;\tau_1}} \\
			\iff \quad&\{x \in Expr_{\tau_1} \mid \letkw\; y = x\; \inkw\; e_2 \in A\} \cap Expr_{s;\tau_1} \in \mathcal{F}_{Expr_{s;\tau_1}} \tag{Definition of $B$}\\
			\iff \quad&\{x \in Expr_{s;\tau_1} \mid \letkw\; y = x\; \inkw\; e_2 \in A\} \in \mathcal{F}_{Expr_{s;\tau_1}} \tag{$Expr_{s;\tau_1} \subseteq Expr_{\tau_1}$}\\
			\iff \quad&\{s[\vec{r}] \in Expr_{s;\tau_1} \mid \letkw\; y = s[\vec{r}]\; \inkw\; e_2 \in A\} \in \mathcal{F}_{Expr_{s;\tau_1}} \tag{rewrite $x$}\\
			\iff \quad&\{\vec{r} \in A_{s;\tau_1} \mid \letkw\; y = s[\vec{r}]\; \inkw\; e_2 \in A\} \in \mathcal{F}_{s;\tau_1} \tag{\Cref{lem:isomorphism}}
		\end{align*}
		Define
		\[
			R_s := \{\vec{r} \in A_{s;\tau_1} \mid \letkw\; y = s[\vec{r}]\; \inkw\; e_2 \in A\}.
		\]
		
		\medskip

		Let $s_2$ be the skeleton of $e_2$, with $m := \holes(s_2)$, and let $\vec{c} \in A_{s_2;\tau_2} \subseteq \R^m$ be the unique tuple such that $e_2 = s_2[\vec{c}]$. Define the combined skeleton
		\[
			s' := \letkw\; y = s\; \inkw\; s_2 \in \skel_{\tau_2},
		\]
		whose holes consist of those from $s$ (filled by $\vec{r} \in \R^n$) followed by those from $s_2$ (filled by $\vec{c} \in \R^m$), so that $A_{s';\tau_2} \subseteq \R^{n+m}$. Thus, for every $\vec{r} \in A_{s;\tau_1}$:
		\[
			\letkw\; y = s[\vec{r}]\; \inkw\; e_2 = s'[(\vec{r}, \vec{c})].
		\]
	
		\medskip
		Define a constant-extension map
		\[
			\phi_{\vec{c}} \colon A_{s;\tau_1} \to A_{s';\tau_2},
			\qquad
			\phi_{\vec{c}}(\vec{r}) := (\vec{r}, \vec{c}).
		\]
			By \Cref{lem:constant-extension-measurable}, $\phi_{\vec{c}}$ is measurable.

		Define the measurable set of hole-assignments corresponding to $A$ at skeleton $s'$:
		\begin{align*}
			A_{s'} := &\{(\vec{r}, \vec{r}\,') \in A_{s';\tau_2} \mid s'[\vec{r}, \vec{r}\,'] \in A\} \\
			:= &\{(\vec{r}, \vec{r}\,') \in A_{s';\tau_2} \mid s'[\vec{r}, \vec{r}\,'] \in A \cap Expr_{\tau_2}\} \tag{\Cref{def:subspace-sigma-algebra}}\\
			:= &\{(\vec{r}, \vec{r}\,') \in A_{s';\tau_2} \mid s'[\vec{r}, \vec{r}\,'] \in A \cap Expr_{s';\tau_2}\} \tag{$s'[\vec{r}, \vec{r}\,']$ ranges over $Expr_{s';\tau_2}$}
		\end{align*}
		Since $A \cap Expr_{s';\tau_2} \in \mathcal{F}_{Expr_{s';\tau_2}}$, by~\Cref{lem:isomorphism}, $A_{s'} \in \mathcal{F}_{s';\tau_2}$.
	
		\medskip
		Then:
		\begin{align*}
			R_s &= \{\vec{r} \in A_{s;\tau_1} \mid \letkw\; y = s[\vec{r}]\; \inkw\; e_2 \in A\} \\
			&= \{\vec{r} \in A_{s;\tau_1} \mid s'[(\vec{r},\vec{c})] \in A\} \\
			&= \{\vec{r} \in A_{s;\tau_1} \mid s'[\phi_{\vec{c}}(\vec{r})] \in A\} \\
			&= \{\vec{r} \in A_{s;\tau_1} \mid \phi_{\vec{c}}(\vec{r}) \in A_{s'}\} \\
			&= \phi_{\vec{c}}^{-1}(A_{s'}).
		\end{align*}
		Since $A_{s'} \in \mathcal{F}_{s';\tau_2}$ and $\phi_{\vec{c}}$ is measurable, we conclude
		\[
			\phi_{\vec{c}}^{-1}(A_{s'}) \in \mathcal{F}_{s;\tau_1},
		\]
		and therefore $B \cap Expr_{s;\tau_1} \in \mathcal{F}_{Expr_{s;\tau_1}}$ for all $s \in \skel_{\tau_1}$, giving $B \in \mathcal{F}_{Expr_{\tau_1}}$ as required.

		\medskip

		Hence, $\{x \in Expr_{\tau_1} \mid g(x) > a\} \in \mathcal{F}_{Expr_{\tau_1}}$ for all $a \in \R$, so $g^{-1}((a,\infty) \cap [0,1]) \in \mathcal{F}_{Expr_{\tau_1}}$ for all $a \in \R$. Hence $g^{-1}(B) \in \mathcal{F}_{Expr_{\tau_1}}$ for every $B \in \mathcal{B}([0,1])$, and $g$ is measurable.

		\medskip

		Thus, $f \colon Expr_{\tau_1} \to \mathcal{D}(Expr_{\tau_2} + \bot)$ is a sub-Markov kernel.

		\item All other cases analogous. \todo[inline]{Maybe show the comparison case}
	\end{itemize}
	By~\Cref{lem:bot-extension-kernel}, the $\bot$-extension $f_\bot \colon (Expr_{\tau_1} + \bot) \to \mathcal{D}(Expr_{\tau_2} + \bot)$ is also a sub-Markov kernel, as required.
\end{proof}

\begin{lemma}[Small-step semantics is a sub-Markov kernel]\label{lem:small-step-kernel}
    For each type $\tau$, the small-step operational semantics $\sem{\cdot} \colon Expr_{\tau} \to \mathcal{D}(Expr_{\tau} + \bot)$ defines a sub-Markov kernel.
\end{lemma}
\begin{proof}
    By induction on the typing derivation of $e : \tau$. We show in each case that the map $e \mapsto \sem{e : \tau}$ is a sub-Markov kernel.

	\begin{itemize}
		\item Consider $e = (\letkw\; y = x : \tau_1 \; \inkw\; e_2 : \tau_2) : \tau_2$. The small-step semantics is:
		\[
			\sem{\letkw \; x = e_1 \; \inkw \; e_2 : \tau_2} =
			\begin{cases}
				\dirac{e_2[v/x]} & \text{if } e_1 = v \\
				\sem{e_1 : \tau_1} \gg\!\!= \lambda g.\; \dirac{\letkw \; x = g \; \inkw \; e_2} & \text{if } e_1 \neq v
			\end{cases}
		\]
		\begin{itemize}
			\item If $e_1 = v$, By~\Cref{lem:dirac-return-kernel} and~\Cref{lem:subst-measurable}, $e \mapsto \dirac{e_2[v/y] : \tau_2}$ is a sub-Markov kernel. Hence, $\sem{\letkw \; x = e_1 \; \inkw \; e_2 : \tau_2} \colon Expr_{\tau_2} \to \mathcal{D}(Expr_{\tau_2} + \bot)$ is a sub-Markov kernel.
			
			\item By the induction hypothesis, 
			\[
				\sem{e_1 : \tau_1} \colon Expr_{\tau_1} \to \mathcal{D}(Expr_{\tau_1} + \bot)
			\]
			is a sub-Markov kernel. By~\Cref{lem:bind-kernel},
			\[
				\lambda g.\; \dirac{\letkw \; x = g \; \inkw \; e_2} \colon 
				(Expr_{\tau_1} + \bot) \to \mathcal{D}(Expr_{\tau_2} + \bot)
			\]
			is a sub-Markov kernel. Hence by~\Cref{lem:kernel-composition}, their composition
			\[
				\sem{e_1 : \tau_1} \gg\!\!= \lambda g.\; \dirac{\letkw \; x = g \; \inkw \; e_2}
				\colon Expr_{\tau_1} \to \mathcal{D}(Expr_{\tau_2} + \bot)
			\]
			is a sub-Markov kernel.
			
			Thus, $\sem{\letkw \; x = e_1 \; \inkw \; e_2 : \tau_2} \colon Expr_{\tau_2} \to \mathcal{D}(Expr_{\tau_2} + \bot)$ is a sub-Markov kernel.
		\end{itemize}

		\item Consider $e = \diverge : \tau$.
		
		The small-step semantics is:
		\[
			\sem{\diverge : \tau} = \lambda A.\, 0.
		\]

		By~\Cref{lem:zero-kernel},
		\[
			\lambda A.\, 0
		\]
		is a sub-Markov kernel.

		Thus, $\sem{\diverge : \tau} \colon Expr_{\tau} \to \mathcal{D}(Expr_{\tau} + \bot)$ is a sub-Markov kernel.
		
		\item All other cases are analogous.
	\end{itemize}
\end{proof}

\subsection{Small-Step Coupling Semantics}

\begin{table}[t]
	\caption{Small-step coupling semantics. Notice that $\discretize{e} : \discretize{\tau}$ is uniquely determined by $e \boldsymbol{:} \boldsymbol{\tau}$ so that it suffices to indicate the latter in the left-hand column. Here, every function of the form $f(g,g') = \lambda (g,g').\dirac{\ldots}$ is extended to $f_{\bot}$ in the sense of \Cref{def:bot-extension}.}
	\label{tab:coupling_smallstep_1}%
	\begin{center}
		\begin{adjustbox}{width=\textwidth}
		\begin{tabular}{l l l l l}
			\toprule
			\toprule
						$\mathbf{e \boldsymbol{:} \boldsymbol{\tau}}$ & \textbf{constraints}  & $\boldsymbol{\sem{e : \tau, \discretize{e} : \discretize{\tau}}}$   \\
			\hline 
			%
			\multirow{1}{*}{$\diverge : \tau$} 
			
			& & $\lambda A.0$  \\
			\midrule
			%
			$v : \tau$ & & $\dirac{(v:\tau,\discretize{v:\tau} : \discretize{\tau})}$ \\[0.5ex]
			\midrule
			%
			\multirow{2}{*}{$\letkw \; x = e_1 : \tau_1 \; \inkw \; e_2 : \tau_2 : \tau_2$} 
			
			& $e_1=v$ & $\dirac{(e_2[v/x] : \tau_2, \discretize{e_2 : \tau_2}[\discretize{v : \tau_1}/x] : \discretize{\tau_2})}$ \\
			& $e_1$ no val.\ & \makecell[l]{$\sem{ e_1 : \tau_1, \discretize{e_1 : \tau_1} : \discretize{\tau_1}}$ \\
			$\gg\!= \lambda (g,g'). \dirac{(\letkw \; x = g\; \inkw \; e_2 : \tau_2 : \tau_2,
			\letkw \; x = g' \; \inkw \; \discretize{e_2 : \tau_2} : \discretize{\tau_2} : \discretize{\tau_2})}$} \\
			\midrule 
			%
			\multirow{3}{*}{$\ifkw \; e_1 : \bool \; \thenkw \; e_2 : \tau \; \elsekw \; e_3 : \tau : \tau$} 
			
			& $e_1=\true$ & $\dirac{(e_2 : \tau, \discretize{e_2 : \tau} : \discretize{\tau})} $ \\
			& $e_1=\false$ & $\dirac{(e_3 : \tau, \discretize{e_3 : \tau} : \discretize{\tau})} $ \\
			& o.w.\ & \makecell[l]{$\sem{ e_1 : \bool, \discretize{e_1 : \bool} : \discretize{\bool}}$ \\
			$\gg\!= \lambda (g,g'). \delta_{\begin{subarray}{l}(\ifkw \; g\; \thenkw \; e_2 : \tau \; \elsekw \; e_3 : \tau : \tau, \\ \ifkw \; g' \; \thenkw \; \discretize{e_2 : \tau} : \discretize{\tau} \; \elsekw \; \discretize{e_3 : \tau} : \discretize{\tau} : \discretize{\tau})\end{subarray}}$} \\
			\midrule 
			%
			\multirow{3}{*}{$\observekw\; (e : \bool)  : \unit$} 
			
			& $e = \true$ & $\dirac{(\unit, \unit)}$ \\[0.5ex]
			& $e=\false$ & $\dirac{\bot}$ \\
			& o.w.\ & \makecell[l]{$\sem{e: \bool, \discretize{e}: \bool}$ \\
			$\gg\!= \lambda (g,g'). \dirac{(\observekw\; (g: \bool)  : \unit, \observekw\; (g' : \bool)  : \unit)}$} \\
			\midrule
			%
			\multirow{4}{*}{\makecell{$e_1 : \tau_1 < e_2 : \tau_2 : \bool$ \\
			for $\tau_1 = \floattype{B}{V_1}, \tau_2 = \floattype{B}{V_2} $ \\
			$B = \top$}} 
			
			& $e_1 = v_1, e_2 = v_2, v_1 < v_2 $ & $\dirac{(\true : \bool, \true : \bool)}$ \\
			& $e_1 = v_1, e_2 = v_2, v_1 \geq v_2$ & $\dirac{(\false : \bool, \false : \bool)}$  \\
			&$e_1 = v_1, e_2 \neq v_2$ & \makecell[l]{$\sem{ e_2 : \tau_2, \discretize{e_2 : \tau_2} : \discretize{\tau_2}}$ \\
			$\gg\!= \lambda (g,g'). \dirac{(v_1 : \tau_1 < g: \bool, \discretize{v_1 : \tau_1} : \discretize{\tau_1} < g' : \discretize{\bool}) }$} \\
			& o.w.\ & \makecell[l]{ $\sem{ e_1 : \tau_1, \discretize{e_1 : \tau_1} : \discretize{\tau_1}}$ \\
			$\gg\!= \lambda (g,g'). \dirac{ (g< e_2 : \tau_2 : \bool, g' < \discretize{e_2 : \tau_2} : \discretize{\tau_2} : \discretize{\bool})}$}  \\
			\midrule
			%
			\multirow{4}{*}{\makecell{$e_1 : \tau_1 < e_2 : \tau_2 : \bool$ \\
			for $\tau_1 = \floattype{B}{V_1}, \tau_2 = \floattype{B}{V_2} $ \\
			$B \neq \top$}} 
			
			& $e_1 = v_1, e_2 = v_2, v_1 < v_2 $ & $\dirac{(\true : \bool, \true : \bool)}$ \\
			& $e_1 = v_1, e_2 = v_2, v_1 \geq v_2$ & $\dirac{(\false : \bool, \false : \bool)}$  \\
			&$e_1 = v_1, e_2 \neq v_2$ & \makecell[l]{$\sem{ e_2 : \tau_2, \discretize{e_2 : \tau_2} : \discretize{\tau_2} }$ \\
			$\gg\!= \lambda (g,g'). \dirac{ (v_1 : \tau_1 < g: \bool, \discretize{v_1 : \tau_1} : \discretize{\tau_1} \finlt{(|B|+1)} g' : \discretize{\bool}) }$} \\
			& o.w.\ & \makecell[l]{ $\sem{ e_1 : \tau_1, \discretize{e_1 : \tau_1} : \discretize{\tau_1}}$ \\
			$\gg\!= \lambda (g,g'). \dirac{ (g< e_2 : \tau_2 : \bool, g' \finlt{(|B|+1)} \discretize{e_2 : \tau_2} : \discretize{\tau_2} : \discretize{\bool})}$}  \\
			\midrule
			%
			\multirow{5}{*}{\makecell{$\uniform(e_1 : \tau_1, e_2 : \tau_2) : \floattype{B}{V}$ \\
			$B=\top$}} 
			
			& $e_1 = v_1 < v_2 = e_2$ &
		    $\text{Uniform}{(v_1, v_2)} \gg\!= \lambda v. \dirac{(v : \floattype{B}{V}, v : \floattype{B}{V})}$  \\
			& $e_1 = v_1 \geq v_2 = e_2$ &
		    $\sem{\diverge : \tau, \discretize{\uniform(v_1 : \tau_1,v_2 : \tau_2)} : \discretize{\tau}}$  \\
			& $e_1=v_1$, $e_2$ no val.\ & \makecell[l]{$\sem{e_2:\tau_2, \discretize{e_2} : \discretize{\tau_2}}$ \\ $\gg\!= \lambda (g,g'). \dirac{(\uniform(v_1:\tau_1, g:\tau_2 ) : \tau, \uniform(\discretize{v_1}:\discretize{\tau_1}, g': \discretize{\tau_2} ) : \discretize{\tau})}$} \\
			& $e_1$ no val.\ & \makecell[l]{$\sem{e_1:\tau_1, \discretize{e_1} : \discretize{\tau_1}}$ \\ $\gg\!= \lambda (g,g'). \dirac{(\uniform(g:\tau_1, e_2:\tau_2) : \tau, \uniform({g'}:\discretize{\tau_1}, \discretize{e_2} : \discretize{\tau_2} ) : \discretize{\tau})}$} \\
			%
			%
			%
			\midrule 
			%
			\multirow{5}{*}{\makecell{$\uniform(e_1 : \tau_1, e_2 : \tau_2) : \floattype{B}{V}$ \\
			$B \neq \top$ \\
			$\tau_1 = \float[B_1 \neq \top,\{v_1,\ldots,v_n\}]$ \\
			$\tau_2 = \float[B_2 \neq \top,\{w_1,\ldots,w_m\}]$}} 
			
			&  $e_1 = c_1 < c_2 = e_2$ values &  \makecell[l]{
			$\text{Uniform}(c_1, c_2) \gg\!= \lambda v.\dirac{(v : \floattype{B}{V}, v_i : \fin{|B| + 1})}$ \\
			where $v_i$ is unique number s.t. $v \in \intervals{B}_{v_i}$} \\
			& $e_1 = v_1 \geq v_2 = e_2$ &
		    $\sem{\diverge : \tau, \discretize{\uniform(v_1 : \tau_1,v_2 : \tau_2)} : \discretize{\tau}}$  \\
			& $e_1=c_1$, $e_2$ no val.\ & \makecell[l]{$\sem{e_2:\tau_2, \discretize{e_2} : \discretize{\tau_2}}$ \\ $\gg\!= \lambda (g,g'). \dirac{(\uniform(v_1:\tau_1, g:\tau_2 ) : \tau, \dc(\discretize{v_1}:\discretize{\tau_1}, g': \discretize{\tau_2} ) : \discretize{\tau})}$} \\
			& $e_1$ no val.\ & \makecell[l]{$\sem{e_1:\tau_1, \discretize{e_1} : \discretize{\tau_1}}$ \\ $\gg\!= \lambda (g,g'). \dirac{(\uniform(g:\tau_1, e_2:\tau_2) : \tau, \dc({g'}:\discretize{\tau_1}, \discretize{e_2} : \discretize{\tau_2} ) : \discretize{\tau})}$} \\
			\bottomrule
			\bottomrule
		\end{tabular}
	\end{adjustbox}
	\end{center}
\end{table}%

\begin{table}[t]
	\caption{Small-step coupling semantics (continued). Here, every function of the form $f(g,g') = \lambda (g,g').\dirac{\ldots}$ is extended to $f_{\bot}$ in the sense of \Cref{def:bot-extension}.}
	\label{tab:coupling_smallstep_2}
	\begin{center}
		\begin{adjustbox}{width=\textwidth}
		\begin{tabular}{l l l l l}
			\toprule
			\toprule
			$\mathbf{e \boldsymbol{:} \boldsymbol{\tau}}$ & \textbf{constraints}  & $\boldsymbol{\sem{e : \tau, \discretize{e} : \discretize{\tau}}}$   \\
			\hline
			%
			\multirow{4}{*}{\makecell{$e_1 : \tau_1 \logand e_2 : \tau_2 : \bool$ \\
			$\tau_1 = \bool, \tau_2 = \bool$}}
			& $e_1 = b_1 = \true,\, e_2 = b_2 = \true$ & $\dirac{(\true : \bool, \true : \bool)}$ \\
			& $e_1 = b_1,\, e_2 = b_2$ & $\dirac{(\false : \bool, \false : \bool)}$ \\
			& $e_1 = b_1,\, e_2 \neq b_2$ &
			\makecell[l]{$\sem{ e_2 : \tau_2, \discretize{e_2 : \tau_2} : \discretize{\tau_2}}$ \\
			$\gg\!= \lambda (g,g'). \dirac{( b_1 : \tau_1 \logand g: \bool,\; b_1 : \discretize{\tau_1} \logand g' : \discretize{\bool})}$} \\
			& o.w.\ &
			\makecell[l]{$\sem{ e_1 : \tau_1, \discretize{e_1 : \tau_1} : \discretize{\tau_1}}$ \\
			$\gg\!= \lambda (g,g'). \dirac{( g\logand e_2 : \tau_2 : \bool,\; g' \logand \discretize{e_2 : \tau_2} : \discretize{\tau_2} : \discretize{\bool})}$} \\
			\midrule
			%
			\multirow{4}{*}{\makecell{$e_1 : \tau_1 \logor e_2 : \tau_2 : \bool$ \\
			$\tau_1 = \bool, \tau_2 = \bool$}}
			& $e_1 = b_1 = \false,\, e_2 = b_2 = \false$ & $\dirac{(\false : \bool, \false : \bool)}$ \\
			& $e_1 = b_1,\, e_2 = b_2$ & $\dirac{(\true : \bool, \true : \bool)}$ \\
			& $e_1 = b_1,\, e_2 \neq b_2$ &
			\makecell[l]{$\sem{ e_2 : \tau_2, \discretize{e_2 : \tau_2} : \discretize{\tau_2}}$ \\
			$\gg\!= \lambda (g,g'). \dirac{( b_1 : \tau_1 \logor g: \bool,\; b_1 : \discretize{\tau_1} \logor g' : \discretize{\bool})}$} \\
			& o.w.\ &
			\makecell[l]{$\sem{ e_1 : \tau_1, \discretize{e_1 : \tau_1} : \discretize{\tau_1}}$ \\
			$\gg\!= \lambda (g,g'). \dirac{( g\logor e_2 : \tau_2 : \bool,\; g' \logor \discretize{e_2 : \tau_2} : \discretize{\tau_2} : \discretize{\bool})}$} \\
			\midrule
			%
			\multirow{3}{*}{\makecell{$\lognot\; e : \tau : \bool$ \\
			$\tau = \bool$}}
			& $e = \true$ & $\dirac{(\false : \bool, \false : \bool)}$ \\
			& $e = \false$ & $\dirac{(\true : \bool, \true : \bool)}$ \\
			& o.w.\ &
			\makecell[l]{$\sem{ e : \tau, \discretize{e : \tau} : \discretize{\tau}}$ \\
			$\gg\!= \lambda (g,g'). \dirac{(\lognot\; g: \bool,\; \lognot\; g' : \discretize{\bool})}$} \\
			\midrule
			%
			\multirow{3}{*}{\makecell{$(e_1 : \tau_1, e_2 : \tau_2) : \tau_1 * \tau_2$}}
			& $e_1 = v_1,\, e_2 = v_2$ &
			$\dirac{((v_1,v_2) : \tau_1 * \tau_2,\; (\discretize{v_1}:\discretize{\tau_1}, \discretize{v_2}:\discretize{\tau_2}) : \discretize{\tau_1} * \discretize{\tau_2})}$ \\
			& $e_1 = v_1,\, e_2 \neq v_2$ &
			\makecell[l]{$\sem{ e_2 : \tau_2, \discretize{e_2 : \tau_2} : \discretize{\tau_2} }$ \\
			$\gg\!= \lambda (g,g'). \dirac{( (v_1 : \tau_1, g: \tau_2) : \tau_1 * \tau_2,\; (\discretize{v_1}:\discretize{\tau_1}, g' : \discretize{\tau_2}) : \discretize{\tau_1} * \discretize{\tau_2})}$} \\
			& o.w.\ &
			\makecell[l]{$\sem{ e_1 : \tau_1, \discretize{e_1 : \tau_1} : \discretize{\tau_1}}$ \\
			$\gg\!= \lambda (g,g'). \dirac{( (g, e_2 : \tau_2) : \tau_1 * \tau_2,\; (g', \discretize{e_2 : \tau_2}) : \discretize{\tau_1} * \discretize{\tau_2})}$} \\
			\midrule
			%
			\multirow{2}{*}{\makecell{$\fstkw\; (e : \tau_1 * \tau_2) : \tau_1$}}
			& $e = (v_1, v_2)$ & $\dirac{(v_1 : \tau_1,\; \discretize{v_1:\tau_1} : \discretize{\tau_1})}$ \\
			& o.w.\ &
			\makecell[l]{$\sem{ e : \tau_1 * \tau_2, \discretize{e} : \discretize{\tau_1 * \tau_2}}$ \\
			$\gg\!= \lambda (g,g'). \dirac{(\fstkw\; (g: \tau_1 * \tau_2) : \tau_1,\; \fstkw\; (g' : \discretize{\tau_1 * \tau_2}) : \discretize{\tau_1})}$} \\
			\midrule
			%
			\multirow{2}{*}{\makecell{$\sndkw\; (e : \tau_1 * \tau_2) : \tau_2$}}
			& $e = (v_1, v_2)$ & $\dirac{(v_2 : \tau_2,\; \discretize{v_2:\tau_2} : \discretize{\tau_2})}$ \\
			& o.w.\ &
			\makecell[l]{$\sem{ e : \tau_1 * \tau_2, \discretize{e} : \discretize{\tau_1 * \tau_2}}$ \\
			$\gg\!= \lambda (g,g'). \dirac{(\sndkw\; (g: \tau_1 * \tau_2) : \tau_2,\; \sndkw\; (g' : \discretize{\tau_1 * \tau_2}) : \discretize{\tau_2})}$} \\
			\midrule
			%
			\multirow{1}{*}{\makecell{$\funkw \; x \to e : \tau_1 \to \tau_2$}}
			& & $\dirac{(\funkw \; x \to e : \tau_1 \to \tau_2,\; \funkw \; x \to \discretize{e} : \discretize{\tau_1} \to \discretize{\tau_2})}$ \\
			\midrule
			%
			\multirow{1}{*}{\makecell{$\fixkw \; g\; x := e : \tau_1 \to \tau_2$}}
			& & $\dirac{(\fixkw \; g\; x := e : \tau_1 \to \tau_2,\; \fixkw \; g\; x := \discretize{e} : \discretize{\tau_1} \to \discretize{\tau_2})}$ \\
			\midrule
			%
			\multirow{3}{*}{\makecell{$(e_1 : \tau_1 \to \tau_2\;\; e_2 : \tau_1) : \tau_2$}}
			& \makecell[l]{$e_1 = \funkw\; x \to e \text{ or } \fixkw\; f\; x := e,$\\ $e_2 = v_2$} &
			\makecell[l]{$\dirac{(e[v_2/x] : \tau_2,\; \discretize{e}[\discretize{v_2} / x] : \discretize{\tau_2})}$} \\
			& \makecell[l]{$e_1 = \funkw\; x \to e \text{ or } \fixkw\; f\; x := e,$\\ $e_2 \neq v_2$} &
			\makecell[l]{$\sem{ e_2 : \tau_1, \discretize{e_2 : \tau_1} : \discretize{\tau_1}}$ \\
			$\monbind \lambda (g,g'). \dirac{( e_1 : \tau_1 \to \tau_2\;\; g : \tau_1 : \tau_2,\; \discretize{e_1} : \discretize{\tau_1} \to \discretize{\tau_2}\;\; g' : \discretize{\tau_1} : \discretize{\tau_2})}$} \\
			& o.w.\ &
			\makecell[l]{$\sem{ e_1 : \tau_1 \to \tau_2, \discretize{e_1} : \discretize{\tau_1 \to \tau_2}}$ \\
			$\monbind \lambda (g,g'). \dirac{( g\; e_2 : \tau_1 : \tau_2,\; g'\; \discretize{e_2} : \discretize{\tau_1} : \discretize{\tau_2})}$} \\
			\midrule
			%
			\multirow{1}{*}{\makecell{$\text{nil} : \listty{\tau}$}}
			& & $\dirac{(\text{nil} : \listty{\tau},\; \text{nil} : \listty{\discretize{\tau}})}$ \\
			\midrule
			%
			\multirow{3}{*}{\makecell{$(e_1 : \tau :: e_2 : \listty{\tau}) : \listty{\tau}$}}
			& $e_1 = v_1,\, e_2 = v_2$ &
			$\dirac{(v_1::v_2 : \listty{\tau},\; \discretize{v_1}::\discretize{v_2} : \listty{\discretize{\tau}})}$ \\
			& $e_1 = v_1,\, e_2 \neq v_2$ &
			\makecell[l]{$\sem{ e_2 : \listty{\tau}, \discretize{e_2} : \listty{\discretize{\tau}}}$ \\
			$\monbind \lambda (g,g'). \dirac{( v_1::g: \listty{\tau},\; \discretize{v_1}::g' : \listty{\discretize{\tau}})}$} \\
			& o.w.\ &
			\makecell[l]{$\sem{ e_1 : \tau, \discretize{e_1} : \discretize{\tau}}$ \\
			$\monbind \lambda (g,g'). \dirac{( g::e_2 : \listty{\tau} : \listty{\tau},\; g'::\discretize{e_2} : \listty{\discretize{\tau}} : \listty{\discretize{\tau}})}$} \\
			\midrule
			%
			\multirow{3}{*}{\makecell{$(\matchkw \; e : \listty{\tau_1}\; \withkw$ \\
			$\matchcase \text{nil} \rightarrow e_1 : \tau_2$ \\
			$\matchcase h :: t \rightarrow e_2 : \tau_2$ \\
			${\textnormal{\ttfamily\bfseries end}}) : \tau_2$}}
			& $e = \text{nil}$ &
			$\dirac{(e_1 : \tau_2,\; \discretize{e_1 : \tau_2} : \discretize{\tau_2})}$ \\
			& $e = v_h :: v_t$ &
			$\dirac{(e_2[v_h/h, v_t/t] : \tau_2,\; \discretize{e_2 : \tau_2}[\discretize{v_h}/h, \discretize{v_t}/t] : \discretize{\tau_2})}$ \\
			& o.w.\ &
			\makecell[l]{%
			$\sem{ e : \listty{\tau_1}, \discretize{e} : \listty{\discretize{\tau_1}} }$ \\
			$\monbind \lambda (g,g').\dirac{\begin{subarray}{l}
			(\matchkw\; g\; \withkw\;
				\matchcase \text{nil} \rightarrow e_1 : \tau_2\;
				\matchcase h::t \rightarrow e_2 : \tau_2\;
			{\textnormal{\ttfamily\bfseries end}} : \tau_2,\\
			\matchkw\; g' \; \withkw\;
				\matchcase \text{nil} \rightarrow \discretize{e_1 : \tau_2} : \discretize{\tau_2}\;
				\matchcase h::t \rightarrow \discretize{e_2 : \tau_2} : \discretize{\tau_2}\;
			{\textnormal{\ttfamily\bfseries end}} : \discretize{\tau_2})
			\end{subarray}}$} \\
			\bottomrule
			\bottomrule
		\end{tabular}
	\end{adjustbox}
	\end{center}
\end{table}

Define $P_\tau = \{ (e:\tau, \discretize{e:\tau} : \discretize{\tau}) ~|~ e:\tau \}$. To define subprobability measures and kernels on the coupling space, we
must first equip it with a $\sigma$-algebra. Since
\[
P_\tau\subseteq Expr_\tau\times Expr_{\discretize{\tau}},
\]
we use the subspace $\sigma$-algebra from
\Cref{def:subspace-sigma-algebra}, inherited from the product measurable
space. We then extend this space with the isolated error element
$\bot$, as for $Expr_\tau+\bot$.

\begin{definition}[$\sigma$-algebra construction on $P_\tau$]
Let
\[
	P_\tau
	=
	\{(e:\tau,\discretize{e}:\discretize{\tau})\mid e:\tau\}
	\subseteq
	Expr_\tau\times Expr_{\discretize{\tau}}.
\]
We equip $P_\tau$ with the subspace $\sigma$-algebra
\[
	\mathcal F_{P_\tau}
	:=
	\left\{
	P_\tau\cap A
	\;\middle|\;
	A\in
	\mathcal F_{Expr_\tau}
	\otimes
	\mathcal F_{Expr_{\discretize{\tau}}}
	\right\}.
\]
We equip $P_\tau+\bot$ with
\[
	\mathcal F_{P_\tau+\bot}
	:=
	\{A,\ A\cup\{\bot\}\mid A\in\mathcal F_{P_\tau}\}.
\]
\end{definition}

\begin{lemma}\label{lem:measurable-coupling-space}
	$(P_\tau,\mathcal F_{P_\tau})$
	and $(P_\tau+\bot,\mathcal F_{P_\tau+\bot})$
	are measurable spaces.
\end{lemma}

\begin{proof}
	The first claim follows from
	\Cref{def:subspace-sigma-algebra}. The second follows from the same
	error-extension construction used for \(Expr_\tau+\bot\).
\end{proof}

Now, exploiting that $e:\tau$ is a value iff $\discretize{e} : \discretize{\tau}$ is a value, we define:
\begin{definition}
	The small-step coupling  $\sem{e :\tau, \discretize{e}: \discretize{\tau}} \in \mathcal{D}(P_{\tau} + \bot)$ is defined recursively on the structure of $e:\tau$ by the rules in \Cref{tab:coupling_smallstep_1,tab:coupling_smallstep_2}. 
\end{definition}
%
%

\begin{lemma}[Monadic bind is a sub-Markov kernel]\label{lem:bind-kernel-coupling}
	For every occurrence of a monadic bind $\mu \monbind (f,f')$ given in the small-step operational coupling semantics in~\Cref{tab:coupling_smallstep_1,tab:coupling_smallstep_2}, $(f,f') \colon (P_{\tau_1} + \bot) \to \mathcal{D}(P_{\tau_2} + \bot)$ is a sub-Markov kernel.
\end{lemma}

\begin{proof}
	Analogous to proof of~\Cref{lem:bind-kernel}.
\end{proof}

\begin{lemma}[Small-step coupling semantics is a sub-Markov kernel]\label{lem:small-step-coupling-kernel}
	For each type $\tau$, the small-step operational coupling semantics $\sem{\cdot, \cdot} \colon P_{\tau} \to \mathcal{D}(P_\tau + \bot)$ defines a sub-Markov kernel.
\end{lemma}

\begin{proof}
	Analogous to proof of~\Cref{lem:small-step-kernel}.
\end{proof}

We now define suitable projection functions on measures $\mu \in \mathcal{D}(P_{\tau} + \bot)$ by 
\begin{align}
	 \pi_1(\mu) &{}={} \mu \gg\!= \lambda(e_1,e_2). \delta_{e_1} \\
	 \pi_2(\mu) &{}={} \mu \gg\!= \lambda(e_1,e_2). \delta_{e_2}~.
\end{align}
We have the following key property:
\begin{lemma}[Small-step coupling projection]
	\label{lem:small_step_coupling_props}
	Let $e:\tau$. We have:
	\begin{enumerate}
		\item\label{lem:small_step_coupling_props1} $\pi_1(\sem{e:\tau, \discretize{e} : \discretize{\tau}}) = \sem{e:\tau}$.
		\item\label{lem:small_step_coupling_props2} $\pi_2(\sem{e:\tau, \discretize{e} : \discretize{\tau}}) = \sem{\discretize{e}:\discretize{\tau}}$.
	\end{enumerate}
\end{lemma}
\begin{proof}
	By induction on the type derivation of $e:\tau$. The most interesting case is 
	\[
	    e:\tau = \uniform(c_1 : \floattype{B_1}{V_1}, c_2 : \floattype{B_2}{V_2}) : \floattype{B}{V}~,
	\]
	where
	\begin{enumerate}
		\item $c_1 < c_2$,
		\item $B =  \{\sim_1\!\!b_1, \ldots, \sim_n\!\!b_n\}$,
		\item $V=\top$, 
		\item $V_1 = \{v_1,\ldots, v_{n_1}\}$ with $\distinguishes{B_1}{V_1}$ and $c_1 \in V_1$ due to $\has{V_1}{c_1}$ by I.H. 
		\item $V_2 = \{w_1,\ldots, w_{n_2}\}$ with $\distinguishes{B_2}{V_2}$ and  $c_2 \in V_2$ due to $\has{V_2}{c_2}$ by I.H. 
	\end{enumerate}
	First observe that we have 
	\begin{align*}
		&\discretize{ e:\tau} : \discretize{\tau} \\
		{}={}& \dc(\discretize{c_1 : \floattype{B_1}{V_1}}, \discretize{c_2 :  \floattype{B_2}{V_2}}) : \fin{|B| + 1}\\
		{}={}&\dc( k_1 : \fin{|B_1| + 1}, k_2:\fin{|B_2|+1} ) : \fin{|B|+1}~,
	\end{align*}
	where $k_i$ is the unique number such that $c_i \in \intervals{B_i}_{k_i}$. By \Cref{tab:smallstep-1} and \Cref{fig:discretization-code}, this yields
	\begin{align*}
		 & \sem{\discretize{e:\tau}} \\
		 {}={}& \sem{\discrete(p_0,\ldots,p_n)}~\text{where}~p_k = \text{CDF}(b_{k+1}) - \text{CDF}(b_k) = \frac{b_{k+1} - c_1}{c_2-c_1} - \frac{b_{k} - c_1}{c_2-c_1}  =  \frac{b_{k+1} - b_k}{c_2-c_1}~. 
	\end{align*}

	For \Cref{lem:small_step_coupling_props}.\ref{lem:small_step_coupling_props2}, consider the following:
	\begin{align*}
		 & \pi_1(\sem{e:\tau, \discretize{e} : \discretize{\tau}}) \\
		 ~{}={}~ &  \sem{e:\tau, \discretize{e} : \discretize{\tau}} \gg\!= \lambda(e_1,e_2). \delta_{e_1}  
		 \tag{definition} \\
		 ~{}={}~ & (\text{Uniform}(c_1, c_2) \gg\!= \lambda v.\delta_{(v : \floattype{B}{V}, v_i : \fin{|B| + 1})})\gg\!= \lambda(e_1,e_2). \delta_{e_1}  
		 \tag{\Cref{tab:coupling_smallstep_1}, $v_i$ is unique number s.t. $v \in \intervals{B}_{v_i}$} \\
		 ~{}={}~ & \text{Uniform}(c_1, c_2) \gg\!= \lambda v.(\delta_{(v : \floattype{B}{V}, v_i : \fin{|B| + 1})} \gg\!= \lambda(e_1,e_2). \delta_{e_1}  )
		 \tag{monad assoc., $v_i$ is unique number s.t. $v \in \intervals{B}_{v_i}$} \\
		 ~{}={}~ & \text{Uniform}(c_1, c_2) \gg\!= \lambda v. \delta_{v : \floattype{B}{V}}  
		 \\
		 ~{}={}~ & \text{Uniform}(c_1, c_2) \\
		 ~{}={}~ & \sem{e:\tau}~.
		 \tag{\Cref{tab:smallstep-1}}
	\end{align*}
	
	For \Cref{lem:small_step_coupling_props}.\ref{lem:small_step_coupling_props2}, consider the following:
	\begin{align*}
		& \pi_2(\sem{e:\tau, \discretize{e} : \discretize{\tau}}) \\
		~{}={}~ &  \sem{e:\tau, \discretize{e} : \discretize{\tau}} \gg\!= \lambda(e_1,e_2). \delta_{e_2}  
		\tag{definition} \\
		~{}={}~ & (\text{Uniform}(c_1, c_2) \gg\!= \lambda v.\delta_{(v : \floattype{B}{V}, v_i : \fin{|B| + 1})})\gg\!= \lambda(e_1,e_2). \delta_{e_2}  
		\tag{\Cref{tab:coupling_smallstep_1}, $v_i$ is unique number s.t. $v \in \intervals{B}_{v_i}$} \\
		~{}={}~ & \text{Uniform}(c_1, c_2) \gg\!= \lambda v.(\delta_{(v : \floattype{B}{V}, v_i : \fin{|B| + 1})} \gg\!= \lambda(e_1,e_2). \delta_{e_2}  )
		\tag{monad assoc., $v_i$ is unique number s.t. $v \in \intervals{B}_{v_i}$} \\
		~{}={}~ & \text{Uniform}(c_1, c_2) \gg\!= \lambda v. \delta_{v_i : \fin{|B| + 1}}  
		\tag{monad return, $v_i$ is unique number s.t. $v \in \intervals{B}_{v_i}$}
		\\
		~{}={}~ & \sem{\discrete(p_0,\ldots,p_n)}~\text{where}~p_k = \text{CDF}(b_{k+1}) - \text{CDF}(b_k) = \frac{b_{k+1} - c_1}{c_2-c_1} - \frac{b_{k} - c_1}{c_2-c_1}  =  \frac{b_{k+1} - b_k}{c_2-c_1} 
		\tag{by construction} \\
		~{}={}~ & \sem{\discretize{e}:\discretize{\tau}}~.
		\tag{\Cref{tab:smallstep-1}}
	\end{align*}
	
	Now, the remaining cases are either (i) trivial base cases or (ii) immediately follow from the I.H.\ and \Cref{lem:bind_project} further below. We demonstrate this for the case 
	\[
	e:\tau = \uniform(e_1 : \floattype{B_1}{V_1}, e_2 : \floattype{B_2}{V_2}) : \floattype{B}{V}~,
	\]
	where
	\begin{enumerate}
		\item $e_1 = c_1$ is a value but $e_2$ is no value,
		\item $B =  \{\sim_1\!\!b_1, \ldots, \sim_n\!\!b_n\}$,
		\item $V=\top$, 
		\item $V_1 = \{v_1,\ldots, v_{n_1}\}$ with $\distinguishes{B_1}{V_1}$ and $c_1 \in V_1$ due to $\has{V_1}{c_1}$ by I.H. 
		\item $V_2 = \{w_1,\ldots, w_{n_2}\}$ with $\distinguishes{B_2}{V_2}$ and  $c_2 \in V_2$ due to $\has{V_2}{c_2}$ by I.H. 
	\end{enumerate}
	For \Cref{lem:small_step_coupling_props}.\ref{lem:small_step_coupling_props2}, consider the following:
	\begin{align*}
		& \pi_1(\sem{e:\tau, \discretize{e} : \discretize{\tau}}) \\
		~{}={}~ &  \sem{e:\tau, \discretize{e} : \discretize{\tau}} \gg\!= \lambda(e_1,e_2). \delta_{e_1}  
		\tag{definition} \\
		~{}={}~ & \big( \sem{e_2:\tau_2, \discretize{e_2} : \discretize{\tau_2}} \\
		& \gg\!= \lambda (f,f'). \delta_{(\uniform(v_1:\tau_1, f:\tau_2 ) : \tau, \dc(\discretize{v_1}:\discretize{\tau_1}, f': \discretize{\tau_2} ) : \discretize{\tau})} \big)
		 \gg\!= \lambda(e_1,e_2). \delta_{e_1}   \\
		 ~{}={}~ &  \sem{e_2:\tau_2, \discretize{e_2} : \discretize{\tau_2}} \\
		 & \gg\!= \lambda (f,f'). \big(\delta_{(\uniform(v_1:\tau_1, f:\tau_2 ) : \tau, \dc(\discretize{v_1}:\discretize{\tau_1}, f': \discretize{\tau_2} ) : \discretize{\tau})} 
		 \gg\!= \lambda(e_1,e_2). \delta_{e_1} \big)
		 \tag{monad assoc.} \\
		 ~{}={}~ &  \sem{e_2:\tau_2, \discretize{e_2} : \discretize{\tau_2}} 
		  \gg\!= \lambda (f,f'). \delta_{\uniform(v_1:\tau_1, f:\tau_2 ) : \tau} 
		 \tag{monad return} \\
		 ~{}={}~ &  \pi_1(\sem{e_2:\tau_2, \discretize{e_2} : \discretize{\tau_2}})
		 \gg\!= \lambda f. \delta_{\uniform(v_1:\tau_1, f:\tau_2 ) : \tau} 
		 \tag{\Cref{lem:bind_project}.\ref{lem:bind_project1}} \\
		 ~{}={}~ &  \sem{e_2:\tau_2}
		 \gg\!= \lambda f. \delta_{\uniform(v_1:\tau_1, f:\tau_2 ) : \tau} ~.
		 \tag{$ \sem{e_2:\tau_2, \discretize{e_2} : \discretize{\tau_2}}  =  \sem{e_2:\tau_2}$ by I.H.} \\
		 ~{}={}~ &  \sem{e : \tau}
		 \tag{\Cref{tab:smallstep-1}} 
	\end{align*}
\end{proof}

\begin{lemma}
	\label{lem:bind_project}
	Let $F\colon P_\tau \to \mathcal{D}(P_\tau + \bot)$, let $G\colon Expr_\tau \to \mathcal{D}(Expr_\tau + \bot)$, and let $G' \colon Expr_{\discretize{\tau}} \to \mathcal{D}(Expr_{\discretize{\tau}} + \bot)$. We have:
	\begin{enumerate}
		\item\label{lem:bind_project1} If  $\forall (f,f') \in P_\tau\colon F(f,f') = G(f)$, then for all $\mu \in \mathcal{D}(P_\tau + \bot)$,
		\[
		\pi_1(\mu) \gg\!= \lambda f. G(f)
		~{}={}~
		\mu \gg\!= \lambda (f,f'). F(f,f')~.
		\]
		\item\label{lem:bind_project2} If  $\forall (f,f') \in P_\tau\colon F(f,f') = G'(f)$, then for all $\mu \in \mathcal{D}(P_\tau + \bot)$,
		\[
		\pi_2(\mu) \gg\!= \lambda f. G'(f)
		~{}={}~
		\mu \gg\!= \lambda (f,f'). F(f,f')~.
		\]
	\end{enumerate}
\end{lemma}
\begin{proof}
	We only prove \Cref{lem:bind_project}.\ref{lem:bind_project1} because the reasoning for \Cref{lem:bind_project}.\ref{lem:bind_project2} is completely analogous. We have
	\begin{align*}
		& \pi_1(\mu) \gg\!= \lambda f. G(f) \\
		{}={}~& (\mu \gg\!= \lambda (f,f'). \delta_f) \gg\!= \lambda f. G(f) 
		\tag{definition} \\
		{}={}~& \mu \gg\!= \lambda (f,f'). (\delta_f \gg\!= \lambda f. G(f))
		\tag{monad assoc.} \\
		{}={}~& \mu \gg\!= \lambda (f,f'). G(f)
		\tag{monad neutral} \\
		{}={}~& \mu \gg\!= \lambda (f,f'). F(f,f')~.
		\tag{assumption} 
	\end{align*}
\end{proof}

\subsection{Big-Step Semantics and Correctness of Discretization}
We extend both of the small-step semantics canonically to an $n$-step semantics by%
\begin{align*}
		\bigsemstep{e : \tau}{n} &{}= 
		\begin{cases}
			\delta_{e : \tau} & \text{if $n = 0$} \\
			\bigsemstep{e : \tau}{n-1} \monbind  \lambda e'. \sem{e' : \tau} & \text{if $n > 0$}~.
		\end{cases} \\
		\bigsemstep{e : \tau, \discretize{e} : \discretize{\tau}}{n} &{}= 
		\begin{cases}
			\delta_{(e : \tau, \discretize{e} : \discretize{\tau})} & \text{if $n = 0$} \\
			\bigsemstep{e : \tau, \discretize{e} : \discretize{\tau}}{n-1} \monbind  \lambda (e', e''). \sem{e' : \tau, e'' : \discretize{\tau}} & \text{if $n > 0$}~.
		\end{cases}
\end{align*}
As a sanity check, we show that the $n$-step semantics defines a sub-Markov kernel.

\begin{lemma}[$n$-step semantics is a sub-Markov kernel]\label{lem:big-step-n-kernel}
	For all $n \in \Nats$, $\bigsemstep{\cdot}{n} : Expr_{\tau} \to \mathcal{D}(Expr_{\tau} + \bot)$ is a sub-Markov kernel.
\end{lemma}

\begin{proof}
	By induction on $n$. 

	Base case: $n = 0$. By definition, 
	
	\[
		\bigsemstep{e : \tau}{n} = \dirac{e : \tau}
	\]
	i.e. the map $e \mapsto \dirac{e}$. By~\Cref{lem:dirac-return-kernel}, this is a sub-Markov kernel.

	Inductive step: $n > 0$. By definition, 
	\[
		\bigsemstep{e : \tau}{n} = \bigsemstep{e : \tau}{n-1} \monbind  \lambda e'. \sem{e' : \tau}
	\]
	
	By the I.H., $\bigsemstep{\cdot}{n-1} \colon Expr_{\tau} \to \mathcal{D}(Expr_{\tau} + \bot)$ is a sub-Markov kernel. By~\Cref{lem:small-step-kernel}, $\lambda e' .\sem{e' : \tau} : Expr_{\tau} \to \mathcal{D}(Expr_{\tau} + \bot)$ is a sub-Markov kernel. By~\Cref{lem:bot-extension-kernel}, its $\bot$-extension from $(Expr_{\tau} + \bot)$ to $\mathcal{D}(Expr_{\tau} + \bot)$ is also a sub-Markov kernel. By~\Cref{lem:kernel-composition}, $\bigsemstep{e : \tau}{n-1} \monbind  \lambda e'. \sem{e' : \tau}$ is a sub-Markov kernel.
\end{proof}

\begin{lemma}[$n$-step coupling semantics is a sub-Markov kernel]\label{lem:big-step-n-coupling-kernel}
	For all $n \in \Nats$, $\bigsemstep{\cdot, \cdot}{n} : P_{\tau} \to \mathcal{D}(P_{\tau} + \bot)$ is a sub-Markov kernel.
\end{lemma}

\begin{proof}
	Analogous to proof of~\Cref{lem:big-step-n-kernel}.
\end{proof}

For every type $\tau$, we denote by $\vals_\tau$ the set of values of
type $\tau$. By~\Cref{lem:measurable-subsets} and closure under finite unions,
$\vals_\tau+\bot$ and $Expr_\tau\setminus\vals_\tau$ are measurable
subsets of $Expr_\tau+\bot$, which we equip with their subspace
$\sigma$-algebras. 

Furthermore, let
\[
K_\tau:(Expr_\tau+\bot)\to\mathcal D(Expr_\tau+\bot)
\]
be the \(\bot\)-extension of the small-step kernel:
\[
K_\tau(x,A)
=
\begin{cases}
\sem{x:\tau}(A) & x\in Expr_\tau,\\
\delta_\bot(A)  & x=\bot,
\end{cases}
\qquad
A\in\mathcal F_{Expr_\tau+\bot}.
\]
Thus, for every \(e:\tau\),
\[
\bigsemstep{e:\tau}{0}=\delta_e
\]
and
\[
\bigsemstep{e:\tau}{n+1}(A)
=
\int_{Expr_\tau+\bot}
K_\tau(x,A)\,
\bigsemstep{e:\tau}{n}(dx).
\]

\begin{lemma}[Monotonicity of n-step semantics]
\label{lem:bigstepsem}
For every \(e:\tau\), \(n\in\Nats\), and
\(A\in\mathcal F_{\vals_\tau+\bot}\),
\[
\bigsemstep{e:\tau}{n}(A)
\leq
\bigsemstep{e:\tau}{n+1}(A).
\]
\end{lemma}

\begin{proof}
Fix \(e:\tau\), \(n\in\Nats\), and
\(A\in\mathcal F_{\vals_\tau+\bot}\). Write
\[
\mu_n:=\bigsemstep{e:\tau}{n}.
\]
Every \(x\in\vals_\tau+\bot\) is absorbing:
\[
K_\tau(x,-)=\delta_x.
\]
Indeed, \(\sem{v:\tau}=\delta_v\) for every
\(v\in\vals_\tau\), while the \(\bot\)-extension maps
\(\bot\) to \(\delta_\bot\). Therefore,
\begin{align*}
\bigsemstep{e:\tau}{n+1}(A)
&=
\int_{Expr_\tau+\bot}K_\tau(x,A)\,\mu_n(dx)
&& \tag{definition of the \(n\)-step semantics}
\\
&=
\int_{\vals_\tau+\bot}K_\tau(x,A)\,\mu_n(dx)
+
\int_{Expr_\tau\setminus\vals_\tau}K_\tau(x,A)\,\mu_n(dx)
&& \tag{partition of \(Expr_\tau+\bot\)}
\\
&=
\int_{\vals_\tau+\bot}\delta_x(A)\,\mu_n(dx)
+
\int_{Expr_\tau\setminus\vals_\tau}K_\tau(x,A)\,\mu_n(dx)
&& \tag{terminal states are absorbing}
\\
&=
\int_{\vals_\tau+\bot}\mathbf{1}_A(x)\,\mu_n(dx)
+
\int_{Expr_\tau\setminus\vals_\tau}K_\tau(x,A)\,\mu_n(dx)
&& \tag{\(\delta_x(A)=\mathbf{1}_A(x)\)}
\\
&=
\mu_n(A)
+
\int_{Expr_\tau\setminus\vals_\tau}K_\tau(x,A)\,\mu_n(dx)
&& \tag{\(A\subseteq\vals_\tau+\bot\)}
\\
&\geq
\mu_n(A)
&& \tag{nonnegativity}
\\
&=
\bigsemstep{e:\tau}{n}(A).
&& \tag{definition of \(\mu_n\)}
\end{align*}
\end{proof}

\begin{lemma}
\label{lem:bigstep-well-defined}
For every \(e:\tau\),
\[
\bigsem{e:\tau}\in\mathcal D(\vals_\tau+\bot).
\]
That is, \(\bigsem{e:\tau}\) is a subprobability measure on
\(\vals_\tau+\bot\).
\end{lemma}

\begin{proof}
Fix \(e:\tau\), and write
\[
\mu_n:=\bigsemstep{e:\tau}{n}.
\]
Each \(\mu_n\) is a subprobability measure on \(Expr_\tau+\bot\). Since
\(\vals_\tau+\bot\) is measurable, every
\(A\in\mathcal F_{\vals_\tau+\bot}\) is also measurable in
\(Expr_\tau+\bot\).

First,
\begin{align*}
\bigsem{e:\tau}(\emptyset)
&=
\lim_{n\to\infty}\mu_n(\emptyset)
&& \tag{definition of \(\bigsem{e:\tau}\)}
\\
&=
\lim_{n\to\infty}0
&& \tag{\(\mu_n\) is a measure}
\\
&=0.
\end{align*}

Now let
\[
(A_i)_{i\in\Nats}\subseteq\mathcal F_{\vals_\tau+\bot}
\]
be pairwise disjoint. Then
\begin{align*}
\bigsem{e:\tau}\!\left(\bigcup_{i\in\Nats}A_i\right)
&=
\lim_{n\to\infty}
\mu_n\!\left(\bigcup_{i\in\Nats}A_i\right)
&& \tag{definition of \(\bigsem{e:\tau}\)}
\\
&=
\lim_{n\to\infty}
\sum_{i\in\Nats}\mu_n(A_i)
&& \tag{countable additivity of \(\mu_n\)}
\\
&=
\sum_{i\in\Nats}
\lim_{n\to\infty}\mu_n(A_i)
&& \tag{monotone convergence on \(\Nats\)}
\\
&=
\sum_{i\in\Nats}\bigsem{e:\tau}(A_i).
&& \tag{definition of \(\bigsem{e:\tau}\)}
\end{align*}
The monotone-convergence step applies because
\[
\mu_n(A_i)\leq\mu_{n+1}(A_i)
\]
for every \(i,n\in\Nats\), by~\Cref{lem:bigstepsem}. Thus
\(\bigsem{e:\tau}\) is countably additive and therefore a measure on
\(\vals_\tau+\bot\).

Finally,
\begin{align*}
\bigsem{e:\tau}(\vals_\tau+\bot)
&=
\lim_{n\to\infty}\mu_n(\vals_\tau+\bot)
&& \tag{definition of \(\bigsem{e:\tau}\)}
\\
&\leq 1.
&& \tag{\(\mu_n\) is a subprobability measure}
\end{align*}
Therefore \(\bigsem{e:\tau}\) is a subprobability measure on
\(\vals_\tau+\bot\).
\end{proof}

Moreover, let
\[
V_\tau
=
\{(v:\tau,\discretize{v}:\discretize{\tau})\mid v:\tau\}.
\]
We equip $V_\tau+\bot$ with the subspace $\sigma$-algebra inherited
from $P_\tau+\bot$.

\begin{lemma}[Monotonicity of n-step coupling semantics]
\label{lem:bigstepsem-pairing}
For every $(e:\tau,\discretize{e}:\discretize{\tau})\in P_\tau$,
$n\in\Nats$, and $A\in\mathcal F_{V_\tau+\bot}$,
\[
\bigsemstep{e:\tau,\discretize{e}:\discretize{\tau}}{n}(A)
\leq
\bigsemstep{e:\tau,\discretize{e}:\discretize{\tau}}{n+1}(A).
\]
\end{lemma}

\begin{proof}
The coupled small-step semantics maps every element of $V_\tau$ to
its Dirac measure, while the $\bot$-extension maps $\bot$ to
$\delta_\bot$. Thus every element of $V_\tau+\bot$ is absorbing. The
remainder is analogous to the proof of~\Cref{lem:bigstepsem}, replacing
$Expr_\tau+\bot$ by $P_\tau+\bot$ and $\vals_\tau+\bot$ by
$V_\tau+\bot$.
\end{proof}

\begin{lemma}
\label{lem:bigstep-coupling-well-defined}
For every $(e:\tau,\discretize{e}:\discretize{\tau})\in P_\tau$,
\[
\bigsem{e:\tau,\discretize{e}:\discretize{\tau}}
\in\mathcal D(V_\tau+\bot).
\]
That is, the coupled big-step semantics is a subprobability measure on
$V_\tau+\bot$.
\end{lemma}

\begin{proof}
Analogous to the proof of~\Cref{lem:bigstep-well-defined}, replacing
$Expr_\tau+\bot$ by $P_\tau+\bot$, $\vals_\tau+\bot$ by
$V_\tau+\bot$, and using~\Cref{lem:bigstepsem-pairing} for
monotonicity.
\end{proof}

\begin{definition}[Big-step semantics]
Let $e:\tau$. For measurable events in the respective terminal spaces,
the source and coupled \emph{(big-step) semantics} are defined by
\begin{align*}
\bigsem{e:\tau}(A)
&:= \lim_{n\to\infty}\bigsemstep{e:\tau}{n}(A),
&& A\in\mathcal F_{\vals_\tau+\bot},
\\
\bigsem{e:\tau,\discretize{e}:\discretize{\tau}}(A)
&:= \lim_{n\to\infty}
\bigsemstep{e:\tau,\discretize{e}:\discretize{\tau}}{n}(A),
&& A\in\mathcal F_{V_\tau+\bot}.
\end{align*}
\end{definition}

By~\Cref{lem:bigstepsem,lem:bigstepsem-pairing}, the source and coupled
$n$-step semantics are monotonically increasing on every measurable
event in $\vals_\tau+\bot$ and $V_\tau+\bot$, respectively.
By~\Cref{lem:bigstep-well-defined,lem:bigstep-coupling-well-defined},
their pointwise limits are subprobability measures.

\begin{lemma}[Continunity of projections]
	For all $n \in \Nats$,
	\begin{enumerate}
		\item $\pi_1(\lim_{n \to \infty} \bigsemstep{e,e'}{n}) = \lim_{n \to \infty} \pi_1(\bigsemstep{e,e'}{n})$
		\item $\pi_2(\lim_{n \to \infty} \bigsemstep{e,e'}{n}) = \lim_{n \to \infty} \pi_2(\bigsemstep{e,e'}{n})$ 
	\end{enumerate}
\end{lemma}

\begin{proof}
	We prove part i) in detail; the proof of part ii) is analogous. Fix an arbitrary measurable $A \in \Sigma_1$:
	\begin{align*}
		&\pi_1(\lim_{n \to \infty} \bigsemstep{e,e'}{n})(A) \\
		&= \left(\lim_{n \to \infty} \bigsemstep{e,e'}{n}\right)(A \times X_2) \tag{by definition of $\pi_1$}\\
		&= \sup_{n \in \Nats} \bigsemstep{e,e'}{n}(A \times X_2) \tag{since $\bigsemstep{e,e'}{n}$ is monotone}\\
		&= \sup_{n \in \Nats} \pi_1(\bigsemstep{e,e'}{n})(A) \tag{by definition of $\pi_1$}\\
		&= \lim_{n \to \infty} \pi_1(\bigsemstep{e,e'}{n})(A) \tag{since $\pi_1(\bigsemstep{e,e'}{n})$ is monotone}
	\end{align*}
	Since $A$ was arbitrary, this implies:
	\[
		\pi_1(\lim_{n \to \infty} \bigsemstep{e,e'}{n}) =  \lim_{n \to \infty} \pi_1(\bigsemstep{e,e'}{n}),
	\]
	as desired.
\end{proof}

\begin{lemma}[$n$-step coupling projection]
	\label{lem:proj_cont}
	For all $n \in \Nats$, 
	\begin{enumerate}
		\item $\pi_1(\bigsemstep{e : \tau,\discretize{e} : \discretize{\tau}}{n}) = \bigsemstep{e : \tau}{n}$
		\item $\pi_2(\bigsemstep{e : \tau,\discretize{e} : \discretize{\tau}}{n}) = \bigsemstep{\discretize{e} : \discretize{\tau}}{n}$
	\end{enumerate}
\end{lemma}

\begin{proof}
	By induction on $n$. We prove i), as the proof for ii) is analogous.

	\underline{Base case:} $n=0$
	\begin{align*}
		&\pi_1\!\left(\bigsemstep{e:\tau,\,\discretize{e'}:\discretize{\tau'}}{0}\right) \\
		= \quad&\pi_1(\delta_{(e:\tau,\,\discretize{e'}:\discretize{\tau'})}) \tag{definition of $n$-step coupling semantics}\\
		= \quad&\delta_{e:\tau} \tag{projection of Dirac} \\
		= \quad&\bigsemstep{e:\tau}{0}. \tag{definition of $n$-step semantics}\\
	\end{align*}
	
	\underline{Inductive step:} $n > 0$
	\begin{align*}
		&\pi_1\!\left(\bigsemstep{e:\tau,\,\discretize{e'}:\discretize{\tau'}}{n}\right) \\
		= \quad&\pi_1(\bigsemstep{e : \tau, \discretize{e} : \discretize{\tau}}{n-1} \monbind  \lambda (f, g). \sem{f : \tau, g : \discretize{\tau}}) \tag{definition of $n$-step coupling semantics}\\
		= \quad&(\bigsemstep{e : \tau, \discretize{e} : \discretize{\tau}}{n-1} \monbind  \lambda (f, g). \sem{f : \tau, g : \discretize{\tau}}) \monbind \lambda (f,g).\dirac{f} \tag{by (2)} \\
		= \quad&\bigsemstep{e : \tau, \discretize{e} : \discretize{\tau}}{n-1} \monbind  \lambda (f, g). (\sem{f : \tau, g : \discretize{\tau}} \monbind \lambda (f,g).\dirac{f}) \tag{monad assoc.}\\
		= \quad&\bigsemstep{e : \tau, \discretize{e} : \discretize{\tau}}{n-1} \monbind  \lambda (f, g). \pi_1(\sem{f,g}) \tag{by (2)}\\
		= \quad&\bigsemstep{e : \tau, \discretize{e} : \discretize{\tau}}{n-1} \monbind  \lambda (f, g). \sem{f} \tag{\Cref{lem:small_step_coupling_props}.\ref{lem:small_step_coupling_props1}}\\
		= \quad&\bigsemstep{e : \tau, \discretize{e} : \discretize{\tau}}{n-1} \monbind  \lambda (f, g). (\dirac{f} \monbind \lambda f.\sem{f}) \tag{monad unit}\\
		= \quad&(\bigsemstep{e : \tau, \discretize{e} : \discretize{\tau}}{n-1} \monbind  \lambda (f, g). \dirac{f}) \monbind \lambda f.\sem{f} \tag{monad assoc.}\\
		= \quad&\pi_1(\bigsemstep{e : \tau, \discretize{e} : \discretize{\tau}}{n-1}) \monbind \lambda f.\sem{f} \tag{by (2)}\\
		= \quad&\bigsemstep{e : \tau}{n-1} \monbind \lambda f.\sem{f} \tag{by I.H.}\\
		= \quad&\bigsemstep{e : \tau}{n} \tag{definition of $n$-step semantics}
	\end{align*}
\end{proof}

Before concluding the soundness proof, we relate the auxiliary
statement $\dc$ used in the coupling semantics to the nested target
expression generated by the discretization function in
\Cref{fig:discretization-code}.

\begin{lemma}
\label{lem:sequential-evaluation}
Suppose an expression $F(e_1,e_2)$ evaluates $e_1$ first, then
$e_2$, and, when both arguments are values $v_1,v_2$, continues
as $q(v_1,v_2)$. Then
\[
	\bigsem{F(e_1,e_2)}
	=
	\bigsem{e_1}
	\monbind
	\lambda v_1.
	\bigsem{e_2}
	\monbind
	\lambda v_2.
	\bigsem{q(v_1,v_2)}.
\]
All functions bound by $\monbind$ are extended to $\bot$ as in
\Cref{def:bot-extension}.
\end{lemma}

\begin{proof}
By induction on the $n$-step semantics, using the small-step rules
for evaluating the first and second arguments and associativity of
$\monbind$. Taking the limit as $n\to\infty$ gives the result.
\end{proof}

\begin{theorem}[$\dc$ is a valid construct]
\label{thm:dc-expansion}
Let
\[
e
=
\uniform(
  e_1 : \float[B_1;V_1],
  e_2 : \float[B_2;V_2]
)
:
\float[B;\top]
\]
be well typed, where $B,B_1,B_2\neq\top$ and
$V_1,V_2\neq\top$. Let $e_{\mathrm{case}}$ be the nested
$\letkw$- and $\ifkw$-expression in the last case of
\Cref{fig:discretization-code}. Then
\[
\bigsem{
  \dc(
    \discretize{e_1 : \float[B_1;V_1]},
    \discretize{e_2 : \float[B_2;V_2]}
  )
  :
  \fin{|B|+1}
}
=
\bigsem{
  e_{\mathrm{case}}
  :
  \fin{|B|+1}
}.
\]
\end{theorem}

\begin{proof}
Write
\[
	d_1=\discretize{e_1:\float[B_1;V_1]}
	\qquad\text{and}\qquad
	d_2=\discretize{e_2:\float[B_2;V_2]}.
\]

For arbitrary values $k_1:\fin{|B_1|+1}$ and
$k_2:\fin{|B_2|+1}$, define
\[
	q(k_1,k_2)
	=
	\begin{cases}
		\discretize{\uniform(v_i,w_j):\float[B;\top]}
		& \begin{array}{l}
			\text{if there exists }(v_i,w_j)\in V_1\times V_2\text{ such that}\\
			k_1=\discretize{v_i:\float[B_1;V_1]}\text{ and }\\
			k_2=\discretize{w_j:\float[B_2;V_2]},
		  \end{array}\\[1ex]
		\diverge & \text{if no such pair exists}.
	\end{cases}
\]
When such a pair exists, it is unique because
$\distinguishes{B_1}{V_1}$ and
$\distinguishes{B_2}{V_2}$; if $V_1=\emptyset$ or
$V_2=\emptyset$, no such pair exists and hence
$q(k_1,k_2)=\diverge$ for all $k_1,k_2$.

By the operational rule for $\dc$ and
\Cref{lem:sequential-evaluation},
\[
	\bigsem{\dc(d_1,d_2)}
	=
	\bigsem{d_1}
	\monbind
	\lambda k_1.
	\bigsem{d_2}
	\monbind
	\lambda k_2.
	\bigsem{q(k_1,k_2)}.
	\tag{1}
\]

Now consider $e_{\mathrm{case}}$. Its two outer
$\letkw$-bindings evaluate $d_1$ and $d_2$ in the same order.
Once they evaluate to $k_1$ and $k_2$, finite equality and
conjunction select the unique branch indexed by $(v_i,w_j)$ when a
matching pair exists; otherwise, including whenever $V_1=\emptyset$
or $V_2=\emptyset$, the final branch yields $\diverge$.
Thus the remaining nested conditional evaluates to $q(k_1,k_2)$.
Applying \Cref{lem:sequential-evaluation} again gives
\[
	\bigsem{e_{\mathrm{case}}}
	=
	\bigsem{d_1}
	\monbind
	\lambda k_1.
	\bigsem{d_2}
	\monbind
	\lambda k_2.
	\bigsem{q(k_1,k_2)}.
	\tag{2}
\]

The right-hand sides of (1) and (2) are identical. Therefore,
\[
	\bigsem{\dc(d_1,d_2)}
	=
	\bigsem{e_{\mathrm{case}}}.
	\qedhere
\]
\end{proof}

The following will imply our desired equivalence claim:

\begin{lemma}[Big-step coupling projection]\label{lem:big_step_coupling_props}
	Let $\bigsem{e : \tau, \discretize{e'} : \discretize{\tau'}} \in \mathcal{D}(P_{\tau} + \bot)$. Then,
	\begin{enumerate}
		\item\label{lem:big_step_coupling_props1} $\pi_1(\bigsem{e : \tau, \discretize{e'} : \discretize{\tau'}}) = \bigsem{e : \tau}$
		\item\label{lem:big_step_coupling_props2} $\pi_2(\bigsem{e : \tau, \discretize{e'} : \discretize{\tau'}}) = \bigsem{\discretize{e'} : \discretize{\tau'}}$
	\end{enumerate}
\end{lemma}

\begin{proof}
	\begin{align*}
		&\pi_1(\bigsem{e : \tau, \discretize{e'} : \discretize{\tau'}}) \\
		&= \pi_1(\lim_{n \to \infty} \bigsemstep{e : \tau, \discretize{e'} : \discretize{\tau'}}{n}) \tag{by defintion of big-step coupling semantics} \\
		&= \lim_{n \to \infty} \pi_1(\bigsemstep{e : \tau, \discretize{e'} : \discretize{\tau'}}{n}) \tag{by communitivity lemma} \\
		&= \lim_{n \to \infty} \bigsemstep{e : \tau}{n} \tag{by definition of n-step projection lemma} \\
		&= \bigsem{e : \tau} \tag{by definition of big-step semantics}
	\end{align*}
\end{proof}


%
%
%
\begin{theorem}\label{thm:contextual-equivalence}
	Let $e:\bool$. Then $\bigsem{e:\bool} = \bigsem{\discretize{e:\bool} : \bool}$.
\end{theorem}
\begin{proof}
	Since $\mathcal{D}(Vals_\tau + \bot)$ contains only discrete distributions, we proceed by case distinction on the probability of the final outcome $v\in\{\true,\false,\bot\}$.\\

	\emph{The case $v=\true$.} We have 
	\begin{align*}
		&  \bigsem{e : \bool}(\true) \\
		{}={}~& \pi_1 (\bigsem{e : \bool, \discretize{e}:\discretize{\bool}})(\true)
		\tag{\Cref{lem:big_step_coupling_props}.\ref{lem:big_step_coupling_props1}} \\
		{}={}~& (\bigsem{e : \bool, \discretize{e}:\discretize{\bool}} \gg\!= \lambda(e_1,e_2). \delta_{e_1})(\true)
		\tag{definition} \\
		{}={}~& \bigsem{e : \bool, \discretize{e}:\discretize{\bool}}((\true,\true))\cdot \delta_{\true}(\true) \\
		&{}+ 
		\bigsem{e : \bool, \discretize{e}:\discretize{\bool}}((\false,\false))\cdot \delta_{\false}(\true) \\
		&{}+ 
		\bigsem{e : \bool, \discretize{e}:\discretize{\bool}}(\bot)\cdot \delta_{\bot}(\true) 
		\tag{definition of $\gg\!=$} \\
		{}={}~& \bigsem{e : \bool, \discretize{e}:\discretize{\bool}}((\true,\true))
		\tag{definition of $\delta$} \\
		{}={}~& \pi_2(\bigsem{e : \bool, \discretize{e}:\discretize{\bool}})(\true)
		\tag{reasoning completely analogous} \\
		{}={}~& \bigsem{\discretize{e}:\discretize{\bool}}(\true)
		\tag{\Cref{lem:big_step_coupling_props}.\ref{lem:big_step_coupling_props2}} ~.
	\end{align*}
	
	The cases $v=\false$ and $v=\bot$ are completely analogous.
	
%
%
\end{proof}

\section{Benchmarks}
\label{app:benchmarks}
\subsection{Asymptotic Scaling}

\subsubsection{Code Sketches for Scaling Benchmarks}\label{app:synthetic-sketches}
This section contains code snippets of the programs  (\Cref{fig:cond-benchmarks-a-example,fig:cond-benchmarks-b-example,fig:cond-benchmarks-c-example,fig:cond-benchmarks-d-example,fig:alt-benchmarks-a-example,fig:alt-benchmarks-b-example,fig:alt-benchmarks-c-example,fig:alt-benchmarks-d-example}) used in the scaling experiments in \Cref{sec:synthetic-benchmarks}. We implemented a Python script that automatically generates these programs based on configurable parameters such as program size, guard structure, and dependency pattern. For reproducibility, we used a fixed random seed to generate each program. All $r_i$ values in the generated programs, where $0 \leq r_i \leq 1$ is randomly chosen, are rounded to two decimal places.

All of these programs follow the same structural pattern: they define a sequence of random variables $X_1,\ldots,X_n$, where each $X_i$ is drawn from a uniform distribution whose parameters depend on the value of a single predecessor variable through a conditional branch of the form:
\[
X_i =
  \begin{cases}
    \mathrm{Unif}(a_i,b_i), & X_{\pi(i)} < r_i, \\[2pt]
    \mathrm{Unif}(a'_i,b'_i), & \text{otherwise},
  \end{cases}
\]
for some parent selection function $\pi : \{2,\ldots,n\} \to \{1,\ldots,i-1\}$.  
Thus each program is a Bayesian network of conditionally independent nodes, differing only in the choice of $\pi(i)$ that determines the dependency pattern.

\begin{figure}
    \centering
    \begin{subfigure}{\textwidth}
        \begin{lstlisting}
            let x1 = uniform(0,1) in
            let x2 = if x1 < 0.56 then uniform(0,2) else uniform(0,3) in
            let x3 = if x2 < 0.92 then uniform(0,4) else uniform(0,5) in
            let x4 = if x3 < 0.47 then uniform(0,6) else uniform(0,7) in
            let x5 = if x4 < 0.51 then uniform(0,8) else uniform(0,9) in
            x5 < 0.5
        \end{lstlisting}
        \caption{Conditional Independent (\Cref{fig:cond-benchmarks-a}): Size 5 chain in which each variable depends on its immediate predecessor; all uniform distributions are distinct.}
        \label{fig:cond-benchmarks-a-example}
    \end{subfigure}
    \begin{subfigure}{\textwidth}
        \begin{lstlisting}
            let x1 = uniform(0,1) in
            let x2 = if x1 < 0.56 then uniform(0,2) else uniform(0,3) in
            let x3 = if x1 < 0.45 then uniform(0,4) else uniform(0,5) in
            let x4 = if x1 < 0.18 then uniform(0,6) else uniform(0,7) in
            let x5 = if x1 < 0.48 then uniform(0,8) else uniform(0,9) in
            x5 < 0.5
        \end{lstlisting}
        \caption{Conditional Independent Fork (\Cref{fig:cond-benchmarks-b}): Size 5 fork in which $x_1$ is the sole parent of all subsequent variables; all uniform distributions are distinct.}
        \label{fig:cond-benchmarks-b-example}
    \end{subfigure}
    \begin{subfigure}{\textwidth}
        \begin{lstlisting}
            let x1 = uniform(0,1) in
            let x2 = if x1 < 0.18 then uniform(0,2) else uniform(0,3) in
            let x3 = if x1 < 0.63 then uniform(0,4) else uniform(0,5) in
            let x4 = if x2 < 0.09 then uniform(0,6) else uniform(0,7) in
            let x5 = if x1 < 0.14 then uniform(0,8) else uniform(0,9) in
            x5 < 0.5
        \end{lstlisting}
        \caption{Conditional Random 1 (\Cref{fig:cond-benchmarks-c}): Size 5 program with randomly selected predecessors. Some variables may be unused; all uniform distributions are distinct.}
        \label{fig:cond-benchmarks-c-example}
    \end{subfigure}
    \begin{subfigure}{\textwidth}
        \begin{lstlisting}
            let x1 = uniform(0, 10) in
            let x2 = if x1 < 0.06 then x1 else x1 in
            let x3 = if uniform(0, 4) < 0.64 then x2 else x1 in
            let x4 = if x3 < 0.71 then uniform(0, 5) else x1 in
            let x5 = if uniform(0, 4) < 0.4 then x1 else uniform(0, 2) in
            x4 < 0.5
        \end{lstlisting}
        \caption{Conditional Random 2 (\Cref{fig:cond-benchmarks-d}): Size 5 program with randomized guards and branches. Uniform upper bounds are random integers from 0 to 10 and may repeat.}
        \label{fig:cond-benchmarks-d-example}
    \end{subfigure}
    \caption{Representative programs from the scaling benchmarks.}
    \label{fig:cond-benchmark-examples}
\end{figure}

\begin{figure}
    \centering
    \begin{subfigure}{\textwidth}
        \begin{lstlisting}
            let x1 = uniform(0,1) in
            let x2 = if x1 < 0.02 then uniform(0,2) else uniform(0,3) in
            let x3 = if x2 < 0.68 then uniform(0,4) else uniform(0,5) in
            let x4 = if x1 < 0.21 then uniform(0,6) else uniform(0,7) in
            let x5 = if x2 < 0.93 then uniform(0,8) else uniform(0,9) in
            x5 < 0.5
        \end{lstlisting}
        \caption{Alternating Guard 1 (\Cref{fig:alt-benchmarks-a}): Size 5 program with guard span 2. Guard variables cycle over the span, so variables beyond it may be unused.}
        \label{fig:alt-benchmarks-a-example}
    \end{subfigure}
    \begin{subfigure}{\textwidth}
        \begin{lstlisting}
            let x1 = uniform(0,1) in
            let x2 = if x1 < 0.38 then uniform(0,2) else x1 in
            let x3 = if x2 < 0.71 then uniform(0,4) else x2 in
            let x4 = if x1 < 0.42 then uniform(0,6) else x3 in
            let x5 = if x2 < 0.57 then uniform(0,8) else x4 in
            x5 < 0.5
        \end{lstlisting}
        \caption{Alternating Guard 2 (\Cref{fig:alt-benchmarks-b}): Size 5 program with guard span 2. Every variable is later used as a guard or else-branch expression; all uniform distributions are distinct.}
        \label{fig:alt-benchmarks-b-example}
    \end{subfigure}
    \begin{subfigure}{\textwidth}
        \begin{lstlisting}
            let x1 = uniform(0,1) in
            let x2 = if x1 < 0.67 then uniform(0,2) else uniform(0,3) in
            let x3 = if x2 < 0.5 then x2 else uniform(0,4) in
            let x4 = if x1 < 0.5 then x3 else uniform(0,6) in
            let x5 = if x2 < 0.5 then x4 else uniform(0,8) in
            x5 < 0.5
        \end{lstlisting}
        \caption{Alternating Guard 3 (\Cref{fig:alt-benchmarks-c}): Size 5 program with guard span 2. Every variable is later used as a guard or then-branch expression; all uniform distributions are distinct.}
        \label{fig:alt-benchmarks-c-example}
    \end{subfigure}
    \begin{subfigure}{\textwidth}
        \begin{lstlisting}
            let x1 = uniform(0,2) in
            let x2 = if x1 < 0.13 then uniform(0,2) else uniform(0,1) in
            let x3 = if x2 < 0.06 then uniform(0,8) else uniform(0,3) in
            let x4 = if x1 < 0.68 then uniform(0,4) else uniform(0,8) in
            let x5 = if x2 < 0.51 then uniform(0,3) else uniform(0,7) in
            x5 < 0.5
        \end{lstlisting}
        \caption{Random Alternating (\Cref{fig:alt-benchmarks-d}): Size 5 program with guard span 2. Uniform upper bounds are random integers from 0 to 10 and may repeat.}
        \label{fig:alt-benchmarks-d-example}
    \end{subfigure}
    \caption{Representative programs from the scaling benchmarks.}
    \label{fig:alt-benchmark-examples}
\end{figure}

The code snippets show representative size $5$ instances from each benchmark family. The scaling experiments generate larger instances of the same form by varying the program size and, for the alternating benchmarks, the guard span. The programs in \Cref{fig:cond-benchmarks-a-example} follow the scheme depicted above where the program size is $n$, i.e., the number of comparisons performed is $n$. In \Cref{fig:cond-benchmarks-b-example}, each $x_i$ does not necessarily depend on $x_{i-1}$ but on some randomly chosen predecessor\footnote{We generate these programs in a randomized manner with a fixed seed for reproducibility}. In \Cref{fig:cond-benchmarks-c-example,fig:cond-benchmarks-d-example}, we additionally choose at random whether the guards of the conditionals and the branches are either some predecessor variable $x_i$ or a  $\uniform$-sampling instruction.  The programs in \Cref{fig:alt-benchmarks-a-example} also follow the scheme depicted above but where the variables used in the guards cycle with a fixed span size of $m$. For instance, for $m=2$, the pattern is $x_1,x_2,x_1,\ldots$. The programs in \Cref{fig:alt-benchmarks-b-example,fig:alt-benchmarks-c-example,fig:alt-benchmarks-d-example} follow a similar pattern.
\begin{figure}[!t]
\centering
\begin{subfigure}{0.4\textwidth}
\includegraphics[width=\textwidth]{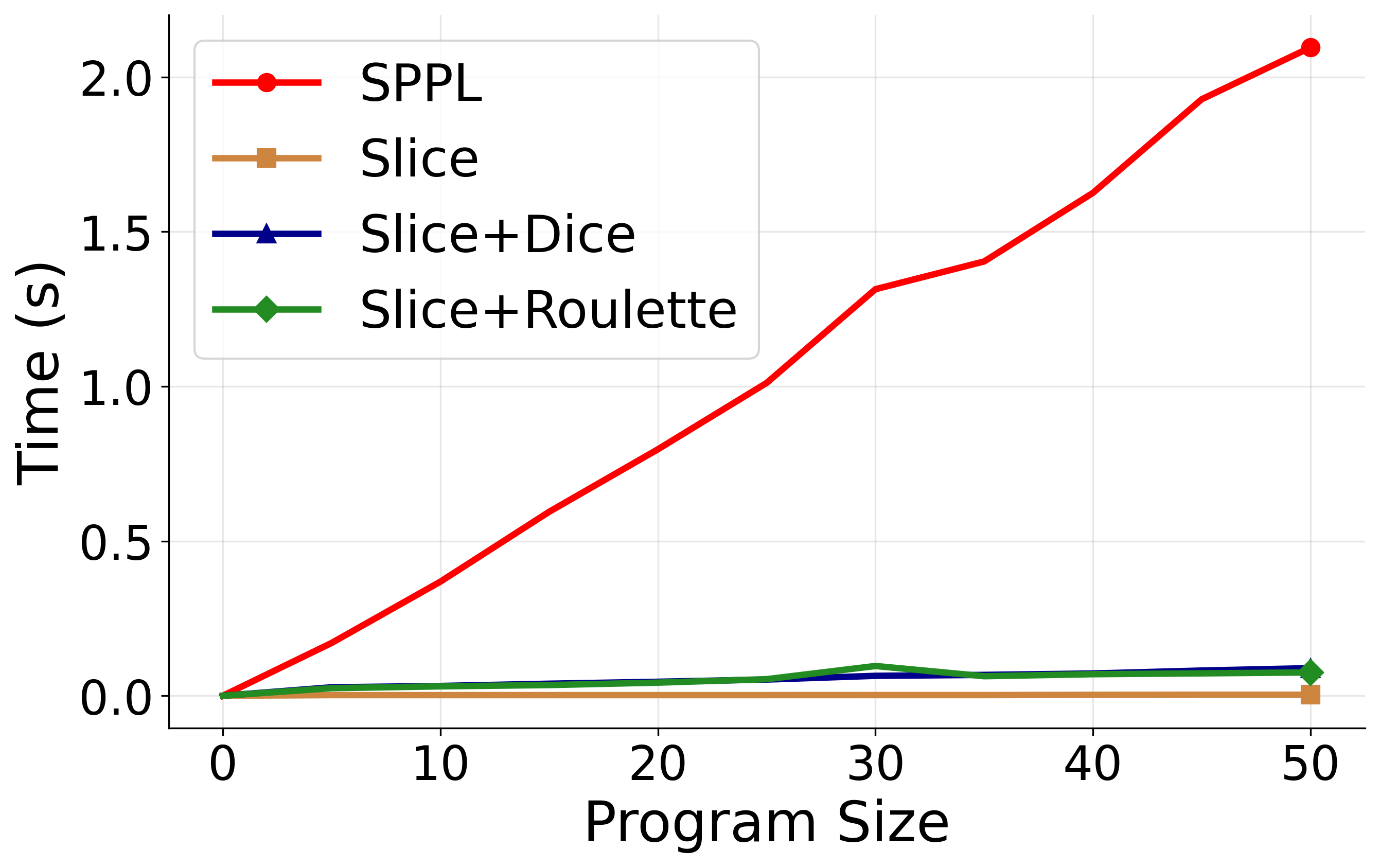}
\caption{Conditional Independent}
\label{fig:cond-benchmarks-a}
\end{subfigure}
\hfill
\begin{subfigure}{0.4\textwidth}
\includegraphics[width=\textwidth]{images/scaling/build_conditional_independent_fork_slice.png}
\caption{Conditional Independent -- Fork}
\label{fig:cond-benchmarks-b}
\end{subfigure}
\hfill
\begin{subfigure}{0.4\textwidth}
\includegraphics[width=\textwidth]{images/scaling/build_conditional_random_independent_slice_1.png}
\caption{Conditional Independent -- Rand. 1}
\label{fig:cond-benchmarks-c}
\end{subfigure}
\hfill
\begin{subfigure}{0.4\textwidth}
\includegraphics[width=\textwidth]{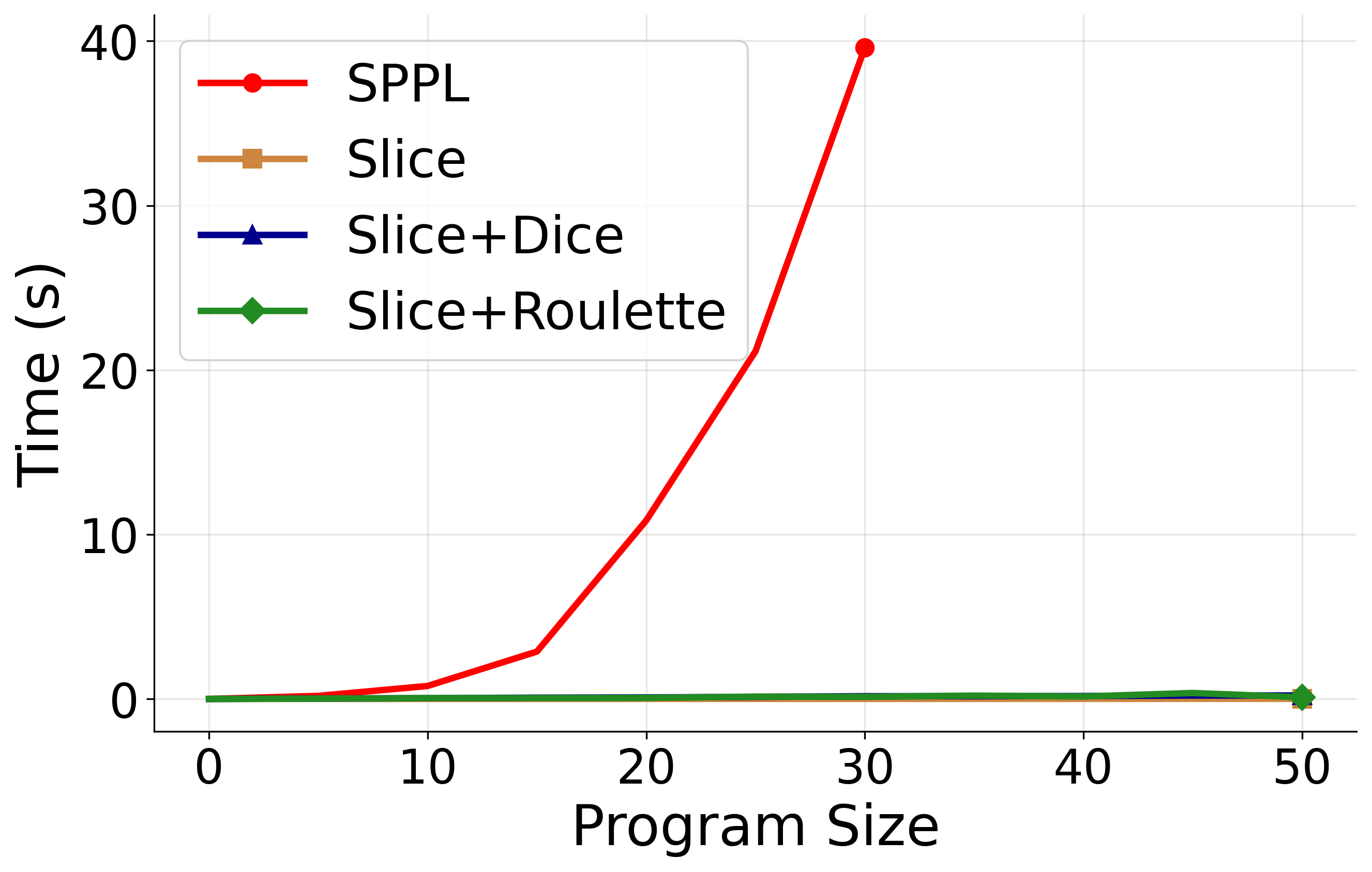}
\caption{Conditional Independent -- Rand. 2}
\label{fig:cond-benchmarks-d}
\end{subfigure}
\caption{Scaling graphs for conditionally independent programs. Time-out=120s; a disappearing line indicates timeout. Graphs (a)--(d) correspond, respectively, to the representative programs in \Cref{fig:cond-benchmarks-a-example,fig:cond-benchmarks-b-example,fig:cond-benchmarks-c-example,fig:cond-benchmarks-d-example}.}
\label{fig:cond-benchmarks}
\end{figure}

\begin{figure}[!t]
\centering
\begin{subfigure}{0.4\textwidth}
\includegraphics[width=\textwidth]{images/scaling/build_alternating_guard_slice_1.png}
\caption{Alt. Guard 1}
\label{fig:alt-benchmarks-a}
\end{subfigure}
\hfill
\begin{subfigure}{0.4\textwidth}
\includegraphics[width=\textwidth]{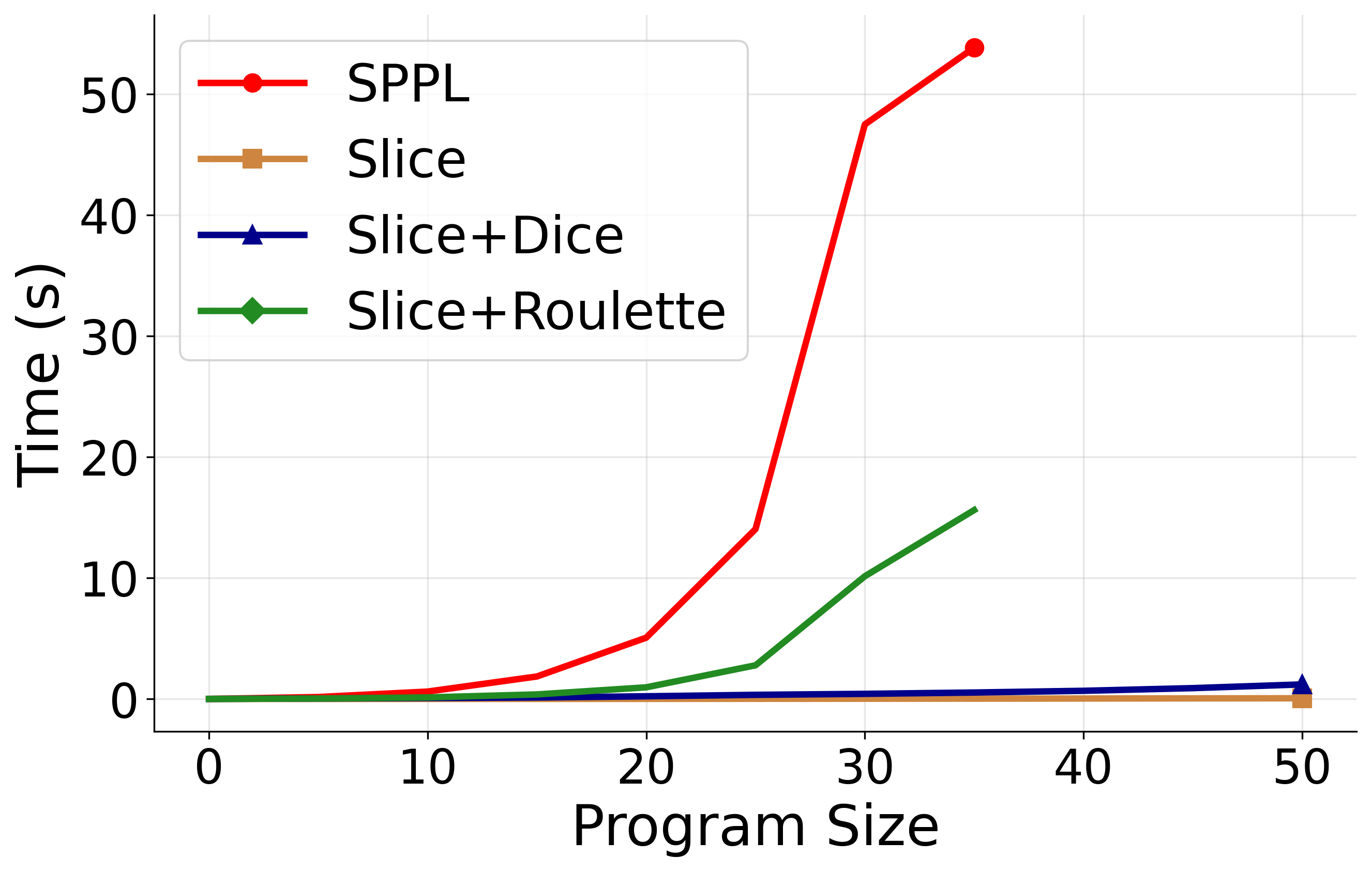}
\caption{Alt. Guard 2}
\label{fig:alt-benchmarks-b}
\end{subfigure}
\hfill
\begin{subfigure}{0.4\textwidth}
\includegraphics[width=\textwidth]{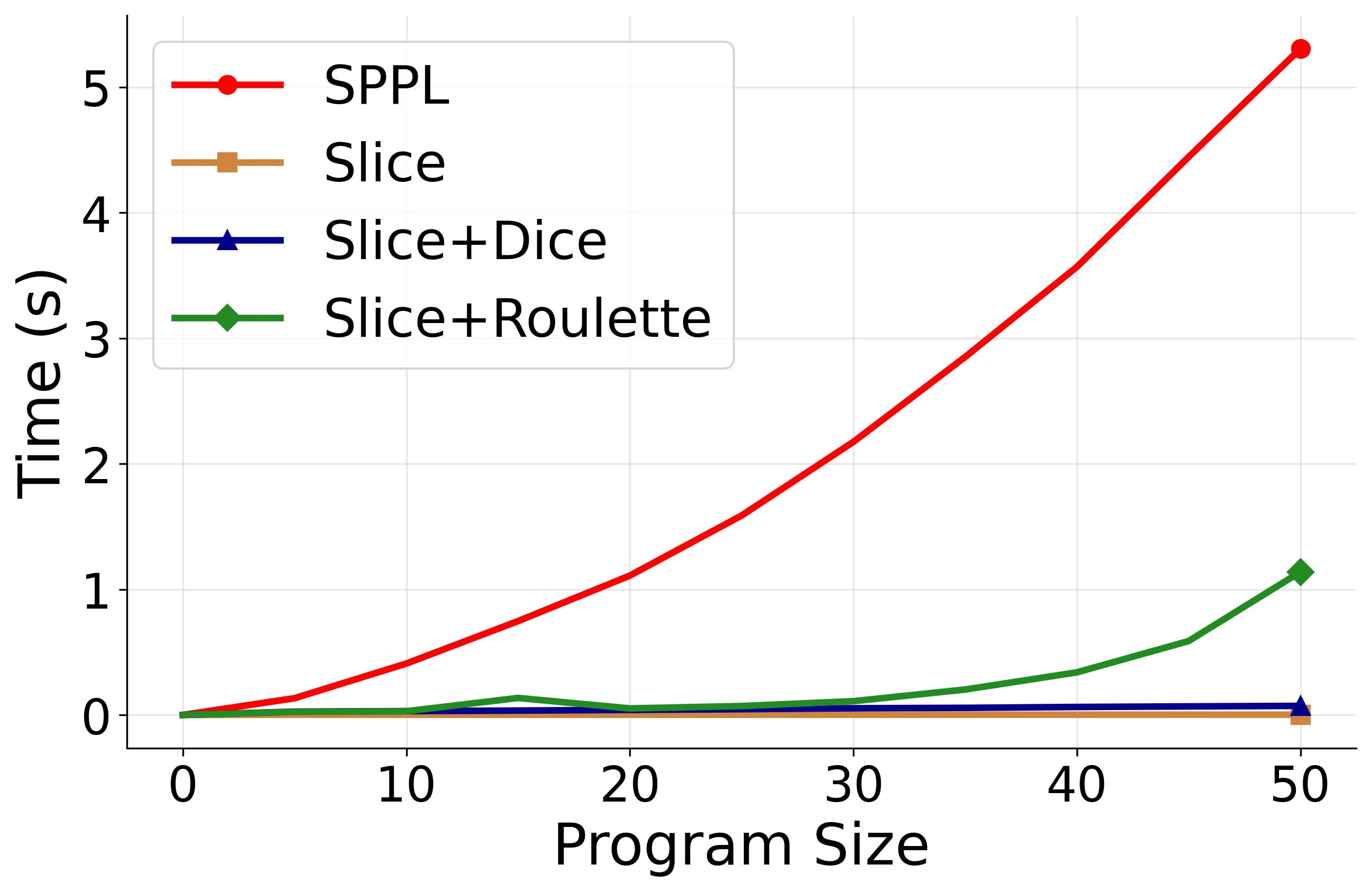}
\caption{Alt. Guard 3}
\label{fig:alt-benchmarks-c}
\end{subfigure}
\hfill
\begin{subfigure}{0.4\textwidth}
\includegraphics[width=\textwidth]{images/scaling/build_random_alternating_guard_slice.png}
\caption{Alt. Guard -- Rand.}
\label{fig:alt-benchmarks-d}
\end{subfigure}
\caption{Scaling graphs for conditionally independent programs with alternating guards. Time-out=120s; a disappearing line indicates timeout. Graphs (a)--(d) correspond, respectively, to the representative programs in \Cref{fig:alt-benchmarks-a-example,fig:alt-benchmarks-b-example,fig:alt-benchmarks-c-example,fig:alt-benchmarks-d-example}.}
\label{fig:alt-benchmarks}
\end{figure}

We report runtime (seconds) for (i) \SPPL{}, (ii) \Slice{} (discretization only), (iii) \Slice{}+\Dice{} (discretization + inference via \Dice{}), and (iv) \Slice{}+\Roulette{} (discretization + inference via \Roulette{}). The timeout is 120 seconds.

%
%

\paragraph{Results.} Discretization remains inexpensive across these benchmark families, and the resulting \Slice{}+\Dice{}/\Roulette{} pipelines scale well with increasing program size. In contrast, exact inference with \SPPL{} exhibits substantially steeper runtime growth and reaches the timeout on several benchmark variants.

\subsection{Decision Trees}
\begin{table}[!t]
    \centering
    \caption{Execution times (seconds) for the benchmarks from \cite{albarghouthi2017fairsquare}, with a timeout of 120s. The \texttt{err.}'s from Hakaru's simplification step represent a crash, caused by an internal symbolic processing exception.}
    \label{tab:benchmarks_dt}
    \small
    \setlength{\tabcolsep}{2.25pt}
    \resizebox{\columnwidth}{!}{%
    \begin{tabular}{llccccccccccccc}
    \toprule
    \makecell{\textbf{Decision}\\\textbf{Program}} &
    \makecell{\textbf{Population}\\\textbf{Model}} &
    \multicolumn{3}{c}{\textbf{Slice+Dice (s)}} &
    \multicolumn{3}{c}{\textbf{Slice+Roulette (s)}} &
    \multicolumn{3}{c}{\textbf{SPPL (s)}} &
    \multirow{2}{*}{\textbf{PSI (s)}} &
    \multicolumn{3}{c}{\textbf{Hakaru (s)}} \\
    \cmidrule(lr){3-5} \cmidrule(lr){6-8} \cmidrule(lr){9-11} \cmidrule(lr){13-15}
     &  & \textbf{Slice} & \textbf{Dice} & \textbf{Total} &
          \textbf{Slice} & \textbf{Roulette} & \textbf{Total} &
          \textbf{Trans.} & \textbf{Judg.} & \textbf{Total} & &
          \textbf{Disint.} & \textbf{Simplify} & \textbf{Total} \\
    \midrule
    DT$_4$            & Ind  & 0.002 & 0.031 & 0.033 & 0.002 & 0.026 & \textbf{0.028} & 0.074 & 0.011 & 0.085 & 3.477 & 0.104 & 1.806 & 1.910 \\
                      & BN1  & 0.002 & 0.033 & 0.035 & 0.002 & 0.031 & \textbf{0.033} & 0.245 & 0.028 & 0.273 & 19.890 & 0.031 & 0.426 & 0.457 \\
                      & BN2  & 0.002 & 0.035 & \textbf{0.037} & 0.003 & 0.034 & \textbf{0.037} & 0.268 & 0.030 & 0.298 & 51.107 & 0.030 & 0.528 & 0.558 \\
    \midrule
    DT$_{14}$         & Ind  & 0.002 & 0.047 & \textbf{0.049} & 0.002 & 0.049 & 0.051 & 0.153 & 0.024 & 0.177 & 26.150 & 0.038 & err. & err. \\
                      & BN1  & 0.002 & 0.056 & \textbf{0.058} & 0.002 & 0.074 & 0.076 & 0.597 & 0.076 & 0.673 & t.o. & 0.029 & 1.067 & 1.096 \\
                      & BN2  & 0.002 & 0.058 & \textbf{0.060} & 0.002 & 0.073 & 0.075 & 0.535 & 0.072 & 0.607 & t.o. & 0.028 & 2.619 & 2.647 \\
    \midrule
    DT$_{16}$         & Ind  & 0.002 & 0.045 & \textbf{0.047} & 0.002 & 0.054 & 0.056 & 0.176 & 0.026 & 0.202 & t.o. & 0.031 & err. & err. \\
                      & BN1  & 0.002 & 0.053 & \textbf{0.055} & 0.002 & 0.079 & 0.081 & 0.535 & 0.088 & 0.623 & t.o. & 0.036 & 1.125 & 1.161 \\
                      & BN2  & 0.003 & 0.053 & \textbf{0.056} & 0.003 & 0.079 & 0.082 & 0.554 & 0.090 & 0.644 & t.o. & 0.035 & 1.594 & 1.629 \\
    \midrule
    DT$^{\alpha}_{16}$ & Ind & 0.002 & 0.034 & \textbf{0.036} & 0.002 & 0.052 & 0.054 & 0.201 & 0.029 & 0.230 & t.o. & 0.031 & err. & err. \\
                       & BN1 & 0.002 & 0.042 & \textbf{0.044} & 0.003 & 0.191 & 0.194 & 0.597 & 0.080 & 0.677 & t.o. & 0.036 & err. & err. \\
                       & BN2 & 0.003 & 0.049 & \textbf{0.052} & 0.003 & 0.224 & 0.227 & 0.642 & 0.075 & 0.717 & t.o. & 0.040 & 2.142 & 2.182 \\
    \midrule
    DT$_{44}$         & Ind  & 0.002 & 0.099 & \textbf{0.101} & 0.002 & 0.276 & 0.278 & 0.402 & 0.073 & 0.475 & t.o. & 0.041 & t.o. & t.o. \\
                      & BN1  & 0.003 & 0.120 & \textbf{0.123} & 0.003 & 0.550 & 0.553 & 1.043 & 0.309 & 1.352 & t.o. & 0.042 & 15.972 & 16.014 \\
                      & BN2  & 0.003 & 0.122 & \textbf{0.125} & 0.003 & 0.546 & 0.549 & 1.135 & 0.346 & 1.481 & t.o. & 0.053 & 44.456 & 44.509 \\
    \bottomrule
    \end{tabular}
    }
  \end{table}

\Cref{tab:benchmarks_dt,tab:benchmarks_dt_extra} report the decision tree benchmarks for varying tree sizes $n$. The first table contains the original decision tree benchmarks from~\cite{albarghouthi2017fairsquare}, while the second extends this evaluation with additional, larger instances that we generated to more systematically study scaling as the tree size increases.

\paragraph{Results.} We observe that \Slice{}+\Dice{} achieves the lowest total runtime in practically all $DT$ programs. \Slice{}+\Roulette{} is competitive only on very small trees, and becomes substantially slower by $DT_{100}$--$DT_{500}$ and times out from $DT_{800}$ onward. Additionally, the discretization cost grows substantially by $DT_{1000}$--$DT_{5000}$. \PSI{} completes the three $DT_4$ variants and the independent $DT_{14}$ variant, but times out on the remaining programs. Hakaru's disintegration phase remains fast, but its subsequent simplification is less robust: it produces errors on several smaller independent variants, approaches the timeout at $DT_{500}$, and times out on all variants from $DT_{800}$ onward.

\begin{table}[!t]
  \centering
  \caption{Execution times (seconds) for the benchmarks from \cite{albarghouthi2017fairsquare}, with a timeout of 120s. For $DT_{5000}$ under our setup, SPPL returned a maximum recursion depth exceeded error (reported as \texttt{err.}).}
  \label{tab:benchmarks_dt_extra}
  \small
  \setlength{\tabcolsep}{2.25pt}
  \resizebox{\columnwidth}{!}{%
  \begin{tabular}{llccccccccccccc}
  \toprule
  \makecell{\textbf{Decision}\\\textbf{Program}} &
  \makecell{\textbf{Population}\\\textbf{Model}} &
  \multicolumn{3}{c}{\textbf{Slice+Dice (s)}} &
  \multicolumn{3}{c}{\textbf{Slice+Roulette (s)}} &
  \multicolumn{3}{c}{\textbf{SPPL (s)}} &
  \multirow{2}{*}{\textbf{PSI (s)}} &
  \multicolumn{3}{c}{\textbf{Hakaru (s)}} \\
  \cmidrule(lr){3-5} \cmidrule(lr){6-8} \cmidrule(lr){9-11} \cmidrule(lr){13-15}
   &  & \textbf{Slice} & \textbf{Dice} & \textbf{Total} &
        \textbf{Slice} & \textbf{Roulette} & \textbf{Total} &
        \textbf{Trans.} & \textbf{Judg.} & \textbf{Total} & &
        \textbf{Disint.} & \textbf{Simplify} & \textbf{Total} \\
  \midrule
  DT$_{50}$         & Ind  & 0.006 & 0.094 & \textbf{0.100} & 0.004 & 0.156 & 0.160 & 0.765 & 0.029 & 0.793 & t.o. & 0.045 & 3.712 & 3.757 \\
                    & BN1  & 0.007 & 0.079 & \textbf{0.086} & 0.006 & 0.451 & 0.457 & 1.646 & 0.085 & 1.731 & t.o. & 0.039 & 3.548 & 3.587 \\
                    & BN2  & 0.007 & 0.083 & \textbf{0.090} & 0.006 & 0.403 & 0.409 & 1.580 & 0.080 & 1.660 & t.o. & 0.037 & 3.635 & 3.672 \\
  \midrule
  DT$_{100}$        & Ind  & 0.011 & 0.138 & \textbf{0.149} & 0.011 & 0.654 & 0.665 & 1.455 & 0.036 & 1.490 & t.o. & 0.057 & 8.490 & 8.547 \\
                    & BN1  & 0.015 & 0.148 & \textbf{0.163} & 0.016 & 1.775 & 1.791 & 2.867 & 0.150 & 3.017 & t.o. & 0.077 & 8.596 & 8.673 \\
                    & BN2  & 0.017 & 0.153 & \textbf{0.170} & 0.016 & 1.775 & 1.791 & 2.914 & 0.152 & 3.066 & t.o. & 0.069 & 8.816 & 8.885 \\
  \midrule
  DT$_{200}$        & Ind  & 0.045 & 0.289 & \textbf{0.334} & 0.045 & 4.089 & 4.134 & 2.710 & 0.059 & 2.770 & t.o. & 0.099 & 24.551 & 24.650 \\
                    & BN1  & 0.064 & 0.322 & \textbf{0.386} & 0.064 & 9.496 & 9.560 & 5.836 & 0.318 & 6.154 & t.o. & 0.111 & 23.741 & 23.852 \\
                    & BN2  & 0.068 & 0.305 & \textbf{0.373} & 0.068 & 9.556 & 9.624 & 5.797 & 0.302 & 6.099 & t.o. & 0.115 & 23.988 & 24.103 \\
  \midrule
  DT$_{500}$        & Ind  & 0.479 & 0.879 & \textbf{1.358} & 0.515 & 55.677 & 56.192 & 7.908 & 0.141 & 8.049 & t.o. & 0.295 & 118.438 & 118.733 \\
                    & BN1  & 0.599 & 0.961 & \textbf{1.560} & 0.636 & 113.426 & 114.062 & 15.017 & 0.787 & 15.804 & t.o. & 0.293 & 117.507 & 117.800 \\
                    & BN2  & 0.616 & 0.999 & \textbf{1.615} & 0.610 & 112.976 & 113.586 & 15.330 & 0.889 & 16.219 & t.o. & 0.301 & 113.667 & 113.968 \\
  \midrule
  DT$_{800}$        & Ind  & 1.738 & 1.592 & \textbf{3.330} & 1.738 & t.o. & t.o. & 14.698 & 0.238 & 14.937 & t.o. & 0.488 & t.o. & t.o. \\
                    & BN1  & 2.184 & 1.690 & \textbf{3.874} & 2.184 & t.o. & t.o. & 26.297 & 1.265 & 27.561 & t.o. & 0.597 & t.o. & t.o. \\
                    & BN2  & 2.151 & 1.704 & \textbf{3.855} & 2.151 & t.o. & t.o. & 26.542 & 1.298 & 27.840 & t.o. & 0.512 & t.o. & t.o. \\
  \midrule
  DT$_{1000}$       & Ind  & 3.584 & 2.096 & \textbf{5.680} & 3.584 & t.o. & t.o. & 17.626 & 0.279 & 17.905 & t.o. & 0.774 & t.o. & t.o. \\
                    & BN1  & 4.125 & 2.192 & \textbf{6.317} & 4.125 & t.o. & t.o. & 33.417 & 1.537 & 34.954 & t.o. & 0.784 & t.o. & t.o. \\
                    & BN2  & 4.052 & 2.152 & \textbf{6.204} & 4.052 & t.o. & t.o. & 34.225 & 1.612 & 35.837 & t.o. & 0.803 & t.o. & t.o. \\
  \midrule
  DT$_{2000}$       & Ind  & 26.791 & 4.607 & \textbf{31.398} & 26.791 & t.o. & t.o. & 44.957 & 0.786 & 45.742 & t.o. & 2.868 & t.o. & t.o. \\
                    & BN1  & 29.148 & 4.726 & \textbf{33.874} & 29.148 & t.o. & t.o. & 84.744 & 3.227 & 87.972 & t.o. & 2.775 & t.o. & t.o. \\
                    & BN2  & 30.353 & 4.715 & \textbf{35.068} & 30.353 & t.o. & t.o. & 76.398 & 3.152 & 79.550 & t.o. & 2.944 & t.o. & t.o. \\
  \midrule
  DT$_{5000}$       & Ind  & t.o. & t.o. & \textbf{t.o.} & t.o. & t.o. & \textbf{t.o.} & err. & err. & \textbf{err.} & t.o. & 16.030 & t.o. & t.o. \\
                  & BN1  & t.o. & t.o. & \textbf{t.o.} & t.o. & t.o. & \textbf{t.o.} & err. & err. & \textbf{err.} & t.o. & 14.422 & t.o. & t.o. \\
                  & BN2  & t.o. & t.o. & \textbf{t.o.} & t.o. & t.o. & \textbf{t.o.} & err. & err. & \textbf{err.} & t.o. & 14.933 & t.o. & t.o. \\
  \bottomrule
  \end{tabular}
  }
\end{table}

\FloatBarrier
\subsection{Baselines}
We finally evaluate \Slice{} on benchmarks drawn from prior work~\cite{gehr2016psi,Saad2021SPPL}. Most benchmark programs are taken directly from the \SPPL{} Github repository and translated into \Slice{}. In some cases, the \SPPL{} versions differ from the original \PSI{} benchmarks, for example by replacing continuous choices with finite ones so that the programs fall within the class supported by \SPPL{}. CoinBias and SurveyUnbias are not included in the \SPPL{} repository, so we adapted these benchmarks in the same style. Consequently, the \Slice{} and \SPPL{} columns in \Cref{tab:benchmarks_psi} evaluate the same benchmark programs and provide an apples-to-apples comparison. The full benchmark programs are provided in \Cref{app:baseline-programs-full}. The first column in the table gives the benchmark name, the second shows the number of datasets on which to condition the program, and the remaining columns display execution times in seconds. For benchmarks with both modified and unmodified versions, the \PSI{} and \Hakaru{} cells report the modified timing marked by $\bullet$ above the unmodified timing marked by $\circ$. The $\bullet$ entries therefore provide the apples-to-apples comparison with \Slice{} and \SPPL{}. For Hakaru, we report disintegration, subsequent simplification, and their total. \texttt{t.o.} denotes the same 120s timeout used elsewhere, and \texttt{err.} denotes a tool error.

\begin{table}[H]
\centering
\caption{Execution times (in seconds) for benchmarks from PSI~\cite{gehr2016psi} and SPPL~\cite{Saad2021SPPL}, with modified and unmodified \PSI{} and Hakaru exact inference baselines. Timeout=120s. For benchmarks with two outputs for \PSI{} and Hakaru, $\bullet$ marks the modified program and $\circ$ marks the unmodified program. \texttt{err.} represents a crash, caused by out of memory.}
    \label{tab:benchmarks_psi}
\small
\setlength{\tabcolsep}{2pt}
\resizebox{\columnwidth}{!}{%
\begin{tabular}{lcccccccc}
    \toprule
    \multirow{2}{*}{\textbf{Benchmark}} &
    \multirow{2}{*}{\textbf{Datasets}} &
    \multirow{2}{*}{\textbf{Slice+Dice (s)}} &
    \multirow{2}{*}{\textbf{Slice+Roulette (s)}} &
    \multirow{2}{*}{\textbf{SPPL (s)}} &
    \multirow{2}{*}{\textbf{PSI (s)}} &
    \multicolumn{3}{c}{\textbf{Hakaru (s)}} \\
    \cmidrule(lr){7-9}
    & & & & & & \textbf{Disintegrate} & \textbf{Simplify} & \textbf{Total} \\
    \midrule
    ClickGraph~\cite{gehr2016psi} & 10 & \textbf{1.469} & \texttt{t.o.} & \texttt{t.o.} &
    \makecell{\texttt{t.o.}\,$\bullet$ \\ \texttt{t.o.}\,$\circ$} &
    \makecell{0.046\,$\bullet$ \\ 0.035\,$\circ$} &
    \makecell{\texttt{t.o.}\,$\bullet$ \\ 11.287\,$\circ$} &
    \makecell{\texttt{t.o.}\,$\bullet$ \\ 11.322\,$\circ$} \\
    \midrule
    ClinicalTrial~\cite{gehr2016psi} & 20 & \textbf{1.582} & 5.020 & 17.895 &
    \makecell{\texttt{t.o.}\,$\bullet$ \\ 1.916\,$\circ$} &
    \makecell{\texttt{t.o.}\,$\bullet$ \\ \texttt{t.o.}\,$\circ$} &
    \makecell{\texttt{t.o.}\,$\bullet$ \\ \texttt{t.o.}\,$\circ$} &
    \makecell{\texttt{t.o.}\,$\bullet$ \\ \texttt{t.o.}\,$\circ$} \\
    \midrule
    CoinBias~\cite{gehr2016psi} & 5 & \textbf{0.466} & 0.590 & 2.382 &
    \makecell{20.138\,$\bullet$ \\ 0.369\,$\circ$} &
    \makecell{0.058\,$\bullet$ \\ 0.025\,$\circ$} &
    \makecell{4.789\,$\bullet$ \\ 0.528\,$\circ$} &
    \makecell{4.847\,$\bullet$ \\ 0.553\,$\circ$} \\
    \midrule
    Election~\cite{Saad2021SPPL} & -- & \textbf{1.587} & \texttt{t.o.} & 19.410 &
    \texttt{t.o.} & 0.019 & \texttt{t.o.} & \texttt{t.o.} \\
    \midrule
    MarkovSwitching~\cite{Saad2021SPPL} & -- & 0.423 & \textbf{0.240} & 19.998 &
    \texttt{t.o.} & 0.020 & 0.275 & 0.295 \\
    \midrule
    StudentInterviews~\cite{Saad2021SPPL} & -- & 3.331 & \texttt{t.o.} & 53.209 &
    \texttt{t.o.} & 0.025 & 0.443 & \textbf{0.468} \\
    \midrule
    SurveyUnbias~\cite{gehr2016psi} & 5 & 75.221 & \textbf{0.651} & 3.519 &
    \makecell{\texttt{t.o.}\,$\bullet$ \\ \texttt{t.o.}\,$\circ$} &
    \makecell{0.027\,$\bullet$ \\ 0.029\,$\circ$} &
    \makecell{\texttt{t.o.}\,$\bullet$ \\ \texttt{err.}\,$\circ$} &
    \makecell{\texttt{t.o.}\,$\bullet$ \\ \texttt{err.}\,$\circ$} \\
    \midrule
    TrueSkill~\cite{gehr2016psi} & 3 & \textbf{0.750} & 68.669 & 32.881 &
    \makecell{\texttt{t.o.}\,$\bullet$ \\ 0.068\,$\circ$} &
    \makecell{0.027\,$\bullet$ \\ 0.018\,$\circ$} &
    \makecell{\texttt{err.}\,$\bullet$ \\ 0.215\,$\circ$} &
    \makecell{\texttt{err.}\,$\bullet$ \\ 0.233\,$\circ$} \\
    \bottomrule
    \end{tabular}  
}  
\end{table}

\paragraph{Results.} Across the modified baseline programs in \Cref{tab:benchmarks_psi}, the \Slice{} pipelines generally provide the fastest successful end-to-end inference: \Slice{}+\Dice{} performs best on most benchmarks, while \Slice{}+\Roulette{} is faster on some instances. \SPPL{} is generally slower, and \PSI{} frequently times out, although both \PSI{} and Hakaru perform better on some unmodified programs. Hakaru's disintegration phase is usually fast, but subsequent Maple simplification often times out or returns an error, limiting its end-to-end reliability.

\clearpage
\subsubsection{Baseline Programs}\label{app:baseline-programs-full}
The exact \Slice{} programs used for these baseline benchmarks are reproduced below.

\lstdefinestyle{baselineprogram}{
  language=Dice,
  basicstyle=\ttfamily\scriptsize,
  columns=fullflexible,
  keepspaces=true,
  breaklines=true,
  breakatwhitespace=false,
  showstringspaces=false,
  numbers=left,
  numberstyle=\tiny\color{gray},
  numbersep=6pt,
  rulecolor=\color{black!25},
  xleftmargin=1.8em,
  aboveskip=0.5em,
  belowskip=1em
}

\begin{lstlisting}[style=baselineprogram,title={ClickGraph}]
let p_similar = beta(1, 1) in

let similar_0 =
    if p_similar <= 0.05 then discrete(0.95, 0.05)
    else if p_similar <= 0.1 then discrete(0.9, 0.1)
    else if p_similar <= 0.15 then discrete(0.85, 0.15)
    else if p_similar <= 0.2 then discrete(0.8, 0.2)
    else if p_similar <= 0.25 then discrete(0.75, 0.25)
    else if p_similar <= 0.3 then discrete(0.7, 0.3)
    else if p_similar <= 0.35 then discrete(0.65, 0.35)
    else if p_similar <= 0.4 then discrete(0.6, 0.4)
    else if p_similar <= 0.45 then discrete(0.55, 0.45)
    else if p_similar <= 0.5 then discrete(0.5, 0.5)
    else if p_similar <= 0.55 then discrete(0.45, 0.55)
    else if p_similar <= 0.6 then discrete(0.4, 0.6)
    else if p_similar <= 0.65 then discrete(0.35, 0.65)
    else if p_similar <= 0.7 then discrete(0.3, 0.7)
    else if p_similar <= 0.75 then discrete(0.25, 0.75)
    else if p_similar <= 0.8 then discrete(0.2, 0.8)
    else if p_similar <= 0.85 then discrete(0.15, 0.85)
    else if p_similar <= 0.9 then discrete(0.1, 0.9)
    else if p_similar <= 0.95 then discrete(0.05, 0.95)
    else discrete(0.0, 1.0)
in

let p_click_url_A_0 = 
    if similar_0 ==#2 1#2 then uniform(0,1) else uniform(0, 1) 
in

let p_click_url_B_0 =
    if similar_0 ==#2 1#2 then p_click_url_A_0 else uniform(0, 1)
in

let click_url_A_0 =
    if p_click_url_A_0 <= 0.1 then discrete(0.9, 0.1)
    else if p_click_url_A_0 <= 0.2 then discrete(0.8, 0.2)
    else if p_click_url_A_0 <= 0.3 then discrete(0.7, 0.3)
    else if p_click_url_A_0 <= 0.4 then discrete(0.6, 0.4)
    else if p_click_url_A_0 <= 0.5 then discrete(0.5, 0.5)
    else if p_click_url_A_0 <= 0.6 then discrete(0.4, 0.6)
    else if p_click_url_A_0 <= 0.7 then discrete(0.3, 0.7)
    else if p_click_url_A_0 <= 0.8 then discrete(0.2, 0.8)
    else if p_click_url_A_0 <= 0.9 then discrete(0.1, 0.9)
    else discrete(0.0, 1.0)
in

let click_url_B_0 =
    if p_click_url_B_0 <= 0.1 then discrete(0.9, 0.1)
    else if p_click_url_B_0 <= 0.2 then discrete(0.8, 0.2)
    else if p_click_url_B_0 <= 0.3 then discrete(0.7, 0.3)
    else if p_click_url_B_0 <= 0.4 then discrete(0.6, 0.4)
    else if p_click_url_B_0 <= 0.5 then discrete(0.5, 0.5)
    else if p_click_url_B_0 <= 0.6 then discrete(0.4, 0.6)
    else if p_click_url_B_0 <= 0.7 then discrete(0.3, 0.7)
    else if p_click_url_B_0 <= 0.8 then discrete(0.2, 0.8)
    else if p_click_url_B_0 <= 0.9 then discrete(0.1, 0.9)
    else discrete(0.0, 1.0)
in

let similar_1 =
    if p_similar <= 0.05 then discrete(0.95, 0.05)
    else if p_similar <= 0.1 then discrete(0.9, 0.1)
    else if p_similar <= 0.15 then discrete(0.85, 0.15)
    else if p_similar <= 0.2 then discrete(0.8, 0.2)
    else if p_similar <= 0.25 then discrete(0.75, 0.25)
    else if p_similar <= 0.3 then discrete(0.7, 0.3)
    else if p_similar <= 0.35 then discrete(0.65, 0.35)
    else if p_similar <= 0.4 then discrete(0.6, 0.4)
    else if p_similar <= 0.45 then discrete(0.55, 0.45)
    else if p_similar <= 0.5 then discrete(0.5, 0.5)
    else if p_similar <= 0.55 then discrete(0.45, 0.55)
    else if p_similar <= 0.6 then discrete(0.4, 0.6)
    else if p_similar <= 0.65 then discrete(0.35, 0.65)
    else if p_similar <= 0.7 then discrete(0.3, 0.7)
    else if p_similar <= 0.75 then discrete(0.25, 0.75)
    else if p_similar <= 0.8 then discrete(0.2, 0.8)
    else if p_similar <= 0.85 then discrete(0.15, 0.85)
    else if p_similar <= 0.9 then discrete(0.1, 0.9)
    else if p_similar <= 0.95 then discrete(0.05, 0.95)
    else discrete(0.0, 1.0)
in

let p_click_url_A_1 = 
    if similar_1 ==#2 1#2 then uniform(0,1) else uniform(0, 1)
in

let p_click_url_B_1 =
    if similar_1 ==#2 1#2 then p_click_url_A_1 else uniform(0, 1)
in

let click_url_A_1 =
    if p_click_url_A_1 <= 0.1 then discrete(0.9, 0.1)
    else if p_click_url_A_1 <= 0.2 then discrete(0.8, 0.2)
    else if p_click_url_A_1 <= 0.3 then discrete(0.7, 0.3)
    else if p_click_url_A_1 <= 0.4 then discrete(0.6, 0.4)
    else if p_click_url_A_1 <= 0.5 then discrete(0.5, 0.5)
    else if p_click_url_A_1 <= 0.6 then discrete(0.4, 0.6)
    else if p_click_url_A_1 <= 0.7 then discrete(0.3, 0.7)
    else if p_click_url_A_1 <= 0.8 then discrete(0.2, 0.8)
    else if p_click_url_A_1 <= 0.9 then discrete(0.1, 0.9)
    else discrete(0.0, 1.0)
in

let click_url_B_1 =
    if p_click_url_B_1 <= 0.1 then discrete(0.9, 0.1)
    else if p_click_url_B_1 <= 0.2 then discrete(0.8, 0.2)
    else if p_click_url_B_1 <= 0.3 then discrete(0.7, 0.3)
    else if p_click_url_B_1 <= 0.4 then discrete(0.6, 0.4)
    else if p_click_url_B_1 <= 0.5 then discrete(0.5, 0.5)
    else if p_click_url_B_1 <= 0.6 then discrete(0.4, 0.6)
    else if p_click_url_B_1 <= 0.7 then discrete(0.3, 0.7)
    else if p_click_url_B_1 <= 0.8 then discrete(0.2, 0.8)
    else if p_click_url_B_1 <= 0.9 then discrete(0.1, 0.9)
    else discrete(0.0, 1.0)
in

let similar_2 =
    if p_similar <= 0.05 then discrete(0.95, 0.05)
    else if p_similar <= 0.1 then discrete(0.9, 0.1)
    else if p_similar <= 0.15 then discrete(0.85, 0.15)
    else if p_similar <= 0.2 then discrete(0.8, 0.2)
    else if p_similar <= 0.25 then discrete(0.75, 0.25)
    else if p_similar <= 0.3 then discrete(0.7, 0.3)
    else if p_similar <= 0.35 then discrete(0.65, 0.35)
    else if p_similar <= 0.4 then discrete(0.6, 0.4)
    else if p_similar <= 0.45 then discrete(0.55, 0.45)
    else if p_similar <= 0.5 then discrete(0.5, 0.5)
    else if p_similar <= 0.55 then discrete(0.45, 0.55)
    else if p_similar <= 0.6 then discrete(0.4, 0.6)
    else if p_similar <= 0.65 then discrete(0.35, 0.65)
    else if p_similar <= 0.7 then discrete(0.3, 0.7)
    else if p_similar <= 0.75 then discrete(0.25, 0.75)
    else if p_similar <= 0.8 then discrete(0.2, 0.8)
    else if p_similar <= 0.85 then discrete(0.15, 0.85)
    else if p_similar <= 0.9 then discrete(0.1, 0.9)
    else if p_similar <= 0.95 then discrete(0.05, 0.95)
    else discrete(0.0, 1.0)
in

let p_click_url_A_2 = 
    if similar_2 ==#2 1#2 then uniform(0,1) else uniform(0, 1)
in

let p_click_url_B_2 =
    if similar_2 ==#2 1#2 then p_click_url_A_2 else uniform(0, 1)
in

let click_url_A_2 =
    if p_click_url_A_2 <= 0.1 then discrete(0.9, 0.1)
    else if p_click_url_A_2 <= 0.2 then discrete(0.8, 0.2)
    else if p_click_url_A_2 <= 0.3 then discrete(0.7, 0.3)
    else if p_click_url_A_2 <= 0.4 then discrete(0.6, 0.4)
    else if p_click_url_A_2 <= 0.5 then discrete(0.5, 0.5)
    else if p_click_url_A_2 <= 0.6 then discrete(0.4, 0.6)
    else if p_click_url_A_2 <= 0.7 then discrete(0.3, 0.7)
    else if p_click_url_A_2 <= 0.8 then discrete(0.2, 0.8)
    else if p_click_url_A_2 <= 0.9 then discrete(0.1, 0.9)
    else discrete(0.0, 1.0)
in

let click_url_B_2 =
    if p_click_url_B_2 <= 0.1 then discrete(0.9, 0.1)
    else if p_click_url_B_2 <= 0.2 then discrete(0.8, 0.2)
    else if p_click_url_B_2 <= 0.3 then discrete(0.7, 0.3)
    else if p_click_url_B_2 <= 0.4 then discrete(0.6, 0.4)
    else if p_click_url_B_2 <= 0.5 then discrete(0.5, 0.5)
    else if p_click_url_B_2 <= 0.6 then discrete(0.4, 0.6)
    else if p_click_url_B_2 <= 0.7 then discrete(0.3, 0.7)
    else if p_click_url_B_2 <= 0.8 then discrete(0.2, 0.8)
    else if p_click_url_B_2 <= 0.9 then discrete(0.1, 0.9)
    else discrete(0.0, 1.0)
in

let similar_3 =
    if p_similar <= 0.05 then discrete(0.95, 0.05)
    else if p_similar <= 0.1 then discrete(0.9, 0.1)
    else if p_similar <= 0.15 then discrete(0.85, 0.15)
    else if p_similar <= 0.2 then discrete(0.8, 0.2)
    else if p_similar <= 0.25 then discrete(0.75, 0.25)
    else if p_similar <= 0.3 then discrete(0.7, 0.3)
    else if p_similar <= 0.35 then discrete(0.65, 0.35)
    else if p_similar <= 0.4 then discrete(0.6, 0.4)
    else if p_similar <= 0.45 then discrete(0.55, 0.45)
    else if p_similar <= 0.5 then discrete(0.5, 0.5)
    else if p_similar <= 0.55 then discrete(0.45, 0.55)
    else if p_similar <= 0.6 then discrete(0.4, 0.6)
    else if p_similar <= 0.65 then discrete(0.35, 0.65)
    else if p_similar <= 0.7 then discrete(0.3, 0.7)
    else if p_similar <= 0.75 then discrete(0.25, 0.75)
    else if p_similar <= 0.8 then discrete(0.2, 0.8)
    else if p_similar <= 0.85 then discrete(0.15, 0.85)
    else if p_similar <= 0.9 then discrete(0.1, 0.9)
    else if p_similar <= 0.95 then discrete(0.05, 0.95)
    else discrete(0.0, 1.0)
in

let p_click_url_A_3 = 
    if similar_3 ==#2 1#2 then uniform(0,1) else uniform(0, 1)
in

let p_click_url_B_3 =
    if similar_3 ==#2 1#2 then p_click_url_A_3 else uniform(0, 1)
in

let click_url_A_3 =
    if p_click_url_A_3 <= 0.1 then discrete(0.9, 0.1)
    else if p_click_url_A_3 <= 0.2 then discrete(0.8, 0.2)
    else if p_click_url_A_3 <= 0.3 then discrete(0.7, 0.3)
    else if p_click_url_A_3 <= 0.4 then discrete(0.6, 0.4)
    else if p_click_url_A_3 <= 0.5 then discrete(0.5, 0.5)
    else if p_click_url_A_3 <= 0.6 then discrete(0.4, 0.6)
    else if p_click_url_A_3 <= 0.7 then discrete(0.3, 0.7)
    else if p_click_url_A_3 <= 0.8 then discrete(0.2, 0.8)
    else if p_click_url_A_3 <= 0.9 then discrete(0.1, 0.9)
    else discrete(0.0, 1.0)
in

let click_url_B_3 =
    if p_click_url_B_3 <= 0.1 then discrete(0.9, 0.1)
    else if p_click_url_B_3 <= 0.2 then discrete(0.8, 0.2)
    else if p_click_url_B_3 <= 0.3 then discrete(0.7, 0.3)
    else if p_click_url_B_3 <= 0.4 then discrete(0.6, 0.4)
    else if p_click_url_B_3 <= 0.5 then discrete(0.5, 0.5)
    else if p_click_url_B_3 <= 0.6 then discrete(0.4, 0.6)
    else if p_click_url_B_3 <= 0.7 then discrete(0.3, 0.7)
    else if p_click_url_B_3 <= 0.8 then discrete(0.2, 0.8)
    else if p_click_url_B_3 <= 0.9 then discrete(0.1, 0.9)
    else discrete(0.0, 1.0)
in

let similar_4 =
    if p_similar <= 0.05 then discrete(0.95, 0.05)
    else if p_similar <= 0.1 then discrete(0.9, 0.1)
    else if p_similar <= 0.15 then discrete(0.85, 0.15)
    else if p_similar <= 0.2 then discrete(0.8, 0.2)
    else if p_similar <= 0.25 then discrete(0.75, 0.25)
    else if p_similar <= 0.3 then discrete(0.7, 0.3)
    else if p_similar <= 0.35 then discrete(0.65, 0.35)
    else if p_similar <= 0.4 then discrete(0.6, 0.4)
    else if p_similar <= 0.45 then discrete(0.55, 0.45)
    else if p_similar <= 0.5 then discrete(0.5, 0.5)
    else if p_similar <= 0.55 then discrete(0.45, 0.55)
    else if p_similar <= 0.6 then discrete(0.4, 0.6)
    else if p_similar <= 0.65 then discrete(0.35, 0.65)
    else if p_similar <= 0.7 then discrete(0.3, 0.7)
    else if p_similar <= 0.75 then discrete(0.25, 0.75)
    else if p_similar <= 0.8 then discrete(0.2, 0.8)
    else if p_similar <= 0.85 then discrete(0.15, 0.85)
    else if p_similar <= 0.9 then discrete(0.1, 0.9)
    else if p_similar <= 0.95 then discrete(0.05, 0.95)
    else discrete(0.0, 1.0)
in

let p_click_url_A_4 = 
    if similar_4 ==#2 1#2 then uniform(0,1) else uniform(0, 1)
in

let p_click_url_B_4 =
    if similar_4 ==#2 1#2 then p_click_url_A_4 else uniform(0, 1)
in

let click_url_A_4 =
    if p_click_url_A_4 <= 0.1 then discrete(0.9, 0.1)
    else if p_click_url_A_4 <= 0.2 then discrete(0.8, 0.2)
    else if p_click_url_A_4 <= 0.3 then discrete(0.7, 0.3)
    else if p_click_url_A_4 <= 0.4 then discrete(0.6, 0.4)
    else if p_click_url_A_4 <= 0.5 then discrete(0.5, 0.5)
    else if p_click_url_A_4 <= 0.6 then discrete(0.4, 0.6)
    else if p_click_url_A_4 <= 0.7 then discrete(0.3, 0.7)
    else if p_click_url_A_4 <= 0.8 then discrete(0.2, 0.8)
    else if p_click_url_A_4 <= 0.9 then discrete(0.1, 0.9)
    else discrete(0.0, 1.0)
in

let click_url_B_4 =
    if p_click_url_B_4 <= 0.1 then discrete(0.9, 0.1)
    else if p_click_url_B_4 <= 0.2 then discrete(0.8, 0.2)
    else if p_click_url_B_4 <= 0.3 then discrete(0.7, 0.3)
    else if p_click_url_B_4 <= 0.4 then discrete(0.6, 0.4)
    else if p_click_url_B_4 <= 0.5 then discrete(0.5, 0.5)
    else if p_click_url_B_4 <= 0.6 then discrete(0.4, 0.6)
    else if p_click_url_B_4 <= 0.7 then discrete(0.3, 0.7)
    else if p_click_url_B_4 <= 0.8 then discrete(0.2, 0.8)
    else if p_click_url_B_4 <= 0.9 then discrete(0.1, 0.9)
    else discrete(0.0, 1.0)
in

let _0 = observe(click_url_A_0 ==#2 1#2) in
let _1 = observe(click_url_B_0 ==#2 1#2) in
let _2 = observe(click_url_A_1 ==#2 1#2) in
let _3 = observe(click_url_B_1 ==#2 1#2) in
let _4 = observe(click_url_A_2 ==#2 0#2) in
let _5 = observe(click_url_B_2 ==#2 0#2) in
let _6 = observe(click_url_A_3 ==#2 1#2) in
let _7 = observe(click_url_B_3 ==#2 1#2) in
let _8 = observe(click_url_A_4 ==#2 0#2) in
let _9 = observe(click_url_B_4 ==#2 0#2) in

p_similar < 0.3
\end{lstlisting}

\lstinputlisting[style=baselineprogram,title={ClinicalTrial}]{baseline_programs/clinical_trial.slice}

\begin{lstlisting}[style=baselineprogram,title={CoinBias}]
let bias = beta(2, 5) in

let tossResults_0 =
    if bias <= 0.05 then discrete(0.95, 0.05)
    else if bias <= 0.1 then discrete(0.9, 0.1)
    else if bias <= 0.15 then discrete(0.85, 0.15)
    else if bias <= 0.2 then discrete(0.8, 0.2)
    else if bias <= 0.25 then discrete(0.75, 0.25)
    else if bias <= 0.3 then discrete(0.7, 0.3)
    else if bias <= 0.35 then discrete(0.65, 0.35)
    else if bias <= 0.4 then discrete(0.6, 0.4)
    else if bias <= 0.45 then discrete(0.55, 0.45)
    else if bias <= 0.5 then discrete(0.5, 0.5)
    else if bias <= 0.55 then discrete(0.45, 0.55)
    else if bias <= 0.6 then discrete(0.4, 0.6)
    else if bias <= 0.65 then discrete(0.35, 0.65)
    else if bias <= 0.7 then discrete(0.3, 0.7)
    else if bias <= 0.75 then discrete(0.25, 0.75)
    else if bias <= 0.8 then discrete(0.2, 0.8)
    else if bias <= 0.85 then discrete(0.15, 0.85)
    else if bias <= 0.9 then discrete(0.1, 0.9)
    else if bias <= 0.95 then discrete(0.05, 0.95)
    else discrete(0.0, 1.0)
in

let tossResults_1 =
    if bias <= 0.05 then discrete(0.95, 0.05)
    else if bias <= 0.1 then discrete(0.9, 0.1)
    else if bias <= 0.15 then discrete(0.85, 0.15)
    else if bias <= 0.2 then discrete(0.8, 0.2)
    else if bias <= 0.25 then discrete(0.75, 0.25)
    else if bias <= 0.3 then discrete(0.7, 0.3)
    else if bias <= 0.35 then discrete(0.65, 0.35)
    else if bias <= 0.4 then discrete(0.6, 0.4)
    else if bias <= 0.45 then discrete(0.55, 0.45)
    else if bias <= 0.5 then discrete(0.5, 0.5)
    else if bias <= 0.55 then discrete(0.45, 0.55)
    else if bias <= 0.6 then discrete(0.4, 0.6)
    else if bias <= 0.65 then discrete(0.35, 0.65)
    else if bias <= 0.7 then discrete(0.3, 0.7)
    else if bias <= 0.75 then discrete(0.25, 0.75)
    else if bias <= 0.8 then discrete(0.2, 0.8)
    else if bias <= 0.85 then discrete(0.15, 0.85)
    else if bias <= 0.9 then discrete(0.1, 0.9)
    else if bias <= 0.95 then discrete(0.05, 0.95)
    else discrete(0.0, 1.0)
in

let tossResults_2 =
    if bias <= 0.05 then discrete(0.95, 0.05)
    else if bias <= 0.1 then discrete(0.9, 0.1)
    else if bias <= 0.15 then discrete(0.85, 0.15)
    else if bias <= 0.2 then discrete(0.8, 0.2)
    else if bias <= 0.25 then discrete(0.75, 0.25)
    else if bias <= 0.3 then discrete(0.7, 0.3)
    else if bias <= 0.35 then discrete(0.65, 0.35)
    else if bias <= 0.4 then discrete(0.6, 0.4)
    else if bias <= 0.45 then discrete(0.55, 0.45)
    else if bias <= 0.5 then discrete(0.5, 0.5)
    else if bias <= 0.55 then discrete(0.45, 0.55)
    else if bias <= 0.6 then discrete(0.4, 0.6)
    else if bias <= 0.65 then discrete(0.35, 0.65)
    else if bias <= 0.7 then discrete(0.3, 0.7)
    else if bias <= 0.75 then discrete(0.25, 0.75)
    else if bias <= 0.8 then discrete(0.2, 0.8)
    else if bias <= 0.85 then discrete(0.15, 0.85)
    else if bias <= 0.9 then discrete(0.1, 0.9)
    else if bias <= 0.95 then discrete(0.05, 0.95)
    else discrete(0.0, 1.0)
in

let tossResults_3 =
    if bias <= 0.05 then discrete(0.95, 0.05)
    else if bias <= 0.1 then discrete(0.9, 0.1)
    else if bias <= 0.15 then discrete(0.85, 0.15)
    else if bias <= 0.2 then discrete(0.8, 0.2)
    else if bias <= 0.25 then discrete(0.75, 0.25)
    else if bias <= 0.3 then discrete(0.7, 0.3)
    else if bias <= 0.35 then discrete(0.65, 0.35)
    else if bias <= 0.4 then discrete(0.6, 0.4)
    else if bias <= 0.45 then discrete(0.55, 0.45)
    else if bias <= 0.5 then discrete(0.5, 0.5)
    else if bias <= 0.55 then discrete(0.45, 0.55)
    else if bias <= 0.6 then discrete(0.4, 0.6)
    else if bias <= 0.65 then discrete(0.35, 0.65)
    else if bias <= 0.7 then discrete(0.3, 0.7)
    else if bias <= 0.75 then discrete(0.25, 0.75)
    else if bias <= 0.8 then discrete(0.2, 0.8)
    else if bias <= 0.85 then discrete(0.15, 0.85)
    else if bias <= 0.9 then discrete(0.1, 0.9)
    else if bias <= 0.95 then discrete(0.05, 0.95)
    else discrete(0.0, 1.0)
in

let tossResults_4 =
    if bias <= 0.05 then discrete(0.95, 0.05)
    else if bias <= 0.1 then discrete(0.9, 0.1)
    else if bias <= 0.15 then discrete(0.85, 0.15)
    else if bias <= 0.2 then discrete(0.8, 0.2)
    else if bias <= 0.25 then discrete(0.75, 0.25)
    else if bias <= 0.3 then discrete(0.7, 0.3)
    else if bias <= 0.35 then discrete(0.65, 0.35)
    else if bias <= 0.4 then discrete(0.6, 0.4)
    else if bias <= 0.45 then discrete(0.55, 0.45)
    else if bias <= 0.5 then discrete(0.5, 0.5)
    else if bias <= 0.55 then discrete(0.45, 0.55)
    else if bias <= 0.6 then discrete(0.4, 0.6)
    else if bias <= 0.65 then discrete(0.35, 0.65)
    else if bias <= 0.7 then discrete(0.3, 0.7)
    else if bias <= 0.75 then discrete(0.25, 0.75)
    else if bias <= 0.8 then discrete(0.2, 0.8)
    else if bias <= 0.85 then discrete(0.15, 0.85)
    else if bias <= 0.9 then discrete(0.1, 0.9)
    else if bias <= 0.95 then discrete(0.05, 0.95)
    else discrete(0.0, 1.0)
in

let _0 = observe(tossResults_0 ==#2 1#2) in
let _1 = observe(tossResults_1 ==#2 1#2) in
let _2 = observe(tossResults_2 ==#2 0#2) in
let _3 = observe(tossResults_3 ==#2 1#2) in
let _4 = observe(tossResults_4 ==#2 0#2) in

bias < 0.3
\end{lstlisting}

\begin{lstlisting}[style=baselineprogram,title={Election}]
let n = 4000 in
let param = discrete(
  0.01: 250,
  0.01: 251,
  0.01: 252,
  0.01: 253,
  0.01: 254,
  0.01: 255,
  0.01: 256,
  0.01: 257,
  0.01: 258,
  0.01: 259,
  0.01: 260,
  0.01: 261,
  0.01: 262,
  0.01: 263,
  0.01: 264,
  0.01: 265,
  0.01: 266,
  0.01: 267,
  0.01: 268,
  0.01: 269,
  0.01: 270,
  0.01: 271,
  0.01: 272,
  0.01: 273,
  0.01: 274,
  0.01: 275,
  0.01: 276,
  0.01: 277,
  0.01: 278,
  0.01: 279,
  0.01: 280,
  0.01: 281,
  0.01: 282,
  0.01: 283,
  0.01: 284,
  0.01: 285,
  0.01: 286,
  0.01: 287,
  0.01: 288,
  0.01: 289,
  0.01: 290,
  0.01: 291,
  0.01: 292,
  0.01: 293,
  0.01: 294,
  0.01: 295,
  0.01: 296,
  0.01: 297,
  0.01: 298,
  0.01: 299,
  0.01: 300,
  0.01: 301,
  0.01: 302,
  0.01: 303,
  0.01: 304,
  0.01: 305,
  0.01: 306,
  0.01: 307,
  0.01: 308,
  0.01: 309,
  0.01: 310,
  0.01: 311,
  0.01: 312,
  0.01: 313,
  0.01: 314,
  0.01: 315,
  0.01: 316,
  0.01: 317,
  0.01: 318,
  0.01: 319,
  0.01: 320,
  0.01: 321,
  0.01: 322,
  0.01: 323,
  0.01: 324,
  0.01: 325,
  0.01: 326,
  0.01: 327,
  0.01: 328,
  0.01: 329,
  0.01: 330,
  0.01: 331,
  0.01: 332,
  0.01: 333,
  0.01: 334,
  0.01: 335,
  0.01: 336,
  0.01: 337,
  0.01: 338,
  0.01: 339,
  0.01: 340,
  0.01: 341,
  0.01: 342,
  0.01: 343,
  0.01: 344,
  0.01: 345,
  0.01: 346,
  0.01: 347,
  0.01: 348,
  0.01: 349
) in

let p =
  if param <= 250 then beta(277, 250)
  else if param <= 251 then beta(277, 251)
  else if param <= 252 then beta(277, 252)
  else if param <= 253 then beta(277, 253)
  else if param <= 254 then beta(277, 254)
  else if param <= 255 then beta(277, 255)
  else if param <= 256 then beta(277, 256)
  else if param <= 257 then beta(277, 257)
  else if param <= 258 then beta(277, 258)
  else if param <= 259 then beta(277, 259)
  else if param <= 260 then beta(277, 260)
  else if param <= 261 then beta(277, 261)
  else if param <= 262 then beta(277, 262)
  else if param <= 263 then beta(277, 263)
  else if param <= 264 then beta(277, 264)
  else if param <= 265 then beta(277, 265)
  else if param <= 266 then beta(277, 266)
  else if param <= 267 then beta(277, 267)
  else if param <= 268 then beta(277, 268)
  else if param <= 269 then beta(277, 269)
  else if param <= 270 then beta(277, 270)
  else if param <= 271 then beta(277, 271)
  else if param <= 272 then beta(277, 272)
  else if param <= 273 then beta(277, 273)
  else if param <= 274 then beta(277, 274)
  else if param <= 275 then beta(277, 275)
  else if param <= 276 then beta(277, 276)
  else if param <= 277 then beta(277, 277)
  else if param <= 278 then beta(277, 278)
  else if param <= 279 then beta(277, 279)
  else if param <= 280 then beta(277, 280)
  else if param <= 281 then beta(277, 281)
  else if param <= 282 then beta(277, 282)
  else if param <= 283 then beta(277, 283)
  else if param <= 284 then beta(277, 284)
  else if param <= 285 then beta(277, 285)
  else if param <= 286 then beta(277, 286)
  else if param <= 287 then beta(277, 287)
  else if param <= 288 then beta(277, 288)
  else if param <= 289 then beta(277, 289)
  else if param <= 290 then beta(277, 290)
  else if param <= 291 then beta(277, 291)
  else if param <= 292 then beta(277, 292)
  else if param <= 293 then beta(277, 293)
  else if param <= 294 then beta(277, 294)
  else if param <= 295 then beta(277, 295)
  else if param <= 296 then beta(277, 296)
  else if param <= 297 then beta(277, 297)
  else if param <= 298 then beta(277, 298)
  else if param <= 299 then beta(277, 299)
  else if param <= 300 then beta(277, 300)
  else if param <= 301 then beta(277, 301)
  else if param <= 302 then beta(277, 302)
  else if param <= 303 then beta(277, 303)
  else if param <= 304 then beta(277, 304)
  else if param <= 305 then beta(277, 305)
  else if param <= 306 then beta(277, 306)
  else if param <= 307 then beta(277, 307)
  else if param <= 308 then beta(277, 308)
  else if param <= 309 then beta(277, 309)
  else if param <= 310 then beta(277, 310)
  else if param <= 311 then beta(277, 311)
  else if param <= 312 then beta(277, 312)
  else if param <= 313 then beta(277, 313)
  else if param <= 314 then beta(277, 314)
  else if param <= 315 then beta(277, 315)
  else if param <= 316 then beta(277, 316)
  else if param <= 317 then beta(277, 317)
  else if param <= 318 then beta(277, 318)
  else if param <= 319 then beta(277, 319)
  else if param <= 320 then beta(277, 320)
  else if param <= 321 then beta(277, 321)
  else if param <= 322 then beta(277, 322)
  else if param <= 323 then beta(277, 323)
  else if param <= 324 then beta(277, 324)
  else if param <= 325 then beta(277, 325)
  else if param <= 326 then beta(277, 326)
  else if param <= 327 then beta(277, 327)
  else if param <= 328 then beta(277, 328)
  else if param <= 329 then beta(277, 329)
  else if param <= 330 then beta(277, 330)
  else if param <= 331 then beta(277, 331)
  else if param <= 332 then beta(277, 332)
  else if param <= 333 then beta(277, 333)
  else if param <= 334 then beta(277, 334)
  else if param <= 335 then beta(277, 335)
  else if param <= 336 then beta(277, 336)
  else if param <= 337 then beta(277, 337)
  else if param <= 338 then beta(277, 338)
  else if param <= 339 then beta(277, 339)
  else if param <= 340 then beta(277, 340)
  else if param <= 341 then beta(277, 341)
  else if param <= 342 then beta(277, 342)
  else if param <= 343 then beta(277, 343)
  else if param <= 344 then beta(277, 344)
  else if param <= 345 then beta(277, 345)
  else if param <= 346 then beta(277, 346)
  else if param <= 347 then beta(277, 347)
  else if param <= 348 then beta(277, 348)
  else beta(277, 349)
in

let votes =
  if p <= 0.05 then binomial(0.025, n)
  else if p <= 0.1 then binomial(0.075, n)
  else if p <= 0.15 then binomial(0.125, n)
  else if p <= 0.2 then binomial(0.175, n)
  else if p <= 0.25 then binomial(0.225, n)
  else if p <= 0.3 then binomial(0.275, n)
  else if p <= 0.35 then binomial(0.325, n)
  else if p <= 0.4 then binomial(0.375, n)
  else if p <= 0.45 then binomial(0.425, n)
  else if p <= 0.5 then binomial(0.475, n)
  else if p <= 0.55 then binomial(0.525, n)
  else if p <= 0.6 then binomial(0.575, n)
  else if p <= 0.65 then binomial(0.625, n)
  else if p <= 0.7 then binomial(0.675, n)
  else if p <= 0.75 then binomial(0.725, n)
  else if p <= 0.8 then binomial(0.775, n)
  else if p <= 0.85 then binomial(0.825, n)
  else if p <= 0.9 then binomial(0.875, n)
  else if p <= 0.95 then binomial(0.925, n)
  else binomial(0.975, n)
in

let win = votes > 2000 in
let _win = observe(win) in
param <= 250
\end{lstlisting}

\lstinputlisting[style=baselineprogram,title={MarkovSwitching (HMM)}]{baseline_programs/hmm.slice}

\begin{lstlisting}[style=baselineprogram,title={StudentInterviews}]
let student_perfect_0 = discrete(0.8, 0.2) in
let student_gpa_0 = if student_perfect_0 ==#2 0#2 then beta(7, 3) else 1.0 in

let student_perfect_1 = discrete(0.8, 0.2) in
let student_gpa_1 = if student_perfect_1 ==#2 0#2 then beta(7, 3) else 1.0 in

let num_recruiters = poisson(25) in
let _recruiters_range = observe(10 <= num_recruiters && num_recruiters <= 40) in

let student_interviews_0 =
    if num_recruiters <= 10 then
        if student_perfect_0 ==#2 1#2 then binomial(0.9, 10)
        else if student_gpa_0 > 0.875 then binomial(0.6, 10)
        else binomial(0.5, 10)
    else if num_recruiters <= 11 then
        if student_perfect_0 ==#2 1#2 then binomial(0.9, 11)
        else if student_gpa_0 > 0.875 then binomial(0.6, 11)
        else binomial(0.5, 11)
    else if num_recruiters <= 12 then
        if student_perfect_0 ==#2 1#2 then binomial(0.9, 12)
        else if student_gpa_0 > 0.875 then binomial(0.6, 12)
        else binomial(0.5, 12)
    else if num_recruiters <= 13 then
        if student_perfect_0 ==#2 1#2 then binomial(0.9, 13)
        else if student_gpa_0 > 0.875 then binomial(0.6, 13)
        else binomial(0.5, 13)
    else if num_recruiters <= 14 then
        if student_perfect_0 ==#2 1#2 then binomial(0.9, 14)
        else if student_gpa_0 > 0.875 then binomial(0.6, 14)
        else binomial(0.5, 14)
    else if num_recruiters <= 15 then
        if student_perfect_0 ==#2 1#2 then binomial(0.9, 15)
        else if student_gpa_0 > 0.875 then binomial(0.6, 15)
        else binomial(0.5, 15)
    else if num_recruiters <= 16 then
        if student_perfect_0 ==#2 1#2 then binomial(0.9, 16)
        else if student_gpa_0 > 0.875 then binomial(0.6, 16)
        else binomial(0.5, 16)
    else if num_recruiters <= 17 then
        if student_perfect_0 ==#2 1#2 then binomial(0.9, 17)
        else if student_gpa_0 > 0.875 then binomial(0.6, 17)
        else binomial(0.5, 17)
    else if num_recruiters <= 18 then
        if student_perfect_0 ==#2 1#2 then binomial(0.9, 18)
        else if student_gpa_0 > 0.875 then binomial(0.6, 18)
        else binomial(0.5, 18)
    else if num_recruiters <= 19 then
        if student_perfect_0 ==#2 1#2 then binomial(0.9, 19)
        else if student_gpa_0 > 0.875 then binomial(0.6, 19)
        else binomial(0.5, 19)
    else if num_recruiters <= 20 then
        if student_perfect_0 ==#2 1#2 then binomial(0.9, 20)
        else if student_gpa_0 > 0.875 then binomial(0.6, 20)
        else binomial(0.5, 20)
    else if num_recruiters <= 21 then
        if student_perfect_0 ==#2 1#2 then binomial(0.9, 21)
        else if student_gpa_0 > 0.875 then binomial(0.6, 21)
        else binomial(0.5, 21)
    else if num_recruiters <= 22 then
        if student_perfect_0 ==#2 1#2 then binomial(0.9, 22)
        else if student_gpa_0 > 0.875 then binomial(0.6, 22)
        else binomial(0.5, 22)
    else if num_recruiters <= 23 then
        if student_perfect_0 ==#2 1#2 then binomial(0.9, 23)
        else if student_gpa_0 > 0.875 then binomial(0.6, 23)
        else binomial(0.5, 23)
    else if num_recruiters <= 24 then
        if student_perfect_0 ==#2 1#2 then binomial(0.9, 24)
        else if student_gpa_0 > 0.875 then binomial(0.6, 24)
        else binomial(0.5, 24)
    else if num_recruiters <= 25 then
        if student_perfect_0 ==#2 1#2 then binomial(0.9, 25)
        else if student_gpa_0 > 0.875 then binomial(0.6, 25)
        else binomial(0.5, 25)
    else if num_recruiters <= 26 then
        if student_perfect_0 ==#2 1#2 then binomial(0.9, 26)
        else if student_gpa_0 > 0.875 then binomial(0.6, 26)
        else binomial(0.5, 26)
    else if num_recruiters <= 27 then
        if student_perfect_0 ==#2 1#2 then binomial(0.9, 27)
        else if student_gpa_0 > 0.875 then binomial(0.6, 27)
        else binomial(0.5, 27)
    else if num_recruiters <= 28 then
        if student_perfect_0 ==#2 1#2 then binomial(0.9, 28)
        else if student_gpa_0 > 0.875 then binomial(0.6, 28)
        else binomial(0.5, 28)
    else if num_recruiters <= 29 then
        if student_perfect_0 ==#2 1#2 then binomial(0.9, 29)
        else if student_gpa_0 > 0.875 then binomial(0.6, 29)
        else binomial(0.5, 29)
    else if num_recruiters <= 30 then
        if student_perfect_0 ==#2 1#2 then binomial(0.9, 30)
        else if student_gpa_0 > 0.875 then binomial(0.6, 30)
        else binomial(0.5, 30)
    else if num_recruiters <= 31 then
        if student_perfect_0 ==#2 1#2 then binomial(0.9, 31)
        else if student_gpa_0 > 0.875 then binomial(0.6, 31)
        else binomial(0.5, 31)
    else if num_recruiters <= 32 then
        if student_perfect_0 ==#2 1#2 then binomial(0.9, 32)
        else if student_gpa_0 > 0.875 then binomial(0.6, 32)
        else binomial(0.5, 32)
    else if num_recruiters <= 33 then
        if student_perfect_0 ==#2 1#2 then binomial(0.9, 33)
        else if student_gpa_0 > 0.875 then binomial(0.6, 33)
        else binomial(0.5, 33)
    else if num_recruiters <= 34 then
        if student_perfect_0 ==#2 1#2 then binomial(0.9, 34)
        else if student_gpa_0 > 0.875 then binomial(0.6, 34)
        else binomial(0.5, 34)
    else if num_recruiters <= 35 then
        if student_perfect_0 ==#2 1#2 then binomial(0.9, 35)
        else if student_gpa_0 > 0.875 then binomial(0.6, 35)
        else binomial(0.5, 35)
    else if num_recruiters <= 36 then
        if student_perfect_0 ==#2 1#2 then binomial(0.9, 36)
        else if student_gpa_0 > 0.875 then binomial(0.6, 36)
        else binomial(0.5, 36)
    else if num_recruiters <= 37 then
        if student_perfect_0 ==#2 1#2 then binomial(0.9, 37)
        else if student_gpa_0 > 0.875 then binomial(0.6, 37)
        else binomial(0.5, 37)
    else if num_recruiters <= 38 then
        if student_perfect_0 ==#2 1#2 then binomial(0.9, 38)
        else if student_gpa_0 > 0.875 then binomial(0.6, 38)
        else binomial(0.5, 38)
    else if num_recruiters <= 39 then
        if student_perfect_0 ==#2 1#2 then binomial(0.9, 39)
        else if student_gpa_0 > 0.875 then binomial(0.6, 39)
        else binomial(0.5, 39)
    else
        if student_perfect_0 ==#2 1#2 then binomial(0.9, 40)
        else if student_gpa_0 > 0.875 then binomial(0.6, 40)
        else binomial(0.5, 40)
in

let student_interviews_1 =
    if num_recruiters <= 10 then
        if student_perfect_1 ==#2 1#2 then binomial(0.9, 10)
        else if student_gpa_1 > 0.875 then binomial(0.6, 10)
        else binomial(0.5, 10)
    else if num_recruiters <= 11 then
        if student_perfect_1 ==#2 1#2 then binomial(0.9, 11)
        else if student_gpa_1 > 0.875 then binomial(0.6, 11)
        else binomial(0.5, 11)
    else if num_recruiters <= 12 then
        if student_perfect_1 ==#2 1#2 then binomial(0.9, 12)
        else if student_gpa_1 > 0.875 then binomial(0.6, 12)
        else binomial(0.5, 12)
    else if num_recruiters <= 13 then
        if student_perfect_1 ==#2 1#2 then binomial(0.9, 13)
        else if student_gpa_1 > 0.875 then binomial(0.6, 13)
        else binomial(0.5, 13)
    else if num_recruiters <= 14 then
        if student_perfect_1 ==#2 1#2 then binomial(0.9, 14)
        else if student_gpa_1 > 0.875 then binomial(0.6, 14)
        else binomial(0.5, 14)
    else if num_recruiters <= 15 then
        if student_perfect_1 ==#2 1#2 then binomial(0.9, 15)
        else if student_gpa_1 > 0.875 then binomial(0.6, 15)
        else binomial(0.5, 15)
    else if num_recruiters <= 16 then
        if student_perfect_1 ==#2 1#2 then binomial(0.9, 16)
        else if student_gpa_1 > 0.875 then binomial(0.6, 16)
        else binomial(0.5, 16)
    else if num_recruiters <= 17 then
        if student_perfect_1 ==#2 1#2 then binomial(0.9, 17)
        else if student_gpa_1 > 0.875 then binomial(0.6, 17)
        else binomial(0.5, 17)
    else if num_recruiters <= 18 then
        if student_perfect_1 ==#2 1#2 then binomial(0.9, 18)
        else if student_gpa_1 > 0.875 then binomial(0.6, 18)
        else binomial(0.5, 18)
    else if num_recruiters <= 19 then
        if student_perfect_1 ==#2 1#2 then binomial(0.9, 19)
        else if student_gpa_1 > 0.875 then binomial(0.6, 19)
        else binomial(0.5, 19)
    else if num_recruiters <= 20 then
        if student_perfect_1 ==#2 1#2 then binomial(0.9, 20)
        else if student_gpa_1 > 0.875 then binomial(0.6, 20)
        else binomial(0.5, 20)
    else if num_recruiters <= 21 then
        if student_perfect_1 ==#2 1#2 then binomial(0.9, 21)
        else if student_gpa_1 > 0.875 then binomial(0.6, 21)
        else binomial(0.5, 21)
    else if num_recruiters <= 22 then
        if student_perfect_1 ==#2 1#2 then binomial(0.9, 22)
        else if student_gpa_1 > 0.875 then binomial(0.6, 22)
        else binomial(0.5, 22)
    else if num_recruiters <= 23 then
        if student_perfect_1 ==#2 1#2 then binomial(0.9, 23)
        else if student_gpa_1 > 0.875 then binomial(0.6, 23)
        else binomial(0.5, 23)
    else if num_recruiters <= 24 then
        if student_perfect_1 ==#2 1#2 then binomial(0.9, 24)
        else if student_gpa_1 > 0.875 then binomial(0.6, 24)
        else binomial(0.5, 24)
    else if num_recruiters <= 25 then
        if student_perfect_1 ==#2 1#2 then binomial(0.9, 25)
        else if student_gpa_1 > 0.875 then binomial(0.6, 25)
        else binomial(0.5, 25)
    else if num_recruiters <= 26 then
        if student_perfect_1 ==#2 1#2 then binomial(0.9, 26)
        else if student_gpa_1 > 0.875 then binomial(0.6, 26)
        else binomial(0.5, 26)
    else if num_recruiters <= 27 then
        if student_perfect_1 ==#2 1#2 then binomial(0.9, 27)
        else if student_gpa_1 > 0.875 then binomial(0.6, 27)
        else binomial(0.5, 27)
    else if num_recruiters <= 28 then
        if student_perfect_1 ==#2 1#2 then binomial(0.9, 28)
        else if student_gpa_1 > 0.875 then binomial(0.6, 28)
        else binomial(0.5, 28)
    else if num_recruiters <= 29 then
        if student_perfect_1 ==#2 1#2 then binomial(0.9, 29)
        else if student_gpa_1 > 0.875 then binomial(0.6, 29)
        else binomial(0.5, 29)
    else if num_recruiters <= 30 then
        if student_perfect_1 ==#2 1#2 then binomial(0.9, 30)
        else if student_gpa_1 > 0.875 then binomial(0.6, 30)
        else binomial(0.5, 30)
    else if num_recruiters <= 31 then
        if student_perfect_1 ==#2 1#2 then binomial(0.9, 31)
        else if student_gpa_1 > 0.875 then binomial(0.6, 31)
        else binomial(0.5, 31)
    else if num_recruiters <= 32 then
        if student_perfect_1 ==#2 1#2 then binomial(0.9, 32)
        else if student_gpa_1 > 0.875 then binomial(0.6, 32)
        else binomial(0.5, 32)
    else if num_recruiters <= 33 then
        if student_perfect_1 ==#2 1#2 then binomial(0.9, 33)
        else if student_gpa_1 > 0.875 then binomial(0.6, 33)
        else binomial(0.5, 33)
    else if num_recruiters <= 34 then
        if student_perfect_1 ==#2 1#2 then binomial(0.9, 34)
        else if student_gpa_1 > 0.875 then binomial(0.6, 34)
        else binomial(0.5, 34)
    else if num_recruiters <= 35 then
        if student_perfect_1 ==#2 1#2 then binomial(0.9, 35)
        else if student_gpa_1 > 0.875 then binomial(0.6, 35)
        else binomial(0.5, 35)
    else if num_recruiters <= 36 then
        if student_perfect_1 ==#2 1#2 then binomial(0.9, 36)
        else if student_gpa_1 > 0.875 then binomial(0.6, 36)
        else binomial(0.5, 36)
    else if num_recruiters <= 37 then
        if student_perfect_1 ==#2 1#2 then binomial(0.9, 37)
        else if student_gpa_1 > 0.875 then binomial(0.6, 37)
        else binomial(0.5, 37)
    else if num_recruiters <= 38 then
        if student_perfect_1 ==#2 1#2 then binomial(0.9, 38)
        else if student_gpa_1 > 0.875 then binomial(0.6, 38)
        else binomial(0.5, 38)
    else if num_recruiters <= 39 then
        if student_perfect_1 ==#2 1#2 then binomial(0.9, 39)
        else if student_gpa_1 > 0.875 then binomial(0.6, 39)
        else binomial(0.5, 39)
    else
        if student_perfect_1 ==#2 1#2 then binomial(0.9, 40)
        else if student_gpa_1 > 0.875 then binomial(0.6, 40)
        else binomial(0.5, 40)
in

let student_offers_0 =
    if student_interviews_0 <= 0 then binomial(0.4, 0)
    else if student_interviews_0 <= 1 then binomial(0.4, 1)
    else if student_interviews_0 <= 2 then binomial(0.4, 2)
    else if student_interviews_0 <= 3 then binomial(0.4, 3)
    else if student_interviews_0 <= 4 then binomial(0.4, 4)
    else if student_interviews_0 <= 5 then binomial(0.4, 5)
    else if student_interviews_0 <= 6 then binomial(0.4, 6)
    else if student_interviews_0 <= 7 then binomial(0.4, 7)
    else if student_interviews_0 <= 8 then binomial(0.4, 8)
    else if student_interviews_0 <= 9 then binomial(0.4, 9)
    else if student_interviews_0 <= 10 then binomial(0.4, 10)
    else if student_interviews_0 <= 11 then binomial(0.4, 11)
    else if student_interviews_0 <= 12 then binomial(0.4, 12)
    else if student_interviews_0 <= 13 then binomial(0.4, 13)
    else if student_interviews_0 <= 14 then binomial(0.4, 14)
    else if student_interviews_0 <= 15 then binomial(0.4, 15)
    else if student_interviews_0 <= 16 then binomial(0.4, 16)
    else if student_interviews_0 <= 17 then binomial(0.4, 17)
    else if student_interviews_0 <= 18 then binomial(0.4, 18)
    else if student_interviews_0 <= 19 then binomial(0.4, 19)
    else if student_interviews_0 <= 20 then binomial(0.4, 20)
    else if student_interviews_0 <= 21 then binomial(0.4, 21)
    else if student_interviews_0 <= 22 then binomial(0.4, 22)
    else if student_interviews_0 <= 23 then binomial(0.4, 23)
    else if student_interviews_0 <= 24 then binomial(0.4, 24)
    else if student_interviews_0 <= 25 then binomial(0.4, 25)
    else if student_interviews_0 <= 26 then binomial(0.4, 26)
    else if student_interviews_0 <= 27 then binomial(0.4, 27)
    else if student_interviews_0 <= 28 then binomial(0.4, 28)
    else if student_interviews_0 <= 29 then binomial(0.4, 29)
    else if student_interviews_0 <= 30 then binomial(0.4, 30)
    else if student_interviews_0 <= 31 then binomial(0.4, 31)
    else if student_interviews_0 <= 32 then binomial(0.4, 32)
    else if student_interviews_0 <= 33 then binomial(0.4, 33)
    else if student_interviews_0 <= 34 then binomial(0.4, 34)
    else if student_interviews_0 <= 35 then binomial(0.4, 35)
    else if student_interviews_0 <= 36 then binomial(0.4, 36)
    else if student_interviews_0 <= 37 then binomial(0.4, 37)
    else if student_interviews_0 <= 38 then binomial(0.4, 38)
    else if student_interviews_0 <= 39 then binomial(0.4, 39)
    else binomial(0.4, 40)
in

let student_offers_1 =
    if student_interviews_1 <= 0 then binomial(0.4, 0)
    else if student_interviews_1 <= 1 then binomial(0.4, 1)
    else if student_interviews_1 <= 2 then binomial(0.4, 2)
    else if student_interviews_1 <= 3 then binomial(0.4, 3)
    else if student_interviews_1 <= 4 then binomial(0.4, 4)
    else if student_interviews_1 <= 5 then binomial(0.4, 5)
    else if student_interviews_1 <= 6 then binomial(0.4, 6)
    else if student_interviews_1 <= 7 then binomial(0.4, 7)
    else if student_interviews_1 <= 8 then binomial(0.4, 8)
    else if student_interviews_1 <= 9 then binomial(0.4, 9)
    else if student_interviews_1 <= 10 then binomial(0.4, 10)
    else if student_interviews_1 <= 11 then binomial(0.4, 11)
    else if student_interviews_1 <= 12 then binomial(0.4, 12)
    else if student_interviews_1 <= 13 then binomial(0.4, 13)
    else if student_interviews_1 <= 14 then binomial(0.4, 14)
    else if student_interviews_1 <= 15 then binomial(0.4, 15)
    else if student_interviews_1 <= 16 then binomial(0.4, 16)
    else if student_interviews_1 <= 17 then binomial(0.4, 17)
    else if student_interviews_1 <= 18 then binomial(0.4, 18)
    else if student_interviews_1 <= 19 then binomial(0.4, 19)
    else if student_interviews_1 <= 20 then binomial(0.4, 20)
    else if student_interviews_1 <= 21 then binomial(0.4, 21)
    else if student_interviews_1 <= 22 then binomial(0.4, 22)
    else if student_interviews_1 <= 23 then binomial(0.4, 23)
    else if student_interviews_1 <= 24 then binomial(0.4, 24)
    else if student_interviews_1 <= 25 then binomial(0.4, 25)
    else if student_interviews_1 <= 26 then binomial(0.4, 26)
    else if student_interviews_1 <= 27 then binomial(0.4, 27)
    else if student_interviews_1 <= 28 then binomial(0.4, 28)
    else if student_interviews_1 <= 29 then binomial(0.4, 29)
    else if student_interviews_1 <= 30 then binomial(0.4, 30)
    else if student_interviews_1 <= 31 then binomial(0.4, 31)
    else if student_interviews_1 <= 32 then binomial(0.4, 32)
    else if student_interviews_1 <= 33 then binomial(0.4, 33)
    else if student_interviews_1 <= 34 then binomial(0.4, 34)
    else if student_interviews_1 <= 35 then binomial(0.4, 35)
    else if student_interviews_1 <= 36 then binomial(0.4, 36)
    else if student_interviews_1 <= 37 then binomial(0.4, 37)
    else if student_interviews_1 <= 38 then binomial(0.4, 38)
    else if student_interviews_1 <= 39 then binomial(0.4, 39)
    else binomial(0.4, 40)
in

let _offer = observe(1 <= student_offers_0 && student_offers_0 <= 1) in
let _recruiters_positive = observe(num_recruiters > 10) in

student_gpa_0 < 0.125
\end{lstlisting}

\begin{lstlisting}[style=baselineprogram,title={SurveyUnbias}]
let bias_0 = beta(1, 1) in
let bias_1 = beta(1, 1) in

let votes_0 =
    if bias_0 <= 0.05 then binomial(0.05, 1000)
    else if bias_0 <= 0.1 then binomial(0.1, 1000)
    else if bias_0 <= 0.15 then binomial(0.15, 1000)
    else if bias_0 <= 0.2 then binomial(0.2, 1000)
    else if bias_0 <= 0.25 then binomial(0.25, 1000)
    else if bias_0 <= 0.3 then binomial(0.3, 1000)
    else if bias_0 <= 0.35 then binomial(0.35, 1000)
    else if bias_0 <= 0.4 then binomial(0.4, 1000)
    else if bias_0 <= 0.45 then binomial(0.45, 1000)
    else if bias_0 <= 0.5 then binomial(0.5, 1000)
    else if bias_0 <= 0.55 then binomial(0.55, 1000)
    else if bias_0 <= 0.6 then binomial(0.6, 1000)
    else if bias_0 <= 0.65 then binomial(0.65, 1000)
    else if bias_0 <= 0.7 then binomial(0.7, 1000)
    else if bias_0 <= 0.75 then binomial(0.75, 1000)
    else if bias_0 <= 0.8 then binomial(0.8, 1000)
    else if bias_0 <= 0.85 then binomial(0.85, 1000)
    else if bias_0 <= 0.9 then binomial(0.9, 1000)
    else if bias_0 <= 0.95 then binomial(0.95, 1000)
    else binomial(1.0, 1000)
in

let votes_1 =
    if bias_1 <= 0.05 then binomial(0.05, 2000)
    else if bias_1 <= 0.1 then binomial(0.1, 2000)
    else if bias_1 <= 0.15 then binomial(0.15, 2000)
    else if bias_1 <= 0.2 then binomial(0.2, 2000)
    else if bias_1 <= 0.25 then binomial(0.25, 2000)
    else if bias_1 <= 0.3 then binomial(0.3, 2000)
    else if bias_1 <= 0.35 then binomial(0.35, 2000)
    else if bias_1 <= 0.4 then binomial(0.4, 2000)
    else if bias_1 <= 0.45 then binomial(0.45, 2000)
    else if bias_1 <= 0.5 then binomial(0.5, 2000)
    else if bias_1 <= 0.55 then binomial(0.55, 2000)
    else if bias_1 <= 0.6 then binomial(0.6, 2000)
    else if bias_1 <= 0.65 then binomial(0.65, 2000)
    else if bias_1 <= 0.7 then binomial(0.7, 2000)
    else if bias_1 <= 0.75 then binomial(0.75, 2000)
    else if bias_1 <= 0.8 then binomial(0.8, 2000)
    else if bias_1 <= 0.85 then binomial(0.85, 2000)
    else if bias_1 <= 0.9 then binomial(0.9, 2000)
    else if bias_1 <= 0.95 then binomial(0.95, 2000)
    else binomial(1.0, 2000)
in

let ansBias_0 = bias_0 in
let ansBias_1 = bias_1 in
let ansBias_2 = bias_0 in
let ansBias_3 = bias_1 in
let ansBias_4 = bias_0 in

let answer_0 =
    if ansBias_0 <= 0.05 then discrete(0.95, 0.05)
    else if ansBias_0 <= 0.1 then discrete(0.9, 0.1)
    else if ansBias_0 <= 0.15 then discrete(0.85, 0.15)
    else if ansBias_0 <= 0.2 then discrete(0.8, 0.2)
    else if ansBias_0 <= 0.25 then discrete(0.75, 0.25)
    else if ansBias_0 <= 0.3 then discrete(0.7, 0.3)
    else if ansBias_0 <= 0.35 then discrete(0.65, 0.35)
    else if ansBias_0 <= 0.4 then discrete(0.6, 0.4)
    else if ansBias_0 <= 0.45 then discrete(0.55, 0.45)
    else if ansBias_0 <= 0.5 then discrete(0.5, 0.5)
    else if ansBias_0 <= 0.55 then discrete(0.45, 0.55)
    else if ansBias_0 <= 0.6 then discrete(0.4, 0.6)
    else if ansBias_0 <= 0.65 then discrete(0.35, 0.65)
    else if ansBias_0 <= 0.7 then discrete(0.3, 0.7)
    else if ansBias_0 <= 0.75 then discrete(0.25, 0.75)
    else if ansBias_0 <= 0.8 then discrete(0.2, 0.8)
    else if ansBias_0 <= 0.85 then discrete(0.15, 0.85)
    else if ansBias_0 <= 0.9 then discrete(0.1, 0.9)
    else if ansBias_0 <= 0.95 then discrete(0.05, 0.95)
    else discrete(0.0, 1.0)
in

let answer_1 =
    if ansBias_1 <= 0.05 then discrete(0.95, 0.05)
    else if ansBias_1 <= 0.1 then discrete(0.9, 0.1)
    else if ansBias_1 <= 0.15 then discrete(0.85, 0.15)
    else if ansBias_1 <= 0.2 then discrete(0.8, 0.2)
    else if ansBias_1 <= 0.25 then discrete(0.75, 0.25)
    else if ansBias_1 <= 0.3 then discrete(0.7, 0.3)
    else if ansBias_1 <= 0.35 then discrete(0.65, 0.35)
    else if ansBias_1 <= 0.4 then discrete(0.6, 0.4)
    else if ansBias_1 <= 0.45 then discrete(0.55, 0.45)
    else if ansBias_1 <= 0.5 then discrete(0.5, 0.5)
    else if ansBias_1 <= 0.55 then discrete(0.45, 0.55)
    else if ansBias_1 <= 0.6 then discrete(0.4, 0.6)
    else if ansBias_1 <= 0.65 then discrete(0.35, 0.65)
    else if ansBias_1 <= 0.7 then discrete(0.3, 0.7)
    else if ansBias_1 <= 0.75 then discrete(0.25, 0.75)
    else if ansBias_1 <= 0.8 then discrete(0.2, 0.8)
    else if ansBias_1 <= 0.85 then discrete(0.15, 0.85)
    else if ansBias_1 <= 0.9 then discrete(0.1, 0.9)
    else if ansBias_1 <= 0.95 then discrete(0.05, 0.95)
    else discrete(0.0, 1.0)
in

let answer_2 =
    if ansBias_2 <= 0.05 then discrete(0.95, 0.05)
    else if ansBias_2 <= 0.1 then discrete(0.9, 0.1)
    else if ansBias_2 <= 0.15 then discrete(0.85, 0.15)
    else if ansBias_2 <= 0.2 then discrete(0.8, 0.2)
    else if ansBias_2 <= 0.25 then discrete(0.75, 0.25)
    else if ansBias_2 <= 0.3 then discrete(0.7, 0.3)
    else if ansBias_2 <= 0.35 then discrete(0.65, 0.35)
    else if ansBias_2 <= 0.4 then discrete(0.6, 0.4)
    else if ansBias_2 <= 0.45 then discrete(0.55, 0.45)
    else if ansBias_2 <= 0.5 then discrete(0.5, 0.5)
    else if ansBias_2 <= 0.55 then discrete(0.45, 0.55)
    else if ansBias_2 <= 0.6 then discrete(0.4, 0.6)
    else if ansBias_2 <= 0.65 then discrete(0.35, 0.65)
    else if ansBias_2 <= 0.7 then discrete(0.3, 0.7)
    else if ansBias_2 <= 0.75 then discrete(0.25, 0.75)
    else if ansBias_2 <= 0.8 then discrete(0.2, 0.8)
    else if ansBias_2 <= 0.85 then discrete(0.15, 0.85)
    else if ansBias_2 <= 0.9 then discrete(0.1, 0.9)
    else if ansBias_2 <= 0.95 then discrete(0.05, 0.95)
    else discrete(0.0, 1.0)
in

let answer_3 =
    if ansBias_3 <= 0.05 then discrete(0.95, 0.05)
    else if ansBias_3 <= 0.1 then discrete(0.9, 0.1)
    else if ansBias_3 <= 0.15 then discrete(0.85, 0.15)
    else if ansBias_3 <= 0.2 then discrete(0.8, 0.2)
    else if ansBias_3 <= 0.25 then discrete(0.75, 0.25)
    else if ansBias_3 <= 0.3 then discrete(0.7, 0.3)
    else if ansBias_3 <= 0.35 then discrete(0.65, 0.35)
    else if ansBias_3 <= 0.4 then discrete(0.6, 0.4)
    else if ansBias_3 <= 0.45 then discrete(0.55, 0.45)
    else if ansBias_3 <= 0.5 then discrete(0.5, 0.5)
    else if ansBias_3 <= 0.55 then discrete(0.45, 0.55)
    else if ansBias_3 <= 0.6 then discrete(0.4, 0.6)
    else if ansBias_3 <= 0.65 then discrete(0.35, 0.65)
    else if ansBias_3 <= 0.7 then discrete(0.3, 0.7)
    else if ansBias_3 <= 0.75 then discrete(0.25, 0.75)
    else if ansBias_3 <= 0.8 then discrete(0.2, 0.8)
    else if ansBias_3 <= 0.85 then discrete(0.15, 0.85)
    else if ansBias_3 <= 0.9 then discrete(0.1, 0.9)
    else if ansBias_3 <= 0.95 then discrete(0.05, 0.95)
    else discrete(0.0, 1.0)
in

let answer_4 =
    if ansBias_4 <= 0.05 then discrete(0.95, 0.05)
    else if ansBias_4 <= 0.1 then discrete(0.9, 0.1)
    else if ansBias_4 <= 0.15 then discrete(0.85, 0.15)
    else if ansBias_4 <= 0.2 then discrete(0.8, 0.2)
    else if ansBias_4 <= 0.25 then discrete(0.75, 0.25)
    else if ansBias_4 <= 0.3 then discrete(0.7, 0.3)
    else if ansBias_4 <= 0.35 then discrete(0.65, 0.35)
    else if ansBias_4 <= 0.4 then discrete(0.6, 0.4)
    else if ansBias_4 <= 0.45 then discrete(0.55, 0.45)
    else if ansBias_4 <= 0.5 then discrete(0.5, 0.5)
    else if ansBias_4 <= 0.55 then discrete(0.45, 0.55)
    else if ansBias_4 <= 0.6 then discrete(0.4, 0.6)
    else if ansBias_4 <= 0.65 then discrete(0.35, 0.65)
    else if ansBias_4 <= 0.7 then discrete(0.3, 0.7)
    else if ansBias_4 <= 0.75 then discrete(0.25, 0.75)
    else if ansBias_4 <= 0.8 then discrete(0.2, 0.8)
    else if ansBias_4 <= 0.85 then discrete(0.15, 0.85)
    else if ansBias_4 <= 0.9 then discrete(0.1, 0.9)
    else if ansBias_4 <= 0.95 then discrete(0.05, 0.95)
    else discrete(0.0, 1.0)
in

let _0 = observe(answer_0 ==#2 1#2) in
let _1 = observe(answer_1 ==#2 0#2) in
let _2 = observe(answer_2 ==#2 1#2) in
let _3 = observe(answer_3 ==#2 1#2) in
let _4 = observe(answer_4 ==#2 0#2) in

bias_0 < 0.5
\end{lstlisting}

\lstinputlisting[style=baselineprogram,title={TrueSkill}]{baseline_programs/trueskill.slice}

\end{document}